\documentclass[article,english,11pt]{article}
\usepackage[utf8]{inputenc}
\usepackage[T1]{fontenc}
\usepackage{comment}
\usepackage[english]{babel}

\usepackage[top=2.5cm, bottom=2.5cm, left=2.0cm, right=2.0cm]{geometry}
\usepackage{graphicx, framed}

\usepackage{amscd,amsfonts,amsmath,amssymb,amsthm,bbm,bm,latexsym,mathrsfs}
\usepackage{algorithm}
  \makeatletter
  \renewcommand{\ALG@name}{Procedure}
  \makeatother
\allowdisplaybreaks
\usepackage{bookmark}
\usepackage{adjustbox}
\usepackage{gensymb}
\usepackage{natbib}
\usepackage{float,rotating}
\usepackage{hyperref}
\hypersetup{hidelinks}

\usepackage{enumerate}
\usepackage[shortlabels]{enumitem}

\usepackage[usenames]{color}
\usepackage{epsfig,epstopdf,graphics,graphicx,subfigure,longtable}
\graphicspath{{./figures/}}
\usepackage{tikz}
\usepackage{tkz-berge}
\newcommand\sbullet[1][.5]{\mathbin{\vcenter{\hbox{\scalebox{#1}{$\bullet$}}}}}

\usetikzlibrary{shapes,arrows,bayesnet}	

\usepackage[toc,page]{appendix}
\usepackage[normalem]{ulem}
\usepackage{cancel}
\usepackage{caption}

\makeatother

\usepackage[section]{placeins}
\usepackage{setspace}
\usepackage{amsmath, amssymb, algorithmicx, algpseudocode}

\newtheorem{proposition}{Proposition}[]
\newtheorem{corollary}{Corollary}[]
\newtheorem{definition}{Definition}[]
\newtheorem{assumption}{Assumption}[]

\theoremstyle{remark}
\newtheorem{rem}{Remark}

\theoremstyle{remark}

\theoremstyle{plain}

\usepackage{nomencl}

\makenomenclature

\makeatletter
\def\@biblabel#1{\hspace*{-\labelsep}}
\makeatother
\makeatletter
\renewcommand*{\@fnsymbol}[1]{\number#1}
\makeatother

\begin{document}

\title{Bayesian Outlier Detection for Matrix--variate Models}

\author{
    Monica Billio\footnote{Department of Economics, Ca' Foscari University of Venice} 
    \hspace{10pt} Roberto Casarin$^{1,}$\footnote{Corresponding author address: r.casarin@unive.it. Other contacts: Monica Billio (billio@unive.it), Fausto Corradin (fausto.corradin@unive.it), Antonio Peruzzi (antonio.peruzzi@unive.it)}
    \hspace{10pt} Fausto Corradin$^{1}$ 
    \hspace{10pt} Antonio Peruzzi$^{1}$
}

\vspace{-5pt}
\date{\today}
\maketitle

\vspace{-10pt}

\begin{abstract}
Anomalies in economic and financial data -- often linked to rare yet impactful events -- are of theoretical interest, but can also severely distort inference. Although outlier-robust methodologies can be used, many researchers prefer pre-processing strategies that remove outliers. In this work, an efficient sequential Bayesian framework is proposed for outlier detection based on the predictive Bayes Factor (BF). The proposed method is specifically designed for large, multidimensional datasets and extends univariate Bayesian model outlier detection procedures to the matrix-variate setting. Leveraging power-discounted priors, tractable predictive BF are obtained, thereby avoiding computationally intensive techniques. The BF finite sample distribution, the test critical region, and robust extensions of the test are introduced by exploiting the sampling variability. The framework supports online detection with analytical tractability, ensuring both accuracy and scalability. Its effectiveness is demonstrated through simulations, and three applications to reference datasets in macroeconomics and finance are provided.
\end{abstract}
 

\newpage

\section{Introduction}

Robust methods are widely used to detect and address outliers -- observations that depart significantly from typical patterns due to measurement error, structural change, or data irregularities. In economics and finance, where such anomalies often signal impactful events like crises or policy shifts, proper identification is essential (\citealp{atkinson1997detecting}, \citealp{GIORDANI2007112}). Many researchers prefer to identify and remove outliers or windsorize the data (e.g., see \citealp{malikov2020estimation}, \citealp{liao2022extrapolative}, \citealp{agarwal2024unobserved}) rather than rely on robust methodologies as the latter are often more complex to implement and may not perform equally well across all scenarios (e.g., see \citealp{zeng2021bayesian} \citealp{chang2024discussion}). Instead, an extensive literature investigates the timing, magnitude, and probability of outliers following various robust model strategies such as change points (\citealp{chopin2004bayesian}, \citealp{koop2007estimation}, \citealp{casini2024change}), Markov switching (\citealp{kole2023moments}, \citealp{casarin2024bayesian}) and Bayesian nonparametrics (\citealp{BASSETTI201449}, \citealp{BILLIO201997}). In this paper, we propose a theoretically founded testing procedure for identifying frequency and periods of outlying observations. It is specifically designed for high-dimensional large datasets. Our test can be useful to the researcher for many purposes. When applied to the raw data, it can be used either to correctly choose a model class in the model specification stage or as a comparison with a model-based structural break identification. Additionally, it can be applied to the residuals of a given model to help identify sources of model mispecification.

Outlier detection is a statistical problem that has received considerable attention from the frequentist and Bayesian perspectives \citep{Bayarri2004}. Despite their popularity, classical outlier detection procedures, such as Grubbs's test (\citealp{Grubb50}), are not designed for high-dimensional or structured data settings, such as multiple index panels or network-valued data. In such environments, observations often exhibit dependence, and univariate tests may yield misleading results. 

Figure~\ref{fig:intro} reports the outliers detected by Grubbs’ test across the three benchmark high-dimensional datasets we will consider in this work, i.e. an Inflation and Unemployment dataset (see \citealp{Can09}), an International Trade dataset (see \citealp{Rose2004}), and a  Volatility Network dataset (see \citealp{billio2018bayesianDT}). The top panels display counts of individual outlying entries over time at different significance thresholds, together with the corresponding sample range (shaded areas). The bottom panels summarize the number of rows and columns containing at least one outlier. As illustrated in Figure~\ref{fig:intro}, applying standard outlier detection methods to the three high-dimensional datasets retrieves many outliers across different significance levels (top panels). This is to be expected, as several testing procedures that rely on independence assumptions tend to overlook interdependencies and therefore overflag outliers. When Bonferroni correction is applied, the number of detected outliers is substantially reduced due to the large number of series involved, highlighting the conservative nature of such inequality-based adjustment. Further results from Grubbs' test reveal structural dependencies and asymmetries between row- and column-wise detections (bottom panels), underscoring the need for joint testing. In addition, outlier detection for high-dimensional datasets calls for efficient and interpretable tools, which in turn motivate the development of tractable, iterative procedures for model fitting and prediction.

This paper addresses these issues and contributes to the Bayesian literature by extending sequential outlier detection procedures for univariate models to matrix-variate models. We follow a Bayesian approach as it naturally allows for sequential updating of the estimators involved in the outlier detection, and analytical results are derived to avoid computationally intensive approximations. The sequential nature of our test is well-suited for large and high-dimensional datasets as it allows reducing the computational cost, while also accommodating time variation in both the conditional mean and variance of the predictive distribution.

\begin{figure}[htb!]
    \centering
    \setlength{\tabcolsep}{3pt}
    \renewcommand{\arraystretch}{1.2}
    \begin{tabular}{c}
   (a) Outliers Count \vspace{5pt}\\
\includegraphics[width = 0.95\textwidth]{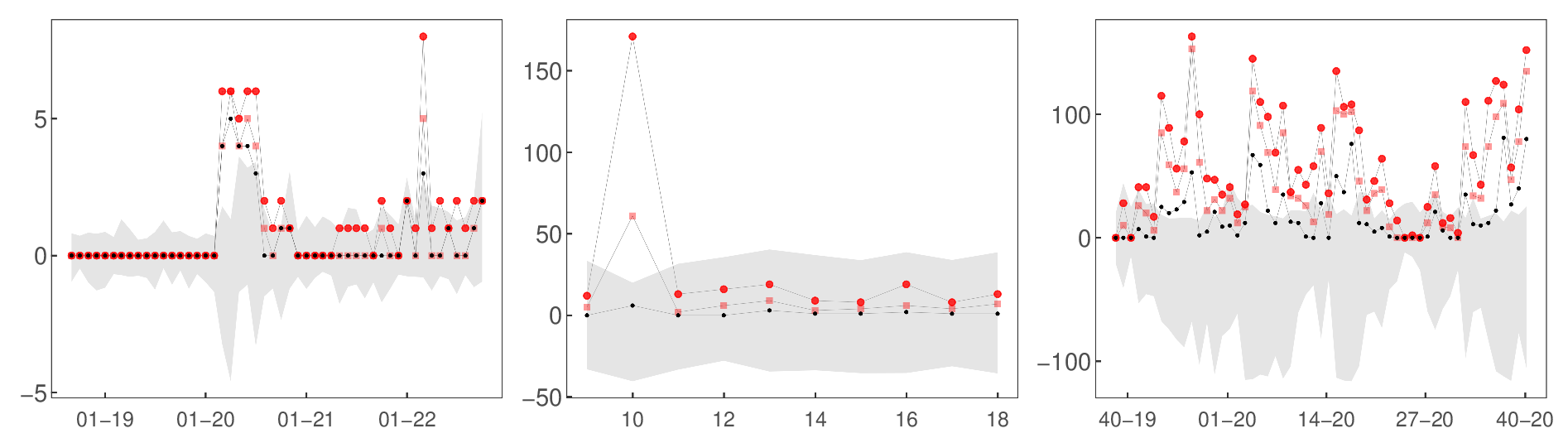}\\
   (b) Row and Column Outliers count \vspace{5pt}\\
\includegraphics[width = 0.95\textwidth]{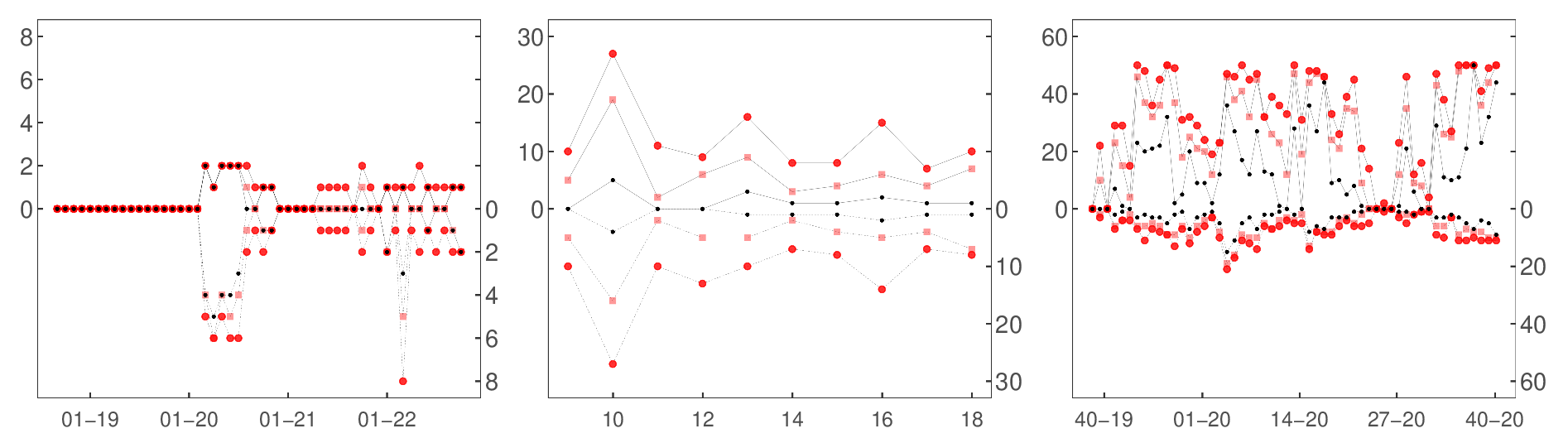}
    \end{tabular}
        \caption{
Grubbs' test for outliers across datasets. Top:  number of outlying observations over time across three benchmark datasets: (i) Inflation and Unemployment (left), (ii) International Trade (centre), and (iii) Volatility Network (right). Counts of outlying matrix entries at different significance levels ({\scriptsize\textcolor{red}{$\blacksquare$}} 1\%, {\footnotesize\textcolor{red}{$\bullet$}} 5\% and  {\footnotesize$\bullet$} 5\% with Bonferroni's correction). The sample min-max gap is denoted by the shaded areas. Bottom: number of rows (solid line, left scale) and columns (dashed line, right scale) with at least one outlier at different significance levels.}
\label{fig:intro}
\end{figure}

Within a Bayesian framework, a common approach is to assume that potential outliers arise from contaminating models distinct from the one generating the bulk of the data. Previous Bayesian works along these lines can be found in \cite{Box1968-gm,Guttman1973-bk, Abraham1979-ba, Guttman1978-be, pettit1985outliers,Pettit1992-bc,Pettit1990-zb,Verdinelli1991-cu,bayarri1994robust,Hoeting1996-ym}. Nonetheless, these approaches are not well-suited for high-dimensional and large datasets since they are designed for a univariate setup.

We follow a Bayes Factor (BF) approach, which is also widely used in Bayesian analysis \citep{Schrider2016-nn,Ly2016-aw,Chen2018-ht,Li2021-pk, Stefan2019-nf}, including selective inference \citep{Yekutieli2012-bd} and outliers detection \citep{Bayarri2003-tb}. The BF is also connected to Schwartz's criterion, also known as the Bayesian Information Criterion, which is another approach commonly used in the outlier detection literature. Due to the extreme sensitivity of the BF to the choice of the alternative distribution, robust methodologies have been developed \citep{Li2021-pk, Schad2022-oj}. In this paper, we follow the power discounting approach of \cite{west1986bayesian} and assume the alternative distribution is proportional to the principal distribution raised to a power. 

The advantage of the perturbation approach based on power discounting is twofold. First, it requires estimating the model only under the null hypothesis, thus reducing computational cost. Second, the specification of the model under the alternative is not required, thus preserving the tractability of the BF testing procedure. We exploit these features to derive the finite-sample distribution of the predictive BF and to provide some frequentist validation of the testing procedure, without relying on asymptotic approximations. A critical region for the test is derived as an alternative to the traditional Jeffreys' scale of evidence, which has been criticized. See \cite{xian209} and referenced therein. We consider the Gaussian family since it is a standard assumption in many fields \cite[e.g., see][]{billio2018bayesianDT,guhaniyogi2017bayesian} and extend the existing procedure for univariate Gaussian models \citep{west1986bayesian} to the matrix--variate case, thus providing an original contribution to the expanding literature in this area \citep[e.g., see][]{landim2000dynamic,triantafyllopoulos2008missing,wang2009bayesian,CarvWest07DynMatNormGraph,Viroli11MatNorm,thompson2020classification,billio2018bayesianDT,tomarchio2022mixtures}. The Gaussian assumption is not restrictive since we show that within a Bayesian framework the scale parameter can be integrated out of the likelihood, thus returning a Student-t predictive distribution, which can account for heavy tails.

This paper presents not only new results for detecting outliers in the matrix case but also uncovers novel results for the univariate Gaussian model. Furthermore, some pre-existing results for the univariate model are recovered as special cases. We also prove analytically that the outcome of the BF procedure for detecting outliers heavily depends on the choice of the discounting factor. Thus, we propose two robust BFs and a BF calibration procedure to alleviate the problem while maintaining a certain degree of tractability. Through some simulation experiments, we investigate the properties of a testing procedure under various outlier generation settings.


The impact of outliers on empirical economic analysis has gained importance in the aftermath of major global disruptions such as the 2008–2009 financial crisis and the COVID-19 pandemic. These episodes encouraged both researchers and official agencies to develop guidelines for outlier detection and adjusting for outliers in macroeconomic and financial data. In response to this growing need, we illustrate the effectiveness of our proposed sequential outlier detection procedure across the three aforementioned representative economic datasets: (i) a panel of inflation and unemployment indicators for European countries, (ii) a dynamic network of international trade flows, and (iii) a dynamic network of financial market volatilities.


The paper is organized as follows. Section \ref{sec:bayesian} introduces the outlier detection procedure based on BF and presents the main results. Section \ref{sec:illus} provides some analytical results on the properties of the test and a simulation study of the procedures. Section \ref{sec:empirical} presents the three real-data illustrations. Section \ref{sec:concl} concludes.
\section{Bayesian Outlier Detection}\label{sec:bayesian}
\subsection{Bayes Factor}
Consider a sequence of observations $\boldsymbol{Y}_t$, $t=1,2,\ldots$ with $\boldsymbol{Y}_{t}\in\mathcal{Y}$ where $\mathcal{Y}$ is a possibly multidimensional sample space and $t$ is a time index. In the following, boldfaced symbols represent vectors or matrices.  We assume the information available at time $t$ is given by the collection of past observations $\mathcal{D}_t=\{\boldsymbol{Y}_1,\ldots,\boldsymbol{Y}_t\}$. Given the parameter and past observations, we assume the conditional sampling distribution does not depend on $\mathcal{D}_{t-1}$ and belongs to a parametric family with unknown parameter $\boldsymbol{\theta}\in \Theta$. The parameter space $\Theta$ is endowed with a prior distribution $p(\boldsymbol{\theta})$.

Our outlier detection procedure is applied sequentially over time to reduce computational cost in large datasets and capture time variations in the moments. We assume at time $t-1$ a posterior distribution for $\boldsymbol{\theta}$ is formed given $\mathcal{D}_{t-1}$ with density given by $p(\boldsymbol{\theta}|\mathcal{D}_{t-1})\propto p(\boldsymbol{\theta})p(\boldsymbol{Y}_1|\boldsymbol{\theta})\cdot\ldots\cdot p(\boldsymbol{Y}_{t-1}|\boldsymbol{\theta})$. We denote with $\boldsymbol{\theta}_t$ the random variable $\boldsymbol{\theta}|\mathcal{D}_{t-1}$ and assume its prior distribution is $p(\boldsymbol{\theta}_t |\mathcal{D}_{t-1})$. Assuming a model $g(\boldsymbol{Y}_t|\boldsymbol{\theta}_t)$ for $\boldsymbol{Y}_t$, the marginal predictive distribution for $\boldsymbol{Y}_t$ has density: $p(\boldsymbol{Y}_t |\mathcal{D}_{t-1})=\int_{\Theta}g(\boldsymbol{Y}_t|\boldsymbol{\theta}_t )p(\boldsymbol{\theta}_t |\mathcal{D}_{t-1})  d\boldsymbol{\theta}_t,\quad \boldsymbol{Y}_t\in\mathcal{Y}$.
The distribution $p(\boldsymbol{Y}_t |\mathcal{D}_{t-1})$ naturally provides a measure of the predictive ability of the model $g(\boldsymbol{Y}_t|\boldsymbol{\theta}_t)$. In the following, we assume the predictive distribution belongs to the same distribution family as the sampling distribution, that is, $g(\boldsymbol{Y}_t|\boldsymbol{\theta}_t)=p(\boldsymbol{Y}_t|\boldsymbol{\theta}_t)$. In hypothesis testing, a model for the alternative hypothesis $\mathcal{H}_{1}$ is assumed, $p_A(\boldsymbol{Y}_t |\mathcal{D}_{t-1})$ and the prior predictive distribution $p_A(\boldsymbol{Y}_t |\mathcal{D}_{t-1})=\int_{\Theta}p(\boldsymbol{Y}_t |\boldsymbol{\theta}_t )p_A(\boldsymbol{\theta}_t |\mathcal{D}_{t-1})  d\boldsymbol{\theta}_t,\quad \boldsymbol{Y}_t\in\mathcal{Y}$ is obtained as the marginal distribution with respect to the prior distribution at time $t$ under the alternative hypothesis. The Bayes Factor is defined as the ratio of the two marginal distributions:
\begin{equation}\label{eq:bf_def}
H_t=\frac{p(\boldsymbol{Y}_t|\mathcal{D}_{t-1})}{p_A(\boldsymbol{Y}_t|\mathcal{D}_{t-1})}=\frac{\int_{\Theta}p(\boldsymbol{Y}_t|\boldsymbol{\theta}_t )p(\boldsymbol{\theta}_t |\mathcal{D}_{t-1})d\boldsymbol{\theta}_t}{\int_{\Theta}p(\boldsymbol{Y}_t|\boldsymbol{\theta}_t )p_A(\boldsymbol{\theta}_t |\mathcal{D}_{t-1})d\boldsymbol{\theta}_t}
\end{equation}
and yields an optimal decision in an inference problem under a 0-1 loss function. If $H_t>1$, the null hypothesis $\mathcal{H}_{0}$ is not rejected, then we conclude there is evidence from the observation $\boldsymbol{Y}_t$ in favor of the null hypothesis. 
This paper considers the BF as a testing procedure for model failure, which includes sensitivity to structural changes and outliers. To define the distribution under the alternative, we follow an approach based on a parametric perturbation $p_{A}(\boldsymbol{\theta}_t|\mathcal{D}_{t-1})= G(p(\boldsymbol{\theta}_t\vert \mathcal{D}_{t-1}),\alpha_t)$ of the posterior predictive distribution, where $\alpha_t$ is a possibly time-varying perturbation parameter and $G(\cdot,\cdot)$ is a perturbation function from $[0,1]\times[0,1]$ onto $[0,1]$. The main advantages of the perturbation approach are two. First, only the estimation of the model under the null is required, which substantially reduces the computational cost in large datasets and high-dimensional models. Second, the specification of the model under the alternative is not required, which usually demands rather involved models and a high computational cost. On the contrary, the perturbation approach preserves the tractability of the model under the null and provides an effective strategy for capturing deviations from the null hypothesis.

Examples of perturbation functions can be derived from the approaches based on mixtures of distributions, $p_{A}(\boldsymbol{\theta}_t\vert \mathcal{D}_{t-1})=\alpha_t p(\boldsymbol{\theta}_t)+(1-\alpha_t)p(\boldsymbol{\theta}_t\vert \mathcal{D}_{t-1})$ \citep{robert202250}, and on distortion of probability measures, $p_{A}(\boldsymbol{\theta}_t\vert \mathcal{D}_{t-1})= g(\int_{-\infty}^{\boldsymbol{\theta}_t}p(u\vert \mathcal{D}_{t-1})du,\alpha_t)p(\boldsymbol{\theta}_t\vert \mathcal{D}_{t-1})$ where $g$ is the partial derivative of $G$ with respect to its first argument.

In this paper, we assume the following perturbation function $G(p,\alpha_t)=p^{\alpha_t} C(\alpha_t)$ where $C(\alpha_t)=\left(\int_{\Theta}p(\boldsymbol{\theta}_t |\mathcal{D}_{t-1})^{\alpha_t}d\boldsymbol{\theta}_t\right)^{-1}$ is the inverse normalizing constant of the density under the alternative hypothesis. The constant satisfies $C(\alpha_t)\rightarrow 1$ as $\alpha_t\rightarrow 1^{-}$ and $C(\alpha_t)\rightarrow\lambda(\Theta)^{-1}$ as $\alpha_t\rightarrow 0^{+}$, where $\lambda(\Theta)<\infty$ denotes the Lebesgue measure of $\Theta$. If $\lambda(\Theta)$ is unbounded, then to prevent $C(\alpha_t)\rightarrow 0$ as $\alpha_t\rightarrow 0^{+}$, some restrictions can be introduced such as $\alpha_t\in(\underline{\alpha},\bar{\alpha})\subset(0,1)$, $\underline{\alpha}>0$. These limits are relevant in the analysis of the BF, denoted as $H_{t}(\alpha_t)$ in what follows. This choice of the calibration function returns the power discounting approach proposed in the seminal paper \cite{west1986bayesian} for eliciting the prior under the alternative hypothesis
\begin{equation}
p_A (\boldsymbol{\theta}_t |\mathcal{D}_{t-1})\propto(p(\boldsymbol{\theta}_t |\mathcal{D}_{t-1}))^{\alpha_t},\label{altdistr}
\end{equation}
where the discounting parameter $\alpha_t$ takes values in the interval $(0,1)$. An advantage of the power discounting approach is that it returns the popular variance contamination model for outlier detection \citep[e.g., see][]{Pettit1992-bc,weiss1997bayesian,Bayarri2003-tb,page2011bayesian,tomarchio2022mixtures} and does not require the estimation of the contamination parameter, which would be difficult, since very little information is available about the outlier generating process. See also \cite{van2021bayesian} for a comparison between Bayesian contamination procedures and alternative approaches.  

We assume a Gaussian likelihood and a conjugate prior distribution to preserve some analytical tractability and extend the univariate discounting approach to a matrix--variate setting.  We will prove that the BF is generally well-defined and bounded by a function of $\alpha_t$. However, the results of the hypothesis testing can be significantly influenced by the value of $\alpha_t$, and the BF can become unbounded in the limit for $\alpha_t\rightarrow 0^{+}$. Consequently, the discounting parameter should be chosen carefully. The new results for the matrix setting readily apply to the previous univariate and multivariate approaches.

\subsection{Robust Bayes Factors}
As one can expect, the outcome of a decision based on the BF may depend on the value of $\alpha_t$. See also the discussion in \cite{west1986bayesian}. We prove that for a given level $H^{\sbullet}$ of the threshold (which is assumed equal to one in the following), for some values of $\alpha_t$, the evidence is against the null hypothesis, and for some others, it is against the alternative. In the outlier detection setting, this means that a new observation $Y_t$ may be considered an outlier or not based on the value of $\alpha_t$. The following result illustrates the indeterminacy of the outcome of the BF procedure.
\begin{proposition}\label{th1}
Assume the likelihood satisfies
\begin{equation}
\int_{\Theta}p(\boldsymbol{Y}_t|\boldsymbol{\theta}_t)(\lambda(\Theta)p(\boldsymbol{\theta}_t|\mathcal{D}_{t-1})-1)\lambda(d\boldsymbol{\theta}_t)>0\label{condLambda1}
\end{equation}
with respect to the reference measure $\lambda$, and there exists a unique $\alpha_{t}^{\sbullet}$ such that $q_1(\alpha_{t})-q_2(\alpha_{t})<0$ for $\alpha_t<\alpha_{t}^{\sbullet}$ and $q_1(\alpha_{t})-q_2(\alpha_{t})>0$ for $\alpha_t>\alpha_{t}^{\sbullet}$ and $q_1(\alpha_{t}^{\sbullet})-q_2(\alpha_{t}^{\sbullet})=0$, where
\begin{eqnarray}
&&q_1(\alpha_t)=\int_{\Theta}p(\boldsymbol{Y}_t|\boldsymbol{\theta}_t)p(\boldsymbol{\theta}_t|\mathcal{D}_{t-1})^{\alpha_t}\lambda(d\boldsymbol{\theta}_t) \int_{\Theta}p(\boldsymbol{\theta}_t|\mathcal{D}_{t-1})^{\alpha_t}\log p(\boldsymbol{\theta}_t|\mathcal{D}_{t-1})\lambda(d\boldsymbol{\theta}_t)\\
&&q_2(\alpha_t)=\int_{\Theta}p(\boldsymbol{\theta}_t|\mathcal{D}_{t-1})^{\alpha_t}\lambda(d\boldsymbol{\theta}_t)\int_{\Theta}p(\boldsymbol{Y}_t|\boldsymbol{\theta}_t)p(\boldsymbol{\theta}_t|\mathcal{D}_{t-1})^{\alpha_t}\log p(\boldsymbol{\theta}_t|\mathcal{D}_{t-1})\lambda(d\boldsymbol{\theta}_t).
\end{eqnarray} Then  there exists a partition $(U_1,U_2)$ of the unit interval such that $H_t(\alpha_t)>1$ for $\alpha_{t}\in U_1$ and $H_t(\alpha_t)<1$ for $\alpha_{t}\in U_2$.
\end{proposition}
Proposition \ref{th1} establishes the conditions such that different values of \( \alpha_t \) may cause the BF \( H_t(\alpha_t) \) to uniquely cross the threshold value \( H_t(\alpha^{\sbullet}_{t}) = 1 \) at \( \alpha_{t}^{\sbullet} \). It can be seen from the proof that for $\lambda(\Theta)$ unbounded, the evidence against the presence of an outlier is negative for $\alpha_t\rightarrow0^{+}$, consistently with the Jeffreys--Lindley Paradox \citep[e.g., see][ ch. 6]{Robert_2014, bernardo2009bayesian}. See also \cite{robert2009harold} and \cite{wag2023} for a review with a historical perspective.  In addition, we shall notice that, in the Bayesian current practice, the BF is usually interpreted following Jeffrey's scale of evidence against the null hypothesis \citep[e.g., see][]{kass1995bayes}, which suggests the evidence is negative for $1/H_t<1$, not worth more than a bare mention for $1<1/H_t<10^{1/2}$, substantial for $10^{1/2}<1/H_t<10$, strong for $10<1/H_t<10^{3/2}$, very strong for $10^{3/2}<1/H_t<10^2$ and decisive for $1/H_t>10^2$. The following example illustrates the result given in the previous proposition and shows that an observation can be regarded as an outlier depending on the value of $\alpha_t$. The illustration is general since a stepwise uniform prior is assumed, and the differentiability of the prior is not required to find the threshold value of $\alpha_t$. 

\begin{rem}[Univariate Gaussian model]\label{rem:exstepwise1}
Assume a Gaussian likelihood with location $\theta_t$ and variance $\sigma^2$, that is $Y_t  \sim\mathcal{N}(\theta_t,\sigma^2)$ independent for $t=1,\ldots,T$. Let us assume $\sigma^2$ is known and $\theta_t$ follows a prior distribution $p(\theta_t\vert \mathcal{D}_{t-1})$ which is absolutely continuous with respect to the Lebesgue measure $\lambda$, and has density function
\begin{equation}
p(\theta_t\vert \mathcal{D}_{t-1})=\sum_{j=1}^{n}g_{t,j}\mathbb{I}(\theta_t\in \Theta_j), \label{ex1:prior}
\end{equation}
where $\{\Theta_j,j=1,\ldots,n\}$ and $\Theta=\Theta_1\cup\ldots\cup\Theta_n$ with $n>1$ and such that $g_{t,1}\lambda(\Theta_1)+\ldots+g_{t,n}\lambda(\Theta_n)=1$. 
In Table \ref{tab:ex2}, we report three cases in which an observation $Y_t$ can be regarded as an outlier depending on the value of $\alpha_t$. The table reports together with  $Y_t$, the value $\alpha^{\sbullet}_t$ such that $H_t(\alpha^{\sbullet}_t) = 1$, the step-function prior specification with possibly disconnected support given by the union of $\Theta_j=(\underline{\theta}_{t,j},\overline{\theta}_{t,j})$, and density $g_{t,j}$ for $j \in 1, \ldots, 3$. The table reports also the likelihood values $p_{t,j}=p(Y_t|\theta_{t,j})$ with $\theta_{t,j} \in \Theta_j$.
\begin{table}[h!]
\renewcommand{\arraystretch}{1.2}
    \centering
    \small
    \setlength{\tabcolsep}{2pt}
    \begin{tabular}{ccc|cccccc|ccc|ccc}
    \hline
    $\alpha_{t}^{\sbullet}$& $Y_t$ &$\sigma$&$\underline{\theta}_{t,1}$ & $\overline{\theta}_{t,1}$ & $\underline{\theta}_{t,2}$ & $\overline{\theta}_{t,2}$ &$\underline{\theta}_{t,3}$ &$\overline{\theta}_{t,3}$&$g_{t,1}$&$g_{t,2}$&$g_{t,3}$&$p_{t,1}$&$p_{t,2}$&$p_{t,3}$\\
    \hline
0.054&-4.422&1.120&0.471&0.565&1.487&1.507&2.863&5.863&0.440&50.000&0.001&0.180 & 0.001&0.001\\
0.076&-5.690&1.054&-1.112&1.689&2.422&2.588&4.780&4.800&0.002&0.492&45.645&0.002 & 0.153&0.002\\
0.412&-8.723&1.939&0.996&1.409&2.628&4.209&5.498&5.549&0.934&0.023&11.510&0.146 & 0.019&0.041\\
\hline
    \end{tabular}
    \caption{Prior hyper-parameters $\underline{\theta}_{t,j}$, $\overline{\theta}_{t,j}$  and $g_{t,j}$ and the likelihood values $p_{t,j}=p(Y_t|\theta_{t,j})$, with $j=1,\ldots,3$.}
    \label{tab:ex2}
\end{table}
The marginal likelihood under the alternative hypothesis is derived as follows. By the mean-value theorem for integrals on a bounded domain, there exists $\theta_{t,j}\in\Theta_j$ such that
\begin{equation}  
\int_{\Theta}p(Y_t|\theta_t)p_A(\theta_t|\mathcal{D}_{t-1})d\theta_t=C(\alpha_t)\sum_{j=1}^{n}\int_{\Theta_j} p(Y_t\vert \theta_t)g_{t,j}^{\alpha_t}d\theta_t =C(\alpha_t)\sum_{j=1}^{n}g_{t,j}^{\alpha_t}p_{t,j}\lambda(\Theta_j),\label{eqex1}
\end{equation}
where $p_{t,j}=p(Y_t|\theta_{t,j})$ and $C(\alpha_t)=g_{t,1}^{\alpha_t}+g_{t,2}^{\alpha_t}+g_{t,3}^{\alpha_t}$ is the normalizing constant.  This expression can is used to show that the conditions in Prop. \ref{th1} are satisfied for the cases provided in Table  \ref{tab:ex2}. The values of $p_{t,j}$ show that the prior predictive assigns a non--negligible probability to the observations considered in the example. Figure \ref{ex2:fig} shows that the observation $Y_t$ is an outlier for values of $\alpha_t < \alpha_t^{\sbullet}$  ($H_t(\alpha_t)>1$), while $Y_t$ is considered an outlier for $\alpha_t>\alpha_t^{\sbullet}$ ($H_t(\alpha_t)<1$). Also, for some observations and settings, comparing the BF with the first threshold of Jeffrey's scale, $10^{1/2}$, returns substantial evidence against the null (first and second panel) depending on $\alpha_t$. The third example shows that the farther the observed value from $\Theta$, e.g. $Y_t=-8.723$, the larger  $\alpha_t$ needs to be in order to make the presence of an outlier a convincing hypothesis, i.e. $\alpha_{t}^{\sbullet}=0.412$.
\end{rem}
\begin{figure}[t]
    \centering
    \setlength{\tabcolsep}{5pt}
    \begin{tabular}{ccc}
    \includegraphics[scale=0.3]{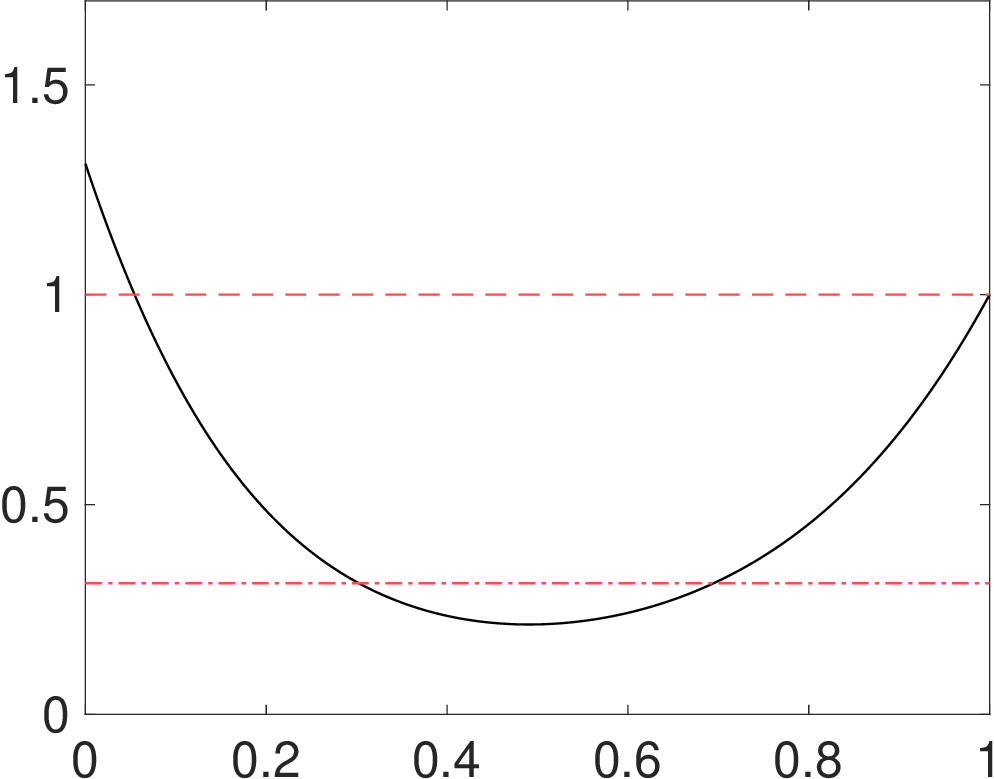}&
        \includegraphics[scale=0.3]{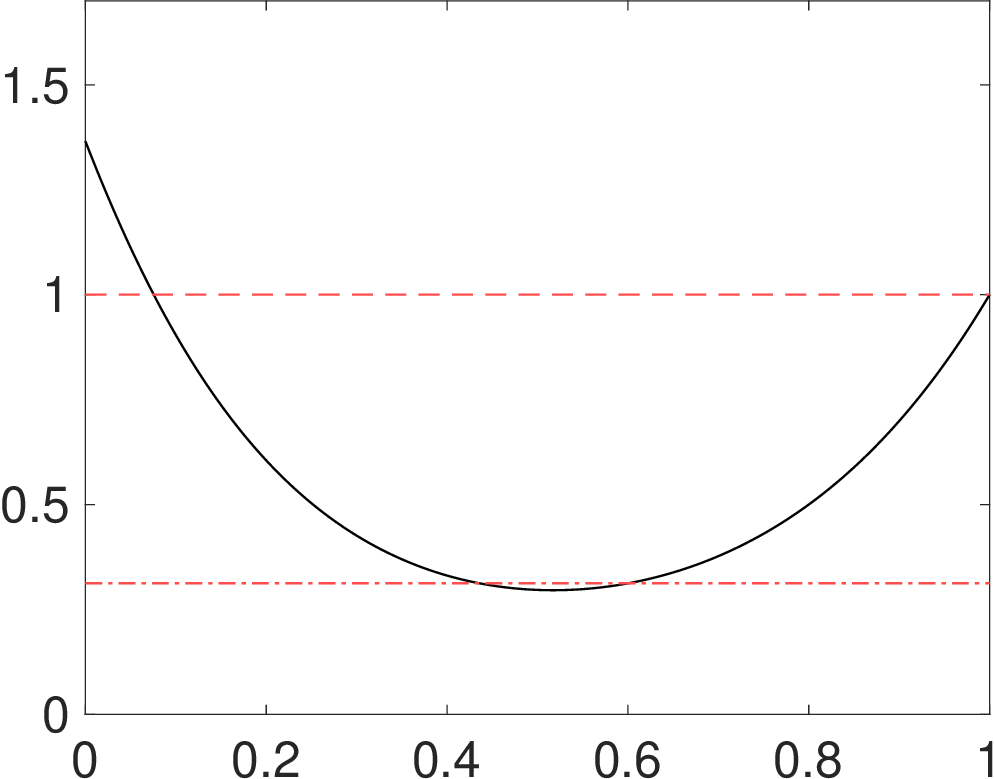}&
    \includegraphics[scale=0.3]{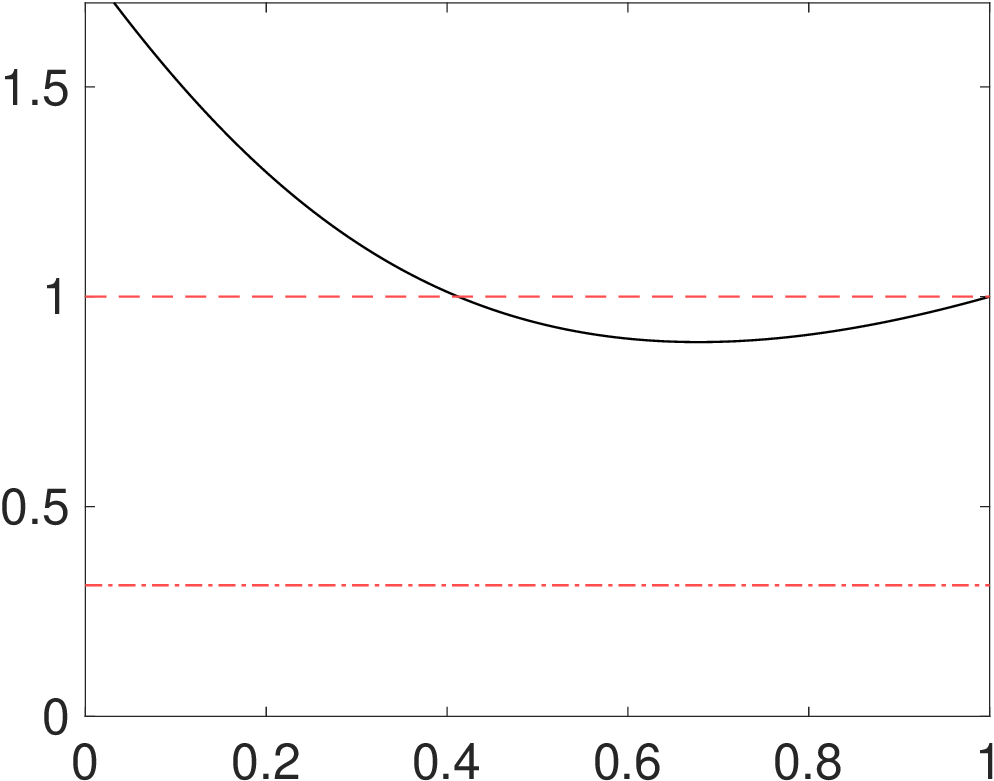}
    \end{tabular}
    \caption{Value of the Bayes Factor $H_t(\alpha_t)$ ({\color{black}\protect\tikz[baseline]{\protect\draw[line width=0.2mm] (0,.6ex)--++(0.5,0) ;}}) as a function of $\alpha_t$ (horizontal axis) for the parameter settings given in Tab. \ref{tab:ex2}. Each plot's thresholds are at 1 and $10^{-1/2}$ ({\color{red}\protect\tikz[baseline]{\protect\draw[line width=0.2mm, dashed] (0,.6ex)--++(0.5,0) ;}}).}
    \label{ex2:fig}
\end{figure}
Since the outcome of the testing procedure strongly depends on the choice of the discounting parameter $\alpha_t$ and since different thresholds for $H_t$ can be chosen, we propose three alternative Bayesian decision rules to conduct outlier detection. These rules are robust and exploit the variability of the BF due to the discounting parameter and the sampling distribution. 

A first decision rule is the Minimum BF (MBF), i.e. the smallest possible BF within the class of alternative distributions specified in Eq. \ref{altdistr}, that is $$H_t^{(MBF)}=\inf_{\alpha_t\in(\underline{\alpha},\overline{\alpha})\subset (0,1)} H_t(\alpha_t).$$
The MBF has been used in other testing problems such as unit root testing \citep{berger1994noninformative}, correlation testing \citep{Chen03062021} and reverse-Bayes procedures \citep{pawel2022sceptical}. See also \cite{HeldOtt2017} for a review of minimum BFs and a comparison with standard BF. The rationale behind this decision rule is that the evidence for the null is at least the MBF. In Remark \ref{rem:exstepwise1}, the three settings return an MBF below one. However, the MBF is not too far below the threshold in the first setting compared to the other two scenarios. These considerations motivate the need for alternative decision rules. 

We consider the Integrated BF (IBF) as a second decision rule. With this rule, one assumes a prior distribution for the discounting parameter $\alpha_t$ and averages the BF over all possible discounting values, that is
$$H_t^{(IBF)}=\int_{\underline{\alpha}}^{\bar{\alpha}} H_t(\alpha_t)\pi_t(\alpha_t)d\alpha_t-1,$$ where $\pi_t(\alpha)$ is a suitable probability density function for $\alpha_t$ with support $(\underline{\alpha},\bar{\alpha})\subset(0,1)$. The choice of $\pi_t$ is crucial to achieve a well-defined IBF in the matrix variate case. 

Since the BF can be bounded from above under mild regularity conditions, the IBF can be modified to account for the relative magnitude of the evidence in favor of the null. We thus introduce the Normalized IBF (NIBF) defined as
$$H_t^{(NIBF)}=\left(\int_{\underline{\alpha}}^{\bar{\alpha}} H_t(\alpha_t)\pi_t(\alpha_t)d\alpha_t-1\right)\left(\int_{\underline{\alpha}}^{\bar{\alpha}} \kappa_t(\alpha_t)\pi_t(\alpha_t)d\alpha_t-1\right)^{-1},$$
where $\kappa_t(\alpha_t) < \infty$ is an upper bound for $H_t(\alpha_t)$. The upper bound will be used later in this paper to show the BF integrability with respect to the discounting parameter.

IBF and NIBF do not account for the sampling variability in hypothesis testing. For this reason, we propose a predictive BF approach to incorporate such variability. As  $Y_t$ is not observed at time $t-1$, the predictive BF is a random variable whose marginal distribution can be used to derive a calibrated value for the discounting parameter $\alpha_t$. With our approach, we account for the sampling variability of the predictive BF and find the analytical distributions $F_{j,t}$ of the random BF, $H_{t}(\alpha_t)$, under the null hypothesis $\mathcal{H}_0$ of the absence of outliers ($j=0$) and the alternative hypothesis $\mathcal{H}_1$ of the presence of outliers ($j=1$). The calibrated predictive BF is derived through the following steps.
    \begin{enumerate}[1.]
        \item Derive the threshold function $\underline{h}(\alpha_t)$, $\alpha_{t}\in(\underline{\alpha},\bar{\alpha})$ for rejecting the null $\mathcal{H}_0$, such that $\mathbb{P}(\{H_t<\underline{h}(\alpha_t)\})=\tau$ under the null for a given test size $\tau$. The probability is evaluated using the $F_{0,t}$ distribution.
        \item Determine the discounting parameter  $\alpha_t^*\in(\underline{\alpha},\bar{\alpha})$ such that for a given power level $\beta$ $\mathbb{P}(\{H_t>\bar{h}(\alpha_t^*)\})=1-\beta$ under the alternative, that is:
        $\beta = 1 - F_{1,t}(\bar{h}(\alpha_t^*))$, $\bar{h}(\alpha_t^*) = 2 - \underline{h}(\alpha_t^*)$,
        where the constraint comes from the threshold symmetric assumption.
        \item Define the inconclusive interval $C_H = (\underline{h}(\alpha_t^*), \bar{h}(\alpha_t^*))$ where the BF is statistically equal to one and define the calibrated BF as $H_t^{*}=H_t(\alpha_t^{*})$.
\end{enumerate}     
The conclusion of the test procedure is to reject the null hypothesis $\mathcal{H}_0$ when $H_t<\underline{h}$, accepting it when $H_t>\bar{h}$ and randomizing when $H_t\in C_{H}$. As randomization is usually not appealing in some applications, an alternative procedure can be used where only one threshold $\bar{h}(\alpha_t^*)=\underline{h}(\alpha_t^*)<1$ is considered to define a critical region $(0,\underline{h}(\alpha_t^*))$. Our procedure for the calibrated value of $\alpha_t$ is similar in spirit to the ones proposed in \cite{weiss1997bayesian} and \cite{pawel2025closed} for determining the optimal sample size. In addition, another procedure, which accounts for the distribution of $\alpha_t$, can be defined by minimizing jointly the first and second type error probabilities.

\section{Matrix Normal Outlier Detection}\label{sec:illus}
\subsection{A Bayesian Matrix Normal Model}
The BF and its properties are derived under the following assumptions. See the Appendix \ref{sec:preliminary} for some background on the matrix distributions and proofs. The normal assumption is standard in outlier detection as it guarantees some finite-sample analytical results.
\begin{assumption}{}\label{ass1}
$\mathbf{Y}=\{\boldsymbol{Y}_{j}\left(p\times n\right),\ j = 1,\ldots,t - 1\}$, are i.i.d. matrix normal variables with distribution $\mathcal{MN}_{p,n}(\boldsymbol{B},\boldsymbol{\Sigma}_L,\boldsymbol{V})$ with $\boldsymbol{B}\in\mathbb{R}^{p\times n}$ $\boldsymbol{\Sigma}_{L}\in\mathbb{R}^{p\times p}_{+}
$ and $\boldsymbol{V}\in\mathbb{R}^{n\times n}_{+}$. 
\end{assumption}
The information set at time $T=t-1$ is given by the sigma-algebra $\mathcal{D}_{t-1}$ generated by the elements of $\mathbf{Y}$. To achieve analytical tractability, and similarly to the univariate setting of \cite{west1986bayesian}, we assume that $\boldsymbol{\Sigma}_{L}$ is $\mathcal{D}_{t-1}$-measurable and $\boldsymbol{B}$ has a conjugate prior. We shall emphasize that the Gaussian assumption is not restrictive, as within a Bayesian framework, the scale parameters can be integrated out of the likelihood, resulting in a Student-t distribution that can account for heavy tails.
\begin{assumption}{}\label{ass2}
$\boldsymbol{B}$ has a conjugate matrix normal prior, i.e. $\boldsymbol{B}\sim \mathcal{N}_{p,n}\left( \boldsymbol{M},\boldsymbol{\Sigma}_{P},\boldsymbol{V} \right)$ and $\boldsymbol{\Sigma}_{L}\otimes \boldsymbol{V}$ is given.
\end{assumption}
The assumption $\boldsymbol{\Sigma}_{P} = \boldsymbol{\Sigma}_{L}/\varphi$, with $\varphi > 0$ is common in Gaussian models \citep[e.g., see][]{zellner1986assessing}. Some analytical results can also be obtained when the covariance $\boldsymbol{V}$ is estimated following a Bayesian procedure. When $\boldsymbol{V}$ is unknown, a conjugate Normal-Inverse Wishart prior for $\boldsymbol{B}$ and $\boldsymbol{V}$ is assumed.

\begin{assumption}{}\label{ass3}
$\boldsymbol{B}$ and $\boldsymbol{\Sigma}_L$ have a conjugate inverse Wishart matrix normal prior, i.e.  $\boldsymbol{B}|\boldsymbol{V}\sim \mathcal{N}_{p,n}\left( \boldsymbol{M},\boldsymbol{\Sigma}_{L}/\varphi,\boldsymbol{V}/\rho \right)$ and $\boldsymbol{V}\sim\mathcal{I}\mathcal{W}_n(\boldsymbol{\Psi},m)$ and $\boldsymbol{\Sigma}_{L}$ is given. 
\end{assumption}

Following the notation in the previous section,  $\boldsymbol{\theta}_{t}$ corresponds to $\boldsymbol{B}_t=\boldsymbol{B}|\mathcal{D}_{t-1}$ and $(\boldsymbol{B}_t,\boldsymbol{V}_t)=(\boldsymbol{B},\boldsymbol{V})|\mathcal{D}_{t-1}$ for the $\boldsymbol{V}$ known and $\boldsymbol{V}$ unknown cases, respectively.

Proposition \ref{prop:posteriorBF} provides the analytical expression of the Bayes Factor derived under the three assumptions outlined above.

\begin{proposition}[BF]\label{prop:posteriorBF}
\begin{enumerate}[i)]
    \item Under the Assumptions \ref{ass1} and \ref{ass2}, $\boldsymbol{Y}_t|\mathcal{D}_{t - 1}\sim \mathcal{N}_{p,n}\left(\boldsymbol{M}_{*},\boldsymbol{\Sigma}_{d}, \boldsymbol{V} \right)$ under the null hypothesis, and $\boldsymbol{Y}_t|\mathcal{D}_{t - 1}\sim \mathcal{N}_{p,n}\left(\boldsymbol{M}_{*},\boldsymbol{\Sigma}_{A,d}, \boldsymbol{V} \right)$ under the alternative and the Bayes Factor $H_t$ is given by   
\begin{equation}
H_{t}(\alpha_{t}) = \frac{\left| \boldsymbol{\Sigma}_{A,d} \right|^{\frac{n}{2}}}{\left| \boldsymbol{\Sigma}_{d} \right|^{\frac{n}{2}}}\exp\left\{ - \frac{1}{2}\text{tr}\left\lbrack \left( \boldsymbol{\Sigma}_{d}^{- 1} - \boldsymbol{\Sigma}_{A,d}^{- 1} \right)\left( \boldsymbol{Y}_t - \boldsymbol{M}_{*} \right)\boldsymbol{V}^{- 1}\left( \boldsymbol{Y}_t - \boldsymbol{M}_{*} \right)^{'} \right\rbrack \right\},\label{BF_matrix}
\end{equation}
where ${\boldsymbol{\Sigma}_{d} = \boldsymbol{\Sigma}}_{L} + \boldsymbol{\Sigma}_{*}$,
$\boldsymbol{\Sigma}_{A,d} = \boldsymbol{\Sigma}_{L} + \boldsymbol{\Sigma}_{A,*}$, $\boldsymbol{\Sigma}_{A,*} = \boldsymbol{\Sigma}_{*}/\alpha_{t}$ and the posterior parameters $\boldsymbol{M}_{\ast}$ and $\boldsymbol{\Sigma}_{*}$ are given in Eq. \ref{eq:Mstar} and \ref{eq:Sigmastar}.
\item Under the the Assumptions \ref{ass1} and \ref{ass3}, $\boldsymbol{Y}_t|\mathcal{D}_{t - 1}\sim\mathcal{T}_{p,n}\left(m_{*}-2n, \boldsymbol{M}_{*}, \boldsymbol{\Sigma}_L , \boldsymbol{L}_{*} \right)$ under the null and $\boldsymbol{Y}_t|\mathcal{D}_{t - 1} \sim \mathcal{T}_{p,n}\left(m_{A, *}-2n, \boldsymbol{M}_{*}, \boldsymbol{\Sigma}_L , \boldsymbol{L}_{A, *}\right)$ under the alternative  where $\mathcal{T}_{p,n}( \nu, \boldsymbol{M},\boldsymbol{\Sigma}, \boldsymbol{\Omega})$ denotes the matrix Student-t distribution, $\boldsymbol{L}_{*} = \boldsymbol{\Psi}_{*}k_d/k_{*}$, and  $\boldsymbol{L}_{A,*} = \boldsymbol{\Psi}_{A,*}k_d/k_{A,*}$. The Bayes Factor $H_t$ is given by \begin{equation}  H_t(\alpha_t)=G\frac{\left|\boldsymbol{\Psi}_{A,*}+\frac{\alpha_t k_{*}}{\left(\alpha_t k_{*}+1\right)}\left(\boldsymbol{M}_{*}-\boldsymbol{Y}_t\right)^{'}\boldsymbol{\Sigma}_L^{-1}\left(\boldsymbol{M}_{*}-\boldsymbol{Y}_t\right)\right|^{\frac{m_{A,d}-n-1}{2}}}{\left|\boldsymbol{\Psi}_{*}+k_{*}\frac{\left(\boldsymbol{M}_{*}-\boldsymbol{Y}_t\right)^{'}\boldsymbol{\Sigma}^{-1}\left(\boldsymbol{M}_{*}-\boldsymbol{Y}_t\right)}{\left(k_{*}+1\right)}\right|^{\frac{m_{d}-n-1}{2}}},
\end{equation}
where
\begin{equation}
G=\frac{\left|\boldsymbol{\Psi}_{*}\right|^{\frac{m_{*}-n-1}{2}}k_{*}^\frac{np}{2}\Gamma_{n}\left(\frac{m_{d}-n-1}{2}\right)k_{A,d}^\frac{np}{2}\Gamma_{n}\left(\frac{m_{A,*}-n-1}{2}\right)}{k_{d}^\frac{np}{2}\Gamma_{n}\left(\frac{m_{*}-n-1}{2}\right)\left|\boldsymbol{\Psi}_{A,*}\right|^{\frac{m_{A,*}-n-1}{2}}k_{A,*}^\frac{np}{2}\Gamma_{n}\left(\frac{m_{A,d}-n-1}{2}\right)}\nonumber\\
\end{equation}
and $k_{d}=k_{*}+1$, $k_{A,*}=\alpha_t k_{*}$, $ k_{A,d}=\alpha_t k_{*}+1$, $m_{d}=m_{*}+p$, $m_{A,*}=\alpha_t( m_{*}+p)-p$, $m_{A,d}=m_{A,*}+p$, 
and $m_{*}, \boldsymbol{M}_{*},\boldsymbol{\Psi}_{*}$ as in Eq. \ref{eq:PsistarV}.
\end{enumerate}
\end{proposition}


\begin{rem}[Interpretability]\label{interpretabiliy}
In the case $\boldsymbol{V}$ is known, it can be easily shown that the discounting approach provides a tractable prior distribution under the alternative. The approach is equivalent to assuming, under the alternative hypothesis, the following hierarchical prior distribution $\boldsymbol{B}|\tilde{\boldsymbol{M}},\mathbf{Y}\sim\mathcal{N}_{p,n}(\tilde{\boldsymbol{M}},\boldsymbol{\Sigma}_{*},\boldsymbol{V})$ with $\tilde{\boldsymbol{M}}|\mathbf{Y}\sim\mathcal{N}_{p,n}(\boldsymbol{M}_{*},\boldsymbol{\Sigma}_{*}(1/\alpha_t-1),\boldsymbol{V})$, which is well-defined for $\alpha_t \in (0,1)$. This can be interpreted as a random perturbation of the posterior distribution, which inflates the variance of the location parameter distribution.
\end{rem}

The normalizing constant of the alternative distribution and its properties can be derived from Prop. \ref{prop:posteriorBF} and are given in the following.
\begin{corollary}\label{coroll}
\begin{enumerate}
    \item Under Assumptions \ref{ass1} and \ref{ass2}, the inverse normalizing constant is
    \begin{equation}
C\left( \alpha_{t} \right)={\alpha_{t}^{\frac{np}{2}}(2\pi)}^{\frac{(\alpha_{t}-1)np}{2}}\left| \boldsymbol{\Sigma}_{*} \right|^{\frac{(\alpha_{t}-1)n}{2}}|\boldsymbol{V}|^{\frac{(\alpha_{t}-1)p}{2}}
\end{equation}
and satisfies $C(\alpha_t)\rightarrow 0$ as $\alpha_{t} \rightarrow 0^{+}$ and $C(\alpha_t)\rightarrow 1$ as $\alpha_t\rightarrow 1^{-}$.
    \item Under Assumptions \ref{ass1} and \ref{ass3}, the inverse normalizing constant is
    \begin{equation}
C(\alpha_t) = \frac{{\left|\boldsymbol{\Psi}_{A,*}\right|^{\frac{m_{A,*}-n-1}{2}}}(2\pi)^{\frac{\alpha_{t}np}{2}}|\boldsymbol{\Sigma}_{*}|^{\frac{\alpha_{t}n}{2}}2^{\frac{\alpha_{t}(m_{*}-n-1)n}{2}}\Gamma_{n}\left(\frac{m_{*}-n-1}{2}\right)^{\alpha_{t}}}{(2\pi)^{\frac{np}{2}}|\boldsymbol{\Sigma}_{A,*}|^{\frac{n}{2}}2^{\frac{(m_{A,*}-n-1)n}{2}}\Gamma_{n}\left(\frac{m_{A,*}-n-1}{2}\right)\left|\boldsymbol{\Psi}_{*}\right|^\frac{\alpha_t (m_{*}-n-1)}{2}}
\end{equation}
\end{enumerate}
and satisfies $C(\alpha_t)\rightarrow 0$ as $\alpha\rightarrow \underline{\alpha_t}^{+}$ and $C(\alpha_t)\rightarrow 1$ as $\alpha_t\rightarrow 1^{-}$, where $\underline{\alpha}=(p+2n)/m_d$.
\end{corollary}

\begin{rem}[Univariate Gaussian model]\label{ex3}
Let $Y_{j}\sim\mathcal{N}(\theta,\sigma^2)$ i.i.d. for $j=1,\ldots,t-1$ and assume a conjugate prior for $\theta$, that is $\theta\sim\mathcal{N}(m,\sigma^2/\varphi)$ with $\varphi>0$. It is easy to show that the predictive under the alternative is given by $p_{A}(Y_{t}|\mathcal{D}_{t-1})=(2\pi\sigma^2_{\ast}/\alpha_t)^{1/2}\exp\{-(\theta-m_{\ast})^2\alpha_t/2\sigma^{2}_{\ast}\}$ with $m_{\ast}=\varphi m/(\varphi+t-1)+(t-1)/(\varphi+t-1)\bar{Y}$ with $\bar{Y} = (t-1)^{-1}\sum_{j = 1}^{t-1}Y_j$, $\sigma^2_{*}=\sigma^2/(\varphi+t-1)$ and the normalizing constant is
$C(\alpha_t)=(2\pi\sigma^2_{*})^{(\alpha_t-1)/2}\alpha_t^{1/2}$. It follows that the BF is $H_t(\alpha_t)=\kappa_t(\alpha_t)\exp\{(Y_t-m_{\ast})^2 A_t^{-1}\}$, where
\begin{equation}
A_t=\frac{2\sigma^2_{\ast}(\varphi+t)(\alpha_t(\varphi+t-1)+1)}{(\alpha_t-1)(\varphi+t-1)}\quad\hbox{and}\quad \kappa_t(\alpha_t)=\left(\frac{\alpha_t(\varphi+t-1)+1}{\alpha_t(\varphi+t)}\right)^{1/2}
\end{equation}
naturally represents an upper bound of $H_{t}(\alpha_t)$. Intuitively, for small values of $\alpha_t$ the null hypothesis is not rejected, that is $H_t(\alpha_t)\rightarrow +\infty$ as $\alpha_t\rightarrow 0^+$, whereas for large values, the BF gets closer to one, that is $H_t(\alpha_t)\rightarrow 1$ as $\alpha_t\rightarrow 1^{-}$. Its derivative $H^{\prime}_{t}(\alpha_t)=0$ for $\alpha_t$ is equal to
$\alpha_{t,0}=((Y_t-m_{\ast})^2/\sigma^2_{\ast}-1)^{-1}(\varphi+t-1)^{-1}$ provided $(Y_{t}-m_{\ast})^2/\sigma^2-1>(\varphi+t-1)^{-1}$. 
\end{rem}

\subsection{Bayes Factor Properties}
\begin{proposition}[BF properties, known $\boldsymbol{V}$]\label{prop:BFmatrix}
Under Assumptions \ref{ass1} and \ref{ass2}, the BF $H_{t}(\alpha_t)$ satisfies the following properties.
\begin{enumerate}[i)]
    \item There exists a function $\kappa_t\left( \alpha_{t} \right)= \left| \boldsymbol{\Sigma}_{A,d} \right|^{\frac{n}{2}}\left| \boldsymbol{\Sigma}_{d} \right|^{-\frac{n}{2}}\geq1$ independent of $\boldsymbol{Y}_t$ such that $H_t(\alpha_t)\leq \kappa_t(\alpha_t)$,  $\kappa_t \rightarrow 1$ as $\alpha_{t} \rightarrow 1^{-}$, $\kappa_t \rightarrow \infty$ as $\alpha_{t} \rightarrow 0^{+}$ and $\kappa_t$ is decreasing in $\alpha_t$;
    \item $H_{t}\left( \alpha_{t} \right) \rightarrow 1$ as $\alpha_{t} \rightarrow 1^{-}$ and $H_{t}\left( \alpha_{t} \right)\rightarrow\infty$ as $\alpha_{t} \rightarrow 0^{+}$;
    \item $\partial_{\alpha_t} H_t(\alpha_t)= H_{t}\left( \alpha_{t} \right)(2\alpha_{t})^{-1}\text{tr}\left(-n\Tilde{\boldsymbol{B}}^{- 1}\boldsymbol{\Sigma}_{*} +\alpha_{t} \boldsymbol{\Upsilon}(\alpha_t) \Tilde{\boldsymbol{A}} + \alpha_{t}\left( {1 - \alpha}_{t} \right)\boldsymbol{\Upsilon}(\alpha_t)\boldsymbol{\Sigma}_{L}\Tilde{\boldsymbol{B}}^{-1}\Tilde{\boldsymbol{A}} \right)$ where we defined $\Tilde{\boldsymbol{A}} = \left( \boldsymbol{Y}_t - \boldsymbol{M}_{*} \right)\boldsymbol{V}^{- 1}\left( \boldsymbol{Y}_t - \boldsymbol{M}_{*} \right)^{'}$, $\Tilde{\boldsymbol{B}}=\left( {\alpha_{t}\boldsymbol{\Sigma}}_{L} + \boldsymbol{\Sigma}_{*} \right)$ and $\boldsymbol{\Upsilon}(\alpha_t)=\left( \boldsymbol{\Sigma}_{L} + \boldsymbol{\Sigma}_{*} \right)^{- 1}\boldsymbol{\Sigma}_{*}\Tilde{\boldsymbol{B}}^{- 1}$.
\end{enumerate}
\end{proposition}
Note that the second and third elements within the trace in the derivative at \textit{iii)} are positive. In contrast, the first one is negative and, for each $\boldsymbol{Y}_t$, $\partial_{\alpha_t}H_t(\alpha_t)\rightarrow -\infty$ as $\alpha_{t}\rightarrow 0^{+}$ and $\partial_{\alpha_t} H_t(\alpha_t)\rightarrow (-n\text{tr}\text(\left( \boldsymbol{\Sigma}_{L} + \boldsymbol{\Sigma}_{*} \right)^{- 1})+\text{tr}(\boldsymbol{\Upsilon}(1)\Tilde{\boldsymbol{A}}))/2$ as $\alpha_{t}\rightarrow 1^{-}$. Thus there is a change of the sign, provided $\text{tr}(\left( \boldsymbol{\Sigma}_{L} + \boldsymbol{\Sigma}_{*} \right)^{- 1}\boldsymbol{\Sigma}_{\ast}(\left( \boldsymbol{\Sigma}_{L} + \boldsymbol{\Sigma}_{*} \right)^{- 1}\tilde{A}-n \boldsymbol{I}))>0$, and the BF has at least one stationary point.

\begin{proposition}[BF properties, unknown $\boldsymbol{V}$]\label{prop:BFmatrixUnknownV}
Under Assumptions \ref{ass1} and \ref{ass3}, the BF $H_{t}(\alpha_t)$ satisfies the following properties.
\begin{enumerate}[i)]
    \item There exists a function $\kappa_t\left( \alpha_{t} \right)= (k_{*}k_{A,d}(k_{d}k_{A,*}))^{-np/2}$ $\Gamma_{n}((m_{d}-n-1)/2)$ $\Gamma_{n}((m_{A,*}-n-1)/2)$ $(\Gamma_{n}((m_{*}-n-1)/2)\Gamma_{n}((m_{A,d}-n-1)/2))^{-1}$ which does not depend on $\boldsymbol{Y}_t$ such that $H_t(\alpha_t)\leq \kappa_t(\alpha_t)$, $\kappa_t$ is decreasing in $\alpha_t$,  $\kappa_t \rightarrow 1$ as $\alpha_{t} \rightarrow 1^{-}$, $\kappa_t \rightarrow +\infty$ as $\alpha_{t} \rightarrow \underline{\alpha}^{+}$ where $\underline{\alpha}=(2n+p)/m_d$;
    \item $H_{t}\left( \alpha_{t} \right) \rightarrow 1$ as $\alpha_{t} \rightarrow 1^{-}$ and $H_{t}\left( \alpha_{t} \right)\rightarrow\infty$ as $\alpha_{t} \rightarrow \underline{\alpha}^{+}$;
    \item $
\partial_{\alpha_t}H_t(\alpha_t)=L(|\boldsymbol{\Psi}_{A,*}|^{(m_{A,*}-n-1)/2}k_{A,*}^\frac{np}{2}\Gamma_n((m_{A,d}-n-1)/2))^{-2}(B_1 \partial_{\alpha_t}A_1 -A_1 \partial_{\alpha_t}B_1)$, where $L=\left|\boldsymbol{\Psi}\right|_*^{(m_*-n-1)/2}k_*^{np/2}\Gamma_n\left((m_d-n-1)/2)\right)(k_d^{np/2}\left|\boldsymbol{\Psi}_d\right|^{(m_d-n-1)/2}\Gamma_n\left((m_*-n-1)/2\right))^{-1}$, $A_1=a_1a_2a_3$, $B_1=b_1b_2b_3$ with $a_1=k_{A,d}^{np/2}$, $a_2=\Gamma_n\left((m_{A,*}-n-1)/2\right)$ $a_3=\left|\boldsymbol{\Psi}_{A,d}\right|^{(m_{A,d}-n-1)/2}$, $b_1=k_{A,*}^{np/2}$, $b_2=\Gamma_n\left((m_{A,d}-n-1)/2\right)$, $b_3=\left|\boldsymbol{\Psi}_{A,*}\right|^{(m_{A,*}-n-1)/2}$.
\end{enumerate}
\end{proposition}

In the derivative $\partial_{\alpha_t}H_t(\alpha_t)$, the terms $a_i$, $b_i$, $i=1,\ldots,3$ are positive, and the terms $\partial_{\alpha_t}A_1$ and $\partial_{\alpha_t}B_1$ can take both positive and negative values since they are sums of digamma functions. Thus, for some values of $\alpha_t$, the $\partial_{\alpha_t}H_t(\alpha_t)$ can go to zero, and the BF has at least one stationary point.

The following remark discusses the relationship with the univariate case, whereas the empirical section will provide further illustrations for the matrix case.

\begin{rem}[Univariate Gaussian model]\label{exUnivGauss}
The BF for the univariate outlier detection and its upper bound $\kappa_t(\alpha_t)$ given in Remark \ref{ex3} can be easily obtained setting $n = p = 1$ and $\boldsymbol{\Sigma}_{P} = \sigma^{2}/\varphi$  in Eq. \ref{BF_matrix}, in Prop. \ref{prop:posteriorBF} and in \textit{i)} of Prop.  \ref{prop:BFmatrix}. Since $A_t$ in Remark \ref{ex3} is strictly negative, then the upper bound satisfies the properties in Prop. \ref{prop:BFmatrix}. The derivative of the BF is
\begin{equation}
    \partial_{\alpha_t} H_t(\alpha_t)=\frac{H_{t}(\alpha_t)}{2}\left(\frac{(\varphi+t-1)(Y_t-m_{\ast})^2}{\sigma^2 (\alpha_t(\varphi+t-1)+1)^2}-\kappa_t(\alpha_t)\frac{1}{\alpha_t^2(\varphi+t)}\right)
\end{equation}
which goes to $-\infty$ as $\alpha_t\rightarrow 0^{+}$ and is strictly positive for $\alpha_t>1/((\varphi+t-1)((Y_t-m_{\ast})^2/\sigma^2-1))$ provided that $(Y_t-m_{\ast})^2/\sigma^2>(\varphi+t)/(\varphi+t-1)$. This implies there exists a stationary point and that $H_{t}<1$ for some values of $\alpha_t$ provided the conditions in Prop. \ref{th1} are satisfied.
\end{rem}

Let us now find the subset of the sample space such that for a given threshold $0 < h_{0}(\alpha_t) < \kappa_t(\alpha_t)$, the BF leads us to accept the null hypothesis, that is, $H_{t} \geq h_{0}(\alpha_t) $. Define $\mathbf{y}_{t} = \text{vec}\left( \boldsymbol{Y}_t \right)$ and $\mathbf{m}_{*} = \text{vec}\left( \boldsymbol{M}_{*} \right)$ and let $\kappa_t(\alpha_t)$ be the upper bound given in Prop. \ref{prop:BFmatrix}. From the expression of the BF given in Prop. \ref{prop:posteriorBF} the condition $H_{t} \geq h_{0}(\alpha_t)$ is satisfied for $\mathbf{y}_t$ in the ellipsoid:
\begin{equation}
\left( \mathbf{y}_{t} - \mathbf{m}_{*} \right)^{'}\left( \boldsymbol{\Sigma}_{H} \otimes \boldsymbol{V} \right)^{- 1}\left( \mathbf{y}_{t} - \mathbf{m}_{*} \right) \leq 2\log\left( {\frac{\kappa_t(\alpha_t)}{h_{0}(\alpha_t)}} \right)
\label{thresMult}
\end{equation}
centered in $\mathbf{m}_{*}$, with axis in the direction of the eigenvector $\boldsymbol{\nu}_{k}$ of $\boldsymbol{\Sigma}_{H} \otimes \boldsymbol{V}$, where $\boldsymbol{\Sigma}_H=(1 - \alpha_{t})^{-1}\left( \alpha_{t}\boldsymbol{\Sigma}_{L} + \boldsymbol{\Sigma}_{*} \right)\boldsymbol{\Sigma}_{*}^{- 1}\left( \boldsymbol{\Sigma}_{L} + \boldsymbol{\Sigma}_{*} \right)$. The length of the ellipsoid axes along the $\nu_{k}$ eigenvector's direction is $\ell_{k} = \ 2\sqrt{2\log\left( {\frac{\kappa_t(\alpha_t)}{h_{0}(\alpha_t)}} \right)\xi_{k}}$, where $\xi_{k}$ is the corresponding eigenvalue of $\boldsymbol{\Sigma}_{H} \otimes \boldsymbol{V}$. Let $\gamma_{i}$ and $\boldsymbol{\zeta}_{i}$ be the eigenvalues and the
corresponding eigenvectors of $\boldsymbol{\Sigma}_{H}$ and $\tau_{j}$ and
$\boldsymbol{\delta}_{j}$ the eigenvalues and eigenvectors
of $\boldsymbol{V}$. Then $\boldsymbol{\Sigma}_{H} \otimes \boldsymbol{V}$ has eigenvalues
$\xi_{k} = \gamma_{i}\tau_{j}$ with corresponding eigenvectors
$\boldsymbol{\nu}_{k}= \boldsymbol{\zeta}_i \otimes \boldsymbol{\delta}_{j}$. If we consider $h_{0} = 1$, then $\boldsymbol{Y}_t$ is not considered an outlier for $H_{t} \geq 1$. Figure \ref{fig:BFelli} provides a numerical illustration for $p=2$ and $n=1$. Increasing the discounting parameter value reduces the evidence in favor of the null for $h_0>1$ (dashed lines, left plot) and increases it for $h_0<1$ (solid lines). For larger values of $h_0$, the evidence against the null becomes stronger (right plot). This effect can also be understood from the following remark on the univariate outlier detection, a special case of Eq. \ref{thresMult}.

\begin{figure}[t]
    \centering
    \begin{tabular}{cc}
    \includegraphics[width=0.35\textwidth]{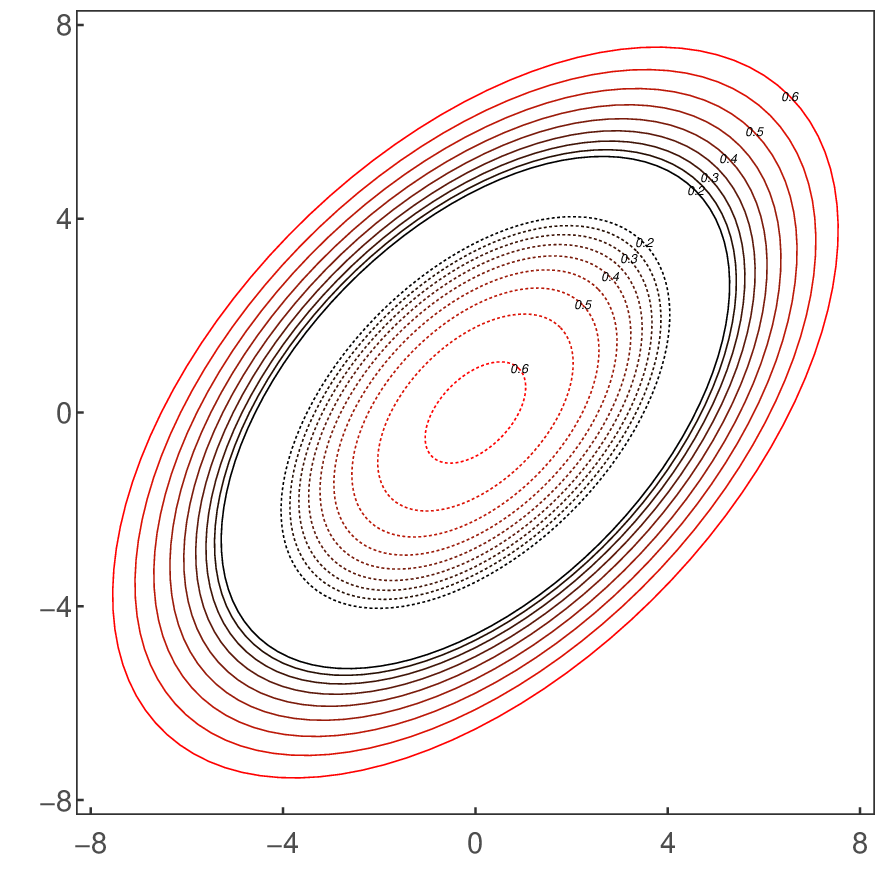}&
        \includegraphics[width=0.35\textwidth]{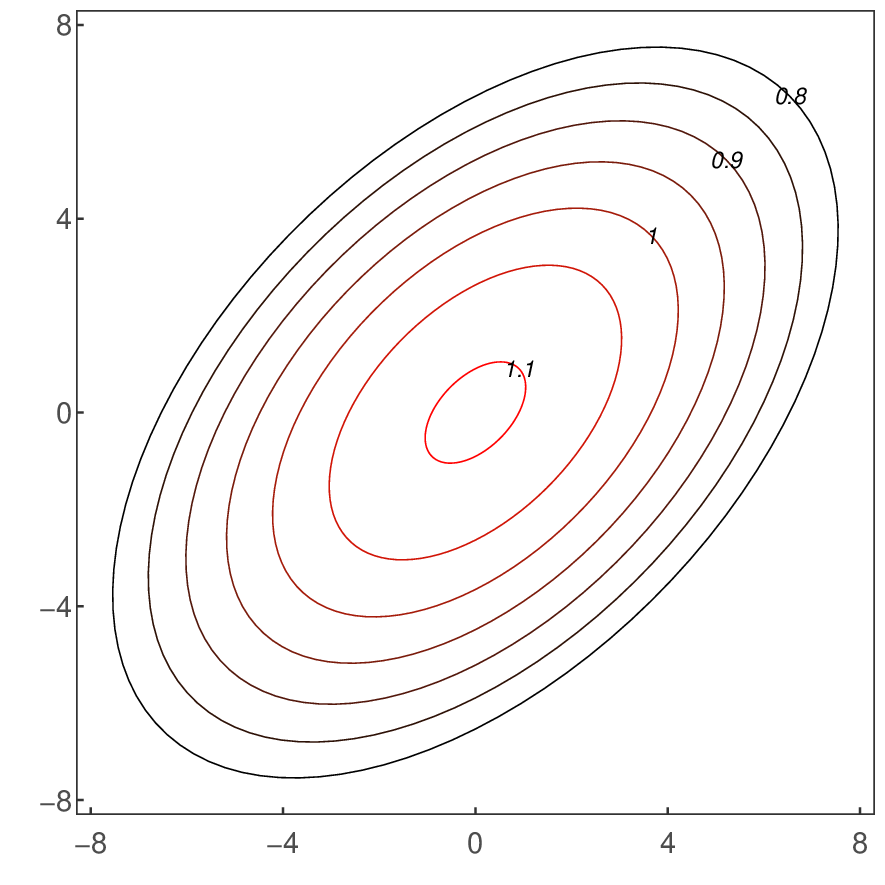}
    \end{tabular}
    \caption{Bayes Factor Critical regions. Left: increasing $\alpha_t$ from 0.2
    ({\color{black}\protect\tikz[baseline]{\protect\draw[line width=0.2mm] (0,.6ex)--++(0.5,0) ;}} $h_0=0.8$, {\color{black}\protect\tikz[baseline]{\protect\draw[line width=0.2mm,dashed] (0,.6ex)--++(0.5,0) ;}} $h_0=1.1$) to 0.6 ({\color{red}\protect\tikz[baseline]{\protect\draw[line width=0.2mm] (0,.6ex)--++(0.5,0) ;}} $h_0=0.8$, {\color{red}\protect\tikz[baseline]{\protect\draw[line width=0.2mm,dashed] (0,.6ex)--++(0.5,0) ;}} $h_0=1.1$). Right: increasing the threshold $h_0$ from 0.8 ({\color{black}\protect\tikz[baseline]{\protect\draw[line width=0.2mm] (0,.6ex)--++(0.5,0) ;}}) to 1.1 ({\color{red}\protect\tikz[baseline]{\protect\draw[line width=0.2mm] (0,.6ex)--++(0.5,0) ;}}).}
    \label{fig:BFelli}
\end{figure}

\begin{rem}[Univariate Gaussian model]\label{exUnivGaussContd}
In the univariate case, the null hypothesis is accepted, i.e. $H_t>h_{0}$, for values of $Y_t$ in the interval $(Y_{1t},Y_{2t})$ with
\begin{equation}
        Y_{jt}=m_{\ast}+(-1)^{j}\sqrt{\log\left(\frac{h_{0}(\alpha_t)}{\kappa_t(\alpha_t)}\right)\frac{2\sigma^2(\varphi+t)(\alpha_t(\varphi+t-1)+1)}{(\alpha_t-1)(\varphi+t-1)}}.
\end{equation}
\end{rem}

The BF and the outcome of the testing procedure depend on the choice of the discounting parameter $\alpha_t$. In this paper, we propose the integrated BF and its normalized version as a solution, which requires that the integral with respect to $\pi_t(\alpha_t)$ is bounded. In the following, we provide existence conditions for the IBF and NIBF under the beta perturbation assumption, a standard prior distribution used in Bayesian inference for parameters on bounded intervals. As stated in the following proposition, the integral of the upper bound for the univariate case can be derived analytically and is well-defined under the assumption of a standard uniform distribution for the discounting parameter.
\begin{proposition}[Integrability]\label{prop:IBFmatrix}
Assume a general beta distribution $\mathcal{B}e(a,b)$ truncated on the interval ($\underline{\alpha},\bar{\alpha})$. The IBF, $H_t=\int_{0}^{1}H_t(\alpha)\pi_{t}(\alpha)d\alpha$, satisfies the following properties.
\begin{enumerate}
    \item Under Assumptions \ref{ass1} and \ref{ass2}, for $a>np/2$, $b>0$, $\underline{\alpha}=0$ and $\bar{\alpha}=1$
\begin{equation}
H_t<\sum_{\ell=0}^{\infty}\sum_{\underline{k}\in\mathcal{K}_{\ell}}d_{\ell,\underline{k}}\bar{c}_{\underline{k}}\frac{\Gamma(n/2+1)B(a-np/2+w_{\underline{k}},b)}{\Gamma(\ell+1)\Gamma(n/2-\ell)B(a,b)}<\infty
\end{equation}
with $\underline{k}=(k_0,\ldots,k_p)$, $\mathcal{K}_{\ell}=\{\underline{k}|k_0+\ldots,k_p=\ell\}$, $\bar{c}_{0}=\tilde{c}_{0}-1$, $\bar{c}_{j}=\tilde{c}_j$, $j=1,\ldots,p$, and
$$
d_{\ell,\underline{k}}=\binom{\ell}{k_0,k_1,k_2,\ldots,k_p},\quad \bar{c}_{\underline{k}}=\prod_{j=0}^{p}\bar{c}_{j}^{k_j},
$$
where $d_{\ell,\underline{k}}$ denotes the multinomial coefficient, $\tilde{c}_j$ $j=0,\ldots,p$ are the coefficients of the characteristic polynomial of $\boldsymbol{\Sigma}_L^{-1}\boldsymbol{\Sigma}_{\ast}$.
\item Under Assumptions \ref{ass1} and \ref{ass3}, for $a>1$ and $b>0$, $\underline{\alpha}=(2n+p)/m_d$ and $\bar{\alpha}=1$.
\begin{align}
&H_t<(\underline{\alpha} k_{*}/(\underline{\alpha}(k_{*}+1)))^{np/2}\Gamma\left(\frac{p}{2}\right)^{-n}\left(\frac{4}{p}\right)^n\frac{B(1-\underline{\alpha};ra-r+1,qb-q+1)}{B(1-\underline{\alpha};a,b)}\xi(r,q)
\end{align}    
where $B(c,a,b)=\int_{0}^{c}x^{a-1}(1-x)^{b-1}dx$ denotes the incomplete beta function, $\xi(r,q)=(((m_{d}(1-\underline{\alpha})+p)^{1-qn}-p^{1-qn})/(m_{d}(1-qn)))^{1/q}$ and $q=r/(r-1)$ with $r > 1$.
\end{enumerate}
\end{proposition}

\begin{corollary}[IBF in the univariate model]\label{corol:integratedBoundUniv}
The integral of the BF $H_t$ given in Remark \ref{ex3} is bounded and satisfies:
\begin{equation}
\int_{0}^{1}{H_{t}\left( \alpha_{t} \right)d\alpha_{t} - 1 \leq \frac{1}{2\sqrt{(\varphi + t-1)(\varphi + t)}}\log\left( \frac{\sqrt{(\varphi + t)} + \sqrt{(\varphi + t-1)}}{\sqrt{(\varphi + t)} - \sqrt{(\varphi + t-1)}} \right)}.
\end{equation}
\end{corollary}


In the empirical illustration, we show that the beta hyper--parameters' choice can affect the hypothesis testing outcome, which calls for calibrated BFs. We assume $\boldsymbol{Y}_t$ in the predictive BF is not observed at time $t$. Thus, the predictive BF is random, and we prove that its distribution is a mixture of gamma distributions. Controlling for the test's size and power while using the BF distribution allows us to derive a calibrated BF and a suitable critical region for the test.

\begin{proposition}[Distribution of the BF, known $\boldsymbol{V}$]\label{DistribBFmatrix} 
Assume $\boldsymbol{Y}_t|\mathcal{D}_{t-1}\sim\mathcal{N}_{p,n}(\tilde{\boldsymbol{M}},\tilde{\boldsymbol{\Sigma}},\boldsymbol{V})$, let $\lambda_1$, $\ldots,\lambda_p$ be the eigenvalues of the covariance matrix $(\boldsymbol{\Sigma}_{H}^{-1/2})'\tilde{\boldsymbol{\Sigma}}\boldsymbol{\Sigma}_{H}^{-1/2}$ with $\boldsymbol{\Sigma}_H= (1 - \alpha_{t})^{-1}\left( \alpha_{t}\tilde{\boldsymbol{\Sigma}} + \boldsymbol{\Sigma}_{*} \right)\boldsymbol{\Sigma}_{*}^{- 1}\left( \tilde{\boldsymbol{\Sigma}} + \boldsymbol{\Sigma}_{*} \right)$ and $Q$ the associated eigenvector matrix. Under Assumptions \ref{ass1} and \ref{ass2}, the distribution of the BF $H_t$ defined in Eq. \ref{BF_matrix} has pdf and cdf:
\begin{eqnarray}
f_{H_{t}}(h|\alpha_t) &=&\frac{2}{h}\sum_{k=0}^{\infty}c_k g(- 2\log\left( h/\kappa_{t} \right);\frac{np}{2}+k,2\lambda),\\
F_{H_{t}}(h|\alpha_t) &=&\sum_{k=0}^{\infty}c_k (1-G(- 2\log\left( h/\kappa_{t} \right);\frac{np}{2}+k,2\lambda)),
\end{eqnarray}
respectively, with support $0<h\leq\kappa_{t}$, where $\kappa_t(\alpha_t)$ is defined in Prop. \ref{prop:BFmatrix}, the parameter $0<\lambda<\infty$ is arbitrarily chosen, $g(x;a,b)$ and $G(x;a,b)$ are the pdf and cdf of a gamma distribution with shape and scale parameters $a>0$ and $b>0$, respectively, and the coefficients $c_k$, $d_k$ and $f_k$ satisfy
\begin{eqnarray*}
    c_{k}&=&\exp(-\sum_{j=1}^p U_{jj})\prod_{j=1}^{p}
(\lambda_j/\lambda)^{-n/2}f_k,\quad
    d_k=\frac{n}{2k}\sum_{j=1}^{p}(1-\lambda/\lambda_j)^{k} + \lambda \sum_{j=1}^{p}\frac{U_{jj}}{\lambda_j}(1-\lambda/\lambda_j)^{k-1},\\
    f_{k+1}&=&\frac{1}{k+1}\sum_{j=1}^{k+1}j d_j f_{k+1-j}, k=0,1,2,\ldots,\quad \text{with}\quad  f_0=1,
\end{eqnarray*}
where $U_{jj}$ are the $(j,j)$-th entry of $\boldsymbol{U}=\boldsymbol{Q}'(\boldsymbol{\Sigma}_{H}^{1/2}\tilde{\boldsymbol{\Sigma}}^{-1}(\tilde{\boldsymbol{M}}-M_{\ast})\boldsymbol{V}^{-1})(\tilde{\boldsymbol{M}}-\boldsymbol{M}_{\ast})'\boldsymbol{\Sigma}^{-1/2}_{H}\boldsymbol{Q}$.
\end{proposition}
If in the previous proposition, we set $\tilde{\boldsymbol{M}}= \boldsymbol{M}_{*}$ and $\tilde{\boldsymbol{\Sigma}}= \boldsymbol{\Sigma}_{d}$, that is the location and scale of the marginal likelihood under the null (from i) in Prop. \ref{prop:posteriorBF}) we denote the distribution with $F_{0,t}(h)$. In contrast, if $\tilde{\boldsymbol{M}}= \boldsymbol{M}_{*}$ and $\tilde{\boldsymbol{\Sigma}}= \boldsymbol{\Sigma}_{A,d}$, that is the location and scale of the marginal likelihood under the alternative (from i) in Prop. \ref{prop:posteriorBF}), we denote the distribution with $F_{1,t}(h)$. 

From the properties of the gamma distribution, $f_{H_{t}}(h|\alpha_t)$ is continuous at $0$ and $\kappa_t$. As stated in the following, the result of Prop. \ref{DistribBFmatrix} provides the BF distribution for the univariate Gaussian model given in previous studies \citep[e.g., see][]{weiss1997bayesian,DESANTIS2004121,pawel2025closed}.
\begin{corollary}[Univariate Gaussian model]\label{CorUnivBFDistr}
In the univariate case, i.e. $n=p=1$, following the notation in Remark \ref{ex3}, we set $\boldsymbol{\Sigma}=\sigma^2$, $\boldsymbol{V}=1$ and $\boldsymbol{B}=\theta$. Assuming $\tilde{\boldsymbol{M}}=\boldsymbol{B}$ and $\tilde{\boldsymbol{\Sigma}}=\boldsymbol{\Sigma}_L$ in Prop. \ref{DistribBFmatrix} the distribution of the BF is
\begin{equation}
F_{H_{t}}\left( h|\alpha_{t} \right) =1-\Phi(\sqrt{-2\log(h/\kappa_t)\sigma^2_H/\sigma^2_L}-\sqrt{\gamma})+\Phi(-\sqrt{-2\log(h/\kappa_t)\sigma^2_H/\sigma^2_L}-\sqrt{\gamma}),    
\end{equation}
where $\gamma=(\theta-m_{\ast})^2/\sigma^2$, the variances $\sigma^2_H=\boldsymbol{\Sigma}_H$ and $\sigma^2_L=\boldsymbol{\Sigma}_L$ are given in Prop. \ref{prop:posteriorBF}, the bound $\kappa_t$ is given in Remark \ref{ex3} and $\Phi(\cdot)$ is the cdf of a standard normal distribution.
\end{corollary}


Our perturbation framework for matrix--variate observations is a form of global regularization that affects all entries within the matrix. Later in this paper, we will explore the sensitivity of the testing procedure to variations in the patterns, proportions, and magnitude of outliers. As an extension, local contamination frameworks can be devised to detect patterns within the outliers and potential dependencies among outlying observations. In the context of matrices, a multiplicative perturbation can be employed to define $p_A(\boldsymbol{Y}_t|\mathcal{D}_t)$, assuming, for instance, the perturbed normal distribution $\mathcal{N}_{p,n}(M_{\ast}, \boldsymbol{A}_1\boldsymbol{\Sigma}_d \boldsymbol{A}_1, \boldsymbol{A}_2 \boldsymbol{V} \boldsymbol{A}_2)$, where $\boldsymbol{A}_1$ and $\boldsymbol{A}_2$ are two diagonal matrices scaling along the rows and columns with varying levels. An additive perturbation would lead to a distribution like $\mathcal{N}_{p,n}(\boldsymbol{M}_{\ast},\boldsymbol{\Sigma}_d+\boldsymbol{A}_1,\boldsymbol{V}+\boldsymbol{A}_2)$, with  $\boldsymbol{A}_1$ and $\boldsymbol{A}_2$ shifting the posterior's covariance matrix diagonal. On the one hand, local contamination allows for patterns among outliers; however, it requires the identification of these patterns, which in turn calls for inference on the contamination parameters. Therefore, alternative approaches, such as multiple-outlier models presented in \cite{page2011bayesian,tomarchio2022mixtures}, might be preferable. This typically involves introducing appropriate prior specifications and has the drawback of requiring computationally intensive procedures for posterior approximation. We will postpone this extension to future research.

\begin{figure}[t]
    \centering
    \setlength{\tabcolsep}{2pt}
    \renewcommand{\arraystretch}{1.2}
    \begin{tabular}{ccc}
    \includegraphics[scale=0.3]{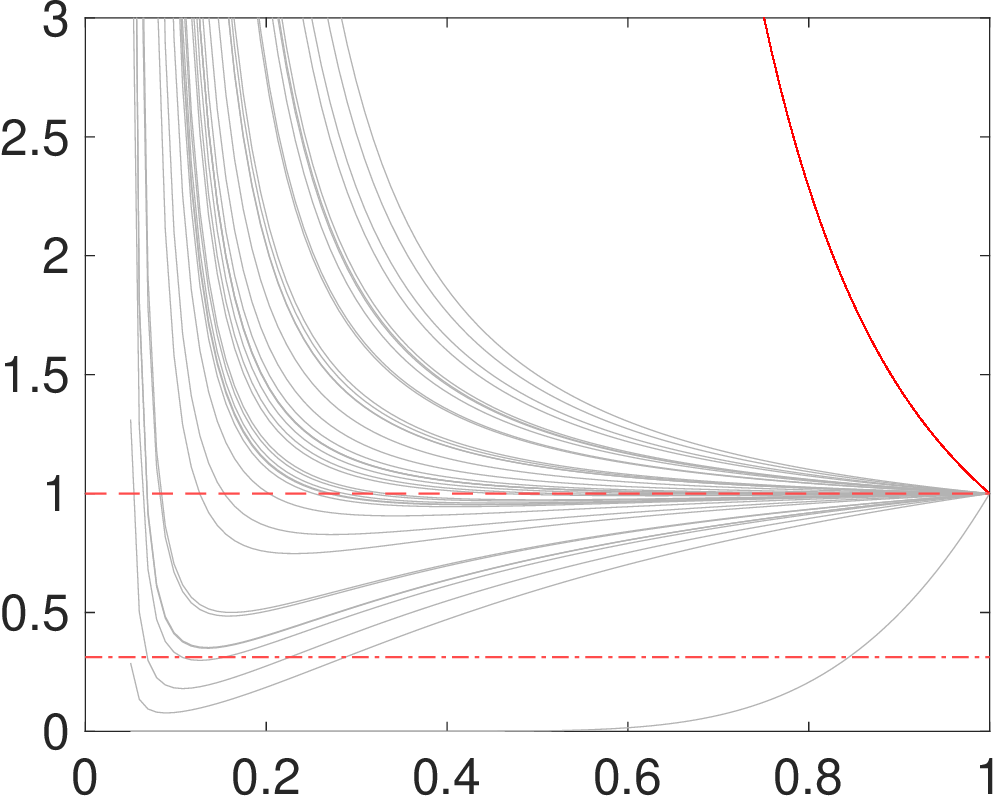}&
    \includegraphics[scale=0.3]{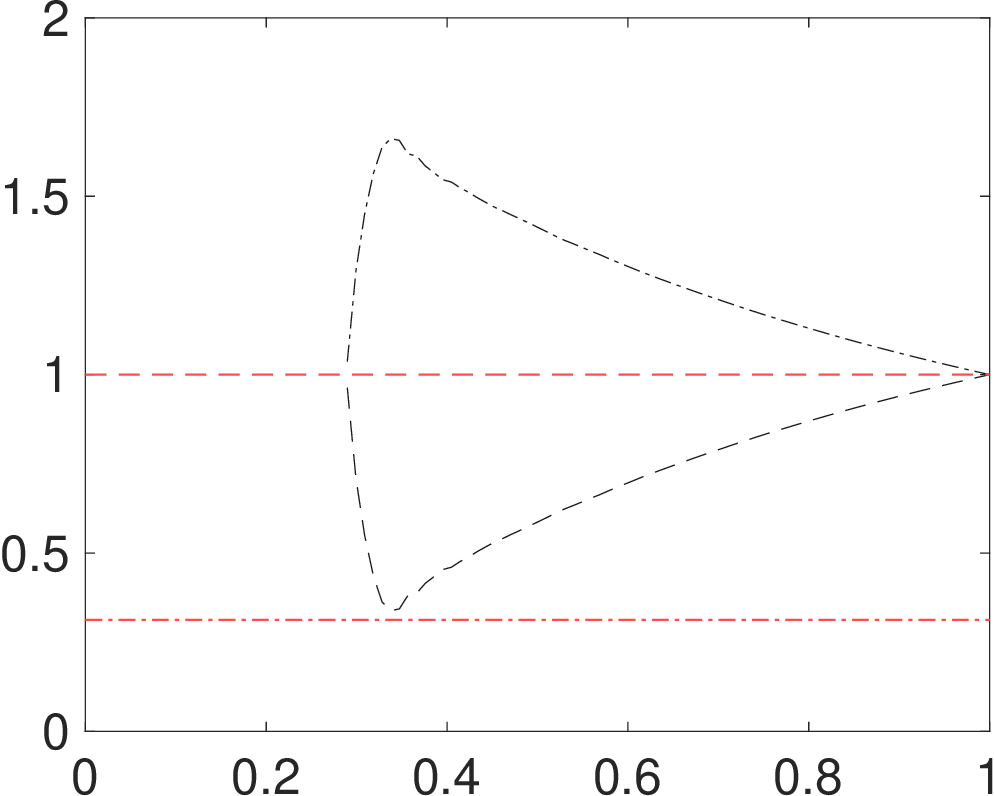}&        \includegraphics[scale=0.3]{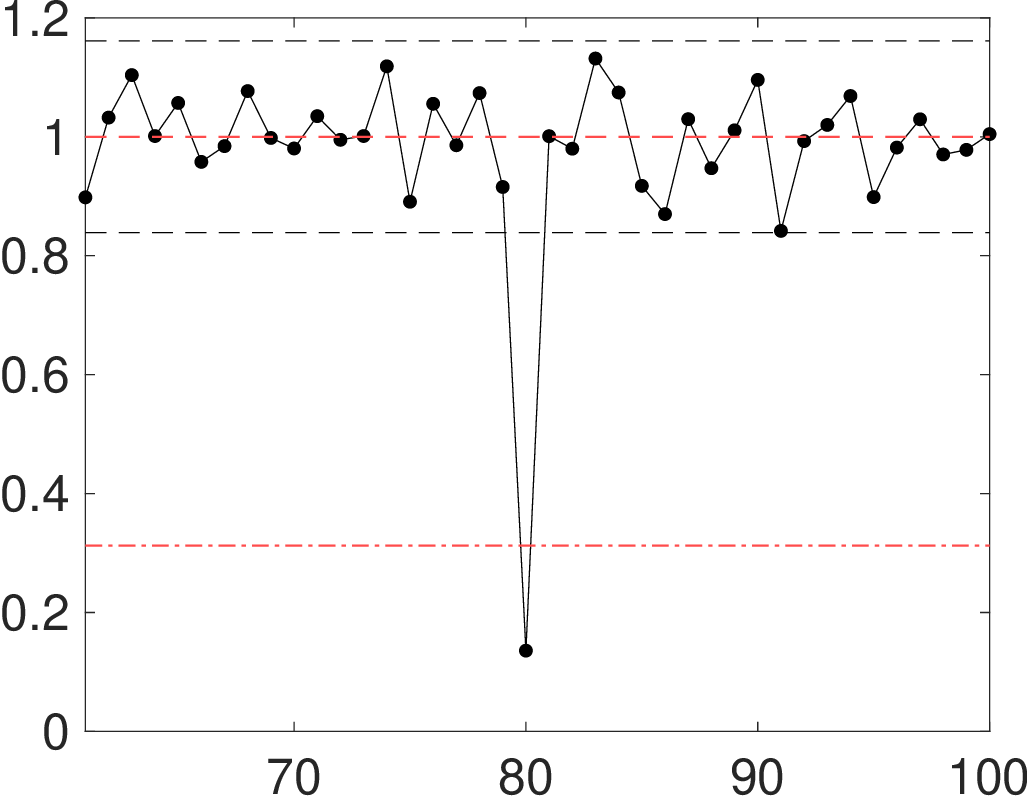}\\
    \includegraphics[scale=0.3]{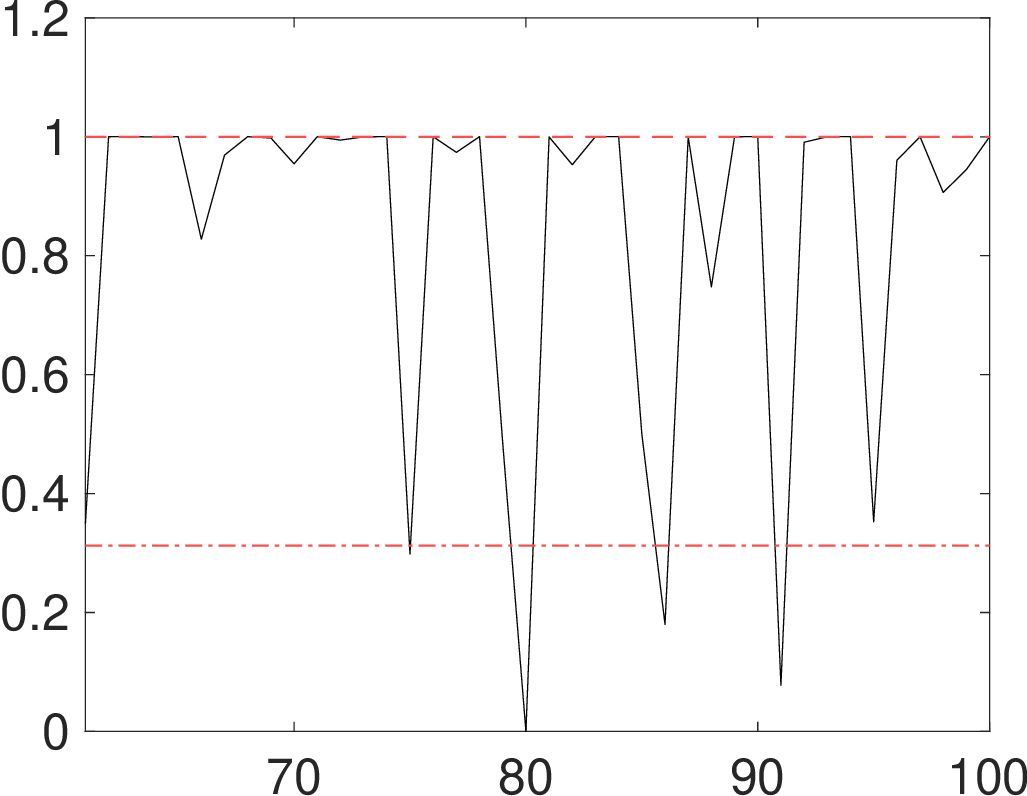}&
    \includegraphics[scale=0.3]{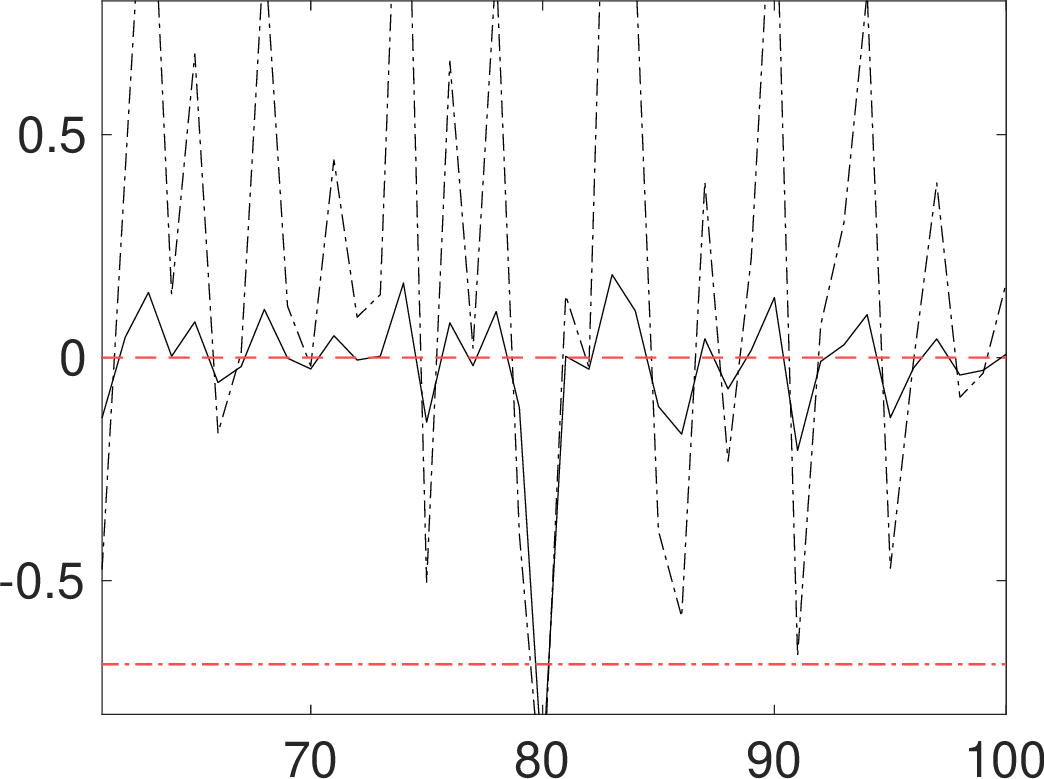}&
    \includegraphics[scale=0.3]{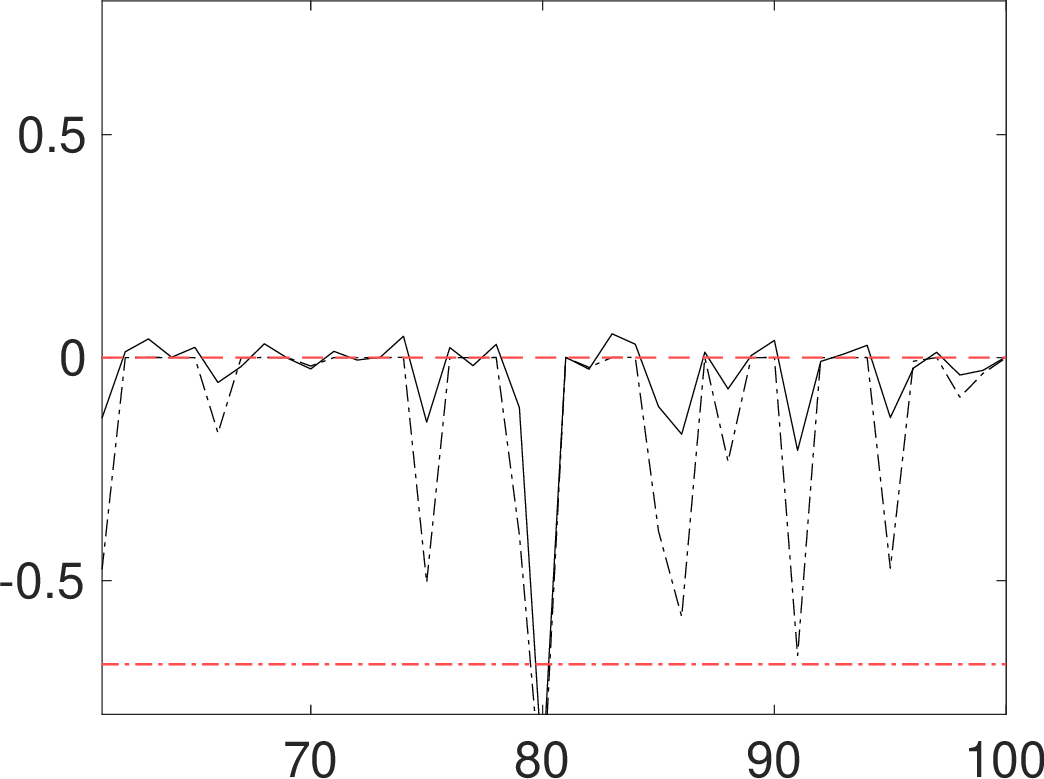}\\
\end{tabular}
        \caption{Simulated dataset. Top: the BFs $H_t(\alpha_t)$ and their upper bounds $\kappa_t(\alpha_t)$ as a function of $\alpha_t$ for different samples (left, {\color{gray}\protect\tikz[baseline]{\protect\draw[line width=0.2mm] (0,.6ex)--++(0.5,0) ;}}, {\color{red}\protect\tikz[baseline]{\protect\draw[line width=0.2mm] (0,.6ex)--++(0.5,0) ;}}); the BF lower and upper thresholds $\underline{h}(\alpha_t)$ and $\bar{h}(\alpha_t)$ (middle, {\color{red}\protect\tikz[baseline]{\protect\draw[line width=0.2mm, black] (0,.6ex)--++(0.5,0) ;}}) as a function of $\alpha_t$ and the BFs for the calibrated discount factor (right, $\bullet$) together with the calibrated lower and upper bounds $\underline{h}(\alpha^*)$ and $\bar{h}(\alpha^*)$ ({\color{black}\protect\tikz[baseline]{\protect\draw[line width=0.2mm, dashed] (0,.6ex)--++(0.5,0) ;}}). Bottom: Minimum BF (left), Integrated BF (middle) and Normalized Integrated BF (right) for different beta distribution settings ({\color{black}\protect\tikz[baseline]{\protect\draw[line width=0.2mm] (0,.6ex)--++(0.5,0) ;}},{\color{black}\protect\tikz[baseline]{\protect\draw[line width=0.2mm, dashed] (0,.6ex)--++(0.5,0) ;}}). In all plots, the reference lines at 1 ({\color{red}\protect\tikz[baseline]{\protect\draw[line width=0.2mm, dashed] (0,.6ex)--++(0.5,0) ;}}) and $10^{-1/2}$ ({\color{red}\protect\tikz[baseline]{\protect\draw[line width=0.2mm, dashdotted] (0,.6ex)--++(0.5,0) ;}}).}
    \label{emp1:figBFsimulation}
\end{figure}

\subsection{Simulation study}
In the following simulation exercise, we generate data with outliers and study some finite-sample properties of the decision procedures based on standard BF, MBF, IBF and NIBF, and the calibrated predictive BF approach we introduced. All analyses were implemented in Matlab 2023a and carried out on a 12-core computing system with 128GB RAM. The parameter values have been randomly generated as follows: 
$\boldsymbol{M}\sim 
\mathcal{N}_{p,n}(\boldsymbol{O}_{p\times n}, \boldsymbol{I}_{p},\boldsymbol{I}_{n})$, $\boldsymbol{\Sigma}=\boldsymbol{SS}'$, $\boldsymbol{S}\sim\mathcal{N}_{p,p}(\boldsymbol{O}_{p\times p}, \boldsymbol{I}_{p},\boldsymbol{I}_{p})$, $\boldsymbol{\Psi}=\boldsymbol{GG}'$, $G\sim\mathcal{N}_{n,n}(\boldsymbol{O}_{n\times n}, \boldsymbol{I}_{n}, \boldsymbol{I}_{n})$. We consider two scenarios: the case of moderate-size observation matrix, denoted as $Case_1$, with $p = 30$ and $n = 10$, and the case of a large-size observation matrix, $Case_2$, with  $p = 50$ and $n = 50$. We generate the synthetic data with an outlier as follows. First a noise sequence is generated, $\boldsymbol{E}_{t}\sim\mathcal{N}_{p,n}(\boldsymbol{O}_{p\times p},\boldsymbol{\Sigma},\boldsymbol{\Psi})$ i.i.d. $t=1,\ldots,100$. Secondly, the observable sequence is defined as $\boldsymbol{X}_{t}=\boldsymbol{M}+\boldsymbol{E}_t$ for $t\neq 80$ and $\boldsymbol{X}_{t}=\boldsymbol{M}+u \boldsymbol{R}_t+\boldsymbol{E}_t$ for $t=80$, where $\boldsymbol{R}_t$ is a binary matrix encoding the position of the outliers and $u$ is 0.5, 1, 1.5, 3, 5 and 15 that are $1/30$, $1/15$, $1/10$, $1/5$, $1/3$  and $1$ average standard deviation of the matrix normal. 

Our simulation study considered various settings obtained by varying the configurations of the matrix $\boldsymbol{R}_t$ and the magnitude of the outliers. For each setting $J$ independent datasets have been generated $\boldsymbol{X}_{t}^{(j)}$, $t=1,\ldots,T$ for $j=1,\ldots,J$, with $T=100$ and the BFs $H_t^{(j)}$ computed. The probabilities $p_{I}=P(H_t>\bar{h})$, $p_{II}=P(H_t<\underline{h})$ and  $p_{III}=P(\bar{h}<H_t<\underline{h})$ have been estimated as follows
\begin{equation*}
    p_{I}^{J}=\frac{1}{J}\sum_{j=1}^{J} \mathbb{I}(H_t^{(n)}>\bar{h}),\,\,\,\, p_{II}^{J}=\frac{1}{J}\sum_{j=1}^{J} \mathbb{I}(\underline{h}<H_t^{(n)}<\bar{h}),\,\,\,\, p_{III}^{J}=\frac{1}{J}\sum_{j=1}^{J} \mathbb{I}(H_t^{(n)}<\underline{h})
\end{equation*}
under the null hypothesis of the absence of outliers at $t\neq 80$ and under the alternative hypothesis of a certain number of outlying observations in the observation matrix at $t=80$. A summary of the results for $Case_1$ and $Case_2$ is provided in Tab. \ref{tab:SimResults} and Tab. \ref{tab:SimResults1} of Appendix \ref{sec:Simulation_appendix} respectively. The results are obtained with $J=100$ experiments for the two outlier settings. The two settings differ for the positions of the outliers within the matrix $R_t$: i) row and column patterns in the positions, panel (a); and ii) completely random positions, panel (b). Appendix \ref{sec:BAYSFWATCH} provides reference to an R package we specifically developed for outlier detection.

Figure \ref{emp1:figBFsimulation} illustrates the result of the BF procedure on one of the simulated datasets for $u=0.5$ (top left panel a in Tab. \ref{tab:SimResults}). In the BF procedure the calibrated value of $\alpha_t$ for $\tau=0.01$ and $\beta=0.8$ is $\alpha^{*}=0.750$ for all $t$ and the inconclusive interval $C_{H}$ has lower and upper bounds $\underline{h}(\alpha^{*})=0.839$ $\bar{h}(\alpha^*)=1.161$, respectively. Starting from the left in the first row, we notice that the observations with a change in the convexity of $H_t(\alpha_t)$ are recognized as outliers based on the value of $\alpha_t$. Notably, for the outlier observation introduced in the simulation, the BF is consistently convex and remains far below 1. The second plot illustrates the behavior of the thresholds $\underline{h}(\alpha_t)$ and $\bar{h}(\alpha_t)$ defined at the end of Section \ref{sec:bayesian} for given power and size values. The length of the randomized decision interval $C_h$ (vertical axis) reduces as $\alpha_t$ goes to 1. The third plot presents the BF and the inconclusive region for $\alpha^{*}=0.750$. The BF of the outlying observation is located outside the inconclusive interval, whereas about 16 observations exhibit a BF below 1 within the inconclusive interval. Moving to the second row, the Minimum BF suggests positive evidence against the null hypothesis of the absence of outliers ($MBF<10^{-1/2}$ following Jeffrey's scale of evidence) for three observations. The IBF and NIBF (solid line in the second and third plot at the bottom) provide strong evidence against the null ($IBF$, $NIBF<10^{-1/2}-1\approx -0.687$) only for the 80th observation and barely worth mentioning evidence ($IBF$, $NIBF>10^{-1/2}-1\approx -0.687$) for the other observations. Nevertheless, comparing the solid and dashed lines shows that the outcome of the IBF and NIBF procedures strongly depends on the choice of the hyperparameters of the beta distribution. Thus, the third robust method should be applied, exploiting the BF's sampling variability to find the calibrated $\alpha_t$ and the reference thresholds. 

In all the experiments, when data are generated under the null, the type I error probability $P(H_t<\underline{h}|\mathcal{H}_0)$ is about 2\%, whereas when data are generated under the alternative hypothesis (first column), the power $P(H_t<\underline{h}|\mathcal{H}_1)$ of the test gets close to one, increasing the number of outliers (e.g., see columns of the panels (b) in Tab. \ref{tab:SimResults}). As we can expect, the convergence is faster for larger number of observations (see panels (b) in Tab. \ref{tab:SimResults1}). Also, the effective size and power may depend on the position of outliers in the rows or columns of the matrix $R_t$. The presence of patterns in the position of the outliers reduces the power compared to the case of a completely random position within the matrix $R_t$ (compare columns $20\times 10$ and $200$ in panels (a) and (b), respectively). The power decreases below 80\% for small outlier amplitude (e.g. $1/10$ standard deviation) and a small number of outlying observations (e.g. 10 out of 300 elements) within the matrix. When all entries are outliers, the power of the test converges to 1, increasing the outliers' magnitude (different rows in the last column in Tab. \ref{tab:SimResults}).

\section{Empirical Illustration}\label{sec:empirical}
We illustrate our sequential matrix outlier detection on three relevant benchmark datasets. See Appendix \ref{sec:datasets} for a detailed description of the datasets and a discussion of the outliers' dating.

\subsection{Inflation and Unemployment Dataset}

The Inflation and Unemployment Dataset (\citealp{Can09}) spans a period from January 2002 to October 2022 ($T =250$) and consists of a sequence of $p\times n$ matrices, covering $p = 11$ EU countries and $n=3$ macroeconomic variables (the Industrial Production Index, the Price Index and the Unemployment Rate,). The total number of observations is $8,250$, representing an example of big data in this field.

 In the testing procedure with BF, we consider rolling windows of $w$ observations each and assumed $\boldsymbol{\Sigma}_P=\boldsymbol{\Sigma}_L /\varphi$, with $\varphi=w$. For each window, the prior mean $\boldsymbol{M}$ is set equal to the posterior mean $M_{\ast}$ of the previous window. The variance-covariance $\boldsymbol{\Sigma}_L$ has been estimated by Least Squares.

The top plots in Figure \ref{emp1:figBFdates} display the BF $H_t(\alpha_t)$ (solid grey) and the upper bound of the BF $\kappa_t(\alpha_t)$ (solid red) as functions of $\alpha_t$ for various dates (each represented by a different line). The findings in the top-left plot suggest that it is crucial to compare the BF to the upper bound. The BF tends to hover around one for higher values of $\alpha_t$, with some exceptions, and it significantly exceeds one only for small values of $\alpha_t$ (illustrated by the dark grey lines). It frequently intersects the threshold at lower values of $\alpha_t$ (light grey) as well as for higher values. Additionally, there are instances where it crosses the threshold intermittently.

The left-bottom plot offers a different illustration of the effect of $\alpha_t$ on the outcome of the sequential outlier detection for the entire sample from September 2018 to October 2022. 
\begin{figure}[t]
    \centering
    \setlength{\tabcolsep}{2pt}
    \renewcommand{\arraystretch}{1.2}
    \begin{tabular}{ccc}
     {\footnotesize (i) Inflation and Unemployment} & {\footnotesize (ii) International Trade} &{\footnotesize  (iii) Volatility Network}\\
    \includegraphics[scale=0.29]{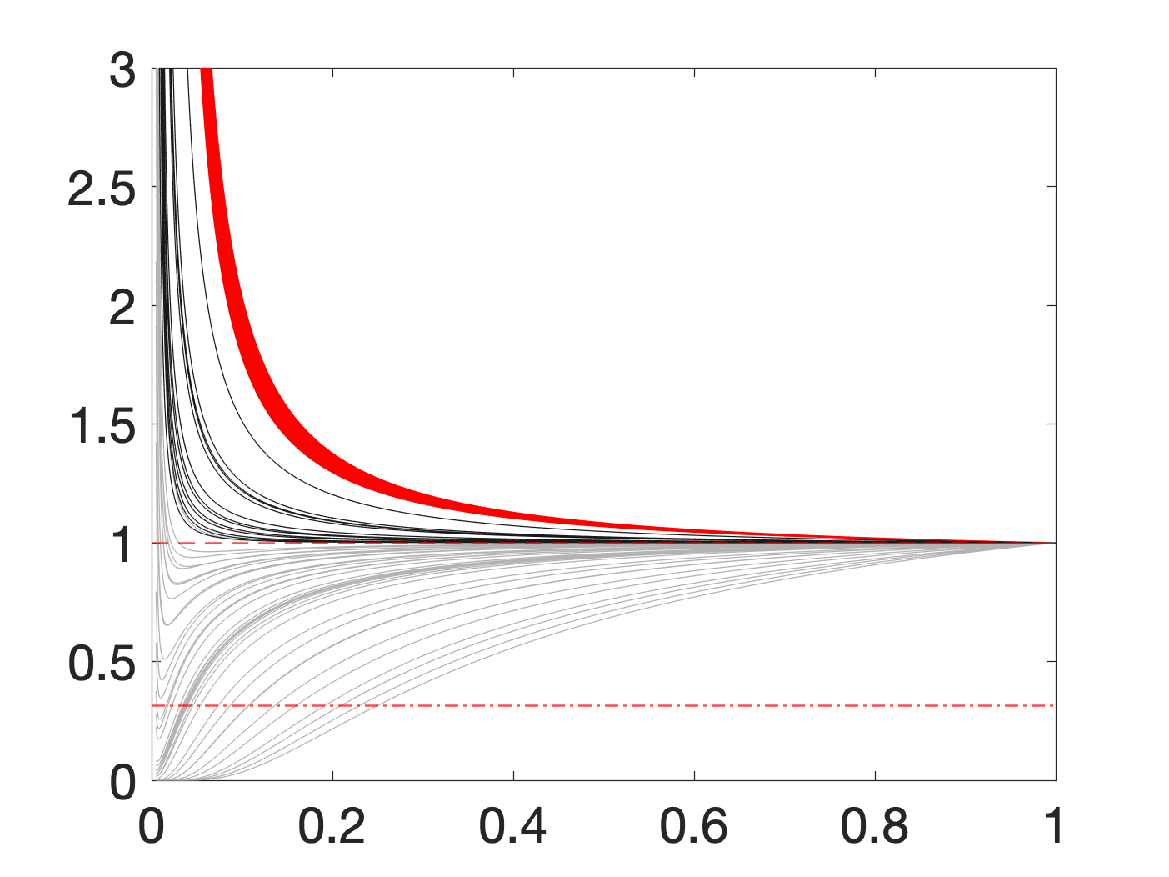}&
    \includegraphics[scale=0.30]{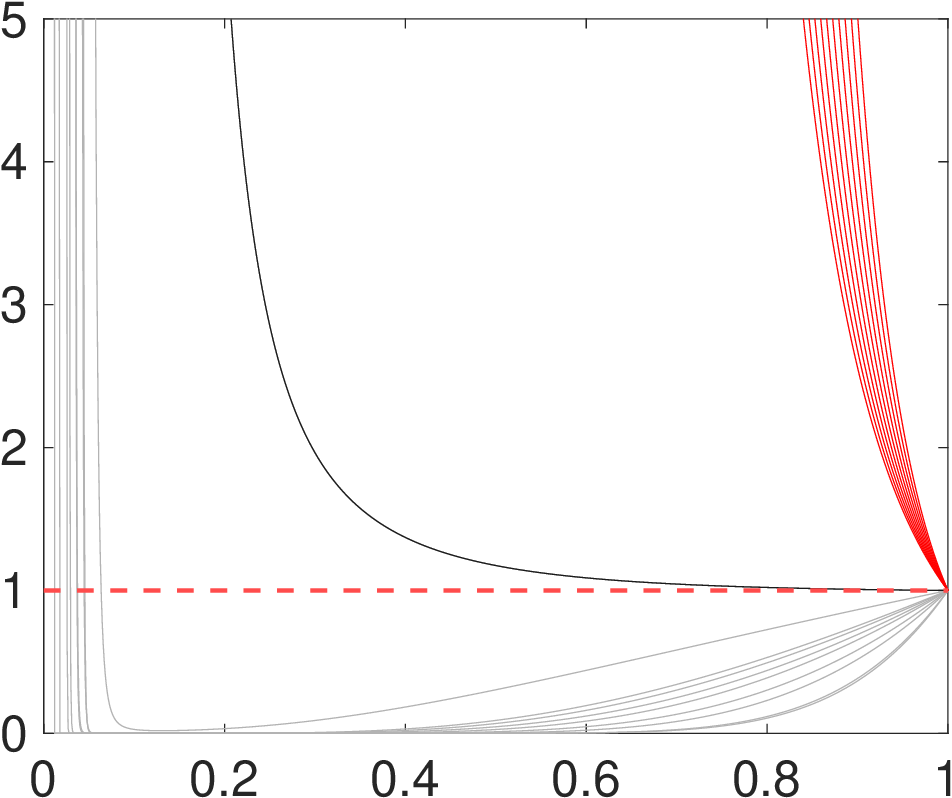}&
    \includegraphics[scale=0.29]{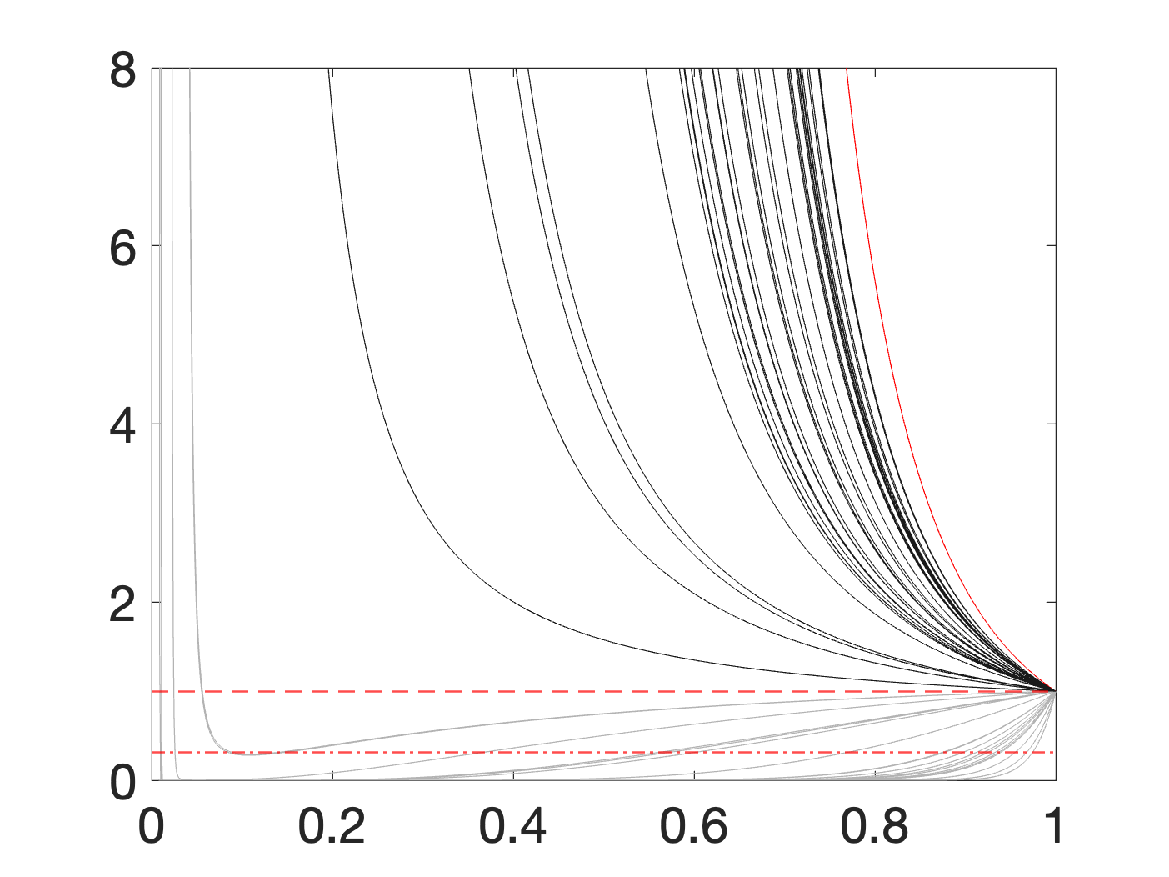}\\
    \includegraphics[scale=0.30]{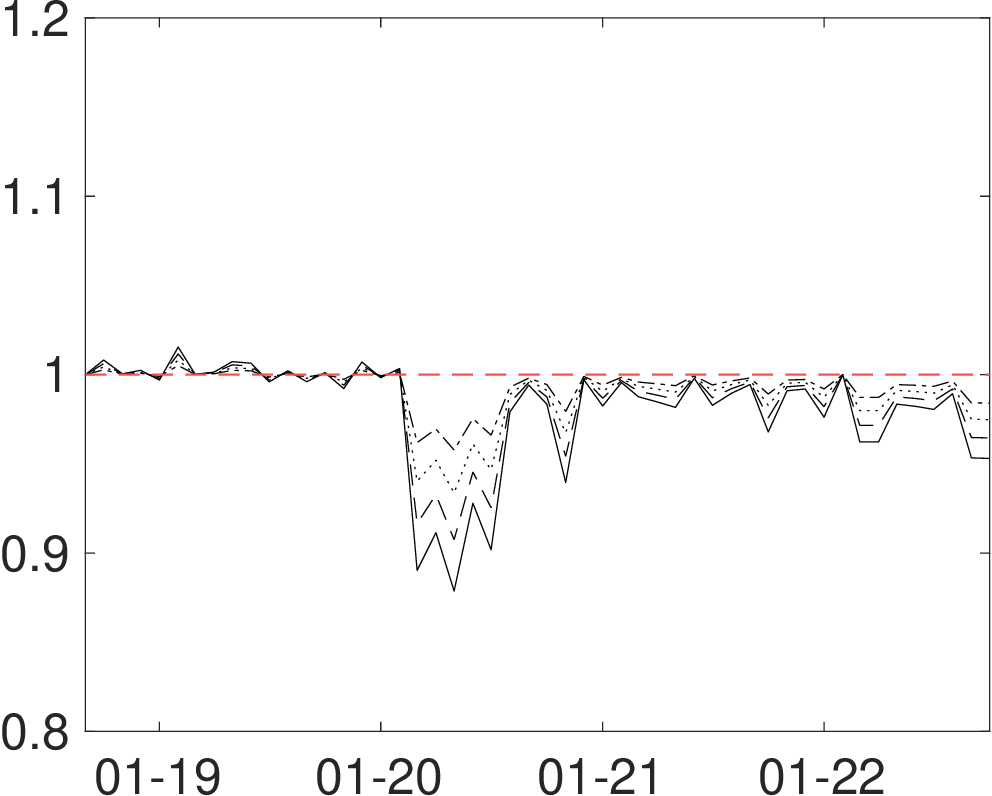}&    \includegraphics[scale=0.30]{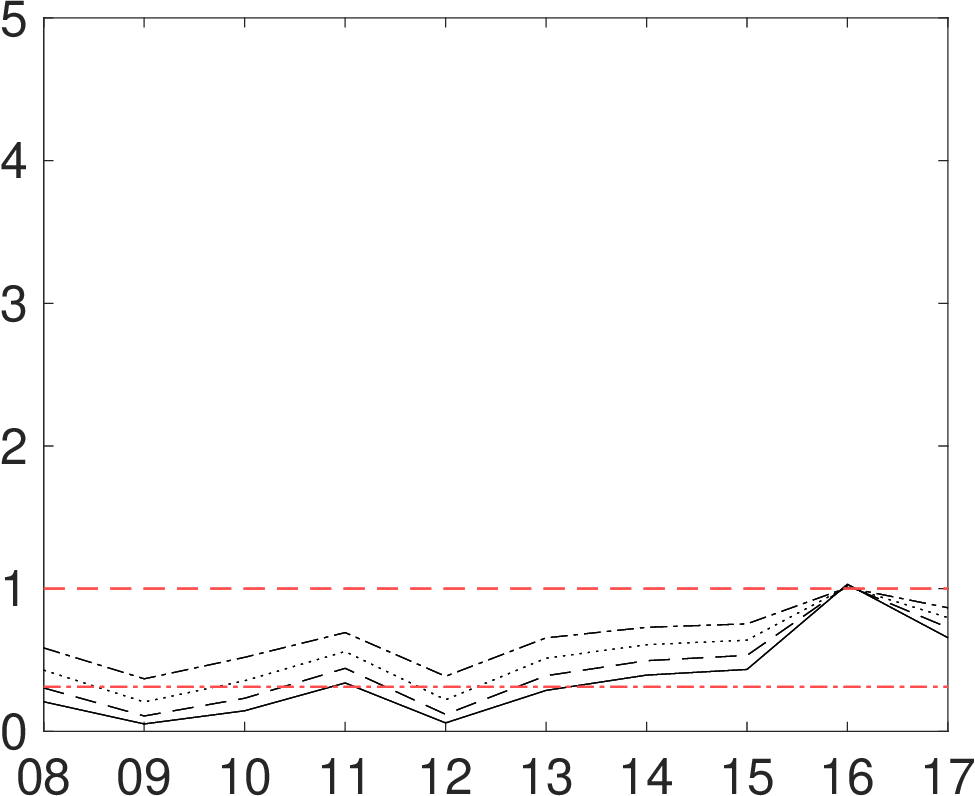}&
    \includegraphics[scale=0.30]{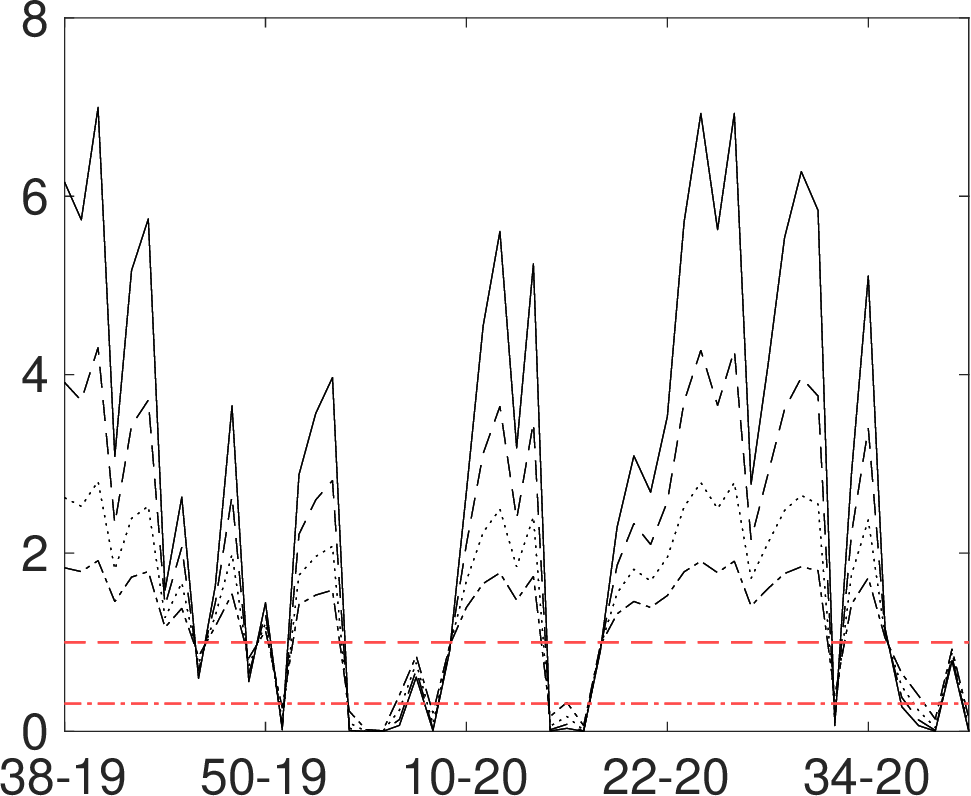}
    \end{tabular}
        \caption{BF for the three datasets: Inflation and Unemployment (left), International Trade (middle), and Volatility Network (right). Top: the BF $H_t(\alpha_t)$ and its upper bound $\kappa_t(\alpha_t)$ as functions of $\alpha_t$  at different dates ({\color{gray}\protect\tikz[baseline]{\protect\draw[line width=0.2mm] (0,.6ex)--++(0.5,0) ;}} and {\color{red}\protect\tikz[baseline]{\protect\draw[line width=0.2mm] (0,.6ex)--++(0.5,0) ;}}, respectively). Bottom: the BF over time for $\alpha_t=0.75$, $0.80$, $0.85$ and $0.90$ ({\color{black}\protect\tikz[baseline]{\protect\draw[line width=0.2mm] (0,.6ex)--++(0.5,0) ;}}, {\color{black}\protect\tikz[baseline]{\protect\draw[line width=0.2mm, dashed] (0,.6ex)--++(0.5,0) ;}}, {\color{black}\protect\tikz[baseline]{\protect\draw[line width=0.2mm, dotted] (0,.6ex)--++(0.5,0) ;}} and {\color{black}\protect\tikz[baseline]{\protect\draw[line width=0.2mm, dashdotted] (0,.6ex)--++(0.5,0) ;}}, respectively). In all plots, the reference lines at 1 ({\color{red}\protect\tikz[baseline]{\protect\draw[line width=0.2mm, dashed] (0,.6ex)--++(0.5,0) ;}}) and $10^{-1/2}$ ({\color{red}\protect\tikz[baseline]{\protect\draw[line width=0.2mm, dashdotted] (0,.6ex)--++(0.5,0) ;}}).}
    \label{emp1:figBFdates}
\end{figure}

In all settings, an outlier was detected in March 2020 (i.e. at the pandemic outbreak). For $\alpha_t=0.054$ (solid line), a sequence of outliers is detected before February 2020 with BF far from 1, whereas, for $\alpha_t$ equal to 0.402, 0.801 and 0.851 (dashed, dotted and dashed-dotted, respectively), the BF is close to one before March. Qualitatively speaking, after the outbreak, for $\alpha_t=0.402$, the BF is close to one after December 2021, whereas larger $\alpha_t$ values return BF close to just one after March 2020. 
For illustrative purposes, we present in detail the results for some relevant dates (Figure \ref{emp1:figBFdates} in the Appendix). For the observations in February 2020, during the pandemic outbreak, the null hypothesis of the absence of an outlier is not rejected for any choice of $\alpha_t$. Nevertheless, the BF is close to one for large values of $\alpha_t$. In contrast, in March 2020, there is strong evidence of an outlier. On the other dates in the figure, the alternative hypothesis is accepted for some values of $\alpha_t$. The dashed vertical lines indicate the stationary point. The stationary point is not near zero on some dates, such as February 2022. On other dates, the stationary point is near zero, and the BF is far below one for a large part of the $\alpha_t$ values (e.g., March 2020, March 2022 and October 2022). The top-left plot in Figure \ref{emp1:figBFdatesIntegratedBF} shows the results of the alternative procedures. 

The MBF and IBFs provide clearer identification of outlying observations, enhancing the standard BF and supporting the outcome of the testing procedure endowed with the inconclusive interval given in the first line of Figure \ref{emp1:figBFdatesThresholds}. Dashed lines indicate the lower and upper bounds of the inconclusive region, the dots show the value of the BF. For $\alpha_t=0.75$ and a test size of 1\%, the power is approximately 96\%. 


The outcome of the sequential test is in line with those of classical frequentist tests for outliers such as the Grubb's (G) test \citep{Grubb50} and the Generalized Extreme Studentized Deviate (GESD) test \cite{rosner1983percentage}. The second row of Figure \ref{emp1:figBFdatesThresholds} reports the number of outliers detected at the 1\% level by applying the two tests element-wise to each entry of the observation matrices.

\begin{figure}[h!]
    \centering
     \setlength{\tabcolsep}{2pt}
     \renewcommand{\arraystretch}{1.2}
    \begin{tabular}{ccc}
    \multicolumn{3}{c}{(a) Minimum BF}\vspace{5pt}\\
     {\footnotesize (i) Inflation and Unemployment} & {\footnotesize (ii) International Trade} &{\footnotesize  (iii) Volatility Network}\\
    \includegraphics[scale=0.3]{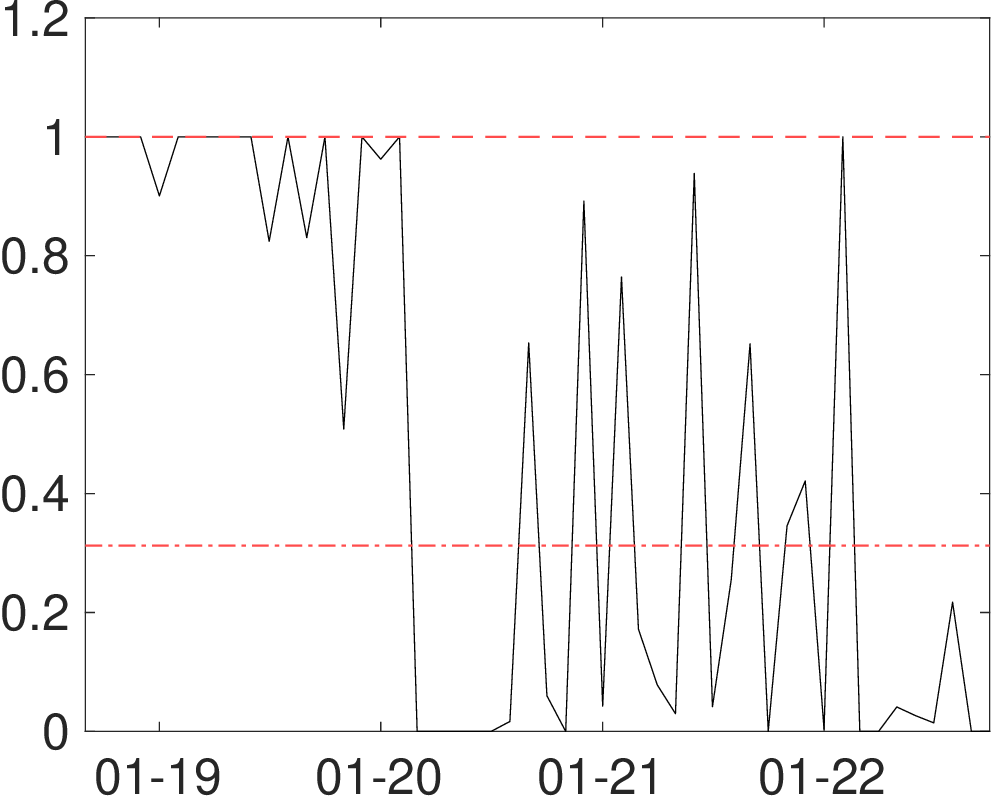}&
    \includegraphics[scale=0.30]{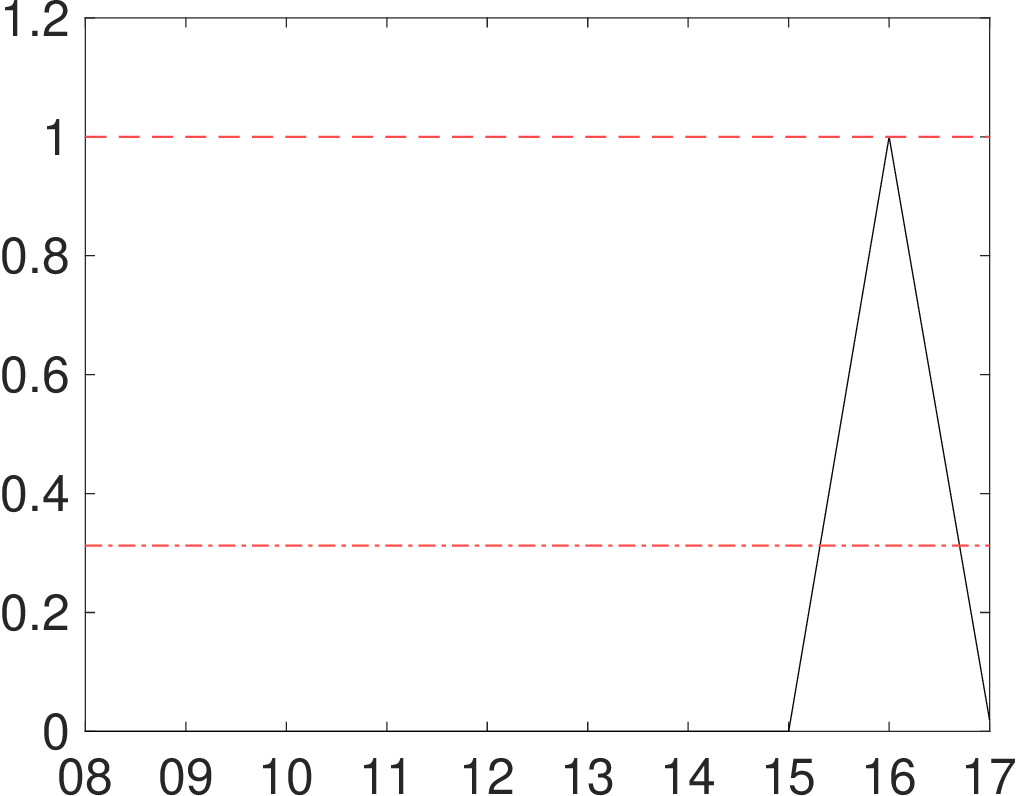}&
    \includegraphics[scale=0.30]{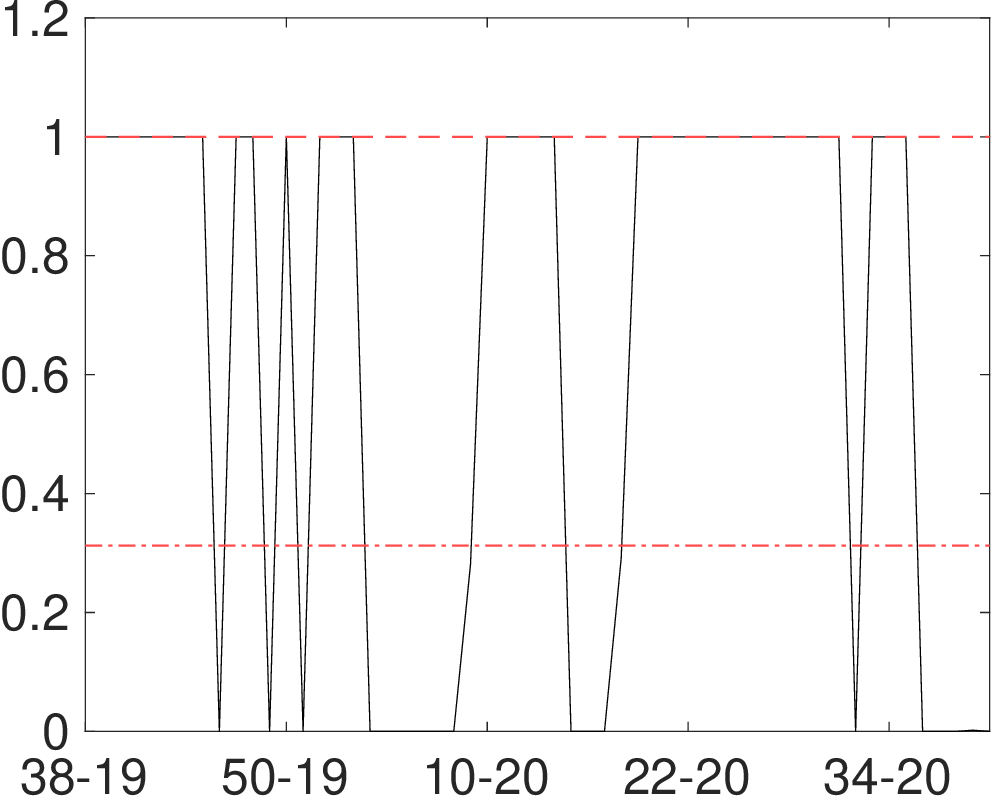}\\
    \multicolumn{3}{c}{(b) Integrated BF}\vspace{5pt}\\
     {\footnotesize (i) Inflation and Unemployment} & {\footnotesize (ii) International Trade} &{\footnotesize  (iii) Volatility Network}\\
   \includegraphics[scale=0.3]{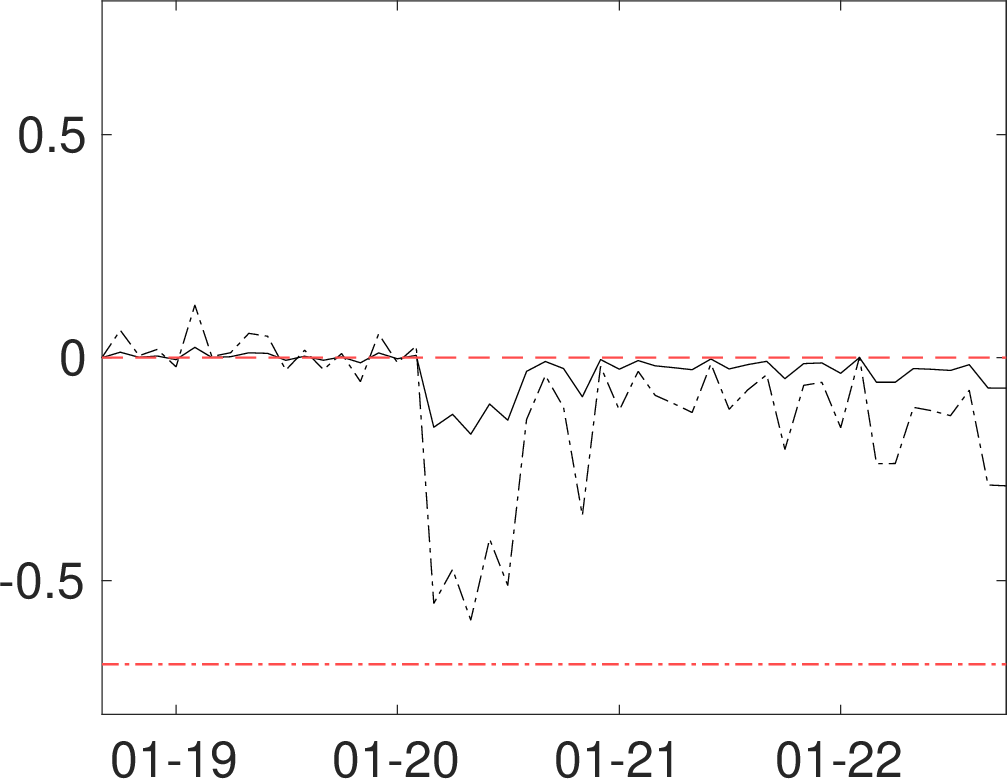}&
    \includegraphics[scale=0.3]{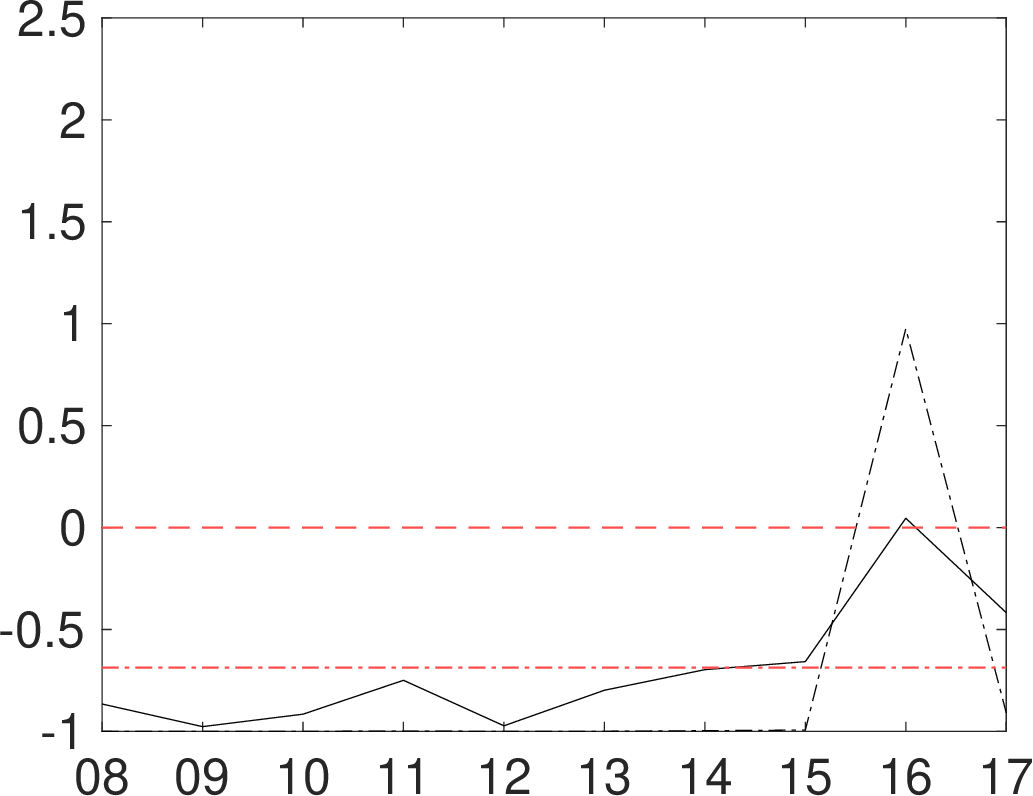}&
    \includegraphics[scale=0.3]{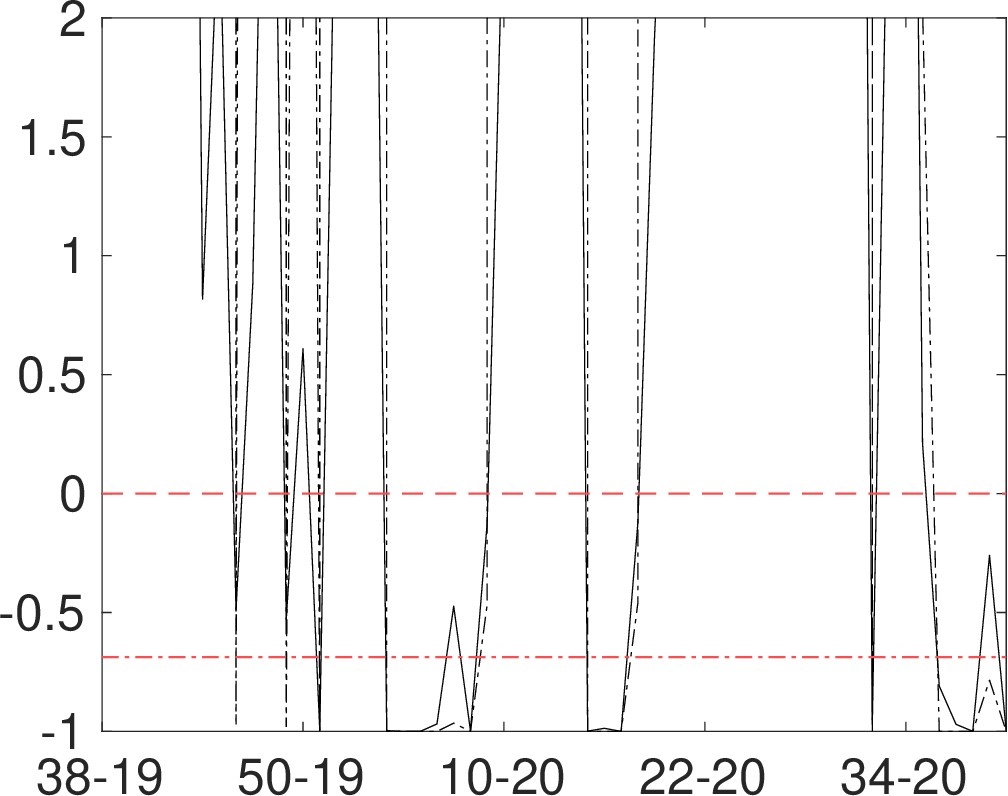}\\    
   \multicolumn{3}{c}{(c) Normalized Integrated BF}\vspace{5pt}\\
     {\footnotesize (i) Inflation and Unemployment} & {\footnotesize (ii) International Trade} &{\footnotesize  (iii) Volatility Network}\\
    \includegraphics[scale=0.3]{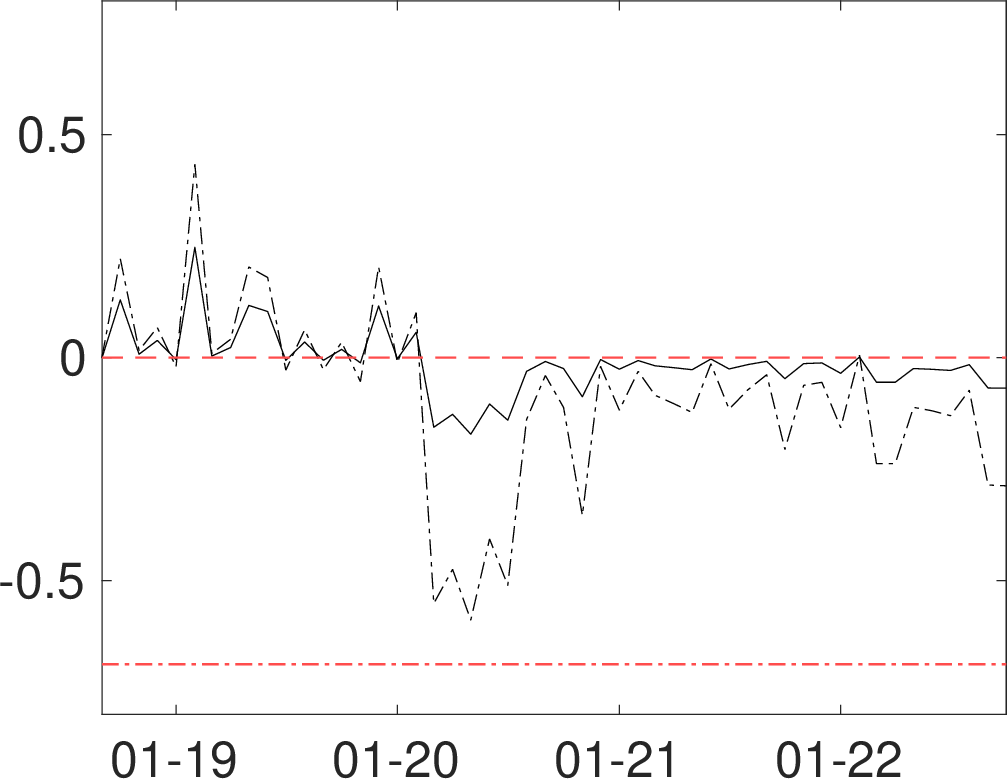}&
    \includegraphics[scale=0.3]{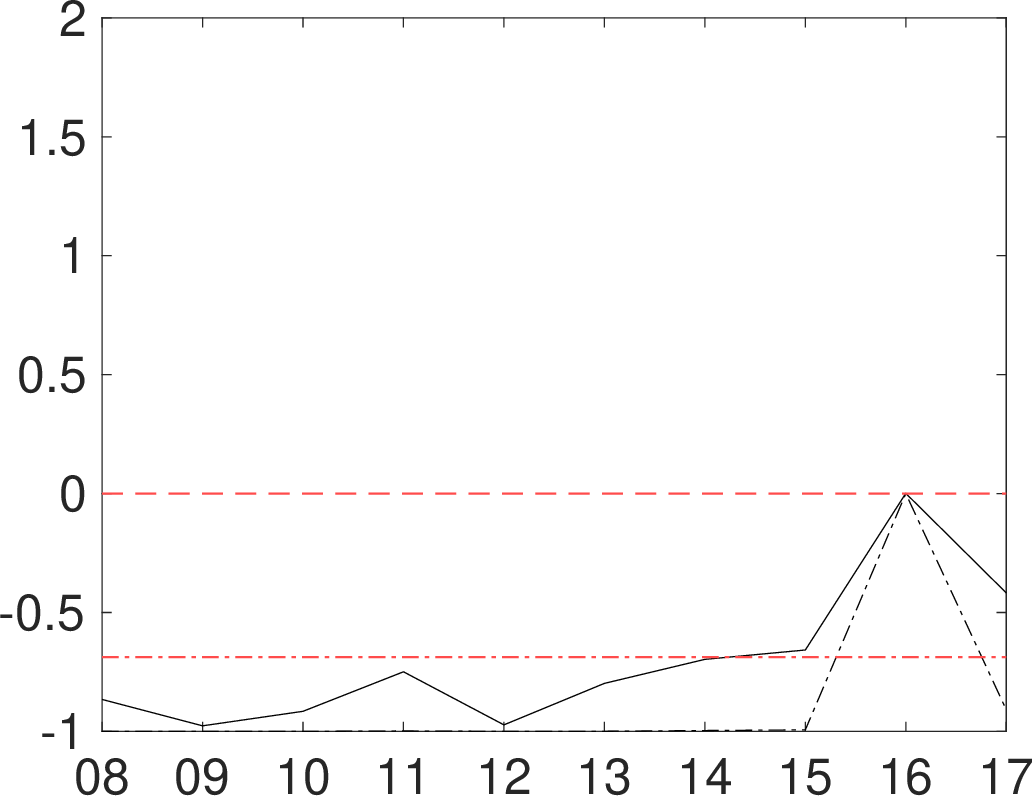}&
    \includegraphics[scale=0.3]{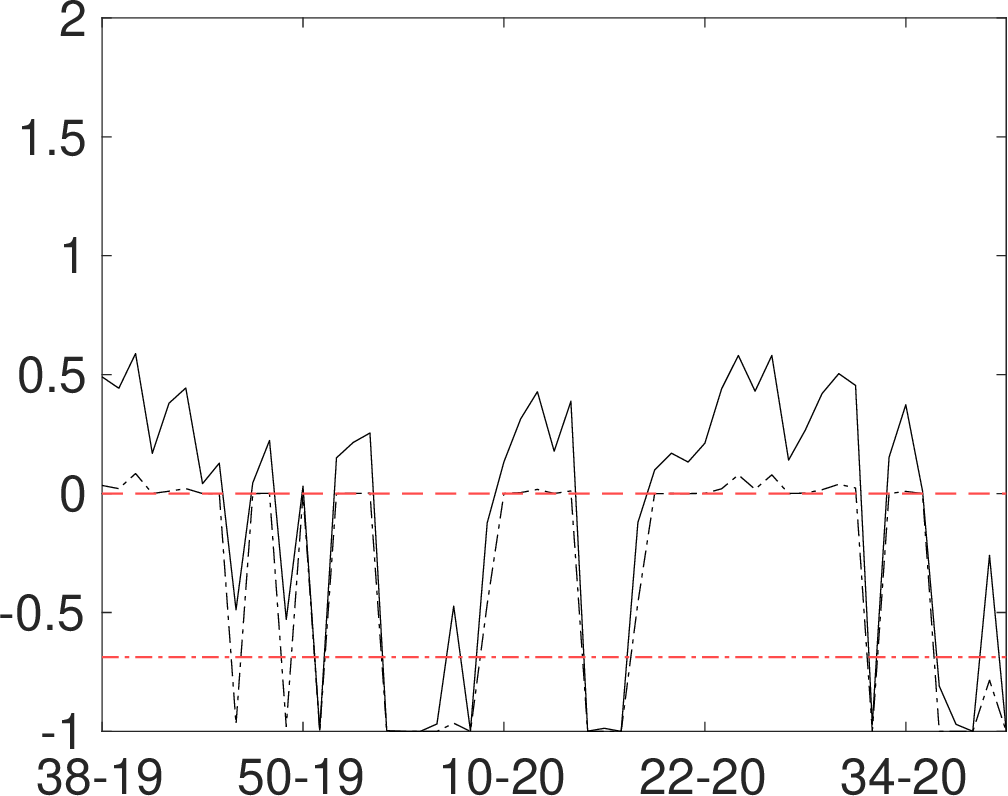}    \end{tabular}
        \caption{Bayes Factors for the three datasets: Inflation and Unemployment (left), International Trade (middle) and Volatility Network (right). The discount factor distribution is $\alpha_t\sim\mathcal{B}e(a,b)$ with the following parameter settings:  $a=16.5001$ and $b=7.6429$ (left, {\color{black}\protect\tikz[baseline]{\protect\draw[line width=0.2mm] (0,.6ex)--++(0.5,0) ;}}) and $b=37.1667$ (left, {\color{black}\protect\tikz[baseline]{\protect\draw[line width=0.2mm, dashed] (0,.6ex)--++(0.5,0) ;}}); $a=364.5001$ and $b=156.7857$ (mid, {\color{black}\protect\tikz[baseline]{\protect\draw[line width=0.2mm] (0,.6ex)--++(0.5,0) ;}}) and $b=849.1667$ (mid, {\color{black}\protect\tikz[baseline]{\protect\draw[line width=0.2mm, dashed] (0,.6ex)--++(0.5,0) ;}}); $a=1250.5001$ and $b=536.2857$ (right, {\color{black}\protect\tikz[baseline]{\protect\draw[line width=0.2mm] (0,.6ex)--++(0.5,0) ;}}) and $b=2915.3000$ (right, {\color{black}\protect\tikz[baseline]{\protect\draw[line width=0.2mm, dashed] (0,.6ex)--++(0.5,0) ;}}). In all plots, the reference lines at 1 ({\color{red}\protect\tikz[baseline]{\protect\draw[line width=0.2mm, dashed] (0,.6ex)--++(0.5,0) ;}}) and $10^{-1/2}$ ({\color{red}\protect\tikz[baseline]{\protect\draw[line width=0.2mm, dashdotted] (0,.6ex)--++(0.5,0) ;}}).}
    \label{emp1:figBFdatesIntegratedBF}
\end{figure}

\begin{figure}[t]
    \centering
    \setlength{\tabcolsep}{3pt}
    \renewcommand{\arraystretch}{1.2}
    \begin{tabular}{ccc}
   \multicolumn{3}{c}{(a) BF with thresholds}\vspace{5pt}\\
    \includegraphics[scale=0.3]{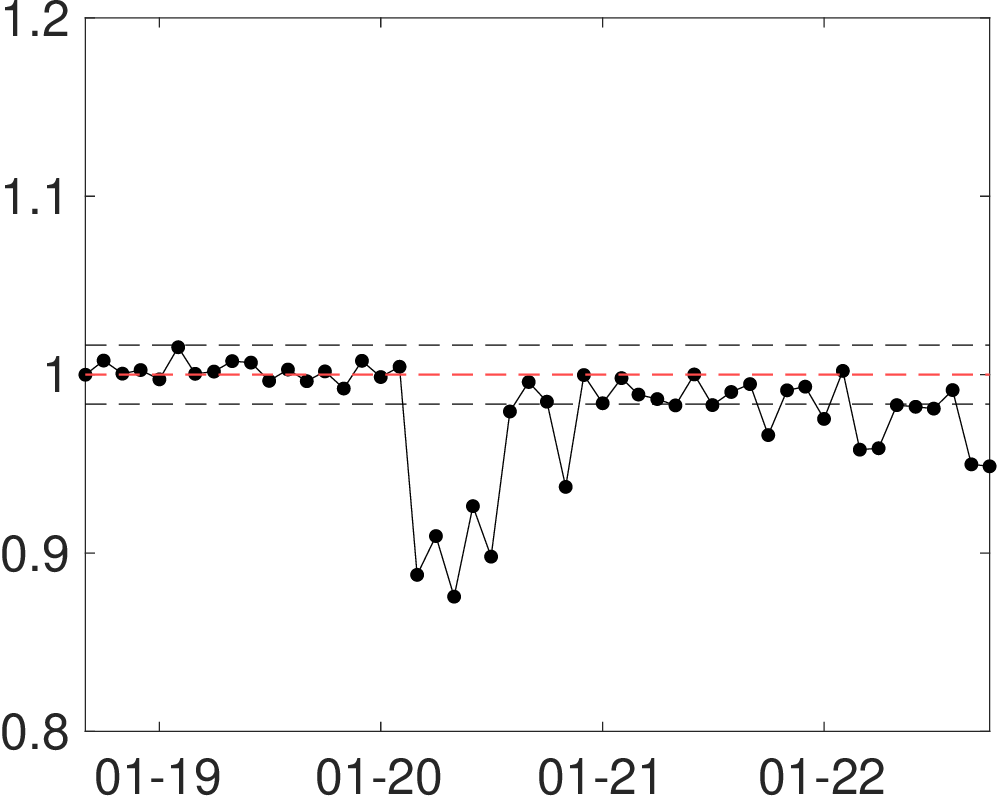}&
    \includegraphics[scale=0.3]{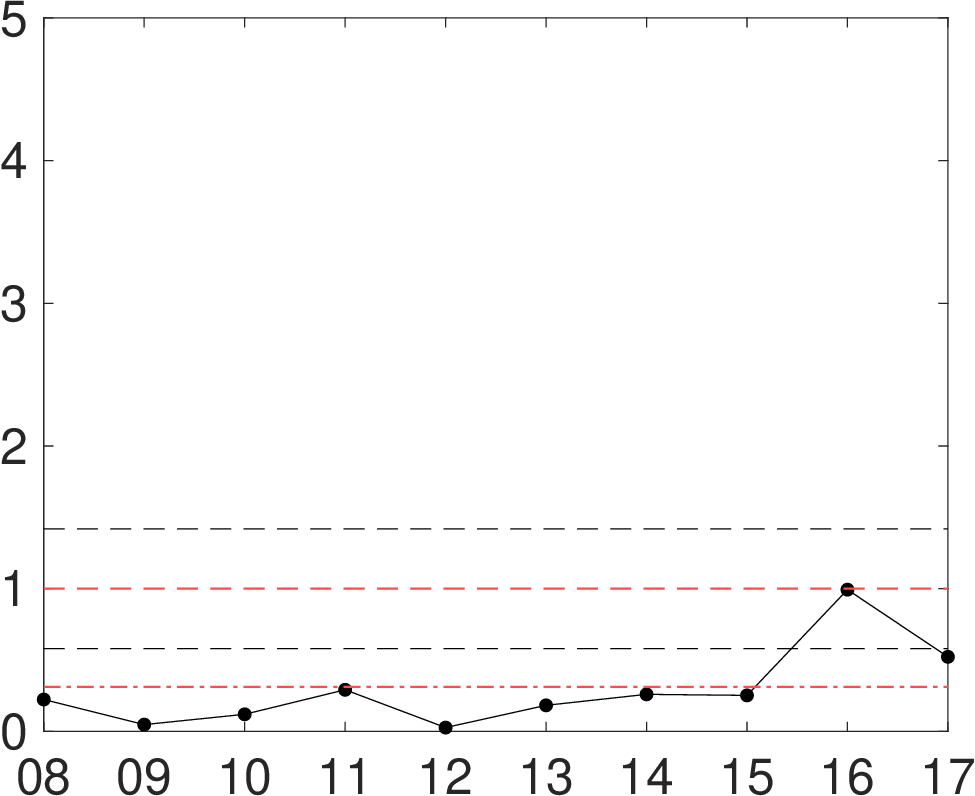}&
    \includegraphics[scale=0.3]{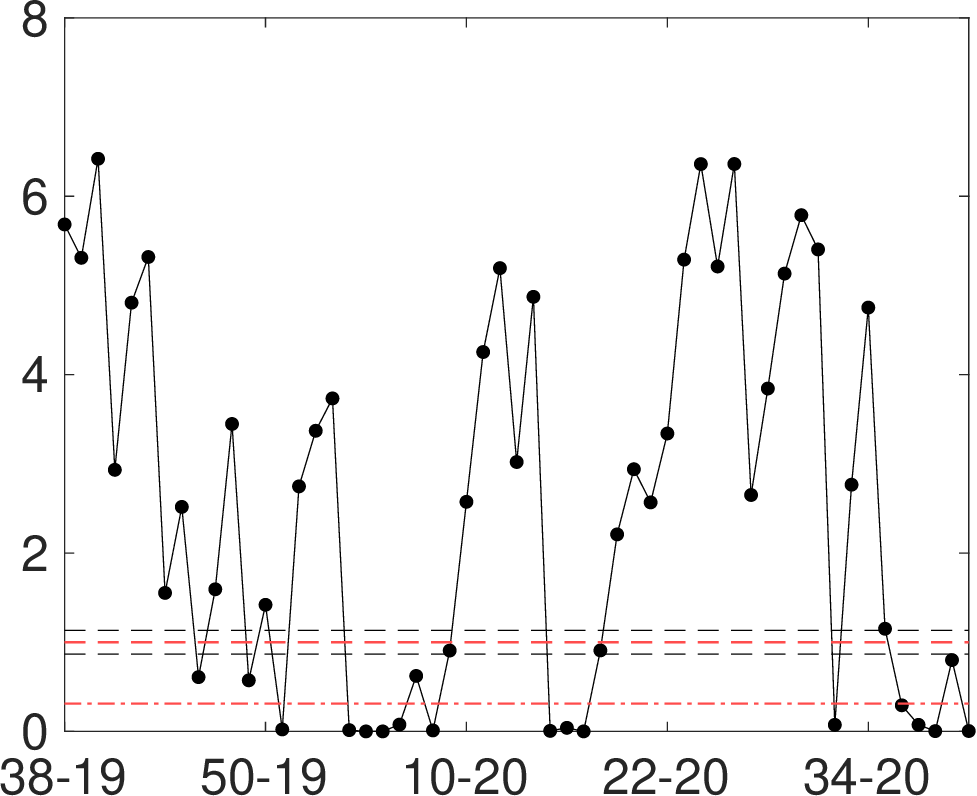}\\
   \multicolumn{3}{c}{(b) Grubb's and GESD tests}\vspace{5pt}\\
    \includegraphics[scale=0.3]{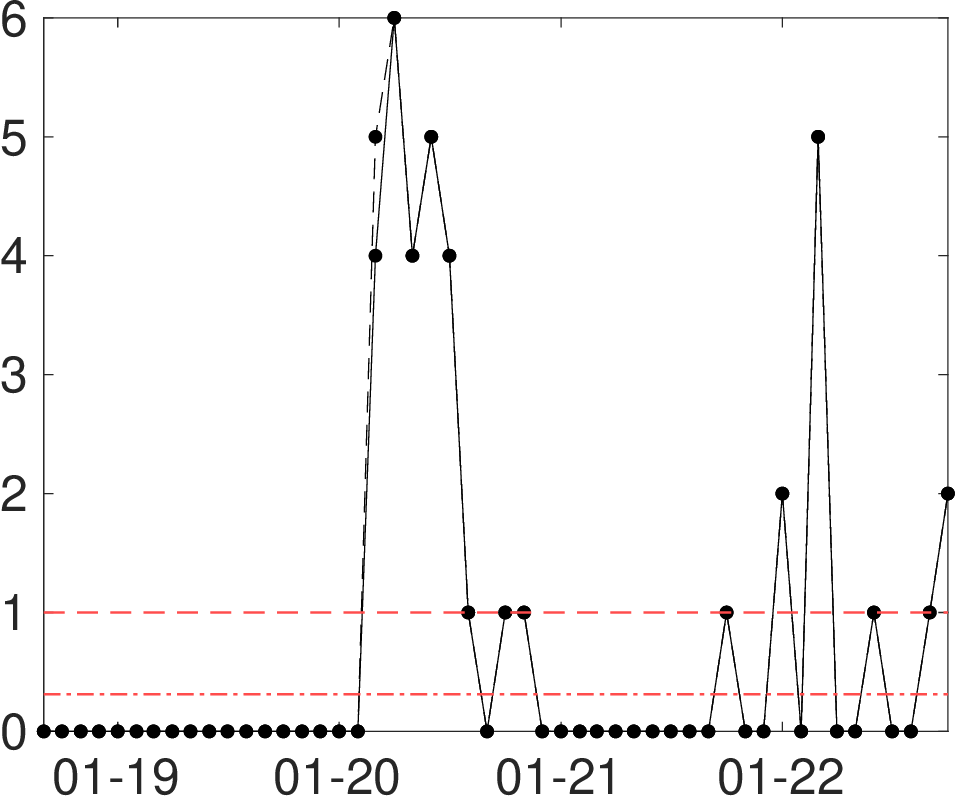}&
    \includegraphics[scale=0.3]{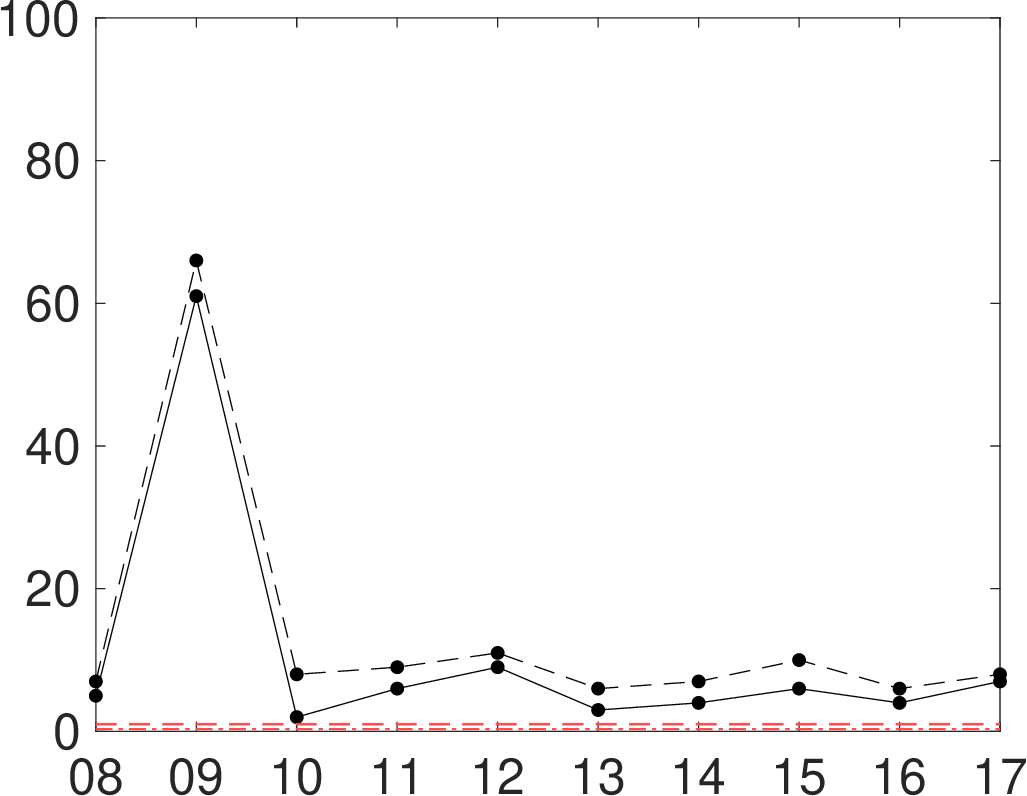}&
    \includegraphics[scale=0.3]{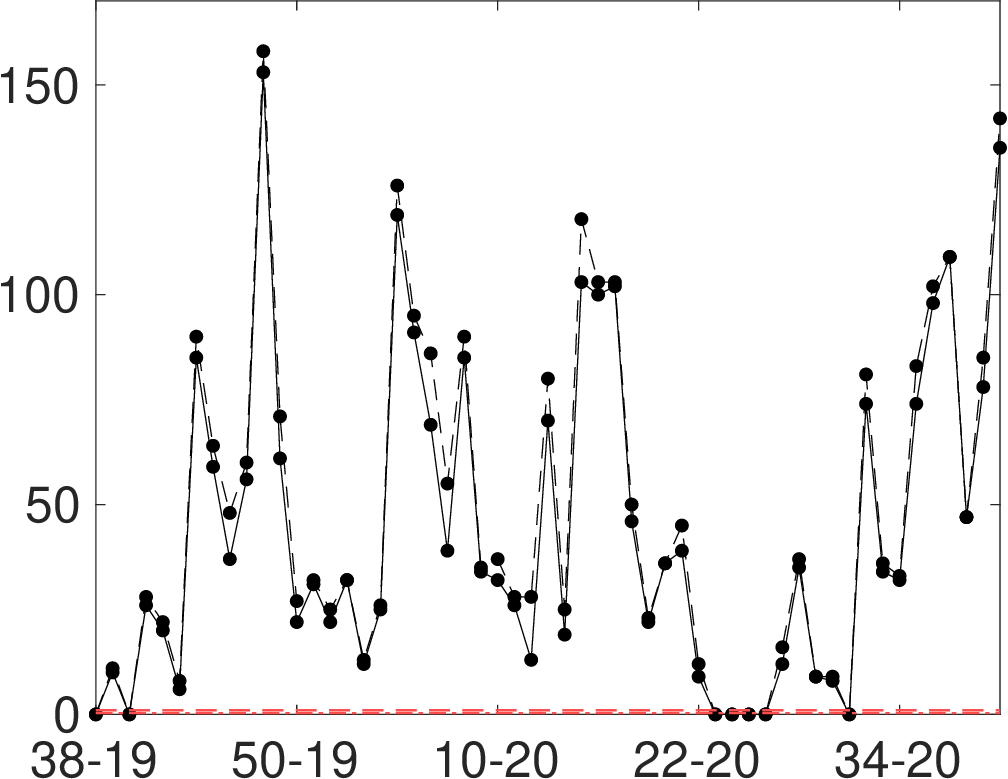}\\
    \end{tabular}
        \caption{Calibrated robust BF (top) and Grubb's test (bottom) for the three datasets: Inflation and Unemployment (left), International Trade (middle) and Volatility Network (right).
        Panel (a): Calibrated BF ({\color{black}\protect\tikz[baseline]{\protect\draw[line width=0.2mm] (0,.6ex)--++(0.5,0) ;}}) with thresholds ({\color{black}\protect\tikz[baseline]{\protect\draw[line width=0.2mm, dashed] (0,.6ex)--++(0.5,0) ;}}). The discounting factor is set at $\alpha_t=0.75$, and the size is less or equal to 1\% and the standard reference lines are at 1 ({\color{red}\protect\tikz[baseline]{\protect\draw[line width=0.2mm, dashed] (0,.6ex)--++(0.5,0) ;}}) and $10^{-1/2}$ ({\color{red}\protect\tikz[baseline]{\protect\draw[line width=0.2mm, dashdotted] (0,.6ex)--++(0.5,0) ;}}). Panel (b): Grubb's ({\color{black}\protect\tikz[baseline]{\protect\draw[line width=0.2mm] (0,.6ex)--++(0.5,0) ;}})  and GESD ({\color{black}\protect\tikz[baseline]{\protect\draw[line width=0.2mm, dashed] (0,.6ex)--++(0.5,0) ;}}) tests at the 1\% significance level}
    \label{emp1:figBFdatesThresholds}
\end{figure}

\subsection{Trade Network Dataset}
We consider the Trade Network Dataset (\citealp{Rose2004}), provided by the IMF, as it is a key reference in international trade studies. It integrates country reports with data from COMTRADE and EUROSTAT, covering 159 countries from 1995 to 2017 with annual frequency. The dataset sample we consider includes a sequence of 22 import networks of trade across 27 countries, i.e. $n=p=27$. 

Following the results in the top-middle plot of Figure \ref{emp1:figBFdates}, we find evidence of outlier observations on all dates except for 2016. Nevertheless, the BF is close to one in some cases, such as 2015 and 2017. From the sensitivity analysis in Figure \ref{emp1:figBFdatesBIStrade}, one can see that the test outcome depends crucially on the choice of $\alpha_t$. The top-middle plot in Figure \ref{emp1:figBFdatesIntegratedBF} indicates that minimum and integrated BFs provide evidence of the absence of an outlier in 2016 and the presence of outlying observations in 2015 and 2017, supporting the conclusion of the calibrated BF procedure with inconclusive intervals presented in the Panel (a) middle plot of Figure \ref{emp1:figBFdatesThresholds}. In the year 2016, the classical G and GESD tests for outliers, applied entry-wise to the observation matrix, detected one outlying series out of 729 series (mid plot in Panel (b).

\subsection{Volatility Network Dataset}
We investigate the presence of outliers in a sequence of volatility networks among European firms with the largest market capitalization (\citealp{billio2021matrix}). The dataset consists of 145 temporal networks (from the 4th of January 2016 to the 30th of September 2020) between 50 firms, i.e. $T=145$ and $n=p=50$, for 362,500 observations. In the sequential outlier detection, we used a rolling window of 90 observations. 

The plots in the top-right and bottom-right sections of Figure \ref{emp1:figBFdates} indicate that the rejection of the hypothesis is independent of the chosen discounting coefficient. However, BF exhibits greater sensitivity to discounting towards the end of 2019. The sensitivity analysis of the BF to $\alpha_t$ is presented for selected dates in Figure \ref{emp1:figBFdatesBISvolatility} in the Appendix.

The top-right plot in Figure \ref{emp1:figBFdatesIntegratedBF} shows the results of the minimum BF, which support the main findings of the calibrated BF procedure (Figure \ref{emp1:figBFdatesThresholds}, Panel(a), right plot). When the BF is far above one, the classical G and GESD tests for outliers, applied entry-wise to the observation matrix, detect a reduction in the number of outliers. In summary, the rapid changes in volatility and the persistence of volatility regimes call for nonlinear models, such as switching or threshold models, which account for structural breaks and recurrent regimes. 

\section{Conclusion}\label{sec:concl}
The assessment of the model performance is relevant in many applications and becomes crucial in forecasting. This paper proposes sequential outlier detection for the matrix normal model. The hypothesis testing procedure extends the predictive Bayes Factor (BF) with the power discounting to matrix models. The proposed approach relies on normality, now a default assumption in many applications, which serves to build a preliminary test for outliers in sequences of matrix--valued data. Some solutions are proposed to mitigate the test outcome's dependence on the discounting coefficient value, such as the minimum and the integrated BFs. The finite--sample distribution of the predictive BF is derived, and a testing procedure is proposed based on calibrated discounting and BF. Simulation experiments are conducted to study the properties of our Bayesian outlier detection. Numerical illustrations with relevant benchmark datasets are given. They include a comparison with classical tests for outliers and a validation based on major global event dates. 



\phantomsection{\large \textbf{Funding}}

This work was funded by the MUR -- PRIN project under g.a. n. 2022CLTYP4 and the Next Generation EU -- `\textit{GRINS -- Growing Resilient, INclusive and Sustainable}' project (PE0000018), National Recovery and Resilience Plan -- PE9. The views and opinions expressed are only those of the authors and do not necessarily reflect those of the EU.

\phantomsection\label{supplementary-material}
\bigskip

{\large\bf Supplementary Materials}

\begin{description}
\item[Supplementary Appendix:] A PDF document containing all mathematical derivations, proofs of theoretical results, and additional simulation evidence supporting the methodology presented in the main text. (PDF file)

\item[R-package for Outlier Detection routine:] The R package \texttt{BAYSFWATCH}, developed for this study, contains functions to implement the proposed Bayesian outlier detection procedures. The package also includes all datasets used in the empirical illustrations. (GNU zipped tar file, available at \url{https://github.com/BayesianEcon/BAYSFWATCH})
\end{description}

  \begin{center}
    {\LARGE\bf Supplementary Materials \\ Bayesian Outlier Detection for Matrix--variate Models}
\end{center}

\appendix
\renewcommand{\thesection}{A}
\renewcommand{\theequation}{A.\arabic{equation}}
\renewcommand{\thefigure}{A.\arabic{figure}}
\renewcommand{\thetable}{A.\arabic{table}}
\renewcommand{\theproposition}{A.\arabic{proposition}}
\setcounter{table}{0}
\setcounter{figure}{0}
\setcounter{equation}{0}
\setcounter{proposition}{0}

\section{Proofs of the Results}\label{app:proof}
\subsection{Proof of Proposition \ref{th1}.}
\begin{proof} For ease of notation in the following, we drop the subscript $t$ from $\alpha_t$. Since $H_t(1)=1$, we need to show that under the proposition assumptions, there exits $\alpha_g\in(0,1)$ such that $H_t(\alpha)>1$ for $0<\alpha<\alpha_g$ and $H_t(\alpha)<1$ for $\alpha_g<\alpha<1$. The limit of the normalizing constant $C(\alpha)=\left(\int_{\Theta}p(\boldsymbol{\theta}_t |\mathcal{D}_{t-1})^{\alpha}\lambda(d\boldsymbol{\theta}_t)\right)^{-1}$ for $\alpha\rightarrow 0^{+}$ is $\lambda(\Theta)^{-1}$ where $\lambda(\Theta)$ is the Lebesgue measure of $\Theta$. It follows that 
\begin{equation}
  \underset{\alpha\rightarrow 0^{+}}{\lim}H_t(\alpha)=\underset{\alpha\rightarrow 0^{+}}{\lim}\frac{1}{C(\alpha)}\frac{\int_{\Theta}p(\boldsymbol{Y}_t| \boldsymbol{\theta}_t)p(\boldsymbol{\theta}_t|\mathcal{D}_{t-1})\lambda(d\boldsymbol{\theta})}{\int_{\Theta}p(\boldsymbol{Y}_t| \boldsymbol{\theta}_t)p(\boldsymbol{\theta}_t|\mathcal{D}_{t-1})^{\alpha}\lambda(d\boldsymbol{\theta})} =\frac{\lambda(\Theta)\int_{\Theta}p(\boldsymbol{Y}_t| \boldsymbol{\theta}_t)p(\boldsymbol{\theta}_t|\mathcal{D}_{t-1})\lambda(d\boldsymbol{\theta})}{\int_{\Theta}p(\boldsymbol{Y}_t| \boldsymbol{\theta}_t)\lambda(d\boldsymbol{\theta})}\geq 1
\end{equation}
from the assumption. The inequality $H_t(\alpha)<1$ can be equivalently written as
\begin{equation}
\int_{\Theta}p(\boldsymbol{Y}_t|\boldsymbol{\theta}_t)(p(\boldsymbol{\theta}_t|\mathcal{D}_{t-1})-C(\alpha)p(\boldsymbol{\theta}_t|\mathcal{D}_{t-1})^{\alpha})\lambda(d\boldsymbol{\theta}_t)<0.\label{condLambda2}
\end{equation}
and is satisfied if $H_t(\alpha)$
has a unique minimum value for some $\alpha$. Let us assume the distribution $p(\boldsymbol{\theta}_t\vert \mathcal{D}_t)$ is absolutely continuous with respect to the Lebesgue measure $\lambda$, and define
\begin{eqnarray*}
m_0=\int_{\Theta}p(\boldsymbol{Y}_t|\boldsymbol{\theta}_t)p(\boldsymbol{\theta}_t|\mathcal{D}_{t-1})\lambda(d\boldsymbol{\theta}_t),\quad m_1(\alpha)=\int_{\Theta}p(\boldsymbol{Y}_t|\boldsymbol{\theta}_t)p(\boldsymbol{\theta}_t|\mathcal{D}_{t-1})^{\alpha}\lambda(d\boldsymbol{\theta}_t).
\end{eqnarray*}
Finding a necessary condition on $\alpha$ for $H_t$ to admit a minimum is possible. The first derivative of the BF is
\begin{eqnarray}
    \partial_{\alpha} H(\alpha)&=&\frac{\partial}{\partial \alpha} \frac{m_0}{C(\alpha)m_1(\alpha)}=\frac{m_0}{C(\alpha)m_1(\alpha)^2}{}\left(C(\alpha)m_1(\alpha) \int_{\Theta}p(\boldsymbol{\theta}_t|\mathcal{D}_{t-1})^{\alpha}\log p(\boldsymbol{\theta}_t|\mathcal{D}_{t-1})\lambda(d\boldsymbol{\theta}_t) -\right.\nonumber\\ &\quad&\left.\int_{\Theta}p(\boldsymbol{Y}_t|\boldsymbol{\theta}_t)p(\boldsymbol{\theta}_t|\mathcal{D}_{t-1})^{\alpha}\log p(\boldsymbol{\theta}_t|\mathcal{D}_{t-1})\lambda(d\boldsymbol{\theta}_t)\right)\nonumber\\
    &=&\frac{m_0}{m_1(\alpha)^2}\left(m_1(\alpha) \int_{\Theta}p(\boldsymbol{\theta}_t|\mathcal{D}_{t-1})^{\alpha}\log p(\boldsymbol{\theta}_t|\mathcal{D}_{t-1})\lambda(d\boldsymbol{\theta}_t) -\right.\nonumber\\ &\quad&\left.\int_{\Theta}p(\boldsymbol{\theta}_t|\mathcal{D}_{t-1})^{\alpha}\lambda(d\boldsymbol{\theta}_t)\int_{\Theta}p(\boldsymbol{Y}_t|\boldsymbol{\theta}_t)p(\boldsymbol{\theta}_t|\mathcal{D}_{t-1})^{\alpha}\log p(\boldsymbol{\theta}_t|\mathcal{D}_{t-1})\lambda(d\boldsymbol{\theta}_t)\right),
\end{eqnarray}
where we used
\begin{equation}
    \partial_{\alpha} C(\alpha)=-C(\alpha)^2\int_{\Theta}p(\boldsymbol{\theta}_t|\mathcal{D}_{t-1})^{\alpha}\log p(\boldsymbol{\theta}_t|\mathcal{D}_{t-1})\lambda(d\boldsymbol{\theta}_t),\nonumber
\end{equation}
which is well-defined  since $x^{\alpha}\log x\rightarrow 0$ as $x\rightarrow 0^{+}$ for all $\alpha>0$. If $\alpha_0$ satisfies the proposition's assumptions, then the first derivative changes sign only once, and the BF $H_t$ takes values above and below one.
\end{proof}

\subsection{Background and Preliminary Results}\label{sec:preliminary}

In this section, we review the notation for the matrix normal model, present some preliminary results, and state the main findings on the BF properties for the normal matrix model.

We adapt to our notation the definition of matrix normal and matrix Student-t distributions provided in Definitions 2.2.1 and 4.2.1 of \citet{Gup99}. In addition, we provide some results on the Bayesian inference for the location matrix of the matrix normal distribution under a conjugate prior assumption.
\begin{definition}[Matrix normal] The
random matrix $\boldsymbol{X}\left(p\times n \right)$ is said to have a matrix variate normal distribution with mean matrix $\boldsymbol{M}(p\times n)$ and covariance matrix
$\boldsymbol{\Sigma}\otimes \boldsymbol{V}$, where
$\boldsymbol{\Sigma}( p\times p ) > 0$ and $\boldsymbol{V}( n\times n ) > 0$, if
$\text{vec}( \boldsymbol{X}^{'})\sim \mathcal{N}_{p,n}( \text{vec}( \boldsymbol{M}^{'}),\boldsymbol{\Sigma}\otimes \boldsymbol{V})$, where $\otimes$ denotes the Kronecker product and $\text{vec}(\cdot)$ the vectorization operator which stacks matrix columns vertically.
We shall use the notation $\boldsymbol{X}\sim \mathcal{N}_{p,n}(\boldsymbol{M},\boldsymbol{\Sigma},\boldsymbol{V})$ and
the pdf of $\boldsymbol{X}$ is given in Theorem 2.2.1, p. 55 of \cite{Gup99}:
\begin{equation}
f\left( \boldsymbol{X} \vert \boldsymbol{M},\boldsymbol{\Psi},\boldsymbol{\Sigma} \right) = \frac{\exp\left\{ - \frac{1}{2}\text{tr}\left\lbrack \boldsymbol{\Sigma}^{- 1}(\boldsymbol{X} - \boldsymbol{M})\boldsymbol{\Psi}^{- 1}(\boldsymbol{X} - \boldsymbol{M})^{'} \right\rbrack \right\}}{(2\pi)^{\frac{np}{2}}|\boldsymbol{\Sigma}|^{\frac{n}{2}}|\boldsymbol{\Psi}|^{\frac{p}{2}}}.
\end{equation}
\end{definition}
We shall notice that the parameters $\boldsymbol{\Sigma}$ and $\boldsymbol{\Psi}$ are not identifiable since the likelihood is invariant to the rescaling of the two parameters. We refer the reader to \cite{And15,Gal18} for a treatment of this identification issue. In this paper, to achieve analytical tractability, we assume $\boldsymbol{\Sigma}$ and $\boldsymbol{\Psi}$ are given, as it is done for the univariate setting in \cite{west1986bayesian}. Nevertheless, some of our analytical results can be extended to cases where the covariance matrix is estimated following a Bayesian procedure. A standard prior assumption for a covariance matrix is the Inverse Wishart distribution, which is defined using the following parametrization.
\begin{definition}[Inverse Wishart] The random matrix $\boldsymbol{V}\left(n\times n \right)$ is said to have an inverse Wishart distribution with location parameter
$\boldsymbol{\Psi}(n\times n)$ and degrees of freedom parameter $m$, where
$\boldsymbol{\Psi}( n\times n ) > 0$ and $m > 2n$, if the  pdf of $\boldsymbol{V}$ is:
\begin{equation}
g\left(\boldsymbol{V}|\boldsymbol{\Psi},m\right) = \frac{\left|\boldsymbol{\Psi} \right|^{\frac{m-n-1}{2}}\exp\left\{ - \frac{1}{2}\text{tr}(\boldsymbol{\Psi} \boldsymbol{V}^{- 1})\right\}}{2^{\frac{(m-n-1)n}{2}}\Gamma_{n}(\frac{m-n-1}{2})\left|\boldsymbol{V}\right|^{\frac{m}{2}}},
\end{equation}
where $\Gamma_n(a)$ denotes the multivariate gamma function $\Gamma_{n}(a)=\pi^{n(n-1)/4}\prod_{j=1}^{n}\Gamma(a+(1-j)/2)$.
We shall use the notation $\boldsymbol{V}\sim \mathcal{I}\mathcal{W}_{n}(\boldsymbol{\Psi},m)$. See Definition 3.4.1, p. 111 of \cite{Gup99}. 
\end{definition}

\begin{definition}[Matrix Student-t] The
random matrix $\boldsymbol{X}\left(p\times n \right)$ is said to have a matrix Student-t distribution with degrees of freedom parameter $\nu > 0$, location parameter $\boldsymbol{M}$, and scale parameters $\boldsymbol{\Sigma}(p\times p) > 0$ and $\boldsymbol{\Omega}(n\times n) > 0$ if the  pdf of $\boldsymbol{X}$ is:
\begin{equation}
f\left( \boldsymbol{X} \vert \nu,\boldsymbol{M},\boldsymbol{\Sigma},\boldsymbol{\Omega} \right) = 
\frac{\Gamma_p\left(\frac{v+n+p-1}{2}\right)}{\pi^{\frac{m p}{2}} \Gamma_p\left(\frac{v+p-1}{2}\right)}\left|\boldsymbol{\Sigma}\right|^{-\frac{n}{2}}\left|\boldsymbol{\Omega}\right|^{-\frac{p}{2}}\left| I_p +\boldsymbol{\Sigma}^{-1}(\boldsymbol{X}-\boldsymbol{M}) \boldsymbol{\Omega}^{-1}(\boldsymbol{X}-\boldsymbol{M})^{\prime}\right|^{-\frac{v+n+p-1}{2}}
\end{equation}
We shall use the notation $\mathcal{T}_{p,n}( \nu, \boldsymbol{M},\boldsymbol{\Sigma}, \boldsymbol{\Omega})$ 
\end{definition}

Since the matrix normal distribution belongs to the exponential family \cite[e.g., see][]{Gup99}, it is possible to show that the matrix normal prior distribution for the location parameter of a matrix normal likelihood is conjugate.

Let us recall the following standard properties for matrix variate normal variables.
\begin{proposition}[Gupta, 1992, Theorem 2.2]\label{prop:gupta2}
Let $\boldsymbol{X}_1$ and $\boldsymbol{X}_2$ be two random matrices of dimension $p\times n$ and $q\times n$, respectively. Assume that
$\boldsymbol{X}_2|\boldsymbol{X}_1\sim \mathcal{N}_{q,n}\left( \boldsymbol{C} + \boldsymbol{DX}_1,\boldsymbol{\Sigma}_{2},\boldsymbol{\Phi} \right)$ and $\boldsymbol{X}_1\sim \mathcal{N}_{p,n}\left( \boldsymbol{F},\boldsymbol{\Sigma}_{1},\boldsymbol{\Phi} \right)$ where $\boldsymbol{C}$ is $q\times n$, $\boldsymbol{D}$ $q\times p$, $\boldsymbol{\Sigma}_{2}$ $q\times q$ $\boldsymbol{F}$ $p\times  n$ $\boldsymbol{\Phi}$ $n\times n$ $\boldsymbol{\Sigma}_{1}$ $ p\times p$, $\boldsymbol{\Sigma}_{2} > 0$, $\boldsymbol{\Sigma}_{1} > 0$, $\boldsymbol{\Phi} > 0$. Then:
\begin{equation}\boldsymbol{Z} = \begin{pmatrix}
\boldsymbol{X}_1 \\
\boldsymbol{X}_2 \\
\end{pmatrix}\sim \mathcal{N}_{q + p,n}\left\lbrack \begin{pmatrix}
\boldsymbol{F} \\
\boldsymbol{DF} + \boldsymbol{C} \\
\end{pmatrix},\begin{pmatrix}
\boldsymbol{\Sigma}_{1} & \boldsymbol{\Sigma}_{1}\boldsymbol{D}^{'} \\
\boldsymbol{D}\boldsymbol{\Sigma}_{1} & \boldsymbol{\Sigma}_{2} + \boldsymbol{D}\boldsymbol{\Sigma}_{1}\boldsymbol{D}^{'} \\
\end{pmatrix},\boldsymbol{\Phi}\right\rbrack.\label{A2.10}\end{equation}
\end{proposition}

\begin{proposition}[Gupta, 1999, Theorem 2.3.12, pg. 65]\label{prop:gupta3}
Let $\boldsymbol{Z}\sim \mathcal{N}_{q,n}(\boldsymbol{H},\boldsymbol{\Sigma}, \boldsymbol{Q})$ and partitions:

\begin{equation} \boldsymbol{Z} = \begin{pmatrix}
\boldsymbol{Z}_{1r} \\
\boldsymbol{Z}_{2r} \\
\end{pmatrix}\text{\ \ \ }\begin{matrix}
m \\
q - m \\
\end{matrix};\ \ \boldsymbol{H} = \begin{pmatrix}
\boldsymbol{H}_{1r} \\
\boldsymbol{H}_{2r} \\
\end{pmatrix}\text{\ \ }\begin{matrix}
m \\
q - m \\
\end{matrix};\ \ \boldsymbol{\Sigma} = \begin{pmatrix}
\boldsymbol{\Sigma}_{11} & \boldsymbol{\Sigma}_{12} \\
\boldsymbol{\Sigma}_{21} & \boldsymbol{\Sigma}_{22} \\
\end{pmatrix}\text{\ \ }\begin{matrix}
m \\
q - m \\
\end{matrix}.\label{A2.11}\end{equation}
Then:
\begin{equation} \boldsymbol{Z}_{1r}\sim \mathcal{N}_{m,n}\left( \boldsymbol{H}_{1r},\boldsymbol{\Sigma}_{11}, \boldsymbol{Q} \right),\ \boldsymbol{Z}_{2r}|\boldsymbol{Z}_{1r}\ \sim \mathcal{N}_{q - m,n}\left( \boldsymbol{H}_{2r} + \boldsymbol{\Sigma}_{21}\boldsymbol{\Sigma}_{11}^{- 1}\left( \boldsymbol{Z}_{1r} - \boldsymbol{H}_{1r} \right),\boldsymbol{\Sigma}_{22.1}, \boldsymbol{Q}\right),\label{A2.12}
\end{equation}
where $\boldsymbol{\Sigma}_{22.1} = \boldsymbol{\Sigma}_{22} - \boldsymbol{\Sigma}_{21}\boldsymbol{\Sigma}_{11}^{- 1}\boldsymbol{\Sigma}_{12}$. It follows that:
\begin{equation}\boldsymbol{Z}_{2r}\sim \mathcal{N}_{q - m,n}\left( \boldsymbol{H}_{2r},\boldsymbol{\Sigma}_{22}, \boldsymbol{Q} \right),\, \boldsymbol{Z}_{1r}|\boldsymbol{Z}_{2r}\ \sim \mathcal{N}_{m,n}\left( \boldsymbol{H}_{1r} + \boldsymbol{\Sigma}_{12}\boldsymbol{\Sigma}_{22}^{- 1}\left( \boldsymbol{Z}_{2r} - \boldsymbol{H}_{2r} \right),\boldsymbol{\Sigma}_{11.2}, \boldsymbol{Q} \right),\label{A2.13}\end{equation}
where $\boldsymbol{\Sigma}_{11.2} = \boldsymbol{\Sigma}_{11} - \boldsymbol{\Sigma}_{12}\boldsymbol{\Sigma}_{22}^{- 1}\boldsymbol{\Sigma}_{21}$.
\end{proposition}


\begin{proposition}\label{prop:posteriormatrix}
Under Assumptions \ref{ass1} and \ref{ass2}, the
posterior distribution of $\boldsymbol{B}$ at time $t$ conditioned to the
information available at time $t - 1$ is the matrix variate normal
$\boldsymbol{B}|\mathbf{Y}\ \sim \mathcal{N}_{p,n}\left(  \boldsymbol{M}_{*},\boldsymbol{\Sigma}_{*}, \boldsymbol{V} \right)$
with:
{\small\begin{eqnarray}
\boldsymbol{M}_{*} &=& \boldsymbol{M} + \boldsymbol{\Sigma}_{P}\left( {\boldsymbol{\iota}_{T}\otimes \boldsymbol{I}}_{p} \right)^{'}\left\lbrack \boldsymbol{I}_{T}\otimes\boldsymbol{\Sigma}_{L} + \left( {\boldsymbol{\iota}_{T}\otimes \boldsymbol{I}}_{p} \right)\boldsymbol{\Sigma}_{P}\left( {\boldsymbol{\iota}_{T}\otimes \boldsymbol{I}}_{p} \right)^{'} \right\rbrack^{- 1}\left( \mathbf{Y} - \boldsymbol{\iota}_{T}\otimes \boldsymbol{M} \right)\label{eq:Mstar}\\
\boldsymbol{\Sigma}_{*} &=& \boldsymbol{\Sigma}_{P} - \boldsymbol{\Sigma}_{P}\left( {\boldsymbol{\iota}_{T}\otimes \boldsymbol{I}}_{p} \right)^{'}\left\lbrack \boldsymbol{I}_{T}\otimes\boldsymbol{\Sigma}_{L} + \left( {\boldsymbol{\iota}_{T}\otimes I}_{p} \right)\boldsymbol{\Sigma}_{P}\left( {\boldsymbol{\iota}_{T}\otimes \boldsymbol{I}}_{p} \right)^{'} \right\rbrack^{- 1}{\left( {\boldsymbol{\iota}_{T}\otimes \boldsymbol{I}}_{p} \right)\boldsymbol{\Sigma}}_{P}.\label{eq:Sigmastar}
\end{eqnarray}}
\end{proposition}

\begin{proof}
Let $\boldsymbol{Y}_{i}\left( p\times n \right)$ $i=1,\ldots,t-1$ be a random matrix i.i.d. sequence from a matrix--variate normal distribution with mean matrix
$\boldsymbol{B}\left( p\times n \right)$ and covariance matrices
$\boldsymbol{\Sigma}_{L}\left( p\times p \right) > 0$ and
$\boldsymbol{V}\left( n\times n \right) > 0$, then
$\text{vec}( \boldsymbol{Y}_{t}^{'} )\sim \mathcal{N}_{pn}( \text{vec}( \boldsymbol{B}^{'} ),\boldsymbol{\Sigma}_{L}\otimes \boldsymbol{V} )$
and the pdf is:
\begin{equation} f\left( \boldsymbol{Y}_i|\boldsymbol{B},\boldsymbol{V},\boldsymbol{\Sigma}_{L} \right) = (2\pi)^{-\frac{np}{2}}\left| \boldsymbol{\Sigma}_{L} \right|^{-\frac{n}{2}}|\boldsymbol{V}|^{-\frac{p}{2}}\exp\left\{ - \frac{1}{2}\text{tr}\left\lbrack \boldsymbol{\Sigma}_{L}^{- 1}\left( \boldsymbol{Y}_i - \boldsymbol{B} \right)\boldsymbol{V}^{- 1}\left( \boldsymbol{Y}_i - \boldsymbol{B} \right)^{'} \right\rbrack \right\}.\label{A2.2}
\end{equation}
Define $T=t-1$, then the likelihood function is:
{\small\begin{eqnarray}
 f\left( \boldsymbol{Y}_{1},\boldsymbol{Y}_{2},\ldots,\boldsymbol{Y}_{T}|\boldsymbol{B},\boldsymbol{V},\boldsymbol{\Sigma}_{L} \right) &&=\prod_{i = 1}^{T}\frac{\exp\left\{ - \frac{1}{2}\text{tr}\left\lbrack \boldsymbol{\Sigma}_{L}^{- 1}\left( \boldsymbol{Y}_i - \boldsymbol{B} \right)\boldsymbol{V}^{- 1}\left( \boldsymbol{Y}_i - \boldsymbol{B} \right)^{'} \right\rbrack \right\}}{(2\pi)^{\frac{np}{2}}\left| \boldsymbol{\Sigma}_{L} \right|^{\frac{n}{2}}|\boldsymbol{V}|^{\frac{p}{2}}}\label{A2.3}\\
&&=\frac{\exp\left\{ - \frac{1}{2}\text{tr}\left\lbrack \sum_{i = 1}^{T}{\boldsymbol{\Sigma}_{L}^{- 1}\left( \boldsymbol{Y}_{t} - \boldsymbol{B} \right)\boldsymbol{V}^{- 1}\left( \boldsymbol{Y}_{t} - \boldsymbol{B} \right)^{'}} \right\rbrack \right\}}{(2\pi)^{\frac{Tnp}{2}}\left| \boldsymbol{\Sigma}_{L} \right|^{T\frac{n}{2}}|\boldsymbol{V}|^{T\frac{p}{2}}}\nonumber\\
&&= \frac{\exp\left\{ - \frac{1}{2}\text{tr}\left\lbrack \left( {\boldsymbol{I}_{T}\otimes\boldsymbol{\Sigma}}_{L}^{- 1} \right)\left( \mathbf{Y} - \boldsymbol{\iota}_{T}\otimes \boldsymbol{B} \right)\boldsymbol{V}^{- 1}\left( \mathbf{Y} - \boldsymbol{\iota}_{T}\otimes \boldsymbol{B} \right)^{'} \right\rbrack \right\}}{(2\pi)^{\frac{Tnp}{2}}\left| \boldsymbol{\Sigma}_{L} \right|^{T\frac{n}{2}}|\boldsymbol{V}|^{\frac{Tp}{2}}},\nonumber
\end{eqnarray}}
where
$\mathbf{Y}' = \left\lbrack \boldsymbol{Y}_{1}^{'},\ \boldsymbol{Y}_{2}^{'},\ \ldots,\ \boldsymbol{Y}_{T}^{'} \right\rbrack$
is a matrix $( n \times Tp)$ obtained by stacking vertically the $T$ matrices $\boldsymbol{Y}_i,\ i = 1,\ldots,T$ and
$\boldsymbol{I}_{T}\otimes\boldsymbol{\Sigma}_{L}$ is $( Tp\times Tp)$.  We can say that the previous expression is the pdf of the random matrix $\mathbf{Y}(Tp\times n)$ with mean $\boldsymbol{\iota}_{T}\otimes \boldsymbol{B}(Tp\times n)$ and a covariance matrix
$( \boldsymbol{I}_{T}\otimes\boldsymbol{\Sigma}_{L} )\otimes \boldsymbol{V}$ where $\boldsymbol{I}_{T}$ and $\boldsymbol{\iota}_{T}$ are the $T$-dimensional identity matrix and unit column vector, respectively. The matrix $\boldsymbol{I}_{T}\otimes\boldsymbol{\Sigma}_{L} ( Tp\times Tp )$ is positive definite, which follows from $\boldsymbol{\Sigma}_{L} > 0$, $\boldsymbol{I}_{T} > 0$ and $\boldsymbol{V}>0$.
Since
$\left| ( \boldsymbol{I}_{T}\otimes\boldsymbol{\Sigma}_{L} )\otimes \boldsymbol{V} \right| = \left| \left( \boldsymbol{I}_{T}\otimes\boldsymbol{\Sigma}_{L} \right) \right|^{n}|\boldsymbol{V}|^{Tp}=
\left| \boldsymbol{I}_{T} \right|^{np}\left| \boldsymbol{\Sigma}_{L} \right|^{nT}|\boldsymbol{V}|^{Tp}=\left| \boldsymbol{\Sigma}_{L} \right|^{Tn}|\boldsymbol{V}|^{Tp}$ the pdf of $\mathbf{Y}$ is:
{\small\begin{eqnarray} f\left( \mathbf{Y}|\boldsymbol{B},\boldsymbol{V},\boldsymbol{\Sigma}_{L} \right) &=& \frac{\exp\left\{ - \frac{1}{2}\text{tr}\left\lbrack \left( \boldsymbol{I}_{T}\otimes\boldsymbol{\Sigma}_{L} \right)^{- 1}\left( \mathbf{Y} - \boldsymbol{\iota}_{T}\otimes \boldsymbol{B} \right)\boldsymbol{V}^{- 1}\left( \mathbf{Y} - \boldsymbol{\iota}_{T}\otimes \boldsymbol{B} \right)^{'} \right\rbrack \right\}}{(2\pi)^{\frac{\text{Tnp}}{2}}\left| \boldsymbol{I}_{T}\otimes\boldsymbol{\Sigma}_{L} \right|^{\frac{n}{2}}|\boldsymbol{V}|^{\frac{\text{Tp}}{2}}}\nonumber\\
&=&\frac{\exp\left\{ - \frac{1}{2}\text{tr}\left\lbrack \left( \boldsymbol{I}_{T}\otimes\boldsymbol{\Sigma}_{L} \right)^{- 1}\left( \mathbf{Y} - \boldsymbol{\iota}_{T}\otimes \boldsymbol{B} \right)\boldsymbol{V}^{- 1}\left( \mathbf{Y} - \boldsymbol{\iota}_{T}\otimes \boldsymbol{B} \right)^{'} \right\rbrack \right\}}{(2\pi)^{\frac{Tnp}{2}}\left| \boldsymbol{\Sigma}_{L} \right|^{\frac{\text{Tn}}{2}}|\boldsymbol{V}|^{\frac{Tp}{2}}}\label{A2.8}
\end{eqnarray}}
which is equal to the pdf in Equation \eqref{A2.3}. In the following, we assume the  $p\times n$ matrix $\boldsymbol{B}$ is random, and the matrices $\boldsymbol{\Sigma}_{P}$ and $\boldsymbol{V}$ are known. Assume a matrix normal prior distribution for $\boldsymbol{B}$, i.e. $\boldsymbol{B}\sim \mathcal{N}_{p,n}\left( \boldsymbol{M},\boldsymbol{\Sigma}_{P}, \boldsymbol{V} \right)$, with pdf
\begin{equation}
f\left( \boldsymbol{B}|\boldsymbol{M},\boldsymbol{V},\boldsymbol{\Sigma}_{P} \right) = \frac{\exp\left\{ - \frac{1}{2}\text{tr}\left\lbrack \boldsymbol{\Sigma}_{P}^{- 1}(\boldsymbol{B} - \boldsymbol{M})\boldsymbol{V}^{- 1}(\boldsymbol{B} - \boldsymbol{M})^{'} \right\rbrack \right\}}{(2\pi)^{\frac{np}{2}}\left| \boldsymbol{\Sigma}_{P} \right|^{\frac{n}{2}}|\boldsymbol{V}|^{\frac{p}{2}}}\label{A2.9}.
\end{equation}
Note that $\boldsymbol{\iota}_{T}\otimes \boldsymbol{B} = \left( {\boldsymbol{\iota}_{T}\otimes \boldsymbol{I}}_{p} \right)\boldsymbol{B}$ and
apply Prop. \ref{prop:gupta3} to \eqref{A2.8} and \eqref{A2.9}, replacing $\boldsymbol{X}_1,\ \boldsymbol{X}_2|\boldsymbol{X}_1,\ \boldsymbol{C},\ \boldsymbol{D},\ \boldsymbol{F},\ \boldsymbol{\Sigma}_{1},\ \boldsymbol{\Sigma}_{2},\ \boldsymbol{\Phi}$ wih $\boldsymbol{B},\ \mathbf{Y}|\boldsymbol{B},\ \boldsymbol{O}_{Tp \times n},\ {\boldsymbol{\iota}_{T}\otimes \boldsymbol{I}}_{p},\ \boldsymbol{M},\ \boldsymbol{\Sigma}_{P},\ \boldsymbol{\Sigma}_{L},\boldsymbol{V}$ , where $\boldsymbol{I}_{p}$ is the identity matrix of order $p$, we get: 
\begin{equation}
\boldsymbol{Z}\sim \mathcal{N}_{Tp + p,n}\left\lbrack \begin{pmatrix}
\boldsymbol{M} \\
\left( {\boldsymbol{\iota}}_{T}\otimes\boldsymbol{I}_{p} \right)\boldsymbol{M} + \boldsymbol{O} \\
\end{pmatrix},\begin{pmatrix}
\boldsymbol{\Sigma}_{P} & \boldsymbol{\Sigma}_{P}\left( {\boldsymbol{\iota}}_{T}\otimes\boldsymbol{I}_{p}\  \right)^{'} \\
\left( {\boldsymbol{\iota}}_{T}\otimes\boldsymbol{I}_{p} \right)\boldsymbol{\Sigma}_{P} & \left( \boldsymbol{I}_{T}\otimes\boldsymbol{\Sigma}_{L} \right) + \left( {\boldsymbol{\iota}_{T}\otimes \boldsymbol{I}}_{p} \right)\boldsymbol{\Sigma}_{P}\left( {\boldsymbol{\iota}_{T}\otimes \boldsymbol{I}}_{p} \right)^{'} \\
\end{pmatrix}\otimes \boldsymbol{V} \right\rbrack,\label{A2.14}
\end{equation}
where $\boldsymbol{Z}' =(\boldsymbol{B}',\mathbf{Y}')$. By applying Prop. \ref{prop:gupta3} to $\boldsymbol{Z}$ and defining:
{\small
\begin{eqnarray}
\boldsymbol{M}_{*} &=& \boldsymbol{M} + \boldsymbol{\Sigma}_{P}\left( {\boldsymbol{\iota}_{T}\otimes \boldsymbol{I}}_{p} \right)^{'}\left\lbrack \boldsymbol{I}_{T}\otimes\boldsymbol{\Sigma}_{L} + \left( {\boldsymbol{\iota}_{T}\otimes \boldsymbol{I}}_{p} \right)\boldsymbol{\Sigma}_{P}\left( {\boldsymbol{\iota}_{T}\otimes I}_{p} \right)^{'} \right\rbrack^{- 1}\left( \mathbf{Y} - \boldsymbol{\iota}_{T}\otimes \boldsymbol{M} \right)\\
\boldsymbol{\Sigma}_{*} &=& \boldsymbol{\Sigma}_{P} - \boldsymbol{\Sigma}_{P}\left( {\boldsymbol{\iota}_{T}\otimes \boldsymbol{I}}_{p} \right)^{'}\left\lbrack \boldsymbol{I}_{T}\otimes\boldsymbol{\Sigma}_{L} + \left( {\boldsymbol{\iota}_{T}\otimes \boldsymbol{I}}_{p} \right)\boldsymbol{\Sigma}_{P}\left( {\boldsymbol{\iota}_{T}\otimes \boldsymbol{I}}_{p} \right)^{'} \right\rbrack^{- 1}{\left( {\boldsymbol{\iota}_{T}\otimes \boldsymbol{I}}_{p} \right)\boldsymbol{\Sigma}}_{P},\label{A2.15}
\end{eqnarray}}
we obtain the posterior distribution of $\boldsymbol{B}$ at time $t-1$ with 
pdf
\begin{equation}f\left( \boldsymbol{B}|\mathbf{Y},\boldsymbol{M},\boldsymbol{V},\boldsymbol{\Sigma}_{*} \right) = \frac{\exp\left\{ - \frac{1}{2}\text{tr}\left\lbrack \boldsymbol{\Sigma}_{*}^{- 1}\left( \boldsymbol{B} - \boldsymbol{M}_{*} \right)\boldsymbol{V}^{- 1}\left( \boldsymbol{B} - \boldsymbol{M}_{*} \right)^{'} \right\rbrack \right\}}{(2\pi)^{\frac{np}{2}}\left| \boldsymbol{\Sigma}_{*} \right|^{\frac{n}{2}}|\boldsymbol{V}|^{\frac{p}{2}}}\label{A2.16}
\end{equation}
that is $\boldsymbol{B}$ follows a matrix normal $\boldsymbol{B}|\mathbf{Y}\ \sim \mathcal{N}_{p,n}\left( \boldsymbol{M}_{*},\boldsymbol{\Sigma}_{*},\boldsymbol{V} \right)$.

\end{proof}


\begin{proposition}\label{prop:posteriormatrixUnknownV}
Under Assumptions \ref{ass1} and \ref{ass3}, the joint posterior distribution of $\boldsymbol{B}$ and $\boldsymbol{V}$ at time $t$ given the information available at time $t - 1$ is the matrix normal inverse Wishart
$\boldsymbol{B}|\boldsymbol{V},\mathbf{Y} \sim \mathcal{N}_{p,n}\left(  \boldsymbol{M}_{*},\boldsymbol{\Sigma}_{*}, \boldsymbol{V} \right)$ and $\boldsymbol{V}|\mathbf{Y}\sim\mathcal{I}\mathcal{W}(\boldsymbol{\Psi}_{*},m_{*})$
with:
\begin{eqnarray}
&&\boldsymbol{M}_{*} =\frac{k\boldsymbol{M}+T\bar{\boldsymbol{Y}}}{k+T},\quad
\boldsymbol{\Sigma}_{*} =\frac{1}{k_{*}}\boldsymbol{\Sigma}_L,\quad\boldsymbol{\Psi}_{*}=\boldsymbol{\Psi}+kT\frac{\left(\boldsymbol{M}-\bar{\boldsymbol{Y}}\right)^{'}\boldsymbol{\Sigma}_L^{-1}\left(\boldsymbol{M}-\bar{\boldsymbol{Y}}\right)}{\left(k+T\right)}+T\boldsymbol{S}\nonumber\\&& k_{*}=k+T,\quad k=\rho\varphi,\quad m_{*}=m+Tp\label{eq:PsistarV}
\end{eqnarray}
where, $\bar{\boldsymbol{Y}}=\frac{1}{T}\sum_{i = 1}^{T}\boldsymbol{Y}_i$ and $\boldsymbol{S}=\frac{1}{T}\sum_{i = 1}^{T}{\left(\boldsymbol{Y}_i - \bar{\boldsymbol{Y}} \right)^{'}\boldsymbol{\Sigma}_L^{- 1}\left(\boldsymbol{Y}_i - \bar{\boldsymbol{Y}} \right)}$.
\end{proposition} 
For the unknown $\boldsymbol{V}$ case, the pair $\boldsymbol{B}_t$ and $\boldsymbol{V}_t$ corresponds to $\boldsymbol{\theta}_{t}$ and their joint distribution corresponds to the posterior $p\left(\boldsymbol{\theta}_{t}\vert \mathcal{D}_{t - 1} \right)$.

\begin{proof}[Proof of Proposition \ref{prop:posteriormatrixUnknownV}]
Let $f\left( \boldsymbol{B},\boldsymbol{V}|\mathbf{Y},\boldsymbol{M},\boldsymbol{\Sigma}_L,\boldsymbol{\Psi},m \right)$ be the density of the Matrix Normal Inverse Wishart prior, the joint posterior density of $\boldsymbol{B}$ and $\boldsymbol{V}$ is:
\begin{eqnarray}f\left( \boldsymbol{B},\boldsymbol{V}|\mathbf{Y},\boldsymbol{M},\boldsymbol{\Sigma}_L,\boldsymbol{\Psi},m \right)&=&\frac{h\left( \boldsymbol{B},\boldsymbol{V}|\boldsymbol{M},\boldsymbol{\Sigma}_L,\boldsymbol{\Psi},m \right)f\left( \mathbf{Y}|\boldsymbol{B},\boldsymbol{V},\boldsymbol{\Sigma}_L \right)}{\int_{\boldsymbol{V}>0}\int_{\boldsymbol{B}}h\left( \boldsymbol{B},\boldsymbol{V}|\boldsymbol{M},\boldsymbol{\Sigma}_L,\boldsymbol{\Psi},m \right)f\left( \mathbf{Y}|\boldsymbol{B},\boldsymbol{V},\boldsymbol{\Sigma}_L \right)d\boldsymbol{B}d\boldsymbol{V}}\nonumber\\
&=&\frac{{\left|\boldsymbol{V}\right|^{-\frac{m+p}{2}}|\boldsymbol{V}|^{-\frac{Tp}{2}}}\exp\left\{ - \frac{1}{2}\text{tr}\left\lbrack\mathcal{A}\right\rbrack\right\}}{\int_{\boldsymbol{V}>0}\int_{\boldsymbol{B}}{\left|\boldsymbol{V}\right|^{-\frac{m+p}{2}}|\boldsymbol{V}|^{-\frac{Tp}{2}}}\exp\left\{ - \frac{1}{2}\text{tr}\left\lbrack\mathcal{A}\right\rbrack\right\}d\boldsymbol{B}d\boldsymbol{V}},\nonumber
\end{eqnarray}
where we defined $\mathcal{A}=(\boldsymbol{B} - \boldsymbol{M})^{'}k\boldsymbol{\Sigma}_L^{- 1}(\boldsymbol{B} - \boldsymbol{M})\boldsymbol{V}^{-1}+\boldsymbol{\Psi} \boldsymbol{V}^{-1}+\sum_{s = 1}^{T}{\left(     \boldsymbol{Y}_{s} - \boldsymbol{B} \right)^{'}\boldsymbol{\Sigma}_L^{- 1}\left(\boldsymbol{Y}_{s} - \boldsymbol{B} \right)\boldsymbol{V}^{- 1}}$ with $k = \rho\varphi$ and $\bar{\boldsymbol{Y}}=T^{-1}\sum_{s = 1}^{T}    \boldsymbol{Y}_{s}$.

Regarding the numerator, since
\begin{equation*}
    \boldsymbol{S} \boldsymbol{V}^{- 1}=\sum_{s = 1}^{T}{\left(     \boldsymbol{Y}_{s} - \boldsymbol{B} \right)^{'}\boldsymbol{\Sigma}_L^{- 1}\left(     \boldsymbol{Y}_{s} - \boldsymbol{B} \right)\boldsymbol{V}^{- 1}} -T\left(\bar{    \boldsymbol{Y}} - \boldsymbol{B} \right)^{'}\boldsymbol{\Sigma}_L^{- 1}\left(\bar{    \boldsymbol{Y}} - \boldsymbol{B} \right)\boldsymbol{V}^{- 1}
\end{equation*}
with $\boldsymbol{S}=\sum_{s = 1}^{T}{\left(    \boldsymbol{Y}_{s} - \bar{\boldsymbol{Y}} \right)^{'}\boldsymbol{\Sigma}_L^{- 1}\left(    \boldsymbol{Y}_{s} - \bar{\boldsymbol{Y}} \right)})/T$ and
{\small\begin{eqnarray*}
%
&&\quad k(\boldsymbol{B}^{'}\boldsymbol{\Sigma}_L^{-1}\boldsymbol{B}-\boldsymbol{M}^{'}\boldsymbol{\Sigma}_L^{-1}\boldsymbol{B}-\boldsymbol{B}^{'}\boldsymbol{\Sigma}_L^{-1}\boldsymbol{M}+\boldsymbol{M}^{'}\boldsymbol{\Sigma}_L^{-1}\boldsymbol{M})+T(\bar{    \boldsymbol{Y}}^{'}\boldsymbol{\Sigma}_L^{-1}\bar{\boldsymbol{Y}}-\boldsymbol{B}^{'}\boldsymbol{\Sigma}_L^{-1}\bar{\boldsymbol{Y}}-\bar{    \boldsymbol{Y}}^{'}\boldsymbol{\Sigma}_L^{-1}\boldsymbol{B}+\boldsymbol{B}^{'}\boldsymbol{\Sigma}_L^{-1}\boldsymbol{B})\nonumber\\
&&\quad+\frac{\left(k\boldsymbol{M}^{'}+T\bar{    \boldsymbol{Y}}^{'}\right)\boldsymbol{\Sigma}_L^{-1}\left(k    \boldsymbol{M}+T\bar{    \boldsymbol{Y}}\right)}{(k+T)}-\frac{\left(k\boldsymbol{M}^{'}+T\bar{    \boldsymbol{Y}}^{'}\right)\boldsymbol{\Sigma}_L^{-1}\left(k\boldsymbol{M}+T\bar{    \boldsymbol{Y}}\right)}{(k+T)}\nonumber\\
&&=\left(k+T\right)\left(\boldsymbol{B}-\frac{k\boldsymbol{M}+T\bar{\boldsymbol{Y}}^{'}}{(k+T)}\right)^{'}\boldsymbol{\Sigma}_L^{-1}\left(\boldsymbol{B}-\frac{k\boldsymbol{M}+T\bar{\boldsymbol{Y}}}{(k+T)}\right)+kT\frac{\left(\boldsymbol{M}-\bar{\boldsymbol{Y}}\right)^{'}\boldsymbol{\Sigma}_L^{-1}\left(\boldsymbol{M}-\bar{\boldsymbol{Y}}\right)}{\left(k+T\right)}
\end{eqnarray*}
the quantity $\mathcal{A}$ in the exponential term can be written as
\begin{eqnarray*}\mathcal{A}&=&\left(k+T\right)\left(\boldsymbol{B}-\frac{k\boldsymbol{M}+T\bar{\boldsymbol{Y}}}{(k+T)}\right)^{'}\boldsymbol{\Sigma}_L^{-1}\left(\boldsymbol{B}-\frac{k\boldsymbol{M}+T\bar{\boldsymbol{Y}}}{(k+T)}\right)\boldsymbol{V}^{-1}+kT\frac{\left(\boldsymbol{M}-\bar{\boldsymbol{Y}}\right)^{'}\boldsymbol{\Sigma}_L^{-1}\left(\boldsymbol{M}-\bar{\boldsymbol{Y}}\right)}{\left(k+T\right)}\boldsymbol{V}^{-1} \nonumber\\
&+&\sum_{i = 1}^{T}{\left(\boldsymbol{Y}_i - \bar{\boldsymbol{Y}} \right)^{'}\boldsymbol{\Sigma}_L^{- 1}\left(\boldsymbol{Y}_i - \bar{\boldsymbol{Y}} \right)} \boldsymbol{V}^{-1}+\boldsymbol{\Psi} \boldsymbol{V}^{-1}=\lbrack k_{*}\left(\boldsymbol{B}-\boldsymbol{M}_{*}\right)^{'}\boldsymbol{\Sigma}_L^{-1}\left(\boldsymbol{B}-\boldsymbol{M}_{*}\right)+\boldsymbol{\Psi}_{*}\rbrack \boldsymbol{V}^{-1},
\end{eqnarray*}}
where $k_{*}=k+T$, $m_{*}=m+Tp$, $\boldsymbol{M}_{*}=\frac{k\boldsymbol{M}+T\bar{\boldsymbol{Y}}}{k+T}$, and
 $\boldsymbol{\Psi}_{*}=\boldsymbol{\Psi}+kT\frac{\left(\boldsymbol{M}-\bar{\boldsymbol{Y}}\right)^{'}\boldsymbol{\Sigma}_L^{-1}\left(\boldsymbol{M}-\bar{\boldsymbol{Y}}\right)}{\left(k+T\right)}+T\boldsymbol{S}$. 
 
The integral at the denominator is
\begin{eqnarray}
&&\int_{\boldsymbol{V}>0}\int_{\boldsymbol{B}}{\left|\boldsymbol{V}\right|^{-\frac{m+p}{2}}|\boldsymbol{V}|^{-\frac{Tp}{2}}}\exp\left\{ - \frac{1}{2}\text{tr}\left\lbrack\mathcal{A}\right\rbrack\right\}d\boldsymbol{B}d\boldsymbol{V}=\nonumber\\
&&=\int_{\boldsymbol{V}>0}\int_{\boldsymbol{B}}{\left|\boldsymbol{V}\right|^{-\frac{m+Tp}{2}}|\boldsymbol{V}|^{-\frac{p}{2}}}\exp\left\{ - \frac{1}{2}\text{tr}\left\lbrack k_{*}\left(\boldsymbol{B}-\boldsymbol{M}_{*}\right)^{'}\boldsymbol{\Sigma}_L^{-1}\left(\boldsymbol{B}-\boldsymbol{M}_{*}\right)+\boldsymbol{\Psi}_{*}\right\rbrack \boldsymbol{V}^{-1}\right\}d\boldsymbol{B}d\boldsymbol{V}\nonumber\\
&&=\int_{\boldsymbol{V}>0}\int_{\boldsymbol{B}}{|\boldsymbol{V}|^{-\frac{p}{2}}}\exp\left\{ - \frac{1}{2}\text{tr}\left\lbrack k_{*}\left(\boldsymbol{B}-\boldsymbol{M}_{*}\right)^{'}\boldsymbol{\Sigma}_L^{-1}\left(\boldsymbol{B}-\boldsymbol{M}_{*}\right)\right\rbrack \boldsymbol{V}^{-1}\right\}d\boldsymbol{B} d\boldsymbol{V}\nonumber\\
&&=(2\pi)^{\frac{np}{2}}|\boldsymbol{\Sigma}_L/k_{*}|^{\frac{n}{2}}\int_{\boldsymbol{V}>0}\left|\boldsymbol{V}\right|^{-\frac{m_{*}}{2}}\exp\left\{ - \frac{1}{2}\text{tr}\left\lbrack \boldsymbol{\Psi}_{*}\right\rbrack \boldsymbol{V}^{-1}\right\}d\boldsymbol{V}\nonumber\\
&&=(2\pi)^{\frac{np}{2}}|\boldsymbol{\Sigma}_L/k_{*}|^{\frac{n}{2}}2^{\frac{(m_{*}-n-1)n}{2}}\Gamma_{n}\left(\frac{m_{*}-n-1}{2}\right)\left|\boldsymbol{\Psi}_{*}\right|^{-\frac{m_{*}-n-1}{2}}.
\end{eqnarray}
Thus the ratio 
\begin{eqnarray}
&&\frac{{\left|\boldsymbol{V}\right|^{-\frac{m+Tp}{2}}|\boldsymbol{V}|^{-\frac{p}{2}}}\exp\left\{ - \frac{1}{2}\text{tr}\left\lbrack k_{*}\left(\boldsymbol{B}-\boldsymbol{M}_{*}\right)^{'}\boldsymbol{\Sigma}_L^{-1}\left(\boldsymbol{B}-\boldsymbol{M}_{*}\right)+\boldsymbol{\Psi}{*}\right\rbrack \boldsymbol{V}^{-1}\right\}}{(2\pi)^{\frac{np}{2}}|\boldsymbol{\Sigma}_L/k_{*}|^{\frac{n}{2}}2^{\frac{(m_{*}-n-1)n}{2}}\Gamma_{n}\left(\frac{m_{*}-n-1}{2}\right)\left|\boldsymbol{\Psi}_{*}\right|^{-\frac{m_{*}-n-1}{2}}}\nonumber\\
&=&\frac{{\left|\boldsymbol{\Psi}_{*}\right|^{\frac{m_{*}-n-1}{2}}}\exp\left\{ - \frac{1}{2}\text{tr}\left\lbrack \left(\boldsymbol{B}-\boldsymbol{M}_{*}\right)^{'}\boldsymbol{\Sigma}_{*}^{-1}\left(\boldsymbol{B}-\boldsymbol{M}_{*}\right)+\boldsymbol{\Psi}{*}\right\rbrack \boldsymbol{V}^{-1}\right\}}{(2\pi)^{\frac{np}{2}}|\boldsymbol{\Sigma}_{*}|^{\frac{n}{2}}2^{\frac{(m_{*}-n-1)n}{2}}\Gamma_{n}\left(\frac{m_{*}-n-1}{2}\right)|\boldsymbol{V}|^{\frac{m_{*}+p}{2}}}
\end{eqnarray}
returns the density of a $\mathcal{MNIW}\left(\boldsymbol{B},\boldsymbol{V}|\boldsymbol{M}_{*},\boldsymbol{\Sigma}_{*},\boldsymbol{\Psi}_{*},k_{*},m_{*}\right)$, where $\boldsymbol{\Sigma}_{*}=\boldsymbol{\Sigma}_L/k_{*}$.

\end{proof}

\subsection{Proof of Proposition \ref{prop:posteriorBF}}

\begin{proof}[Proof of Proposition \ref{prop:posteriorBF}]
\begin{enumerate}[i)]
    \item Under Assumptions \ref{ass1} and \ref{ass2},  apply Propositions \ref{prop:gupta2} and \ref{prop:posteriormatrix}  with
$\boldsymbol{X}_1=\boldsymbol{B}|\mathbf{Y}$, 
$\boldsymbol{X}_2=\boldsymbol{Y}_{t}$, $\boldsymbol{C}= \boldsymbol{O}_{p, n}$, $\boldsymbol{D}=\boldsymbol{I}_{p}$,$\boldsymbol{F}=\boldsymbol{M}_{\ast}$, $\boldsymbol{\Sigma}_1=\boldsymbol{\Sigma}_{\ast}$ and $\boldsymbol{\Sigma}_2=\boldsymbol{\Sigma}_L$, where
$\boldsymbol{Y}_{t}$ is the sample at time $t$ and find:
\begin{equation} 
\boldsymbol{Z} = \begin{pmatrix}
\boldsymbol{B}|\mathbf{Y} \\
\boldsymbol{Y}_{t} \\
\end{pmatrix}
\sim 
N_{p + p,n}\left\lbrack \begin{pmatrix}
\boldsymbol{M}_{*} \\
\boldsymbol{M}_{*} \\
\end{pmatrix},\begin{pmatrix}
\boldsymbol{\Sigma}_{*} & \boldsymbol{\Sigma}_{*} \\
\boldsymbol{\Sigma}_{*} & \boldsymbol{\Sigma}_{d} \\
\end{pmatrix}\otimes \boldsymbol{V} \right\rbrack, \label{A2.17}
\end{equation}
where ${\boldsymbol{\Sigma}_{d} = \boldsymbol{\Sigma}}_{L} + \boldsymbol{\Sigma}_{*}$. The posterior predictive distribution returns the numerator of the BF:
\begin{eqnarray}
&&p\left( \boldsymbol{Y}_{t}|\mathcal{D}_{t - 1} \right) = \int_{}^{}{p\left( \boldsymbol{Y}_{t}|\boldsymbol{\theta}_{t} \right)p\left( \boldsymbol{\theta}_{t}|\mathcal{D}_{t - 1} \right)}d\boldsymbol{\theta}_{t} \\
&&= \int_{}^{}{f\left( \boldsymbol{Y}_{t}|\boldsymbol{B},\boldsymbol{V},\boldsymbol{\Sigma}_{L} \right)f\left( \boldsymbol{B}|\mathbf{Y},\boldsymbol{M},\boldsymbol{V},\boldsymbol{\Sigma}_{P} \right)}d\boldsymbol{B}\\
&&= \frac{\exp\left\{ - \frac{1}{2}\text{tr}\left\lbrack \boldsymbol{\Sigma}_{d}^{- 1}\left( \boldsymbol{Y}_{t} - \boldsymbol{M}_{*} \right)\boldsymbol{V}^{- 1}\left( \boldsymbol{Y}_{t} - \boldsymbol{M}_{*} \right)^{'} \right\rbrack \right\}}{(2\pi)^{\frac{np}{2}}\left| \boldsymbol{\Sigma}_{d} \right|^{\frac{n}{2}}|\boldsymbol{V}|^{\frac{p}{2}}}\label{A2.18}
\end{eqnarray}
that is $\boldsymbol{Y}_{t}|\mathcal{D}_{t - 1} \sim \mathcal{N}_{p,n}\left( \boldsymbol{M}_{*},\boldsymbol{\Sigma}_{d}, \boldsymbol{V} \right)$. The denominator of the BF is derived from the alternative distribution of the parameter
\begin{equation}
p_{A}\left(\boldsymbol{\theta}_{t}|\mathcal{D}_{t - 1} \right)= C(\alpha_t)p\left( \boldsymbol{B}|\mathbf{Y},\boldsymbol{M},\boldsymbol{V},\boldsymbol{\Sigma}_{P} \right)^{\alpha_t}= \frac{\exp\left\{ - \frac{1}{2}\text{tr}\left\lbrack \boldsymbol{\Sigma}_{A,*}^{- 1}\left( \boldsymbol{B} - \boldsymbol{M}_{*} \right)\boldsymbol{V}^{- 1}\left( \boldsymbol{B} - \boldsymbol{M}_{*} \right)^{'} \right\rbrack \right\}}{(2\pi)^{\frac{np}{2}}\left| \boldsymbol{\Sigma}_{A,*} \right|^{\frac{n}{2}}|\boldsymbol{V}|^{\frac{p}{2}}}
\label{A3.1}
\end{equation}
with $\boldsymbol{\Sigma}_{A,*} = \boldsymbol{\Sigma}_{*}/\alpha_{t}$ and the constant

\begin{eqnarray}
    C(\alpha_t) &= \frac{(2\pi)^{\frac{\alpha_{t}np}{2}}\alpha_{t}^{\frac{\alpha_{t}np}{2}}\left| \frac{\boldsymbol{\Sigma}_{*}}{\alpha_{t}} \right|^{\frac{\alpha_{t}n}{2}}|\boldsymbol{V}|^{\frac{\alpha_{t}p}{2}}}{(2\pi)^{\frac{np}{2}}\left| \frac{\boldsymbol{\Sigma}_{*}}{\alpha_{t}} \right|^{\frac{n}{2}}|\boldsymbol{V}|^{\frac{p}{2}}}= \frac{{\alpha_{t}^{p\frac{n}{2}}(2\pi)}^{\frac{\alpha_{t}np}{2}}\left| \boldsymbol{\Sigma}_{*} \right|^{\frac{\alpha_{t}n}{2}}|\boldsymbol{V}|^{\frac{\alpha_{t}p}{2}}}{(2\pi)^{\frac{np}{2}}\left| \boldsymbol{\Sigma}_{*} \right|^{\frac{n}{2}}|\boldsymbol{V}|^{\frac{p}{2}}}.
\end{eqnarray}

We apply Prop. \ref{prop:gupta2} with $i = t$:
\begin{equation}
Z = \begin{bmatrix}
\left( \boldsymbol{B}|\mathbf{Y} \right)_{A} \\
\boldsymbol{Y}_{t} \\
\end{bmatrix}\sim 
\mathcal{N}_{p + p,n}\left\lbrack \begin{pmatrix}
\ \boldsymbol{M}_{*} \\
\ \boldsymbol{M}_{*} \\
\end{pmatrix},\begin{pmatrix}
\boldsymbol{\Sigma}_{A,*} & \boldsymbol{\Sigma}_{A,*} \\
\boldsymbol{\Sigma}_{A,*} & \boldsymbol{\Sigma}_{A,d} \\
\end{pmatrix},\boldsymbol{V} \right\rbrack,
\label{A3.2}
\end{equation}
where $\left(\boldsymbol{B}|\mathbf{Y} \right)_{A}$ means that we consider the
alternative distribution as is in \eqref{A3.1} and $\boldsymbol{\Sigma}_{A,d} = \boldsymbol{\Sigma}_{L} + \boldsymbol{\Sigma}_{A,*}$. We get:
\begin{eqnarray}
&&p_{A}\left( \boldsymbol{Y}_{t}|\mathcal{D}_{t - 1} \right) = \int_{}^{}{p\left( \boldsymbol{Y}_{t}|\boldsymbol{\theta}_{t} \right)p_{A}\left( \boldsymbol{\theta}_{t}|\mathcal{D}_{t - 1} \right)}d\boldsymbol{\theta}_{t} = \int_{}^{}{f\left( \boldsymbol{Y}_{t}|\boldsymbol{B},\boldsymbol{V},\boldsymbol{\Sigma}_{L} \right)f_{A}\left( \boldsymbol{B}|\mathbf{Y},\boldsymbol{M},\boldsymbol{V},\boldsymbol{\Sigma}_{P} \right)}d\boldsymbol{B}\nonumber\\
&&= \frac{\exp\left\{ - \frac{1}{2}\text{tr}\left\lbrack \boldsymbol{\Sigma}_{A,d}^{- 1}\left( \boldsymbol{Y}_{t} - \ \boldsymbol{M}_{*} \right)\boldsymbol{V}^{- 1}\left( \boldsymbol{Y}_{t} -  M_{*} \right)^{'} \right\rbrack \right\}}{(2\pi)^{\frac{np}{2}}\left| \boldsymbol{\Sigma}_{A,d} \right|^{\frac{n}{2}}|\boldsymbol{V}|^{\frac{p}{2}}}\label{A3.3}
\end{eqnarray}
that is the pdf of $\mathcal{N}_{p,n}\left(\boldsymbol{M}_{*},\boldsymbol{\Sigma}_{A,d}, \boldsymbol{V} \right)$. The BF follows from its definition
\begin{equation}
H_{t} = \frac{p\left( \boldsymbol{Y}_{t}|\mathcal{D}_{t - 1} \right)}{p_{A}\left( \boldsymbol{Y}_{t}|\mathcal{D}_{t - 1} \right)} = 
\frac{\left| \boldsymbol{\Sigma}_{A,d} \right|^{\frac{n}{2}}}{\left| \boldsymbol{\Sigma}_{d} \right|^{\frac{n}{2}}}\exp\left\{ - \frac{1}{2}\text{tr}\left\lbrack \left( \boldsymbol{\Sigma}_{d}^{- 1} - \boldsymbol{\Sigma}_{A,d}^{- 1} \right)\left( \boldsymbol{Y}_{t} - \ \boldsymbol{M}_{*} \right)\boldsymbol{V}^{- 1}\left( \boldsymbol{Y}_{t} - \ \boldsymbol{M}_{*} \right)^{'} \right\rbrack \right\}\label{A4.1}.
\end{equation}
\item Under Assumptions \ref{ass1} and \ref{ass3}, since the observation $\boldsymbol{Y}_{t}$ follows a matrix normal with density $f\left( \boldsymbol{Y}_{t}|\boldsymbol{B}_t,\boldsymbol{V}_t,\boldsymbol{\Sigma} \right)$, the posterior predictive at the numerator of the BF is:
{\small
\begin{eqnarray}p\left( \boldsymbol{Y}_{t}|\mathcal{D}_{t - 1} \right) = \int_{\boldsymbol{V}_t>0}\int_{\boldsymbol{B}_t}{f\left( \boldsymbol{Y}_{t}|\boldsymbol{B}_t,\boldsymbol{V}_t,\boldsymbol{\Sigma}_L \right)f\left( \boldsymbol{B}_t,\boldsymbol{V}_t|\mathbf{Y},\boldsymbol{M}_{*},\boldsymbol{\Sigma}_{*},\boldsymbol{\Psi}_{*},m_{*} \right)}d\boldsymbol{B}_td\boldsymbol{V}_t,
\end{eqnarray}}
where, from Proposition \ref{prop:posteriormatrixUnknownV}, the argument of the integral is:
\begin{equation}
\frac{{\left|\boldsymbol{\Psi}_{*}\right|^{\frac{m_{*}-n-1}{2}}}\exp\left\{ - \frac{1}{2}\text{tr}\left\lbrack\left( k_{*}\left(\boldsymbol{B}_t-\boldsymbol{M}_{*}\right)^{'}\boldsymbol{\Sigma}_L^{-1}\left(\boldsymbol{B}_t-\boldsymbol{M}_{*}\right)+\boldsymbol{\Psi}_{*}+\left( \boldsymbol{Y}_{t} - \boldsymbol{B}_t \right)^{'}\boldsymbol{\Sigma}_L^{- 1}\left( \boldsymbol{Y}_{t} - \boldsymbol{B}_t \right)\right)\boldsymbol{V}_t^{-1}\right\rbrack \right\}}{(2\pi)^{\frac{np}{2}}|\boldsymbol{\Sigma}_L/k_{*}|^{\frac{n}{2}}2^{\frac{(m_{*}-n-1)n}{2}}\Gamma_{n}\left(\frac{m_{*}-n-1}{2}\right)|\boldsymbol{V}_t|^{\frac{m_{*}+p}{2}}(2\pi)^{\frac{np}{2}}\left| \boldsymbol{\Sigma}_L \right|^{\frac{n}{2}}|\boldsymbol{V}_t|^{\frac{p}{2}}}.\nonumber
\end{equation}
The matrix premultiplying $\boldsymbol{V}_t^{-1}$ in the argument of the $tr$ is:
\begin{eqnarray}
&& k_{*}\left(\boldsymbol{B}_t-\boldsymbol{M}_{*}\right)^{'}\boldsymbol{\Sigma}_L^{-1}\left(\boldsymbol{B}_t-\boldsymbol{M}_{*}\right)+\boldsymbol{\Psi}_{*}+\left( \boldsymbol{Y}_{t} - \boldsymbol{B}_t \right)^{'}\boldsymbol{\Sigma}_L^{- 1}\left( \boldsymbol{Y}_{t} - \boldsymbol{B}_t \right)=\nonumber\\
&&\quad
k_{d}\left(\boldsymbol{B}_t-\boldsymbol{M}_{d}\right)^{'}\boldsymbol{\Sigma}_L^{-1}\left(\boldsymbol{B}_t-\boldsymbol{M}_{d}\right)+\boldsymbol{\Psi}_{d}= \\ && \quad \left(\boldsymbol{B}_t-\boldsymbol{M}_{d}\right)^{'}\boldsymbol{\Sigma}_{d}^{-1}\left(\boldsymbol{B}_t-\boldsymbol{M}_{d}\right)+\boldsymbol{\Psi}_{d},
\end{eqnarray}
where the equality follows from the fact that $\boldsymbol{S}$ in $\boldsymbol{\Psi}_{d}$ is equal to zero and $\bar{\boldsymbol{Y}}=\boldsymbol{Y}_{t}$ for $T=1$, and where we defined $k_{d}=k_{*}+1,\,   m_{d}=m_{*}+p,\,$ and
\begin{equation}
\boldsymbol{M}_{d}=\frac{k_{*}\boldsymbol{M}_{*}+\boldsymbol{Y}_{t}}{k_{*}+1},\quad \boldsymbol{\Sigma}_{d}=\boldsymbol{\Sigma}_L/k_{d},\quad
\boldsymbol{\Psi}_{d}=\boldsymbol{\Psi}_{*}+k_{*}\frac{\left(\boldsymbol{M}_{*}-\boldsymbol{Y}_{t}\right)^{'}\boldsymbol{\Sigma}_L^{-1}\left(\boldsymbol{M}_{*}-\boldsymbol{Y}_{t}\right)}{\left(k_{*}+1\right)}.
\end{equation}
The integral becomes:
\begin{eqnarray}
&&\int_{\boldsymbol{V}_t>0}\frac{\left|\boldsymbol{\Psi}_{*}\right|^{\frac{m_{*}-n-1}{2}}k_{*}^\frac{np}{2}}{(2\pi)^{\frac{np}{2}}k_{d}^\frac{np}{2}2^{\frac{(m_{*}-n-1)n}{2}}\Gamma_{n}\left(\frac{m_{*}-n-1}{2}\right)|\boldsymbol{V}_t|^{\frac{m_{d}}{2}}\left| \boldsymbol{\Sigma}_L \right|^{\frac{n}{2}}}\nonumber\\
&&\quad\quad\int_{\boldsymbol{B}_t}\frac{\exp\left\{ - \frac{1}{2}\text{tr}\left\lbrack\left( \left(\boldsymbol{B}_t-\boldsymbol{M}_{d}\right)^{'}\boldsymbol{\Sigma}_d^{-1}\left(\boldsymbol{B}_t-\boldsymbol{M}_{d}\right)+\boldsymbol{\Psi}_{d}\right) \boldsymbol{V}_t^{-1}\right\rbrack \right\}}{|\boldsymbol{\Sigma}_{d}|^{\frac{n}{2}} |\boldsymbol{V}_t|^{\frac{p}{2}}(2\pi)^{\frac{np}{2}}}d\boldsymbol{B}_td\boldsymbol{V}_t\nonumber\\
&&=\int_{\boldsymbol{V}_t>0}\frac{{\left|\boldsymbol{\Psi}_{*}\right|^{\frac{m_{*}-n-1}{2}}}k_{*}^\frac{np}{2}\exp\left\{ - \frac{1}{2}\text{tr}\left\lbrack\boldsymbol{\Psi}_{d} \boldsymbol{V}_t^{-1}\right\rbrack \right\}}{(2\pi)^{\frac{np}{2}}k_{d}^\frac{np}{2}2^{\frac{(m_{*}-n-1)n}{2}}\Gamma_{n}\left(\frac{m_{*}-n-1}{2}\right)|\boldsymbol{V}_t|^{\frac{m_{d}}{2}}\left| \boldsymbol{\Sigma}_L \right|^{\frac{n}{2}}}d\boldsymbol{V}_{t}\nonumber\\
&&=\int_{\boldsymbol{V}_t>0}\frac{\left|\boldsymbol{\Psi}_{*}\right|^{\frac{m_{*}-n-1}{2}}\left|\boldsymbol{\Psi}_{d}\right|^{\frac{m_{d}-n-1}{2}}k_{*}^\frac{np}{2}\Gamma_{n}\left(\frac{m_{d}-n-1}{2}\right)\exp\left\{ - \frac{1}{2}\text{tr}\left\lbrack\boldsymbol{\Psi}_{d} \boldsymbol{V}_t^{-1}\right\rbrack \right\}}{(2\pi)^{\frac{np}{2}}k_{d}^\frac{np}{2}\left|\boldsymbol{\Psi}_{d}\right|^{\frac{m_{d}-n-1}{2}}2^{\frac{(m_{d}-n-1)n}{2}}2^{-\frac{pn}{2}}\Gamma_{n}\left(\frac{m_{*}-n-1}{2}\right)\Gamma_{n}\left(\frac{m_{d}-n-1}{2}\right)|\boldsymbol{V}_t|^{\frac{m_{d}}{2}}\left| \boldsymbol{\Sigma}_L \right|^{\frac{n}{2}}}d\boldsymbol{V}_{t}\nonumber\\
&&=\frac{\left|\boldsymbol{\Psi}_{*}\right|^{\frac{m_{*}-n-1}{2}}k_{*}^\frac{np}{2}\Gamma_{n}\left(\frac{m_{d}-n-1}{2}\right)}{\pi^{\frac{np}{2}}k_{d}^\frac{np}{2}\left|\boldsymbol{\Psi}_{d}\right|^{\frac{m_{d}-n-1}{2}}\Gamma_{n}\left(\frac{m_{*}-n-1}{2}\right)\left| \boldsymbol{\Sigma}_L \right|^{\frac{n}{2}}}
\end{eqnarray}
which is the density of a matrix Student-t $\mathcal{T}_{p,n}(m_{*}-2n, \boldsymbol{M}_{*}, \boldsymbol{\Sigma}_L, \boldsymbol{L}_{*})$ following the definition in \cite{Gup99} p. 134, 
given:
\begin{eqnarray}
&&\left|\boldsymbol{\Psi}_{d}\right|^{\frac{m_{d}-n-1}{2}}=\left|\boldsymbol{\Psi}_{*}+k_{*}\frac{\left(\boldsymbol{M}_{*}-\boldsymbol{Y}_{t}\right)^{'}\boldsymbol{\Sigma}_L^{-1}\left(\boldsymbol{M}_{*}-\boldsymbol{Y}_{t}\right)}{\left(k_{*}+1\right)}\right|^{\frac{m_{d}-n-1}{2}}\nonumber\\
&&=\left|\boldsymbol{\Psi}_{*}\left[ \boldsymbol{I}_n+\boldsymbol{L}_{*}^{-1}\left(\boldsymbol{M}_{1,*}-\boldsymbol{Y}_{1,t}\right)\boldsymbol{\Sigma}_L^{-1}\left(\boldsymbol{M}_{1,*}-\boldsymbol{Y}_{1,t}\right)^{'}\right]\right|^{\frac{m_{d}-n-1}{2}}\nonumber\\
&&=\left|\boldsymbol{\Psi}_{*}\right|^{\frac{m_*+p-n-1}{2}}\left| \boldsymbol{I}_n+\boldsymbol{L}_{*}^{-1}\left(\boldsymbol{M}_{1,*}-\boldsymbol{Y}_{1,t}\right)\boldsymbol{\Sigma}_L^{-1}\left(\boldsymbol{M}_{1,*}-\boldsymbol{Y}_{1,t}\right)^{'}\right|^{\frac{m_*+p-n-1}{2}}\nonumber\\
&&=\frac{(\pi)^{-\frac{np}{2}}\left|\boldsymbol{L}_{*}\right|^{-\frac{p}{2}}\left| \boldsymbol{\Sigma}_L \right|^{-\frac{n}{2}}\Gamma_{n}\left(\frac{m_{1}+p+n-1}{2}\right)}{\left| \boldsymbol{I}_n+\boldsymbol{L}_{*}^{-1}\left(\boldsymbol{M}_{1,*}-\boldsymbol{Y}_{1,t}\right)\boldsymbol{\Sigma}_L^{-1}\left(\boldsymbol{M}_{1,*}-\boldsymbol{Y}_{1,t}\right)^{'}\right|^{\frac{m_1+p+n-1}{2}}\Gamma_{n}\left(\frac{m_{1}+n-1}{2}\right)},
\end{eqnarray}
is the probability density function of a $\mathcal{T}_{n,p}(m_{*}-2n, \boldsymbol{M}_{1,*}, \boldsymbol{L}_{*}, \boldsymbol{\Sigma}_L)$ for the quantity $\boldsymbol{Y}_{1,t}$,
where we have defined $\boldsymbol{L}_{*}=\boldsymbol{\Psi}_{*}k_{d}/k_{*}$, $\boldsymbol{M}_{1,*}=\boldsymbol{M}_{*}^{'}$, $\boldsymbol{Y}_{1,t}=\boldsymbol{Y}_{t}^{'}$, $m_d=m_*+p$, and $m_1=m_{*}-2n=m+Tp-2n$ (assuming $m>2n-Tp$) and given the Theorem 4.4.4 in \citet{Gup99} about the probability density function of the transposed of a matrix Student-t distribution. 

For what concerns the denominator of the BF, after some algebraic manipulations, one obtains:
\begin{eqnarray}
&&p_{A}(\boldsymbol{\theta}_t |\mathcal{D}_{t-1})=C(\alpha_t)p(\boldsymbol{B}_t,\boldsymbol{V}_t|\mathbf{Y},\boldsymbol{M}_{*},\boldsymbol{\Sigma}_{*},\boldsymbol{\Psi}_{*},m_{*})^{\alpha_t}\nonumber\\ &&=\frac{{\left|\alpha_t\boldsymbol{\Psi}_{*}\right|^{\frac{m_{A,*}-n-1}{2}}}\exp\left\{ - \frac{1}{2}\text{tr}\left\lbrack \alpha_t k_{*}\left(\boldsymbol{B}_t-\boldsymbol{M}_{*}\right)^{'}\boldsymbol{\Sigma}_L^{-1}\left(\boldsymbol{B}_t-\boldsymbol{M}_{*}\right)+\alpha_t\boldsymbol{\Psi}_{*}\right\rbrack \boldsymbol{V}_t^{-1}\right\}}{(2\pi)^{\frac{np}{2}}|\boldsymbol{\Sigma}_L/\alpha_t k_{*}|^{\frac{n}{2}}2^{\frac{(m_{A,*}-n-1)n}{2}}\Gamma_{n}\left(\frac{m_{A,*}-n-1}{2}\right)|\boldsymbol{V}_t|^{\frac{m_{A,*}+p}{2}}}\nonumber\\
&&=\frac{{\left|\boldsymbol{\Psi}_{A,*}\right|^{\frac{m_{A,*}-n-1}{2}}}\exp\left\{ - \frac{1}{2}\text{tr}\left\lbrack \left(\boldsymbol{B}_t-\boldsymbol{M}_{*}\right)^{'}\boldsymbol{\Sigma}_{A,*}^{-1}\left(\boldsymbol{B}_t-\boldsymbol{M}_{*}\right)+\boldsymbol{\Psi}_{A,*}\right\rbrack \boldsymbol{V}_t^{-1}\right\}}{(2\pi)^{\frac{np}{2}}|\boldsymbol{\Sigma}_{A,*}|^{\frac{n}{2}}2^{\frac{(m_{A,*}-n-1)n}{2}}\Gamma_{n}\left(\frac{m_{A,*}-n-1}{2}\right)|\boldsymbol{V}_t|^{\frac{m_{A,*}+p}{2}}},\label{Pmatrix}
\end{eqnarray}
where we defined $m_{A,*}=\alpha_t\left(m_*+p\right)=\alpha_t m_d$, $\boldsymbol{\Sigma}_{A,*}=\boldsymbol{\Sigma}_{*}/\alpha_{t}$,  $\boldsymbol{\Psi}_{A,*}=\alpha_{t}\boldsymbol{\Psi}_{*}$, and the constant
\begin{equation}
C(\alpha_t) = \frac{{\left|\boldsymbol{\Psi}_{A,*}\right|^{\frac{m_{A,*}-n-1}{2}}}(2\pi)^{\frac{\alpha_{t}np}{2}}|\boldsymbol{\Sigma}_{*}|^{\frac{\alpha_{t}n}{2}}2^{\frac{\alpha_{t}(m_{*}-n-1)n}{2}}\Gamma_{n}\left(\frac{m_{*}-n-1}{2}\right)^{\alpha_{t}}}{(2\pi)^{\frac{np}{2}}|\boldsymbol{\Sigma}_{A,*}|^{\frac{n}{2}}2^{\frac{(m_{A,*}-n-1)n}{2}}\Gamma_{n}\left(\frac{m_{A,*}-n-1}{2}\right)\left|\boldsymbol{\Psi}_{*}\right|^\frac{\alpha_t (m_{*}-n-1)}{2}}.
\end{equation}
The posterior predictive under the alternative is
\begin{eqnarray}
p_A\left(\boldsymbol{Y}_{t}|\mathcal{D}_{t - 1} \right) &=& \frac{\left|\boldsymbol{\Psi}_{A,*}\right|^{\frac{m_{A,*}-n-1}{2}}k_{A,*}^\frac{np}{2}\Gamma_{n}\left(\frac{m_{A,d}-n-1}{2}\right)}{(\pi)^{\frac{np}{2}}k_{A,d}^\frac{np}{2}\left|\boldsymbol{\Psi}_{A,d}\right|^{\frac{m_{A,d}-n-1}{2}}\Gamma_{n}\left(\frac{m_{A,*}-n-1}{2}\right)\left| \boldsymbol{\Sigma}_L \right|^{\frac{n}{2}}}\label{PAmatrix}
\end{eqnarray}
which is a Student-t distribution $\mathcal{T}(m_{A,*} -2n, \boldsymbol{M}_{*}, \boldsymbol{\Sigma}_L, \boldsymbol{L}_{A, *})$ where $\boldsymbol{L}_{A, * } = \boldsymbol{\Psi}_{A, *}k_{A, d}/k_{A, *}$,
where $k_{A,*}=\alpha_t k_{*}$, $k_{A,d}=\alpha_t k_{*}+1$, $m_{A,d}=m_{A,*}+p$, and
\begin{equation*}\boldsymbol{\Psi}_{A,d}=\boldsymbol{\Psi}_{A,*}+\alpha_t k_{*}\left(\boldsymbol{M}_{*}-\boldsymbol{Y}_{t}\right)^{'}\boldsymbol{\Sigma}_L^{-1}\left(\boldsymbol{M}_{*}-\boldsymbol{Y}_{t}\right)/(\left(\alpha_t k_{*}+1\right))\end{equation*}
following an argument analogous to the one used for the numerator.
We conclude with the BF, which is the ratio of $p(\boldsymbol{Y}_{t} |\mathcal{D}_{t-1})$ given in Eq. \ref{Pmatrix} and $p_A(\boldsymbol{Y}_{t} |\mathcal{D}_{t-1})$ given in Eq. \ref{PAmatrix}: 
\begin{equation}
H_t(\alpha_t)=G\left|\boldsymbol{\Psi}_{A,d}\right|^{\frac{m_{A,d}-n-1}{2}}\left|\boldsymbol{\Psi}_{d}\right|^{-\frac{m_{d}-n-1}{2}}\end{equation}
where
\begin{equation}
G=\frac{\left|\boldsymbol{\Psi}_{*}\right|^{\frac{m_{*}-n-1}{2}}k_{*}^\frac{np}{2}\Gamma_{n}\left(\frac{m_{d}-n-1}{2}\right)k_{A,d}^\frac{np}{2}\Gamma_{n}\left(\frac{m_{A,*}-n-1}{2}\right)}{k_{d}^\frac{np}{2}\Gamma_{n}\left(\frac{m_{*}-n-1}{2}\right)\left|\boldsymbol{\Psi}_{A,*}\right|^{\frac{m_{A,*}-n-1}{2}}k_{A,*}^\frac{np}{2}\Gamma_{n}\left(\frac{m_{A,d}-n-1}{2}\right)}.
\end{equation}
$H_t\left(\alpha_{t}\right)$ goes to 1 when $\alpha_{t}\rightarrow 1^{-}$, $H_t(\alpha_t)\rightarrow \infty$ as $\alpha_t\rightarrow \underline{\alpha}^{+}$, due to the fact that $G\rightarrow \infty$ as $\alpha_t\rightarrow \underline{\alpha}^{+}$ because $\Gamma_n\left(\frac{m_{A,*}-n-1}{2}\right) \rightarrow \infty$. 
\end{enumerate}
\end{proof}

\subsection{Proof of Corollary \ref{coroll}}
\begin{proof}[Proof of Corollary \ref{coroll}]
\begin{enumerate}[i)]
    \item In the first case
    \begin{eqnarray}
&&C\left( \alpha_{t} \right)
= \frac{{\alpha_{t}^{p\frac{n}{2}}(2\pi)}^{\frac{\alpha_{t}np}{2}}\left| \boldsymbol{\Sigma}_{*} \right|^{\frac{\alpha_{t}n}{2}}|\boldsymbol{V}|^{\frac{\alpha_{t}p}{2}}}{(2\pi)^{\frac{np}{2}}\left| \boldsymbol{\Sigma}_{*} \right|^{\frac{n}{2}}|\boldsymbol{V}|^{\frac{p}{2}}}
\end{eqnarray}
which goes to zero for $\alpha_t\rightarrow 0^{+}$ and goes to one for $\alpha_t\rightarrow 1^{-}$.
\item In the second case, from the proof of Prop. \ref{prop:posteriorBF} the constant is
\begin{equation}
C(\alpha_t) = \frac{{\left|\boldsymbol{\Psi}_{A,*}\right|^{\frac{m_{A,*}-n-1}{2}}}(2\pi)^{\frac{\alpha_{t}np}{2}}|\boldsymbol{\Sigma}_{*}|^{\frac{\alpha_{t}n}{2}}2^{\frac{\alpha_{t}(m_{*}-n-1)n}{2}}\Gamma_{n}\left(\frac{m_{*}-n-1}{2}\right)^{\alpha_{t}}}{(2\pi)^{\frac{np}{2}}|\boldsymbol{\Sigma}_{A,*}|^{\frac{n}{2}}2^{\frac{(m_{A,*}-n-1)n}{2}}\Gamma_{n}\left(\frac{m_{A,*}-n-1}{2}\right)\left|\boldsymbol{\Psi}_{*}\right|^\frac{\alpha_t (m_{*}-n-1)}{2}}.
\end{equation}
Note that $m_{A,*}=\alpha_t m_d\rightarrow 0$ for $\alpha_t \rightarrow 0^{+}$ and the argument of the multivariate Gamma function $\Gamma_n \left(\left( m_{A,*}-n-1\right)/2\right)$ becomes negative. To guarantee the positivity of the constant, the condition $\Gamma_n((m_{A,\ast}-n-1)/2)>0$ is needed noticing that $m_{A,\ast} \leq m_{\ast} \leq m_d \leq m_{A,d}$. A sufficient condition is that the $n$-th element of the multivariate gamma product $\Gamma((m_{A,*}-n-1/2)$ is strictly positive, that is
$(m_{A,*}-n-1)/2+(1-n)/2=(\alpha_t(m+Tp+p)-p-n-1+1-n)/2>0$, or equivalently $\alpha_t>\underline{\alpha}$ where $\underline{\alpha}=(p+2n)/m_d$. Under this condition $C(\alpha_t)\rightarrow 0^+$ for $\alpha_t\rightarrow \underline{\alpha}^{+}$. Moreover, $C(\alpha_t)\rightarrow 1$ for $\alpha_t\rightarrow 1^{-}$.
\end{enumerate}
\end{proof}

\subsection{Proof of Proposition \ref{prop:BFmatrix}}
\begin{proof}[Proof of Proposition \ref{prop:BFmatrix}]
\textit{i)} Define:
\begin{eqnarray}
&&\boldsymbol{\Sigma}_{H} = \left\lbrack \left( \frac{1}{\alpha_{t}} - 1 \right)\left( \boldsymbol{\Sigma}_{L} + \boldsymbol{\Sigma}_{*} \right)^{- 1}\boldsymbol{\Sigma}_{*}\left( \boldsymbol{\Sigma}_{L} + \frac{\boldsymbol{\Sigma}_{*}}{\alpha_{t}} \right)^{- 1} \right\rbrack^{- 1}\label{A4.4}\\
&&= \left( \frac{1}{1 - \alpha_{t}} \right)\left( \alpha_{t}\boldsymbol{\Sigma}_{L} + \boldsymbol{\Sigma}_{*} \right)\boldsymbol{\Sigma}_{*}^{- 1}\left( \boldsymbol{\Sigma}_{L} + \boldsymbol{\Sigma}_{*} \right)
\end{eqnarray}
and $\kappa_t(\alpha_t)=\left| \boldsymbol{\Sigma}_{L} + \frac{\boldsymbol{\Sigma}_{*}}{\alpha_{t}} \right|^{\frac{n}{2}}\left| \boldsymbol{\Sigma}_{L} + \boldsymbol{\Sigma}_{*} \right|^{-\frac{n}{2}} = \left| \boldsymbol{\Sigma}_{A,d} \right|^{\frac{n}{2}}\left| \boldsymbol{\Sigma}_{d} \right|^{-\frac{n}{2}}$. Since: $(\boldsymbol{\Sigma}_{L} + \boldsymbol{\Sigma}_{*})^{- 1} - ( \boldsymbol{\Sigma}_{L} + \boldsymbol{\Sigma}_{*}/\alpha_{t})^{- 1} = (\boldsymbol{\Sigma}_{L} + \boldsymbol{\Sigma}_{*})^{- 1}\lbrack \boldsymbol{\Sigma}_{*}/\alpha_{t} - \boldsymbol{\Sigma}_{*}\rbrack (\boldsymbol{\Sigma}_{L} + \boldsymbol{\Sigma}_{*}/\alpha_{t})^{- 1}$ $=(1/(\alpha_{t} - 1))( \boldsymbol{\Sigma}_{L} + \boldsymbol{\Sigma}_{*})^{- 1}\boldsymbol{\Sigma}_{*}(\boldsymbol{\Sigma}_{L} + \boldsymbol{\Sigma}_{*}/\alpha_{t})^{- 1}$, then the BF can be written as
{\footnotesize
\begin{eqnarray}
&&H_{t} = \kappa_t(\alpha_t)\exp\left\{ - \frac{1}{2}\text{tr}\left\lbrack \left\lbrack \left( \boldsymbol{\Sigma}_{L} + \boldsymbol{\Sigma}_{*} \right)^{- 1} - \left( \boldsymbol{\Sigma}_{L} + \frac{\boldsymbol{\Sigma}_{*}}{\alpha_{t}} \right)^{- 1} \right\rbrack\left( \boldsymbol{Y}_{t} - \boldsymbol{M}_{*} \right)\boldsymbol{V}^{- 1}\left( \boldsymbol{Y}_{t} - \boldsymbol{M}_{*} \right)^{'} \right\rbrack \right\}\label{A4.2}\\
&& = 
\kappa_t(\alpha_t)\exp\left\{ - \frac{1}{2}\text{tr}\left\lbrack \boldsymbol{\Sigma}_{H}^{- 1}\left( \boldsymbol{Y}_{t} - \boldsymbol{M}_{*} \right)\boldsymbol{V}^{- 1}\left( \boldsymbol{Y}_{t} - \boldsymbol{M}_{*} \right)^{'} \right\rbrack \right\}.\label{A4.5}
\end{eqnarray}}
The assumptions $\left(1/\alpha_{t} - 1\right) > 0$, $\boldsymbol{\Sigma}_{*} > 0$, $\boldsymbol{\Sigma}_{L} > 0$,
$\boldsymbol{\Sigma}_{L} + \boldsymbol{\Sigma}_{*} > 0$ imply the eigenvalues of
$\boldsymbol{\Sigma}_{L} + \boldsymbol{\Sigma}_{*}$ and of $\boldsymbol{\Sigma}_{L} + \boldsymbol{\Sigma}_{*}/\alpha_{t}$ are greater than zero, and the product of the three matrices $\left( \boldsymbol{\Sigma}_{L} + \boldsymbol{\Sigma}_{*} \right)^{- 1}$,
$\boldsymbol{\Sigma}_{*}$, $\left( \boldsymbol{\Sigma}_{L} + \boldsymbol{\Sigma}_{*}/\alpha_{t} \right)^{- 1}$ has
positive eigenvalues. Moreover the assumption
$\boldsymbol{\Sigma}_{L},\ \boldsymbol{\Sigma}_{*}$ are symmetric and positive definite implies $\left( (1-\alpha_{t})^{-1} \right)\left( \boldsymbol{\Sigma}_{L} + \boldsymbol{\Sigma}_{*} \right)^{- 1}\boldsymbol{\Sigma}_{*}\left( \boldsymbol{\Sigma}_{L} + \boldsymbol{\Sigma}_{*}/\alpha_{t} \right)^{- 1}$ is symmetric and has positive eigenvalues, it is positive definite. 
Since the argument of the exponential function is negative, we get an upper bound for $H_{t}$:
\begin{equation}
H_{t} \leq \kappa_t(\alpha_t)\label{A4.6}
\end{equation}
such that $\lim_{\alpha_{t} \rightarrow 1^{-}}{\kappa_t =}1$. From Th. 7.7 in \cite{zhang2011matrix} it follows:
\begin{equation}\left| \boldsymbol{\Sigma}_{L} + \boldsymbol{\Sigma}_{*} \right| \geq \left| \boldsymbol{\Sigma}_{L} \right| + \left| \boldsymbol{\Sigma}_{*} \right|\Leftrightarrow \left| \boldsymbol{\Sigma}_{L} + \frac{\boldsymbol{\Sigma}_{*}}{\alpha_{t}} \right| \geq \left| \boldsymbol{\Sigma}_{L} \right| + \left| \frac{\boldsymbol{\Sigma}_{*}}{\alpha_{t}} \right| = \left| \boldsymbol{\Sigma}_{L} \right| + \frac{\left| \boldsymbol{\Sigma}_{*} \right|}{\alpha_{t}^{p}},\label{A4.7}
\end{equation}
or equivalently
\begin{equation}
\left| \boldsymbol{\Sigma}_{L} + \frac{\boldsymbol{\Sigma}_{*}}{\alpha_{t}} \right| - \left| \boldsymbol{\Sigma}_{L} + \boldsymbol{\Sigma}_{*} \right| \geq \frac{\left| \boldsymbol{\Sigma}_{*} \right|}{\alpha_{t}^{p}} - \left| \boldsymbol{\Sigma}_{*} \right| = \left( \frac{1}{\alpha_{t}^{p}} - 1 \right)\left| \boldsymbol{\Sigma}_{*} \right| \geq 0 \label{A4.8}
\end{equation}
and $\kappa_t \geq 1$. By Jacobi's formula on the derivative of the matrix determinant \citep[][]{magneu99}, we obtain
{\small \begin{eqnarray}
&&\partial_{\alpha_t} \kappa_t\left( \alpha_{t} \right) = \frac{n}{2}\frac{\left| \boldsymbol{\Sigma}_{L} + \frac{\boldsymbol{\Sigma}_{*}}{\alpha_{t}} \right|^{\frac{n}{2} - 1}}{\left| \boldsymbol{\Sigma}_{L} + \boldsymbol{\Sigma}_{*} \right|^{\frac{n}{2}}}\frac{d}{d\alpha_{t}}\left| \boldsymbol{\Sigma}_{L} + \frac{\boldsymbol{\Sigma}_{*}}{\alpha_{t}} \right|\nonumber\\
&&= \frac{n}{2}\kappa_t(\alpha_t)\text{tr}\left\lbrack \left( \boldsymbol{\Sigma}_{L} + \frac{\boldsymbol{\Sigma}_{*}}{\alpha_{t}} \right)^{- 1}\left( - \frac{\boldsymbol{\Sigma}_{*}}{\alpha_{t}^{2}} \right) \right\rbrack = - \frac{n}{2\alpha_{t}^{2}}\kappa_t(\alpha_t)\text{tr}\left\lbrack \left( \boldsymbol{\Sigma}_{L} + \frac{\boldsymbol{\Sigma}_{*}}{\alpha_{t}} \right)^{- 1}\boldsymbol{\Sigma}_{*} \right\rbrack.
\end{eqnarray}}
Note that  $\left( \boldsymbol{\Sigma}_{L} + \frac{\boldsymbol{\Sigma}_{*}}{\alpha_{t}} \right)^{- 1} > 0$ and
$\boldsymbol{\Sigma}_{*} > 0$, thus the trace of the product of positive matrices is
positive and $\partial_{\alpha_t} \kappa_t\left( \alpha_{t} \right) < 0$ for every
$\alpha_{t}$. In conclusion, we showed $\kappa_t\left( \alpha_{t} \right)$ is a decreasing function that goes to infinite when $\alpha_{t} \rightarrow 0^{+}$ and goes to 1 when $\alpha_{t} \rightarrow 1^{-}$.
\medskip

\textit{ii)} From its definition it trivially follows $\lim_{\alpha_{t} \rightarrow 1^{-}}{H_{t}(\alpha_t)}=1$.  Consider now the term inside the trace and write:
\begin{equation}\left( \boldsymbol{\Sigma}_{L} + \frac{\boldsymbol{\Sigma}_{*}}{\alpha_{t}} \right)^{- 1} = \boldsymbol{\Sigma}_{L}^{- 1} - \boldsymbol{\Sigma}_{L}^{- 1}\left\lbrack \boldsymbol{\Sigma}_{L}^{- 1} + \left( \frac{\boldsymbol{\Sigma}_{*}}{\alpha_{t}} \right)^{- 1} \right\rbrack^{- 1}\boldsymbol{\Sigma}_{L}^{- 1} = \boldsymbol{\Sigma}_{L}^{- 1} - \boldsymbol{\Sigma}_{L}^{- 1}\left\lbrack \boldsymbol{\Sigma}_{L}^{- 1} + {\alpha_{t}\boldsymbol{\Sigma}}_{*}^{- 1} \right\rbrack^{- 1}\boldsymbol{\Sigma}_{L}^{- 1}.\label{A4.9}
\end{equation}
This term goes to zero as $\alpha_{t} \rightarrow 0^{+}$, and the
exponential in \eqref{A4.5} has a bounded limit since the other terms
don't depend by $\alpha_{t}$. Furthermore, due to the:
\begin{equation}
\left| \boldsymbol{\Sigma}_{L} + \frac{\boldsymbol{\Sigma}_{*}}{\alpha_{t}} \right| = \left| \frac{{{\alpha_{t}\boldsymbol{\Sigma}}_{L} + \boldsymbol{\Sigma}}_{*}}{\alpha_{t}} \right| = \frac{1}{\alpha_{t}^{p}}\left| {{\alpha_{t}\boldsymbol{\Sigma}}_{L} + \boldsymbol{\Sigma}}_{*} \right|.\label{A4.10}
\end{equation}
We have:
\begin{equation}\lim_{\alpha_{t} \rightarrow 0^{+}}{\left| \boldsymbol{\Sigma}_{L} + \frac{\boldsymbol{\Sigma}_{*}}{\alpha_{t}} \right|^{\frac{n}{2}} = \lim_{\alpha_{t} \rightarrow 0^{+}}{\frac{1}{\alpha_{t}^{\frac{np}{2}}}\left| {{\alpha_{t}\boldsymbol{\Sigma}}_{L} + \boldsymbol{\Sigma}}_{*} \right|^{\frac{n}{2}} = \infty}}.\label{A4.11}
\end{equation}
From \eqref{A4.5} and \eqref{A4.11} if follows that $\lim_{\alpha_{t} \rightarrow 0^{+}}{\text{H}_{t}\left( \alpha_{t} \right)} = \infty$.
\medskip

\textit{iii)} Let us define
$\Tilde{\boldsymbol{A}} = \left( \boldsymbol{Y}_{t} -  \boldsymbol{M}_{*} \right)\boldsymbol{V}^{- 1}\left( \boldsymbol{Y}_{t} - \boldsymbol{M}_{*} \right)^{'}$
and $\kappa_t\left( \alpha_{t} \right) = \left| \boldsymbol{\Sigma}_{L} + \boldsymbol{\Sigma}_{*}/\alpha_{t} \right|^{\frac{n}{2}}\left| \boldsymbol{\Sigma}_{L} + \boldsymbol{\Sigma}_{*} \right|^{-\frac{n}{2}}$
then the derivative of the BF is

\begin{equation}
\partial_{\alpha_t} H_{t}\left( \alpha_{t} \right)= \exp\left\{ - \frac{1}{2}\text{tr}\left\lbrack \boldsymbol{\Sigma}_{H}^{- 1}\Tilde{\boldsymbol{A}}  \right\rbrack \right\}\frac{d}{d\alpha_{t}}\kappa_t(\alpha_{t}) + \kappa_t(\alpha_{t})\frac{d}{d\alpha_{t}}\exp\left\{ - \frac{1}{2}\text{tr}\left\lbrack \boldsymbol{\Sigma}_{H}^{- 1}\Tilde{\boldsymbol{A}}  \right\rbrack \right\}.
\end{equation}
From \textit{i)} we know:
\begin{equation}
\partial_{\alpha_t}\kappa_t\left( \alpha_{t} \right)=- \frac{n\left| \boldsymbol{\Sigma}_{L} + \frac{\boldsymbol{\Sigma}_{*}}{\alpha_{t}} \right|^{\frac{n}{2}}}{2\alpha_{t}^{2}\left| \boldsymbol{\Sigma}_{L} + \boldsymbol{\Sigma}_{*} \right|^{\frac{n}{2}}}\text{tr}\left\lbrack \left( \boldsymbol{\Sigma}_{L} + \frac{\boldsymbol{\Sigma}_{*}}{\alpha_{t}} \right)^{- 1}\boldsymbol{\Sigma}_{*} \right\rbrack
\end{equation}
and since
$\boldsymbol{\Sigma}_{H}^{- 1} = (1/\alpha_{t} - 1)(\boldsymbol{\Sigma}_{L} + \boldsymbol{\Sigma}_{*})^{- 1}\boldsymbol{\Sigma}_{*}(\boldsymbol{\Sigma}_{L} + \boldsymbol{\Sigma}_{*}/\alpha_{t})^{- 1} = ( 1 - \alpha_{t}) \boldsymbol{\Upsilon}$, where we define $\boldsymbol{\Upsilon}=\left( \boldsymbol{\Sigma}_{L} + \boldsymbol{\Sigma}_{*} \right)^{- 1}\boldsymbol{\Sigma}_{*}\left( {\alpha_{t}\boldsymbol{\Sigma}}_{L} + \boldsymbol{\Sigma}_{*} \right)^{- 1}$, then we obtain
\begin{eqnarray*}
&&\hspace{-20pt}\frac{d}{d\alpha_{t}}\exp\left\{ - \frac{1}{2}\text{tr}\left\lbrack \left( 1 - \alpha_{t} \right)\boldsymbol{\Upsilon} \Tilde{\boldsymbol{A}}  \right\rbrack \right\}=\exp\left\{ - \frac{1}{2}\text{tr}\left\lbrack \boldsymbol{\Sigma}_{H}^{- 1}\Tilde{\boldsymbol{A}} \left( \alpha_{t} \right) \right\rbrack \right\}\left\{ \frac{1}{2}\text{tr}\left\lbrack \boldsymbol{\Upsilon} \Tilde{\boldsymbol{A}}  \right\rbrack + \frac{\left( \alpha_{t} - 1 \right)}{2} \frac{d}{d\alpha_{t}}\text{tr}\left\lbrack \boldsymbol{\Upsilon} \Tilde{\boldsymbol{A}}  \right\rbrack \right\} \\
&&\hspace{-20pt}=\exp\left\{ - \frac{1}{2}\text{tr}\left\lbrack \boldsymbol{\Sigma}_{H}^{- 1}\Tilde{\boldsymbol{A}} \left( \alpha_{t} \right) \right\rbrack \right\}\left\{ \frac{1}{2}\text{tr}\left\lbrack \boldsymbol{\Upsilon} \Tilde{\boldsymbol{A}}  \right\rbrack + \frac{\left( \alpha_{t} - 1 \right)}{2} tr\left\lbrack \left( \boldsymbol{\Sigma}_{L} + \boldsymbol{\Sigma}_{*} \right)^{- 1}\boldsymbol{\Sigma}_{*}\frac{d}{d\alpha_{t}}\left( {\alpha_{t}\boldsymbol{\Sigma}}_{L} + \boldsymbol{\Sigma}_{*} \right)^{- 1}\Tilde{\boldsymbol{A}}  \right\rbrack \right\}.
\end{eqnarray*}
The derivative in the last line is:
\begin{equation}
\frac{d}{d\alpha_{t}}\left( {\alpha_{t}\boldsymbol{\Sigma}}_{L} + \boldsymbol{\Sigma}_{*} \right)^{- 1} = - \left( {\alpha_{t}\boldsymbol{\Sigma}}_{L} + \boldsymbol{\Sigma}_{*} \right)^{- 1}\boldsymbol{\Sigma}_{L}\left( {\alpha_{t}\boldsymbol{\Sigma}}_{L} + \boldsymbol{\Sigma}_{*} \right)^{- 1}.
\end{equation}
See \cite{abamag05} and \cite{magneu99}. Thus, we get the result:
\begin{eqnarray*}
&&\partial_{\alpha_t} H_{t}\left( \alpha_{t} \right) = \kappa_t\left( \alpha_{t} \right)\exp\left\{ - \frac{1}{2}\text{tr}\left\lbrack \boldsymbol{\Sigma}_{H}^{- 1}\Tilde{\boldsymbol{A}} \right\rbrack \right\}\left\{-\frac{n}{{2\alpha}_{t}^{2}}\text{tr}\left\lbrack \left( \boldsymbol{\Sigma}_{L} + \frac{\boldsymbol{\Sigma}_{*}}{\alpha_{t}} \right)^{- 1}\boldsymbol{\Sigma}_{*} \right\rbrack +\right.\\
&&\left. \frac{1}{2}\text{tr}\left\lbrack \boldsymbol{\Upsilon} \Tilde{\boldsymbol{A}}  \right\rbrack + \frac{\left( {1 - \alpha}_{t} \right)}{2}\text{tr}\left\lbrack \boldsymbol{\Upsilon} \boldsymbol{\Sigma}_{L}\left( {\alpha_{t}\boldsymbol{\Sigma}}_{L} + \boldsymbol{\Sigma}_{*} \right)^{- 1}\Tilde{\boldsymbol{A}}  \right\rbrack \right\} =\\
&&= \frac{H_{t}\left( \alpha_{t} \right)}{2}\text{tr}\left\{ {- \frac{n}{\alpha_{t}^{2}}\left( \boldsymbol{\Sigma}_{L} + \frac{\boldsymbol{\Sigma}_{*}}{\alpha_{t}} \right)}^{- 1}\boldsymbol{\Sigma}_{*} + \boldsymbol{\Upsilon} \Tilde{\boldsymbol{A}}  + \left( {1 - \alpha}_{t} \right)\boldsymbol{\Upsilon} \boldsymbol{\Sigma}_{L}\left( {\alpha_{t}\boldsymbol{\Sigma}}_{L} + \boldsymbol{\Sigma}_{*} \right)^{- 1}\Tilde{\boldsymbol{A}}  \right\}.
\end{eqnarray*}
\end{proof}

\subsection{Proof of Proposition \ref{prop:BFmatrixUnknownV}}

\begin{proof}[Proof of Proposition \ref{prop:BFmatrixUnknownV}]
Note that from Prop. \ref{prop:posteriorBF} the BF can be written as $H_t(\alpha_t)=\chi(\alpha_t)\kappa_{1t}(\alpha_t)\kappa_{2t}(\alpha_t)$ where the first term depend on $\boldsymbol{Y}_{t}$ and $\bar{\boldsymbol{Y}}$
\begin{equation}
\chi(\alpha_t)=\frac{\left|\boldsymbol{\Psi}_{A,d}\right|^\frac{m_{A,d}-n-1}{2}\left|\boldsymbol{\Psi}_*\right|^\frac{m_*-n-1}{2}}{\left|\boldsymbol{\Psi}_{d}\right|^\frac{m_{d}-n-1}{2}\left|\boldsymbol{\Psi}_{A,*}\right|^\frac{m_{A,*}-n-1}{2}}=\frac{\left|\boldsymbol{\Psi}_{A,d}\right|^\frac{\alpha_{t}m_{d}-n-1}{2}\left|\boldsymbol{\Psi}_*\right|^\frac{m_*-n-1}{2}}{\left|\boldsymbol{\Psi}_{d}\right|^\frac{m_{d}-n-1}{2}\left|\boldsymbol{\Psi}_{A,*}\right|^\frac{\alpha_{t}m_{d}-p-n-1}{2}}
\end{equation}
and the remaining two terms depend only on $\alpha_t$
\begin{equation}
\kappa_{1t}(\alpha_t)=\left(\frac{k_{*}k_{A,d}}{k_{d}k_{A,*}}\right)^{\frac{np}{2}},\quad \kappa_{2t}(\alpha_t) = \frac{\Gamma_{n}\left(\frac{m_{d}-n-1}{2}\right)\Gamma_{n}\left(\frac{m_{A,*}-n-1}{2}\right)}{\Gamma_{n}\left(\frac{m_{*}-n-1}{2}\right)\Gamma_{n}\left(\frac{m_{A,d}-n-1}{2}\right)}.
\end{equation}
First, we prove that $\chi(\alpha_t)<1$ to find the upper bound to the BF and then study the properties of the remaining terms to find the limits of the BF.

\textit{i)} Define the positive definite matrix $\boldsymbol{E}=\left(\boldsymbol{M}_{*}-\boldsymbol{Y}_{t}\right)^{'}\boldsymbol{\Sigma}^{-1}\left(\boldsymbol{M}_{*}-\boldsymbol{Y}_{t}\right)$ and write
\begin{eqnarray}
\boldsymbol{\Psi}_{A,d}&=&\alpha_t\boldsymbol{\Psi}_{*}+\frac{\alpha_tk_*}{\alpha_{t}k_{*}+1}\boldsymbol{E}=\alpha_{t}\boldsymbol{\Psi}_{*}\left(\boldsymbol{I}_n+\frac{\alpha_tk_*}{\alpha_{t}k_{*}+1}\boldsymbol{\Psi}_*^{-1}\boldsymbol{E}\right)\\
\boldsymbol{\Psi}_{d}&=&\boldsymbol{\Psi}_{*}+\frac{k_*}{k_{*}+1}\boldsymbol{E}=\boldsymbol{\Psi}_{*}\left(\boldsymbol{I}_n+\frac{k_*}{k_{*}+1}\boldsymbol{\Psi}_*^{-1}\boldsymbol{E}\right)
\end{eqnarray}
and $\boldsymbol{\Psi}_{A,*}=\alpha_{t}\boldsymbol{\Psi}_{*}$. Substituting in the expression of $\chi(\alpha_t)$ their determinants 
\begin{eqnarray}
\left|\boldsymbol{\Psi}_{A,d}\right|^{\frac{\alpha_{t}m_{d}-n-1}{2}}&=&\alpha_{t}^{\frac{n(\alpha_{t}m_{d}-n-1)}{2}}\left|\boldsymbol{\Psi}_{*}\right|^{\frac{\alpha_{t}m_{d}-n-1}{2}}\left|\boldsymbol{I}_n+\frac{\alpha_tk_*}{\alpha_{t}k_{*}+1}\boldsymbol{\Psi}_*^{-1}\boldsymbol{E}\right|^{\frac{\alpha_{t}m_{d}-n-1}{2}}\\
\left|\boldsymbol{\Psi}_{d}\right|^{\frac{m_{d}-n-1}{2}}&=&\left|\boldsymbol{\Psi}_{*}\right|^{\frac{m_{d}-n-1}{2}}\left|\boldsymbol{I}_n+\frac{k_*}{k_{*}+1}\boldsymbol{\Psi}_*^{-1}\boldsymbol{E}\right|^{\frac{m_{d}-n-1}{2}}\\
\left|\boldsymbol{\Psi}_{A,*}\right|^\frac{\alpha_{t}m_{d}-p-n-1}{2}&=&\alpha_{t}^\frac{n(\alpha_{t}m_{d}-p-n-1)}{2}\left|\boldsymbol{\Psi}_*\right|^\frac{\alpha_{t}m_{d}-p-n-1}{2}
\end{eqnarray}
and using $m_d=m_*+p$ yield
\begin{eqnarray}
\chi(\alpha_t)
&=&\frac{\alpha_{t}^{\frac{np}{2}}\left|\boldsymbol{\Psi}_{*}\right|^{\frac{\alpha_{t}m_{d}+m_*-m_d-\alpha_{t}m_d+p}{2}}\left|\boldsymbol{I}_n+\frac{\alpha_tk_*}{\alpha_{t}k_{*}+1}\boldsymbol{\Psi}_*^{-1}\boldsymbol{E}\right|^{\frac{\alpha_{t}m_{d}-n-1}{2}}}{\left|\boldsymbol{I}_n+\frac{k_*}{k_{*}+1}\boldsymbol{\Psi}_*^{-1}\boldsymbol{E}\right|^{\frac{m_{d}-n-1}{2}}}\nonumber\\
&=&
\alpha_{t}^{\frac{np}{2}}\left|\boldsymbol{I}_n+\frac{\alpha_tk_*}{\alpha_{t}k_{*}+1}\boldsymbol{\Psi}_*^{-1}\boldsymbol{E}\right|^{\frac{\alpha_{t}m_{d}-n-1}{2}}
\left|\boldsymbol{I}_n+\frac{k_*}{k_{*}+1}\boldsymbol{\Psi}_*^{-1}\boldsymbol{E}\right|^{-\frac{m_{d}-n-1}{2}}.
\end{eqnarray}
In the sense of positive definite matrices, the following inequality is satisfied:
\begin{equation}
\boldsymbol{I}_n+\frac{\alpha_{t}k_*}{\alpha_{t}k_{*}+1}\boldsymbol{\Psi}_*^{-1}\boldsymbol{E}<\boldsymbol{I}_n+\frac{k_*}{k_{*}+1}\boldsymbol{\Psi}_*^{-1}\boldsymbol{E}
\end{equation}
since $\alpha_t<1$, and $\alpha_{t}k_*/(\alpha_{t}k_{*}+1)-k_*/(k_{*}+1)=k_*(\alpha_{t}-1)/((\alpha_{t}k_{*}+1)(k_{*}+1))<0$. From Th. 7.8 in \cite{zhang2011matrix} on the relationship between determinants of positive definite matrices we obtain:
$|\boldsymbol{I}_n+\alpha_tk_*/(\alpha_{t}k_{*}+1)\boldsymbol{\Psi}_*^{-1}\boldsymbol{E}|<|\boldsymbol{I}_n+k_*/(k_{*}+1)\boldsymbol{\Psi}_*^{-1}\boldsymbol{E}|$ which implies
\begin{equation}
\left|\boldsymbol{I}_n+\frac{\alpha_tk_*}{\alpha_{t}k_{*}+1}\boldsymbol{\Psi}_*^{-1}\boldsymbol{E}\right|^{\frac{\alpha_{t}m_{d}-n-1}{2}}<\left|\boldsymbol{I}_n+\frac{k_*}{k_{*}+1}\boldsymbol{\Psi}_*^{-1}\boldsymbol{E}\right|^{\frac{m_{d}-n-1}{2}}
\end{equation}
since the arguments of the power functions are larger than one. For the first argument, this can be proved using Th. 7.7 and 7.8 in \cite{zhang2011matrix} and the positive definiteness:
\begin{equation}
\left|\boldsymbol{I}_n+\frac{\alpha_tk_*}{\alpha_{t}k_{*}+1}\boldsymbol{\Psi}_*^{-1}\boldsymbol{E}\right|>\left|\boldsymbol{I}_n\right|+\left|\frac{\alpha_tk_*}{\alpha_{t}k_{*}+1}\boldsymbol{\Psi}_*^{-1}\boldsymbol{E}\right|>1.
\end{equation}
For the second argument, the same inequality for determinants can be used together with $\alpha_{t}m_d+n-1>0$ and the assumption $\alpha_t>\underline{\alpha}=(2n+p)/m_d$. It follows that $\chi(\alpha_t)<1$ and $H_t(\alpha_t)<\kappa_{1t}(\alpha_t)\kappa_{2t}(\alpha_t)$ and the upper bound $\kappa_t(\alpha_t)=\kappa_{1t}(\alpha_t)\kappa_{2t}(\alpha_t)$ does not depends on $\boldsymbol{Y}_{t}$ and $\bar{\boldsymbol{Y}}$.

We show now that $\kappa_t(\alpha_t)$ is monotone decreasing since $\kappa_{1t}(\alpha_t)$ and $\kappa_{2t}(\alpha_t)$ are monotone decreasing. Regarding the first derivative of $\kappa_{1t}(\alpha_t)$ 
\begin{equation}
\partial_{\alpha_t}\left(\frac{\alpha_{t}k_*+1}{\alpha_t(k_*+1)}\right)^{\frac{np}{2}}=\frac{np}{2}\left(\frac{\alpha_{t}k_*+1}{\alpha_t(k_*+1)}\right)^{\frac{np}{2}-1}\left(-\alpha_t^{-2}(k_*+1)^{-1}\right)<0.
\end{equation}
Regarding the first derivative of $\kappa_2(\alpha_t)$ consider the following preliminary results, which can be easily proved using the properties of the multivariate  function of order $n$ \citep[e.g., see][] {Gup99}
{\footnotesize\begin{eqnarray}
\partial_{\alpha_t}\Gamma_{n}\left(\frac{m_{A,*}-n-1}{2}\right)&=&\frac{m_*+p}{2}\Gamma_{n}\left(\frac{m_{A,*}-n-1}{2}\right)\sum_{i = 1}^{n}\gamma_{n}\left(\frac{m_{A,*}-n-1}{2}+\frac{1-i}{2}\right)\label{gam1}\\
\partial_{\alpha_t}\Gamma_{n}\left(\frac{m_{A,d}-n-1}{2}\right)&=&\frac{m_*+p}{2}\Gamma_{n}\left(\frac{m_{A,d}-n-1}{2}\right)\sum_{i = 1}^{n}\gamma_{n}\left(\frac{m_{A,d}-n-1}{2}+\frac{1-i}{2}\right),\label{gam2}
\end{eqnarray}}
where $\gamma_n(a)$ is the multivariate digamma function of order $n$. Note that since $m_{A,d}\geq m_{A,*}$ the smallest argument of $\gamma_n$ in Eq. \ref{gam1}-\ref{gam2} is $
 (\alpha_t(m_*+p)-p-n-1+1-n)/2$ which is positive for  $\alpha_t\geq \underline{\alpha}=(p+2n)/(m_*+p)$. The first derivative of $\kappa_{2t}(\alpha_t)$ becomes:
 \begin{eqnarray}
\partial_{\alpha_t}\kappa_2(\alpha_t)&=&\frac{\Gamma_{n}\left(\frac{m_{d}-n-1}{2}\right)}{\Gamma_{n}\left(\frac{m_{*}-n-1}{2}\right)}\partial_{\alpha_t}\frac{\Gamma_{n}\left(\frac{m_{A,*}-n-1}{2}\right)}{\Gamma_{n}\left(\frac{m_{A,d}-n-1}{2}\right)}\nonumber\\
 &=&\frac{\Gamma_{n}\left(\frac{m_{d}-n-1}{2}\right)}{\Gamma_{n}\left(\frac{m_{*}-n-1}{2}\right)}\frac{(m_*+p)}{2}\Gamma_{n}\left(\frac{m_{A,*}-n-1}{2}\right)\Gamma_{n}\left(\frac{m_{A,d}-n-1}{2}\right)\nonumber\\
&\times&\frac{\sum_{i=1}^{n}\left[\gamma_n\left(\frac{m_{A,*}-n-1}{2}+\frac{i-1}{2}\right)-\gamma_n\left(\frac{m_{A,d}-n-1}{2}+\frac{i-1}{2}\right)\right]}{\Gamma_{n}\left(\frac{m_{A,d}-n-1}{2}\right)^2}.
 \end{eqnarray}
Since $m_{A,d}\geq m_{A,*}$ and the $\gamma(a)$ function is an increasing function for $a\geq0$ the differences of gamma functions at the numerator are negative, provided $\alpha_t\geq \underline{\alpha}$. This means that  $\partial_{\alpha_t}\kappa_2(\alpha_t)\leq 0$.

Finally, the limits of $\kappa_{1t}(\alpha_t)$ as $\alpha_t\rightarrow 1^{-}$ is 1, since $\kappa_{1t}(\alpha_t)$ is decreasing and the lower bound is
\begin{equation}
\kappa_{1t}(\alpha_t)=\left(\frac{k_{*}k_{A,d}}{k_{d}k_{A,*}}\right)^{\frac{np}{2}}=\left(\frac{k_{*}(\alpha_{t}k_*+1)}{(k_*+1)\alpha_{t}k_*}\right)^{\frac{np}{2}}=   \left(\frac{\alpha_{t}k_{*}+1}{\alpha_t(k_*+1)}\right)^{\frac{np}{2}}\ge 1
\end{equation}
whereas $\kappa_{1t}(\alpha_t)$ converges to $(\underline{\alpha}_tk_{*}/(\underline{\alpha_t}(k_{*}+1)))^{np/2}$ that is larger than 1 if $\alpha_t\rightarrow\underline{\alpha}$. We also get $\kappa_{2t}(\alpha_t)\rightarrow 1$ as $\alpha_t\rightarrow 1$ , since $m_{A,d}=m_d$ and $m_{A,*}=m_*$. For $\alpha_t\rightarrow\underline{\alpha}$ we get $\kappa_2(\alpha_t)\rightarrow +\infty$ because $\Gamma_n[(m_{A,*}-n-1)/2]\rightarrow+\infty$ for $\alpha_t\rightarrow\underline{\alpha}$ and $m_{A,*}\leq m_*\leq m_d \leq m_{A,d}$ for $\alpha_t\leq1$. 

In conclusion we get that $\kappa(\alpha_t)=\kappa_{1t}(\alpha_t)\kappa_{2t}(\alpha_t)$ is the product of two functions monotone decreasing, thus $\kappa_(\alpha_t)$ is monotone decreasing and 
$\kappa(\alpha_t)\rightarrow\infty$ as $\alpha\rightarrow \underline{\alpha}_{t}^{+}$ and $\kappa(\alpha_t)\rightarrow 1$ as $\alpha\rightarrow 1^{-}$.

\textit{ii)} From its definition it trivially follows $\lim_{\alpha_{t} \rightarrow 1^{-}}{H_{t}(\alpha_t)}=1$. Regarding the second limit, note that all terms at the numerator and denominator of $H_t(\alpha_t)$ have bounded limit for $\alpha_t\rightarrow \underline{\alpha}$, expect for $\Gamma_{n}\left(\frac{m_{A,*}-n-1}{2}\right)$ which is unbounded since $m_{A,*}-n-1=\alpha_t(m_*+p)-p-n-1\rightarrow 0$. Thus $H_t(\alpha_t)\rightarrow\infty$ for $\alpha_t\rightarrow \underline{\alpha}$.

\textit{iii)} We factorize out the terms that depend on $\alpha_t$ and write $H_t$ as $ H_t(\alpha_t)=L (a_1 a_2 a_3)/(b_1 b_2 b_3)$, where $L=\left|\boldsymbol{\Psi}_*\right|^\frac{m_*-n-1}{2}k_*^\frac{np}{2}\Gamma_n\left(\frac{m_d-n-1}{2}\right) k_d^\frac{-np}{2}\left|\boldsymbol{\Psi}_d\right|^{-\frac{m_d-n-1}{2}}\Gamma_n\left(\frac{m_*-n-1}{2}\right)^{-1}$
and
\begin{eqnarray}
&&a_1=k_{A,d}^\frac{np}{2},\,a_2=\Gamma_n\left(\frac{m_{A,*}-n-1}{2}\right),\,a_3=\left|\boldsymbol{\Psi}_{A,d}\right|^\frac{m_{A,d}-n-1}{2},\nonumber\\
&&b_1=k_{A,*}^\frac{np}{2},\,b_2=\Gamma_n\left(\frac{m_{A,d}-n-1}{2}\right),\,b_3=\left|\boldsymbol{\Psi}_{A,*}\right|^\frac{m_{A,*}-n-1}{2}.
\end{eqnarray}
The first derivative of $H(\alpha_t)$ is:
$\partial_{\alpha_t}H(\alpha_t)=L(b_1b_2b_3)^{-2}(\left(a_1^{'}a_2a_3+a_1a_2^{'}a_3+a_1a_2a_3^{'}\right)b_1b_2b_3$ $-\left(b_1^{'}b_2b_3+b_1b_2^{'}b_3+b_1b_2b_3^{'}\right)a_1a_2a_3)$ where
\begin{eqnarray}
a_1^{'}&=&\partial_{\alpha_t} k_{A,d}^\frac{np}{2}=\partial_{\alpha_t} \left(\alpha_t k_*+1\right)^\frac{np}{2}=\frac{np}{2}k_*\left(\alpha_t k_*+1\right)^\frac{np-2}{2}\\
a_2^{'} &=& \partial_{\alpha_t}\Gamma_n\left(\frac{m_{A,*}-n-1}{2}\right)=\partial_{\alpha_t}\Gamma_n\left[\frac{\alpha_t(m_{*}+p)-p-n-1}{2}\right]\nonumber\\
&=&(m_*+p)\Gamma_n\left(\frac{m_{A,*}-n-1}{2}\right)\sum_{i = 1}^{n}\psi\left(\frac{m_{A,*}-n-i}{2}\right)\nonumber\\
a_3^{'} &=& \partial_{\alpha_t}\left|\boldsymbol{\Psi}_{A,d}\right|^\frac{m_{A,d}-n-1}{2}\nonumber\\
&=&\left|\boldsymbol{\Psi}_{A,d}\right|^\frac{m_{A,d}-n-1}{2}\left[\partial_{\alpha_t}\left(\frac{m_{A,d}-n-1}{2}\right)log\left|\boldsymbol{\Psi}_{A,d}\right|+\frac{m_{A,d}-n-1}{2}\frac{\partial_{\alpha_t}\left|\boldsymbol{\Psi}_{A,d}\right|}{\left|\boldsymbol{\Psi}_{A,d}\right|}\right]\nonumber
\end{eqnarray}
with $\psi(a)$ the digamma function and
\begin{eqnarray}
\partial_{\alpha_t}\frac{m_{A,d}-n-1}{2}&=&\frac{m_*+p}{2}\nonumber\\
\partial_{\alpha_t}\left|\boldsymbol{\Psi}_{A,d}\right|&=&\left|\boldsymbol{\Psi}_{A,d}\right|tr\left[\boldsymbol{\Psi}_{A,d}^{-1}\left(\boldsymbol{\Psi}_*+\frac{k_*\left(\boldsymbol{M}_*-\boldsymbol{Y}_{t}\right)^{'}\boldsymbol{\Sigma}^{-1}\left(\boldsymbol{M}_*-\boldsymbol{Y}_{t}\right)}{\left(\alpha_tk_*+1\right)^2}\right)\right].\nonumber
\end{eqnarray}
The remaining three terms have derivatives:
$b_1^{'}$:
{\small\begin{eqnarray}
b_1^{'}&=&\partial_{\alpha_t} k_{A,*}^\frac{np}{2}=\partial_{\alpha_t} \alpha_t k_*^\frac{np}{2}=\frac{np}{2}k_*k_{A,*}^\frac{np-2}{2}=\frac{np}{2}k_*\left(\alpha_t k_*\right)^\frac{np-2}{2}\\
b_2^{'} &=& \partial_{\alpha_t}\Gamma_n\left(\frac{m_{A,d}-n-1}{2}\right)=\partial_{\alpha_t}\Gamma_n\left[\frac{\alpha_t(m_{*}+p)-p+p-n-1}{2}\right]\nonumber\\
&=&(m_*+p)\Gamma_n\left(\frac{m_{A,d}-n-1}{2}\right)\sum_{i = 1}^{n}\psi\left(\frac{m_{A,d}-n-i}{2}\right)\\
b_3^{'} &=& \partial_{\alpha_t}\left|\boldsymbol{\Psi}_{A,*}\right|^\frac{m_{A,*}-n-1}{2}\nonumber\\
&=&\left|\boldsymbol{\Psi}_{A,*}\right|^\frac{m_{A,*}-n-1}{2}\left[\partial_{\alpha_t}\left(\frac{m_{A,*}-n-1}{2}\right)\log\left|\boldsymbol{\Psi}_{A,*}\right|+\frac{m_{A,*}-n-1}{2}\frac{\partial_{\alpha_t}\left|\boldsymbol{\Psi}_{A,*}\right|}{\left|\boldsymbol{\Psi}_{A,*}\right|}\right],
\end{eqnarray}}
where $\partial_{\alpha_t}(m_{A,*}-n-1)/2=(m_*+p)/2$ and $\partial_{\alpha_t}\left|\boldsymbol{\Psi}_{A,*}\right|=\left|\boldsymbol{\Psi}_{A,*}\right|tr\left(\boldsymbol{\Psi}_{A,*}^{-1}\boldsymbol{\Psi}_*\right)=\left|\boldsymbol{\Psi}_{A,*}\right|\alpha_t^{-1}$ since $\boldsymbol{\Psi}_{A,*}=\alpha_t\boldsymbol{\Psi}_{*}$.  In conclusion, the derivative of $H_t(\alpha_t)$ is:
\begin{equation}
\partial_{\alpha_t}H_t(\alpha_t)=L\left[\left|\boldsymbol{\Psi}_{A,*}\right|^\frac{m_{A,*}-n-1}{2}k_{A,*}^\frac{np}{2}\Gamma_n\left(\frac{m_{A,d}-n-1}{2}\right)\right]^{-2}(A_1^{'}B_1-A_1B_1^{'})
\end{equation}
where $A_1=a_1a_2a_3$, $B_1=b_1b_2b_3$ and their derivatives $A_1{'}, B_1{'}$ are defined as:
{\footnotesize \begin{eqnarray*}
A_1&=&k_{A,d}^\frac{np}{2}\Gamma_n\left(\frac{m_{A,*}-n-1}{2}\right)\left|\boldsymbol{\Psi}_{A,d}\right|^\frac{m_{A,d}-n-1}{2}\\
B_1&=&k_{A,*}^\frac{np}{2}\Gamma_n\left(\frac{m_{A,d}-n-1}{2}\right)\left|\boldsymbol{\Psi}_{A,*}\right|^\frac{m_{A,*}-n-1}{2}\\
A_1^{'}&=&a_1^{'}a_2a_3+a_1a_2^{'}a_3+a_1a_2a_3^{'}\nonumber\\
&=&\frac{np}{2}k_*k_{A,d}^\frac{np-2}{2}\Gamma_n\left(\frac{m_{A,*}-n-1}{2}\right)\left|\boldsymbol{\Psi}_{A,d}\right|^\frac{m_{A,d}-n-1}{2}+\nonumber\\
&&k_{A,d}^\frac{np}{2}(m_*+p)\Gamma_n\left(\frac{m_{A,*}-n-1}{2}\right)\sum_{i = 1}^{n}\psi\left(\frac{m_{A,*}-n-i}{2}\right)\left|\boldsymbol{\Psi}_{A,d}\right|^\frac{m_{A,d}-n-1}{2}+\nonumber\\
&&k_{A,d}^\frac{np}{2}\Gamma_n\left(\frac{m_{A,*}-n-1}{2}\right)\left|\boldsymbol{\Psi}_{A,d}\right|^\frac{m_{A,d}-n-1}{2}.\nonumber\\
&&\left\lbrace\frac{m_*+p}{2}log\left|\boldsymbol{\Psi}_{A,d}\right|+\frac{m_{A,d}-n-1}{2}tr\left[\boldsymbol{\Psi}_{A,d}^{-1}\left(\boldsymbol{\Psi}_*+\frac{k_*\left(\boldsymbol{M}_*-\boldsymbol{Y}_{t}\right)^{'}\boldsymbol{\Sigma}^{-1}\left(\boldsymbol{M}_*-\boldsymbol{Y}_{t}\right)}{\left(\alpha_tk_*+1\right)^2}\right)\right]\right\rbrace\\
B_1^{'}&=&b_1^{'}b_2b_3+b_1b_2^{'}b_3+b_1b_2b_3^{'}\nonumber\\
&=&\frac{np}{2}k_*k_{A,*}^\frac{np-2}{2}\Gamma_n\left(\frac{m_{A,d}-n-1}{2}\right)\left|\boldsymbol{\Psi}_{A,*}\right|^\frac{m_{A,*}-n-1}{2}+\nonumber\\
&&k_{A,*}^\frac{np}{2}(m_*+p)\Gamma_n\left(\frac{m_{A,d}-n-1}{2}\right)\sum_{i = 1}^{n}\psi\left(\frac{m_{A,d}-n-i}{2}\right)\left|\boldsymbol{\Psi}_{A,*}\right|^\frac{m_{A,*}-n-1}{2}+\nonumber\\
&&k_{A,*}^\frac{np}{2}\Gamma_n\left(\frac{m_{A,d}-n-1}{2}\right)\left|\boldsymbol{\Psi}_{A,*}\right|^\frac{m_{A,*}-n-1}{2}\left[\left(\frac{m_*+p}{2}\right)log\left|\boldsymbol{\Psi}_{A,*}\right|+\frac{m_{A,*}-n-1}{2\alpha_t}\right].
\end{eqnarray*}}
\end{proof}

\subsection{Proof of Corollary \ref{exUnivGauss}}

\begin{proof}[Proof of Corollary \ref{exUnivGauss}]
Follows from some long but straightforward algebra from \textit{iv) of Prop. \ref{prop:BFmatrix}}
\end{proof}

\subsection{Proof of Remark \ref{exUnivGaussContd}}

\begin{proof}[Proof of Remark \ref{exUnivGaussContd}]
From the BF in Eq. \ref{BF_matrix}, assuming $n=p=1$, $\boldsymbol{\Sigma}_P=\sigma^2/\varphi$, $\boldsymbol{\Sigma}_L=\sigma^2$, $\boldsymbol{V}=1$, $\boldsymbol{M}=m$ and $\mathbf{Y}= (Y_{1},\ldots,Y_{t-1})$ one obtains $H_t(\alpha_t)>H_0$ for
\begin{eqnarray}
    &&\exp\left\{\frac{(\alpha_t-1)(\varphi+T)(Y_t-m_{\ast})^2}{2\sigma^2(\varphi+T+1)(\alpha_t\varphi+\alpha_t T+1)}\right\}>H_0 \frac{\sqrt{\alpha_t(\varphi+T+1)}}{\sqrt{\alpha_t(\varphi+T)+1}}\Longleftrightarrow\nonumber\\
    &&(Y_t-m_{\ast})^2 \leq \frac{2\sigma^2(\varphi+T+1)(\alpha_t\varphi+\alpha_t T+1)}{(\alpha_t-1)(\varphi+T)}\log\left\{H_0\frac{\sqrt{\alpha_t(\varphi+T+1)}}{\sqrt{\alpha_t(\varphi+T)+1}}\right\},
\end{eqnarray}
where $m_{\ast}=\varphi/(\varphi+T)m+T/(\varphi+T)\bar{Y}$, $\bar{Y}=(Y_1+\ldots+Y_{T})/T$ and $T=t-1$. The inequality is satisfied for $Y_{1t}\leq Y_t\leq Y_{2t}$ there $Y_{jt}$ are the solutions of the second order equation associated with the above inequality.
\end{proof}

\subsection{Proof of Proposition \ref{prop:IBFmatrix}}

\begin{proof}[Proof of Proposition \ref{prop:IBFmatrix}]
i) Since $H_{t}(\alpha_t)<\kappa_t(\alpha_t)$ it is sufficient to find $a$, $b$, $\underline{\alpha}$ and $\bar{\alpha}$ such that $\int_{0}^{1}\kappa_t(\alpha_t)\pi_t(\alpha_t)d\alpha_t<\infty$. By multinomial theorem, it follows
\begin{eqnarray}
    &&\frac{|\boldsymbol{\Sigma}_L+\boldsymbol{\Sigma}_{\ast}/\alpha_t|^{n/2}}{|\boldsymbol{\Sigma}_L+\boldsymbol{\Sigma}_{\ast}|^{n/2}}=|\boldsymbol{\Sigma}_L+\boldsymbol{\Sigma}_{\ast}|^{-n/2} \left(\sum_{j=0}^{p}c_j \alpha_t^{-j}\right)^{n/2}=|\boldsymbol{\Sigma}_L+\boldsymbol{\Sigma}_{\ast}|^{-n/2}\left(\sum_{\underline{k}\in\mathcal{K}_n}d_{n,\underline{k}}\prod_{j=0}^{p}c_{j}^{k_j}\alpha_t^{-j k_j}\right)^{1/2}\nonumber\\
&&=|\boldsymbol{\Sigma}_L+\boldsymbol{\Sigma}_{\ast}|^{-n/2}\left(\sum_{\underline{k}\in\mathcal{K}}d_{n,\underline{k}}\alpha_t^{-\sum_{j=0}^{p}j k_j}c_{\underline{k}}\right)^{1/2}<c \alpha_t^{-pn/2}
\end{eqnarray}
with $c_j$ the $j$-th coefficient of the characteristic polynomial of $\boldsymbol{\Sigma}_L+\boldsymbol{\Sigma}_{\ast}$, $\underline{k}=(k_0,\ldots,k_p)$ a multi-index with values in $\mathcal{K}_n=\{\underline{k}|k_0+\ldots+k_p=n\}$ and
$$
c=|\boldsymbol{\Sigma}_L+\boldsymbol{\Sigma}_{\ast}|^{-n/2}\left(\sum_{\underline{k}\in\mathcal{K}_n}d_{\underline{k}}c_{\underline{k}}\right)^{1/2},\quad d_{n,\underline{k}}=\binom{n}{k_0,k_1,k_2,\ldots,k_p},\quad c_{\underline{k}}=\prod_{j=0}^{p}c_{j}^{k_j}.
$$
We used  the relationship $\sum_{j=1}^{p}j k_j=\sum_{\ell=1}^{p}\sum_{j=\ell}^{p}k_j<\sum_{\ell=1}^{p}\sum_{j=0}^{p}k_j=pn$, along with the condition $\alpha_t<1$, to derive the inequality here above. Thus we conclude
\begin{equation}
\int_{0}^{1}\kappa_t(\alpha)\pi_t(\alpha_t)d\alpha_t<c \int_{0}^{1}\frac{1}{B(a,b)}\alpha_t^{-pn/2}\alpha_t^{a-1}(1-\alpha_t)^{b-1} d\alpha_t<\infty 
\end{equation}
for $a-pn/2-1>-1$, that is $a>pn/2$, and $b>0$. Under the integrability conditions given above, the integral of the upper bound is
{\small
\begin{eqnarray*}
    &&\int_0^{1} \frac{|\boldsymbol{\Sigma}_L+\boldsymbol{\Sigma}_{\ast}/\alpha_t|^{n/2}}{|\boldsymbol{\Sigma}_L+\boldsymbol{\Sigma}_{\ast}|^{n/2}}\frac{1}{B(a,b)}\alpha_t^{a-1}(1-\alpha_t)^{b-1}d\alpha_t=\nonumber\\
    &&=\int_0^{1}\alpha^{-np/2}\left(1+\sum_{j=0}^{p}\tilde{c}_j\alpha_t^{j}-1\right)^{n/2}\frac{1}{B(a,b)}\alpha_t^{a-1}(1-\alpha_t)^{b-1}d\alpha_t\nonumber\\
&&=\sum_{\ell=0}^{\infty}\frac{\Gamma(n/2+1)}{\Gamma(\ell+1)\Gamma(n/2-\ell+1)} \int_0^{1}\alpha_t^{-np/2}\left(\sum_{j=0}^{p}\bar{c}_j\alpha_t^{j}+1\right)^{\ell}\frac{1}{B(a,b)}\alpha_t^{a-1}(1-\alpha_t)^{b-1}d\alpha_t
\end{eqnarray*}}
with $\tilde{c}_{j}$ the $j$-th coefficient of the characteristic polynomial of $\boldsymbol{\Sigma}_L^{-1}\boldsymbol{\Sigma}_{\ast}$, $\bar{c}_{0}=\tilde{c}_{0}-1$, $\bar{c}_{j}=\tilde{c}_j$, $j=1,\ldots,p$, and where we used $0<|\alpha_t \boldsymbol{I}_p+\boldsymbol{A}|/| \boldsymbol{I}_p+\boldsymbol{A}|<|\alpha_t \boldsymbol{I}_p+\boldsymbol{A}|/(| \boldsymbol{I}_p\alpha|+ |\boldsymbol{I}_p(1-\alpha_t)+\boldsymbol{A}|)<1$ with $\boldsymbol{A}=\boldsymbol{\Sigma}_L^{-1}\boldsymbol{\Sigma}_{\ast}$, and the generalized binomial formula. The above expression becomes
\begin{eqnarray}
&&\sum_{\ell=0}^{\infty}\sum_{\underline{k}\in\mathcal{K}_{\ell}}\frac{\Gamma(n/2+1)}{\Gamma(\ell+1)\Gamma(n/2-\ell+1)} d_{\ell,\underline{k}}\bar{c}_{\underline{k}}\int_0^{1}\frac{1}{B(a,b)}\alpha_t^{a-np/2+w_{\underline{k}}-1}(1-\alpha_t)^{b-1}d\alpha\nonumber\\
&&=\sum_{\ell=0}^{\infty}\sum_{\underline{k}\in\mathcal{K}_{\ell}+1}d_{\ell,\underline{k}}\bar{c}_{\underline{k}}\frac{\Gamma(n/2+1)B(a-np/2+w_{\underline{k}},b)}{\Gamma(\ell+1)\Gamma(n/2-\ell+1)B(a,b)}
\end{eqnarray}
by multinomial theorem, with
$$
\bar{c}_{\underline{k}}=\prod_{j=0}^{p}\bar{c}_{j}^{k_j}.
$$

ii) Since from Corollary \ref{coroll} the normalizing constant is well defined for $\alpha_t>(p+2n)/m_d$ it is natural to assume in the general beta distribution $\underline{\alpha}=(p+2n)/m_d$. It follows from Prop. \ref{prop:BFmatrixUnknownV} the BF $H_{t}(\alpha_t)=\chi(\alpha_t)\kappa_{1t}(\alpha_t)\kappa_{2t}(\alpha_t)$, with $\chi(\alpha_t)<1$, $\kappa_{1t}(\alpha_t)<(\underline{\alpha} k_{*}/(\underline{\alpha}(k_{*}+1)))^{np/2}$. To show the integrability it is sufficient to prove that $\int \kappa_{2t}(\alpha)\pi_t(\alpha)d\alpha<\infty$. We need to show that 
\begin{align}\label{A85}
    &\frac{\Gamma_{n}((m_d-n-1)/2)}{\Gamma_{n}((m_*-n-1)/2)}\int_{\underline{\alpha}}^{1} \frac{\Gamma_n\left(\frac{m_{A,*}-n-1}{2}\right)}{\Gamma_n\left(\frac{m_{A,d}-n-1}{2}\right)}\pi(\alpha_t)d\alpha_t<\infty.
\end{align}
Note that we can write:
\begin{align}
&\int_{\underline{\alpha}}^{1}\frac{\Gamma_n\left(\frac{m_{A,*}-n-1}{2}\right)}{\Gamma_n\left(\frac{m_{A,d}-n-1}{2}\right)}\pi(\alpha_t)d\alpha_t=\int_{\underline{\alpha}}^{1}\frac{\pi^{\frac {n(n-1)}{4}}\prod_{j=1}^n\Gamma\left(\frac{m_{A,*}-n-1}{2}+\frac{1-j}{2}\right)}{\pi^{\frac {n(n-1)}{4}}\prod_{j=1}^n\Gamma\left(\frac{m_{A,d}-n-1}{2}+\frac{1-j}{2}\right)}\pi(\alpha_t)d\alpha_t\nonumber\\
&=\int_{\underline{\alpha}}^{1}\frac{\prod_{j=1}^{n}\Gamma\left(\frac{\alpha_{t}m_{d}-p-n-j}{2}\right)}{\prod_{j=1}^{n}\Gamma\left(\frac{\alpha_{t}m_{d}-n-j}{2}\right)}\pi(\alpha_t)d\alpha_t=\int_{\underline{\alpha}}^{1}\prod_{j=1}^{n}\frac{B\left((\alpha_{t}m_{d}-n-j)/2,\frac{p}{2}\right)}{\Gamma\left(\frac{p}{2}\right)}\pi(\alpha_t)d\alpha_t,
\end{align}
where the last line follows from the property: $\Gamma\left(c-d\right)\Gamma\left(d\right)=B\left(c,d\right)\Gamma(c)$.
Given the upper bound for the beta function, $B(x,y) \leq \frac{1}{xy}$ \citep[e.g., see][Th. 1.2]{from2022some}, we can write
\begin{align}
&B\left(\frac{\alpha_{t}m_{d}-n-j}{2},\frac{p}{2}\right)=
\int_{0}^{1}t^{\frac{m_{d}(\alpha_{t}-\underline{\alpha})+m_d\underline{\alpha}-n-j}{2}-1}(1-t)^{\frac{p}{2}-1}dt\nonumber\\
&=\int_{0}^{1}t^{\frac{m_d}{2}(\alpha_{t}-\underline{\alpha})+\frac{m_d\underline{\alpha}-n-j}{2}-1}(1-t)^{\frac{p}{2}-1}dt
\leq\frac{2}{p\left[\frac{m_d}{2}(\alpha_{t}-\underline{\alpha})+\frac{m_d\underline{\alpha}-n-j}{2}\right]}.\nonumber
\end{align}
From the above inequality and using $j\leq n$ and by H\"older's inequality 
\begin{align}
&\int_{\underline{\alpha}}^{1}\frac{\Gamma_n\left(\frac{m_{A,*}-n-1}{2}\right)}{\Gamma_n\left(\frac{m_{A,d}-n-1}{2}\right)}\pi(\alpha_t)d\alpha_t
\leq\Gamma\left(\frac{p}{2}\right)^{-n}\int_{\underline{\alpha}}^{1}\prod_{j=1}^{n}\frac{4}{p\left[m_{d}(\alpha_{t}-\underline{\alpha})+(m_d\underline{\alpha}-n-j)\right]}\pi(\alpha_t)d\alpha_t\nonumber\\
&\leq\Gamma\left(\frac{p}{2}\right)^{-n}\int_{\underline{\alpha}}^{1}\prod_{j=1}^{n}\frac{4}{p\left[m_{d}(\alpha_{t}-\underline{\alpha})+(m_d\underline{\alpha}-n-n)\right]}\pi(\alpha_t)d\alpha_t\nonumber\\
&=\Gamma\left(\frac{p}{2}\right)^{-n}\left(\frac{4}{p}\right)^n\int_{0}^{1-\underline{\alpha}}\frac{x^{(a-1)}\left(1-x\right)^{(b-1)}}{B(1-\underline{\alpha};a,b)\left(m_{d}x+p\right)^n}dx\nonumber\\
&\leq\Gamma\left(\frac{p}{2}\right)^{-n}\left(\frac{4}{p}\right)^nB(1-\underline{\alpha};a,b)^{-1}\left[\int_{0}^{1-\underline{\alpha}}x^{r(a-1)}\left(1-x\right)^{r(b-1)}dx\right]^{1/r}\left[\int_{0}^{1-\underline{\alpha}}\left(m_{d}x+p\right)^{-qn}dx\right]^{1/q}\nonumber\\
&=\Gamma\left(\frac{p}{2}\right)^{-n}\left(\frac{4}{p}\right)^n\frac{B(1-\underline{\alpha};ra-r+1,qb-q+1)^{1/r}}{B(1-\underline{\alpha};a,b)}\left[\int_{0}^{1-\underline{\alpha}}\left(m_{d}x+p\right)^{-qn}dx\right]^{1/q}\nonumber\\
&=\Gamma\left(\frac{p}{2}\right)^{-n}\left(\frac{4}{p}\right)^n\frac{B(1-\underline{\alpha};ra-r+1,qb-q+1)}{B(1-\underline{\alpha};a,b)}\left[\frac{(m_{d}(1-\underline{\alpha})+p)^{-qn+1}-p^{-qn+1}}{m_{d}(-qn+1)}\right]^{1/q}<\infty.\nonumber
\end{align}
Note that, since $r$ and $q$ are such that $r^{-1}+q^{-1}=1$, the upper bound in the last line be minimize in $r$ after choosing $q=r/(1-r)$.
\end{proof}

\subsection{Proof of Corollary \ref{corol:integratedBoundUniv}}

\begin{proof}[Proof of Corollary \ref{corol:integratedBoundUniv}]
From Remark \ref{ex3} an upper bound of the BF is $A_t(\alpha_t)=(\alpha_t(\varphi+t-1)+1)/(\alpha_t(\varphi+t))^{1/2}$ and its integral is bounded since $\int_{0}^{1}{A_{t}(\alpha_{t})d\alpha_{t}}\  \leq \ \ \ \int_{0}^{1}{\alpha_{t}^{- 1/2}d\alpha_{t} = 2}$ even if $\alpha_{t} = 0$ is not a continuity point for $A_{t}$. By
the change of variable $A_{t}(\alpha_t) = z$ we obtain
\begin{eqnarray*}
&&\int_{0}^{1}{A_{t}(\alpha_t)d\alpha_{t}} = \int_{-\infty}^{1}{\frac{- 2(\varphi + t)z^{2}}{( (\varphi + t)z^{2} - (\varphi + t-1))^{2}}dz} = 1 - \int_{-\infty}^{1}{\frac{1}{(\varphi + t)z^{2} - (\varphi + t-1)}dz}\\
&&=1 - \frac{1}{(\varphi + t-1)}\int_{-\infty}^{\sqrt{\frac{(\varphi + t-1 + 1)}{(\varphi + t-1)}}}{\sqrt{\frac{(\varphi + t-1)}{(\varphi + t)}}\frac{1}{w^{2} - 1}}dw\\ 
&&= 1 + \frac{1}{2\sqrt{(\varphi + t-1)(\varphi + t)}}\log\left( \frac{\sqrt{(\varphi + t)} + \sqrt{(\varphi + t-1)}}{\sqrt{(\varphi + t)} - \sqrt{(\varphi + t-1)}} \right),
\end{eqnarray*}
where the second line follows from integration by parts and the last line from the change of variable
$w=z(\varphi + t)^{1/2}/(\varphi + t-1)^{1/2}$.
\end{proof}

\subsection{Proof of Proposition \ref{DistribBFmatrix}}

\begin{proof}[Proof of Proposition \ref{DistribBFmatrix}]
Let us define $\boldsymbol{E}_t=\boldsymbol{Y}_{t} - \boldsymbol{M}_{*}$, then following the expression of BF given in i) Prop. \ref{prop:BFmatrix}, the cdf of the BF is
\begin{eqnarray}
&&F_{H_{t}}(h|\alpha_t) =\mathbb{P}\left(H_{t} < h|\alpha_t \right)=\mathbb{P}\left( \frac{\left| \boldsymbol{\Sigma}_{A,d} \right|^{\frac{n}{2}}}{\left| \boldsymbol{\Sigma}_{d} \right|^{\frac{n}{2}}}\exp\left\{ - \frac{1}{2}\text{tr}\left\lbrack \boldsymbol{\Sigma}_{H}^{- 1}\boldsymbol{E}_t \boldsymbol{V}^{- 1} \boldsymbol{E}_t^{'} \right\rbrack \right\} < h|\alpha_t \right)\nonumber\\
&&=\mathbb{P}\left(\exp\left\{ - \frac{1}{2}\text{tr}\left\lbrack \boldsymbol{\Sigma}_{H}^{- 1}\boldsymbol{E}_t \boldsymbol{V}^{- 1}\boldsymbol{E}_t^{'} \right\rbrack \right\} < \frac{h}{\kappa_{t}}\left|\right.\alpha_t \right)\nonumber\\
&&=\mathbb{P}\left(\text{tr}\left\lbrack (\boldsymbol{\Sigma}_{H}^{-1/2})\boldsymbol{E}_t \boldsymbol{V}^{-1/2}(\boldsymbol{V}^{-1/2})'\boldsymbol{E}_t^{'}(\boldsymbol{\Sigma}_{H}^{-1/2})' \right\rbrack > - 2\log\left( h/\kappa_{t} \right)|\alpha_t \right),
\end{eqnarray}
where $- \log\left(h/\kappa_{t}\right) > 0$ since
$h_{t} < \kappa_{t}$, and $\boldsymbol{A}^{1/2}$ denotes the Cholesky decomposition of $\boldsymbol{A}$ and $\boldsymbol{A}^{-1/2}$ denotes its inverse. From Th. 2.3.10 in \cite{Gup99}, 
the random matrix $\boldsymbol{W}=(\boldsymbol{\Sigma}_{H}^{-1/2})\boldsymbol{E}_t \boldsymbol{V}^{-1/2}$ follows the matrix normal $\mathcal{N}_{p,n}(\boldsymbol{M},\boldsymbol{\Sigma}, \boldsymbol{I}_n)$ with $\boldsymbol{M}=(\boldsymbol{\Sigma}_{H}^{-1/2})(\tilde{\boldsymbol{M}}-\boldsymbol{M}_{\ast})\boldsymbol{V}^{-1/2}$ and $\boldsymbol{\Sigma}=(\boldsymbol{\Sigma}_{H}^{-1/2})\tilde{\boldsymbol{\Sigma}} \boldsymbol{\Sigma}_{H}^{-1/2}$. The product $\boldsymbol{WW}'$ follows a non-central Wishart distribution $\mathcal{W}_{p,n}(\boldsymbol{\Sigma},\boldsymbol{\Sigma}^{-1} \boldsymbol{M} \boldsymbol{M}')$ and following the results in \cite{mathai1980moments} and \cite{pham2015trace} its trace $r=\text{tr}(\boldsymbol{WW}')$ is the weighted sum of non-central chi-square variables $r=\lambda_1 r_{1}+\ldots+\lambda_p r_p$, where $r_{j}$ follows a non-central chi-square $f_{n,U_{jj}}$ with non-centrality parameter $U_{jj}$, with $U_{jj}$ being the element $j$-th diagonal element of $\boldsymbol{U}=\boldsymbol{Q}'\boldsymbol{\Sigma}^{-1}\boldsymbol{MM}'\boldsymbol{Q}=\boldsymbol{Q}'(\boldsymbol{\Sigma}_{H}^{1/2}\tilde{\boldsymbol{\Sigma}}^{-1}(\tilde{\boldsymbol{M}}-\boldsymbol{M}_{\ast})\boldsymbol{V}^{-1})(\tilde{\boldsymbol{M}}-\boldsymbol{M}_{\ast})'\boldsymbol{\Sigma}^{-1/2}_{H}\boldsymbol{Q}$, and $\boldsymbol{Q}$ an orthogonal matrix such that $\boldsymbol{Q}' \boldsymbol{\Sigma} \boldsymbol{Q} =\text{diag}(\lambda_1,\ldots,\lambda_n)$. Then $F_{H_{t}}(h|\alpha_t)=1-F_{r}(- 2\log\left( h/\kappa_{t} \right))$ and from \cite{kourouklis1985distribution} the density of $H_t$ is
\begin{eqnarray}
f_{H_{t}}(h|\alpha_t) &=&\frac{2}{h}\sum_{k=0}^{\infty}c_k g(- 2\log\left( h/\kappa_{t} \right);\frac{np}{2}+k,2\lambda)
\end{eqnarray}
with support $0<h<\kappa_{t}$,  where $0<\lambda<\infty$ is arbitrarily chosen, $g(x;a,b)$ is the density of a gamma distribution with shape and scale parameters $a>0$ and $b>0$, respectively, and the coefficients $c_k$, $d_k$ and $f_k$ satisfy
\begin{eqnarray*}
    c_{k}&=&\exp(-\sum_{j=1}^p U_{jj})\prod_{j=1}^{p}
(\lambda_j/\lambda)^{-n/2}f_k,\quad d_k=\frac{n}{2k}\sum_{j=1}^{p}(1-\lambda/\lambda_j)^{k} + \lambda \sum_{j=1}^{p}\frac{U_{jj}}{\lambda_j}(1-\lambda/\lambda_j/)^{k-1},\\
    f_{k+1}&=&\frac{1}{k+1}\sum_{j=1}^{k+1}j d_j f_{k+1-j}, k=0,1,2,\ldots\quad \text{with}\quad  f_0=1.
\end{eqnarray*}
The cdf can be obtained by a change of variable argument as
\begin{eqnarray}
    F_{H_{t}}(h|\alpha_t) &=&\int_{0}^{h}f_{H_{t}}(u|\alpha_t)du=\sum_{k=0}^{\infty}c_k\int_{- 2\log\left( h/\kappa_{t} \right)}^{\infty} g(u;\frac{np}{2}+k,2\lambda))du\nonumber\\    
    &=&\sum_{k=0}^{\infty}c_k (1-G(- 2\log\left( h/\kappa_{t} \right);\frac{np}{2}+k,2\lambda)),
\end{eqnarray}
where $G(x;a,b)$ is the cdf of the gamma density $g(x;a,b)$.
\end{proof}
\subsection{Proof of Corollary \ref{CorUnivBFDistr}}
\begin{proof}[Proof of Corollary \ref{CorUnivBFDistr}]
In the univariate case, following the notation in Example \ref{ex3}, we set $n=p=1$, $\tilde{\boldsymbol{\Sigma}}=\sigma_L^2$, $\boldsymbol{V}=1$ and $\tilde{\boldsymbol{M}}=\theta$. From Prop. \ref{DistribBFmatrix} one obtains $r=\lambda_1 r_1$ where $\lambda_1=\sigma^2_L/\sigma^2_H$ and $U_{11}=(\theta-m_{\ast})^2/\sigma^2_{L}$ and $r_1$ follows a non-central chi-square $f_{n,\gamma}(z)$ with $\gamma=U_{11}$. Then it follows $F_{H_{t}}(h|\alpha_t)=1-F_{r}(- 2\log\left( h/\kappa_{t} \right))=1-F_{n,U_{11}}(- 2\log\left(h/\kappa_{t} \right)\sigma^2_H/\sigma^2_L)$ with density
\begin{equation}
    f_{H_{t}}(h|\alpha_t)=2\frac{\sigma^2_H}{h\sigma^2_L}f_{n,U_{11}}(- 2\log\left(h/\kappa_{t} \right)\sigma^2_H/\sigma^2_L).
\end{equation}
From Th. 1.3 in \cite{ANDRAS2008395} the distribution of a noncentral chi-square with 1 degree of freedom can be written as $f_{n,\gamma}(z)=\exp(-(z+\gamma)/2)(2\pi z)^{-1/2}\cosh(\sqrt{\gamma z})$
for $n=1$, and the cdf of $H_t$ is
\begin{eqnarray}
&&    F_{H_t}(h|\alpha_t)=\int_{0}^{h}\frac{\sigma^2_H}{u\sigma^2_L}\exp(-(-2\log(u/\kappa_t)\sigma^2_H/\sigma^2_L+\gamma)/2)\frac{\cosh(\sqrt{-2\gamma\log(u/\kappa_t)\sigma^2_H/\sigma^2_L)}}{\sqrt{-4\pi\log(u/\kappa_t)\sigma^2_H/\sigma^2_L}}du\nonumber\\
&&=\int_{\sqrt{-2\log(h/\kappa_t)\sigma^2_H/\sigma^2_L}}^{+\infty}\exp(-(u^2+\gamma)/2)\frac{1}{\sqrt{2\pi}}\exp(\sqrt{\gamma}u)du+\nonumber
   \\ 
&&\quad \int_{\sqrt{-2\log(h/\kappa_t)\sigma^2_H/\sigma^2_L}}^{+\infty}\exp(-(u^2+\gamma)/2)\frac{1}{\sqrt{2\pi}}\exp(-\sqrt{\gamma}u)du\nonumber\\
&&=1-\Phi(\sqrt{-2\log(h/\kappa_t)\sigma^2_H/\sigma^2_L}-\sqrt{\gamma})+\Phi(-\sqrt{-2\log(h/\kappa_t)\sigma^2_H/\sigma^2_L}-\sqrt{\gamma}),   
\end{eqnarray}
where $\gamma=(\theta-m_{\ast})^2/\sigma^2$ and $\Phi(\cdot)$ is the cdf of a standard normal distribution.
\end{proof}

\vfill
\newpage

\renewcommand{\thesection}{B}
\renewcommand{\theequation}{B.\arabic{equation}}
\renewcommand{\thefigure}{B.\arabic{figure}}
\renewcommand{\thetable}{B.\arabic{table}}
\renewcommand{\theproposition}{.\arabic{proposition}}
\setcounter{table}{0}
\setcounter{figure}{0}
\setcounter{equation}{0}
\setcounter{proposition}{0}
\section{Further Simulation Results}
\label{sec:Simulation_appendix}
This section presents the simulation results for the following settings:  
$\boldsymbol{M}\sim \mathcal{N}_{p,n}(\boldsymbol{O}_{p\times n}, \boldsymbol{I}_{p},\boldsymbol{I}_{n})$,  
$\boldsymbol{\Sigma}=\boldsymbol{SS}'$, where $\boldsymbol{S}\sim\mathcal{N}_{p,p}(\boldsymbol{O}_{p\times p}, \boldsymbol{I}_{p},\boldsymbol{I}_{p})$,  
$\boldsymbol{\Psi}=\boldsymbol{GG}'$, where $\boldsymbol{G}\sim\mathcal{N}_{n,n}(\boldsymbol{O}_{n\times n}, \boldsymbol{I}_{n}, \boldsymbol{I}_{n})$.  

Two scenarios are considered:  
- $Case_1$: A moderate-sized observation matrix with $p = 30$ and $n = 10$.  
- $Case_2$: A large-sized observation matrix with $p = 50$ and $n = 50$.  

The synthetic dataset, including an outlier, is generated as follows:

\begin{enumerate}
    \item A noise sequence $\boldsymbol{E}_{t}$ is sampled independently as  $\boldsymbol{E}_{t}\sim\mathcal{N}_{p,n}(\boldsymbol{O}_{p\times p},\boldsymbol{\Sigma},\boldsymbol{\Psi})$ for $t=1,\ldots,100$.  

    \item The observed sequence is defined as follows:  
    - For most time points ($t\neq 80$): $\boldsymbol{X}_{t}=\boldsymbol{M}+\boldsymbol{E}_t$.  
    - At time $t=80$, an outlier is introduced: $\boldsymbol{X}_{t}=\boldsymbol{M}+u \boldsymbol{R}_t+\boldsymbol{E}_t$,  
      where $\boldsymbol{R}_t$ is a binary matrix specifying outlier positions and  
      $u$ takes values $0.5, 1, 1.5, 3, 5$, and $15$.  
      These values correspond to $1/30$, $1/15$, $1/10$, $1/5$, $1/3$,  
      and $1$ times the average standard deviation of the matrix normal distribution.  
\end{enumerate}

We considered various configurations of the matrix $\boldsymbol{R}_t$ and the magnitude of the outliers. For each setting $J$ independent datasets have been generated $X_{t}^{(j)}$, $t=1,\ldots,T$ for $j=1,\ldots,J$, with $T=100$ and the BFs $H_t^{(j)}$ computed. The probabilities $p_{I}=P(H_t>\bar{h})$, $p_{II}=P(H_t<\underline{h})$ and  $p_{III}=P(\bar{h}<H_t<\underline{h})$ have been estimated as follows
\begin{equation*}
    p_{I}^{J}=\frac{1}{J}\sum_{j=1}^{J} \mathbb{I}(H_t^{(n)}>\bar{h}),\,\,\,\, p_{II}^{J}=\frac{1}{J}\sum_{j=1}^{J} \mathbb{I}(\underline{h}<H_t^{(n)}<\bar{h}),\,\,\,\, p_{III}^{J}=\frac{1}{J}\sum_{j=1}^{J} \mathbb{I}(H_t^{(n)}<\underline{h})
\end{equation*}
under the null hypothesis of the absence of outliers at $t\neq 80$ and under the alternative hypothesis of a certain number of outlying observations in the observation matrix at $t=80$.
 A summary of the results for $Case_1$ and $Case_2$ is provided in Tab. \ref{tab:SimResults} and Tab. \ref{tab:SimResults1} respectively.

\begin{sidewaystable}[h!]
    \centering
    \setlength{\tabcolsep}{7pt}
    \captionsetup{width=0.9\textwidth}
    \renewcommand{\arraystretch}{1.15}
    \begin{small}
     \resizebox{0.9\textwidth}{!}{  
    \begin{tabular}{l|c|cccccccc|cccccc}
         &$\mathcal{H}_{0}$ & \multicolumn{14}{c}{$\mathcal{H}_{1}$} \\
         \hline
         & & \multicolumn{8}{c|}{(a) Patterns in the binary matrix $R_t\in\mathbb{R}^{30\times 10}$} & \multicolumn{6}{c}{(b) Random Positions within $R_t\in\mathbb{R}^{30\times 10}$}\\ 
         & & \multicolumn{8}{c|}{(number of random rows $\times$ number of random columns)} & \multicolumn{6}{c}{(number of randomly selected entries)} \\       
        & & $3\times 2$ &$1\times 10$ & $7\times 5$ & $9\times 7$ & $11\times 9$ & $15\times 9$ & $20\times 10$ & $25\times 10$ & 50 & 100 & 150 & 200 &250 & 300\\
        \hline\vspace{-20pt}\\
        \multicolumn{16}{c}{$\,$}\\                       
        \hline
        & & \multicolumn{14}{c}{One thirtieth ($1/30$) of the standard deviation, i.e. $u=0.5$ in $\boldsymbol{X}_{t}=\boldsymbol{M}+u \boldsymbol{R}_t+\boldsymbol{E}_t$ for $t=80$}\\
        \hline                          
$P(H_t>\bar{h})$
&0.05&0.06&0&0&0&0&0&0&0&0&0&0&0&0&0\\
$P(\underline{h}<H_t<\bar{h})$
&0.93&0.72&0.76&0.6&0.58&0.52&0.3&0.3&0.3&0.44&0.32&0.32&0.3&0.24&0.22\\
$P(H_t<\underline{h})$
&0.02&0.22&0.24&0.4&0.42&0.48&0.7&0.7&0.7&0.56&0.68&0.68&0.7&0.76&0.78\\
        \hline\vspace{-20pt}\\
        \multicolumn{16}{c}{$\,$}\\                  
        \hline
        & & \multicolumn{14}{c}{One fifteenth ($1/15$) of the standard deviation, i.e. $u=1$ in $\boldsymbol{X}_{t}=\boldsymbol{M}+u \boldsymbol{R}_t+\boldsymbol{E}_t$ for $t=80$}\\
        \hline 
$P(H_t>\bar{h})$&0.05&0.01&0&0&0.01&0&0&0&0&0&0&0&0&0&0\\
$P(\underline{h}<H_t<\bar{h})$&0.93&0.55&0.35&0.24&0.18&0.09&0.1&0.04&0.04&0.14&0.02&0.01&0.01&0.02&0.02\\
$P(H_t<\underline{h})$&0.02&0.44&0.65&0.76&0.81&0.91&0.9&0.96&0.96&0.86&0.98&0.99&0.99&0.98&0.98\\        
        \hline\vspace{-20pt}\\
        \multicolumn{16}{c}{$\,$}\\                        
        \hline
        & & \multicolumn{14}{c}{One tenth ($1/10$) of the standard deviation, i.e. $u=1.5$ in $\boldsymbol{X}_{t}=\boldsymbol{M}+u \boldsymbol{R}_t+\boldsymbol{E}_t$ for $t=80$}\\
        \hline                          
$P(H_t>\bar{h})
$&0.05&0.06&0&0&0&0&0&0&0&0&0&0&0&0&0\\
$P(\underline{h}<H_t<\bar{h})$
&0.93&0.36&0.26&0.09&0.07&0.09&0.05&0.03&0.02&0.04&0.01&0&0&0&0\\
$P(H_t<\underline{h})$
&0.02&0.63&0.74&0.91&0.93&0.91&0.95&0.97&0.98&0.96&0.99&1&1&1&1\\
        \hline \vspace{-20pt}\\                      
        \multicolumn{16}{c}{$\,$}\\                        
        \hline
        & & \multicolumn{14}{c}{One fifth ($1/5$) of the standard deviation, i.e. $u=3$ in $\boldsymbol{X}_{t}=\boldsymbol{M}+u \boldsymbol{R}_t+\boldsymbol{E}_t$ for $t=80$}\\
        \hline
$P(H_t>\bar{h})$&0.05&0.06&0&0&0&0&0&0&0&0&0&0&0&0&0\\
$P(\underline{h}<H_t<\bar{h})$&0.93&0.22&0.1&0.05&0.02&0&0&0&0&0&0&0&0&0&0\\
$P(H_t<\underline{h})$&0.02&0.78&0.9&0.95&0.98&1&1&1&1&1&1&1&1&1&1\\
        \hline\vspace{-20pt}\\
        \multicolumn{16}{c}{$\,$}\\                        
        \hline
        & & \multicolumn{14}{c}{One tenth ($1/3$) of the standard deviation, i.e. $u=5$ in $\boldsymbol{X}_{t}=\boldsymbol{M}+u \boldsymbol{R}_t+\boldsymbol{E}_t$ for $t=80$}\\
        \hline                          
$P(H_t>\bar{h})$&0.05&0&0&0&0&0&0&0&0&0&0&0&0&0&0\\
$P(\underline{h}<H_t<\bar{h})$&0.93&0.08&0.02&0.02&0&0.01&0&0&0&0&0&0&0&0&0\\
$P(H_t<\underline{h})$&0.02&0.92&0.98&0.98&1&0.99&1&1&1&1&1&1&1&1&1\\
        \hline\vspace{-20pt}\\
        \multicolumn{16}{c}{$\,$}\\                        
        \hline
        & & \multicolumn{14}{c}{One ($1$) standard deviation, i.e. $u=15$ in $\boldsymbol{X}_{t}=\boldsymbol{M}+u \boldsymbol{R}_t+\boldsymbol{E}_t$ for $t=80$}\\
        \hline
$P(H_t>\bar{h})$ 
&0.05&0 &0 &0 &0 &0 &0 &0 &0 &0 &0 &0 &0 &0 &0 \\
$P(\underline{h}<H_t<\bar{h})$ 
&0.93&0 &0 &0 &0 &0 &0 &0 &0 &0 &0 &0 &0 &0 &0 \\ 
$P(H_t<\underline{h})$ 
&0.02 &1 &1 &1 &1 &1 &1 &1 &1 &1 &1 &1 &1 &1 &1 \\ 
    \end{tabular}}
    \end{small}
    \caption{\small\textbf{Moderate-size observation matrix}. The observable sequence is defined as $\boldsymbol{X}_{t}=\boldsymbol{M}+\boldsymbol{E}_t$ for $t\neq 80$ and $\boldsymbol{X}_{t}=\boldsymbol{M}+u \boldsymbol{R}_t+\boldsymbol{E}_t$ for $t=80$, with $\boldsymbol{X}_{t},\boldsymbol{M},\boldsymbol{R}_t, \boldsymbol{E}_t\in\mathbb{R}^{30\times 10}$, where $u$ is the magnitude of the outlier and $\boldsymbol{R}_t$ is a random binary matrix. Note: the position of the outliers in entries of $\boldsymbol{R}_t$ can follow row and column patterns ($r \times c$ in (a) with $r$ and $c$ the given number of rows and columns, respectively, selected randomly) or can be completely random ($r$ in (b) with $r$ the given total number of outlying entries selected randomly).}
    \label{tab:SimResults}
\end{sidewaystable}

\begin{sidewaystable}[p]
    \centering
    \captionsetup{width=0.9\textwidth}
    \setlength{\tabcolsep}{7pt}
    \begin{small}
    \resizebox{0.9\textwidth}{!}{  
    \begin{tabular}{l|c|cccccccc|cccccc}
         &$\mathcal{H}_{0}$ & \multicolumn{14}{c}{$\mathcal{H}_{1}$} \\
         \hline
         & & \multicolumn{8}{c|}{(a) Patterns in the binary matrix $R_t\in\mathbb{R}^{50\times 50}$} & \multicolumn{6}{c}{(b) Random Positions within $R_t\in\mathbb{R}^{50\times 50}$}\\ 
         & & \multicolumn{8}{c|}{(number of random rows $\times$ number of random columns)} & \multicolumn{6}{c}{(number of randomly selected entries)} \\       
        & & $3\times 2$ &$1\times 10$ & $7\times 5$ & $9\times 7$ & $11\times 9$ & $15\times 9$ & $20\times 10$ & $25\times 10$ & 50 & 100 & 150 & 200 &250 & 300\\
        \hline\vspace{-20pt}\\
        \multicolumn{16}{c}{$\,$}\\                       
        \hline
        & & \multicolumn{14}{c}{One thirtieth ($1/30$) of the standard deviation, i.e. $u=0.5$ in $\boldsymbol{X}_{t}=\boldsymbol{M}+u \boldsymbol{R}_t+\boldsymbol{E}_t$ for $t=80$}\\
        \hline                          
$P(H_t>\bar{h})$&0.05&0.03&0&0&0&0&0&0&0&0&0&0&0&0&0\\
$P(\underline{h}<H_t<\bar{h})$&0.93&0.57&0.45&0.35&0.15&0.12&0.09&0.06&0.06&0.09&0.02&0.04&0.01&0.01&0\\
$P(H_t<\underline{h})$&0.02&0.4&0.55&0.65&0.85&0.88&0.91&0.94&0.94&0.91&0.98&0.96&0.99&0.99&1\\
        \hline\vspace{-20pt}\\
        \multicolumn{16}{c}{$\,$}\\                  
        \hline
        & & \multicolumn{14}{c}{One fifteenth ($1/15$) of the standard deviation, i.e. $u=1$ in $\boldsymbol{X}_{t}=\boldsymbol{M}+u \boldsymbol{R}_t+\boldsymbol{E}_t$ for $t=80$}\\
        \hline 
$P(H_t>\bar{h})$&0.05&0.02&0.01&0&0&0&0&0&0&0&0&0&0&0&0\\
$P(\underline{h}<H_t<\bar{h})$&0.93&0.33&0.23&0.11&0.06&0.03&0.01&0&0&0.01&0&0&0&0&0\\
$P(H_t<\underline{h})$&0.02&0.65&0.76&0.89&0.94&0.97&0.99&1&1&0.99&1&1&1&1&1\\
        \hline\vspace{-20pt}\\
        \multicolumn{16}{c}{$\,$}\\                        
        \hline
        & & \multicolumn{14}{c}{One tenth ($1/10$) of the standard deviation, i.e. $u=1.5$ in $\boldsymbol{X}_{t}=\boldsymbol{M}+u \boldsymbol{R}_t+\boldsymbol{E}_t$ for $t=80$}\\
        \hline                          
$P(H_t>\bar{h})$&0.05&0&0.01&0&0&0&0&0&0&0&0&0&0&0&0\\
$P(\underline{h}<H_t<\bar{h})$&0.93&0.23&0.16&0.04&0.01&0.02&0&0&0&0&0&0&0&0&0\\
$P(H_t<\underline{h})$&0.02&0.77&0.83&0.96&0.99&0.98&1&1&1&1&1&1&1&1&1\\
        \hline\vspace{-20pt}\\                       
        \multicolumn{16}{c}{$\,$}\\                        
        \hline
        & & \multicolumn{14}{c}{One fifth ($1/5$) of the standard deviation, i.e. $u=3$ in $\boldsymbol{X}_{t}=\boldsymbol{M}+u \boldsymbol{R}_t+\boldsymbol{E}_t$ for $t=80$}\\
        \hline
$P(H_t>\bar{h})$&0.05&0&0&0&0&0&0&0&0&0&0&0&0&0&0\\
$P(\underline{h}<H_t<\bar{h})$&0.93&0.08&0.04&0.01&0&0&0&0&0&0&0&0&0&0&0\\
$P(H_t<\underline{h})$&0.02&0.92&0.96&0.99&1&1&1&1&1&1&1&1&1&1&1\\
        \hline\vspace{-20pt}\\
        \multicolumn{16}{c}{$\,$}\\                        
        \hline
        & & \multicolumn{14}{c}{One tenth ($1/3$) of the standard deviation, i.e. $u=5$ in $\boldsymbol{X}_{t}=\boldsymbol{M}+u \boldsymbol{R}_t+\boldsymbol{E}_t$ for $t=80$}\\
        \hline                          
$P(H_t>\bar{h})$&0.05&0&0&0&0&0&0&0&0&0&0&0&0&0&0\\
$P(\underline{h}<H_t<\bar{h})$&0.93&0.02&0.01&0&0&0&0&0&0&0&0&0&0&0&0\\
$P(H_t<\underline{h})$&0.02&0.98&0.99&1&1&1&1&1&1&1&1&1&1&1&1\\
        \hline\vspace{-20pt}\\
        \multicolumn{16}{c}{$\,$}\\                        
        \hline
        & & \multicolumn{14}{c}{One ($1$) standard deviation, i.e. $u=15$ in $\boldsymbol{X}_{t}=\boldsymbol{M}+u \boldsymbol{R}_t+\boldsymbol{E}_t$ for $t=80$}\\
        \hline
$P(H_t>\bar{h})$&0.05&0&0&0&0&0&0&0&0&0&0&0&0&0&0\\
$P(\underline{h}<H_t<\bar{h})$&0.93&0&0&0&0&0&0&0&0&0&0&0&0&0&0\\
$P(H_t<\underline{h})$&0.02&1&1&1&1&1&1&1&1&1&1&1&1&1&1\\
\end{tabular}}
    \end{small}
    \caption{\small\textbf{Large-size observation matrix}. The observable sequence is defined as $\boldsymbol{X}_{t}=\boldsymbol{M}+\boldsymbol{E}_t$ for $t\neq 80$ and $\boldsymbol{X}_{t}=\boldsymbol{M}+u \boldsymbol{R}_t+\boldsymbol{E}_t$ for $t=80$, with $\boldsymbol{X}_{t},\boldsymbol{M},\boldsymbol{R}_t, \boldsymbol{E}_t\in\mathbb{R}^{50\times 50}$, where $u$ is the magnitude of the outlier, and $\boldsymbol{R}_t$ is a random binary matrix. Note: the position of the outliers in entries of $\boldsymbol{R}_t$ can follow row and column patterns ($r \times c$ in (a) with $r$ and $c$ the given number of rows and columns, respectively, selected randomly) or can be completely random ($r$ in (b) with $r$ the given total number of outlying entries selected randomly).}
    \label{tab:SimResults1}
\end{sidewaystable}

\pagebreak
\renewcommand{\thesection}{C}
\renewcommand{\theequation}{C.\arabic{equation}}
\renewcommand{\thefigure}{C.\arabic{figure}}
\renewcommand{\thetable}{C.\arabic{table}}
\renewcommand{\theproposition}{C.\arabic{proposition}}
\setcounter{table}{0}
\setcounter{figure}{0}
\setcounter{equation}{0}
\setcounter{proposition}{0}

\section{Datasets Description and Outlier Detection}\label{sec:datasets}
In this paper, we consider a multi-country dataset of macroeconomic variables (see \citealp{Can09}), an international trade network dataset (see \citealp{Rose2004}) and a financial network dataset (\citealp{billio2018bayesianDT}).

\subsection{Inflation and Unemployment Dataset}
\label{subsec:datset_infl}
The Inflation and Unemployment Dataset includes three relevant variables for macroeconomic research: the Industrial Production Index, the Price Index and the Unemployment Rate, available from the Federal Reserve Economic Data and the EUROSTAT databases. We selected observations from 11 countries of the EU (Austria, Belgium, Estonia, Finland, France, Germany, Greece, Ireland, Italy, The Netherlands, Portugal, and Spain).  Series are seasonally and working day adjusted, and when necessary, they have been differentiated to get stationarity. Variables are sampled at a monthly frequency from January 2002 to October 2022.

Outlying observations are detected in periods related to three main events: i) the COVID-19 pandemic (from March 2020 to August 2020 and in November 2020); ii) the Ukrainian conflict (March and April 2022); iii) and the rise of inflation (September and October 2022). Our procedure results indicate that the pandemic outbreak's main effects are in March 2020 since the BF of February 2020 falls within the inconclusive interval. Also, the effects of the Ukrainian conflict on the global economy began in March 2020, and the beginning of inflation rise was detected in September 2020. Another outlying observation was detected in October 2021 and is related to significant changes in the unemployment series. 

\subsection{Trade Network Dataset}
\label{subsec:dataset_trade}
We consider an international trade dataset made available by the International Monetary Fund (IMF). This dataset is now a reference in international trade  studies\cite[e.g., see] []{Rose2004,Alcala2004,Brancaccio2020} and combines the country reports with other sources such as United Nations COMTRADE and EUROSTAT. The data frequency is annual. The dataset includes 159 countries from 1995 to 2017. The trade network corresponds to the value of all goods imported (CIF) between pairs of countries collected by the IMF in U.S. dollars, which are transformed into constant terms using the US GDP PPP deflator from the same source with the base year 2010. All series have been differentiated to remove unit roots, and the units sub-selected to include larger GDP countries and to filter out series with more than 8\% of zeroes. The final sample includes a sequence of 22 networks between 27 countries, i.e. $n=p=27$, with the following ISO 3166--1 alpha--3 codes: AUS, AUT, BRA, CAN, CHE, CHN, DEU, DNK, ESP, FIN, FRA, GBR, GRC, IDN, IND, IRL, ITA, JPN, KOR, NLD, NOR, NZL, POL, PRT, SWE, TUR, and USA. 

BFs provide evidence of the absence of an outlier in 2016 and the presence of outlying observations in 2015 and 2017.

\subsection{Volatility Network Dataset}
\label{subsec:dataset_vol}
 The dataset consists of volatility networks composed of 50 firms (22 German, 24 French, 4 Italian) belonging to 11 GICS sectors: Financials (7 firms), Communication Services (4 firms), Consumer Discretionary (11 firms), Consumer Staples (5 firms), Health Care (6 firms), Energy (2 firms), Industrials (5 firms), Information Technology (3 firms), Materials (2 firms), Real Estate (1 firm), Utilities (3 firms), and not classified in a specific GICS sector (1 firm). The temporal volatility networks have been extracted by pairwise Granger causality approach  \citep[e.g., see][]{billio2021matrix} at a weekly frequency (Friday-Friday). The sample period ranges from the 4th of January 2016 to the 30th of September 2020, and periods before and after the outbreak of COVID-19 are included in the time interval.

 The identified outliers in 2019 are associated with transitions between different volatility levels: i) from ``low'' to ``very low" volatility on 11-15 November; ii) from ``low'' to ``moderate'' on 2-6 December; iii) from ``moderate" to ``low'' on the week of 16-20 December. In 2020, the following transitions were detected: i) from ``low'' to ``high'' during the week of 17-21 February; ii) a progressive transition from ``high'' to ``moderate'' volatility between 7 and 24 April; iii) from ``moderate'' to ``low'' volatility during the week of 3-7 August; iv) a sequence of outliers corresponding to a period of switches between ``low'' and ``moderate'' between 1 September and 2 October. The BFs of the transitory periods mentioned belong to the inconclusive interval: 25-28 February 2020 (``low'' to ``high"), 27-30 April (``high'' to ``moderate''), and 21-25 September (``low" to ``moderate").

\pagebreak
\renewcommand{\thesection}{D}
\renewcommand{\theequation}{D.\arabic{equation}}
\renewcommand{\thefigure}{D.\arabic{figure}}
\renewcommand{\thetable}{D.\arabic{table}}
\renewcommand{\theproposition}{D.\arabic{proposition}}
\setcounter{table}{0}
\setcounter{figure}{0}
\setcounter{equation}{0}
\setcounter{proposition}{0}

\section{Further Empirical Results}
\begin{figure}[h!]
    \centering
    \setlength{\tabcolsep}{5pt}

    \begin{tabular}{ccc}
\small    September 2018 & \small January 2019 & \small July 2019\\
    \includegraphics[scale=0.3]{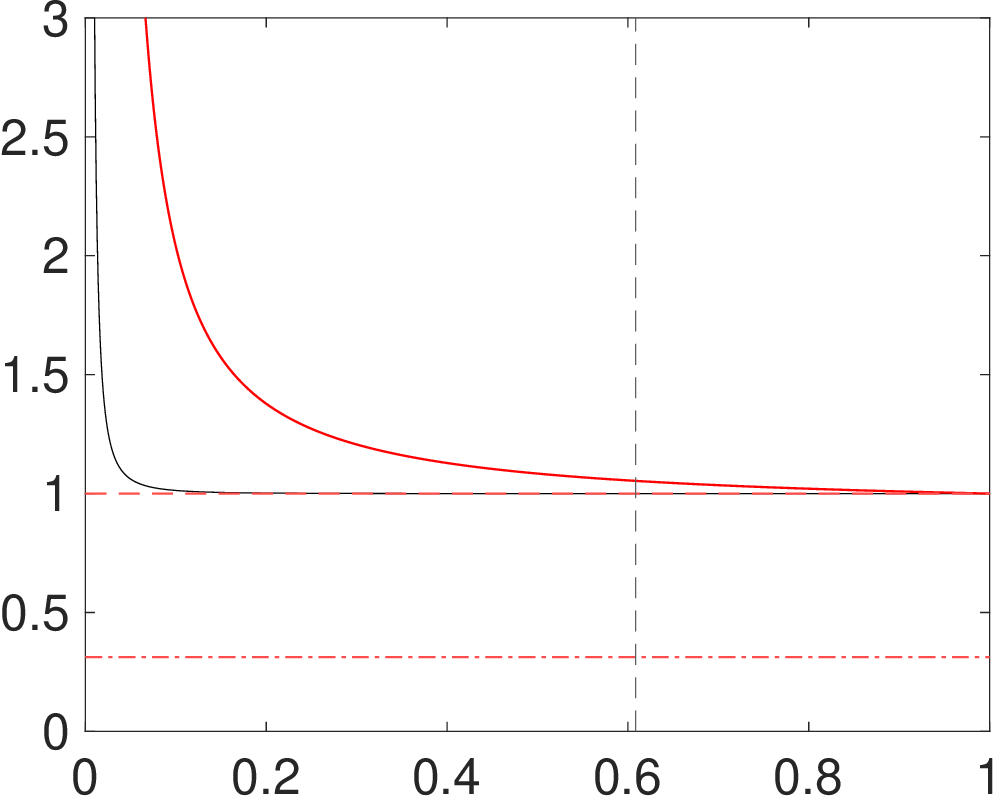}&
        \includegraphics[scale=0.3]{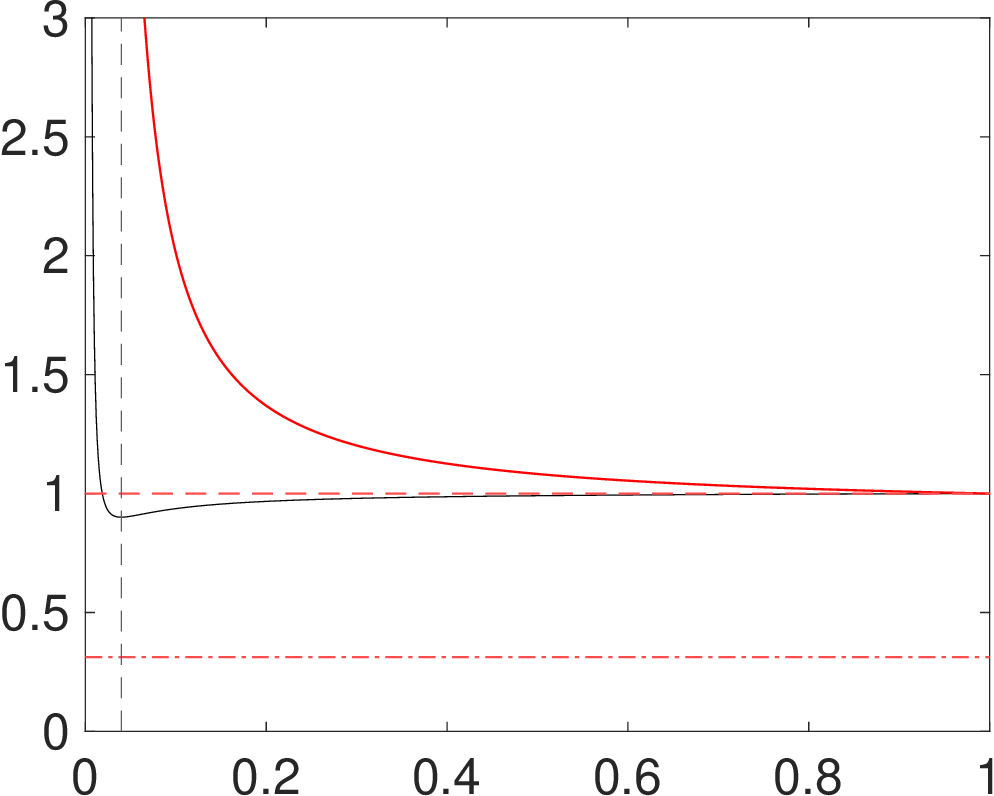}&
    \includegraphics[scale=0.3]{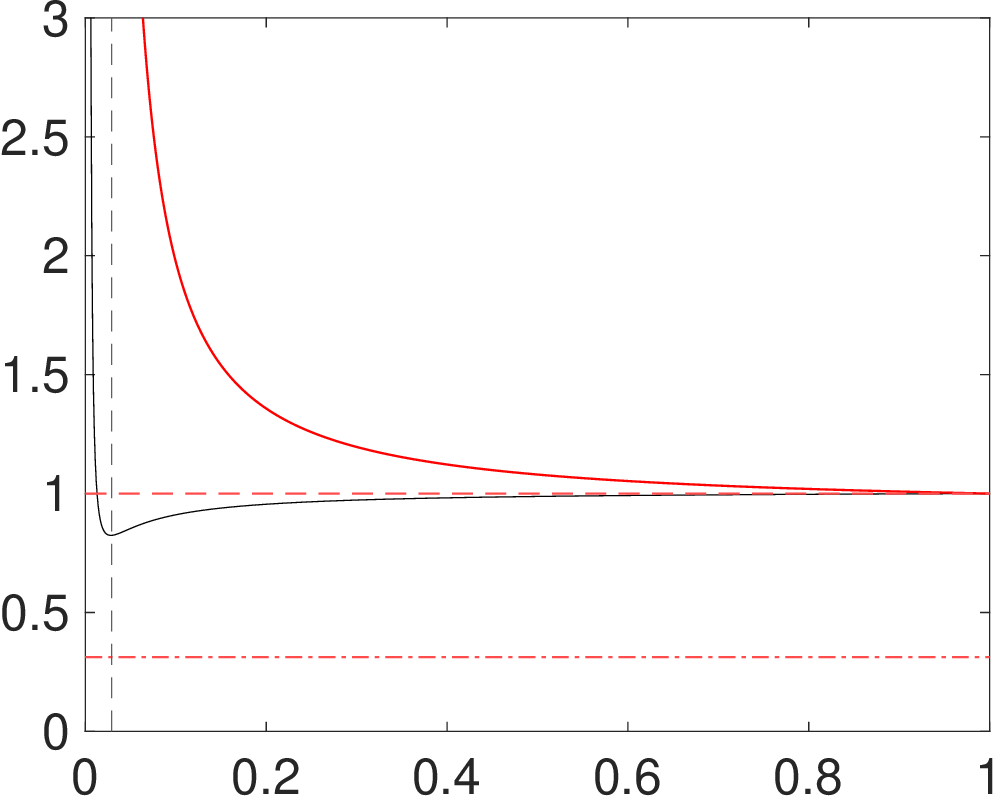}\\
\small    September 2019 & \small January 2020 & \small February 2020\\        \includegraphics[scale=0.3]{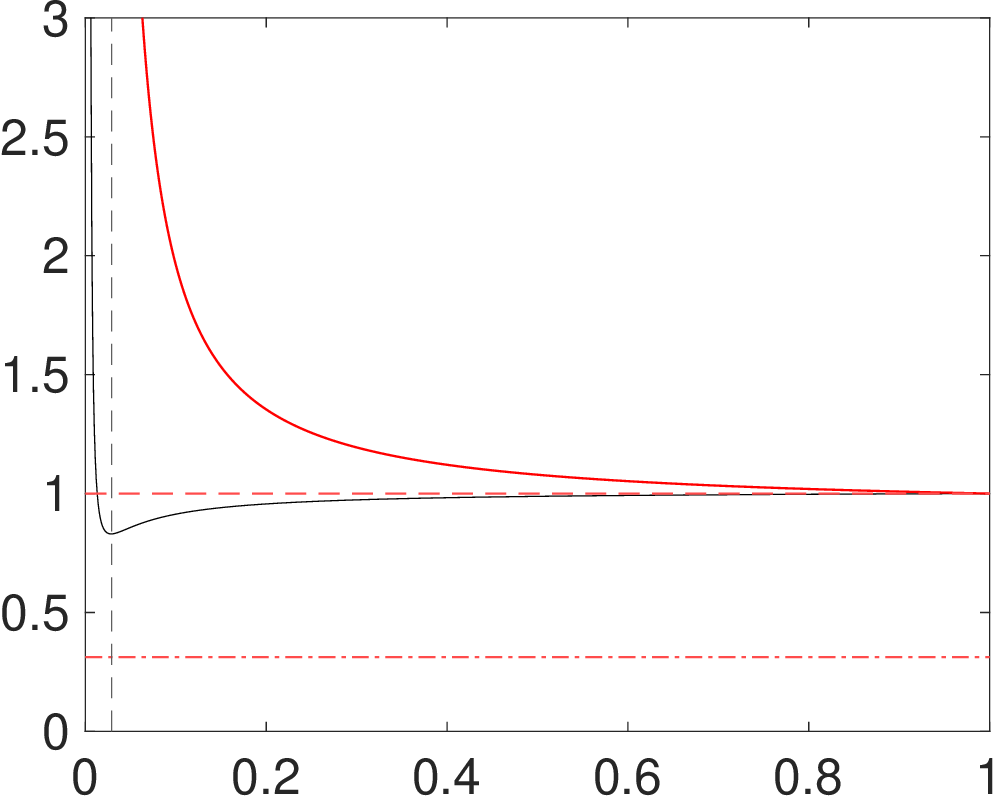}&
        \includegraphics[scale=0.3]{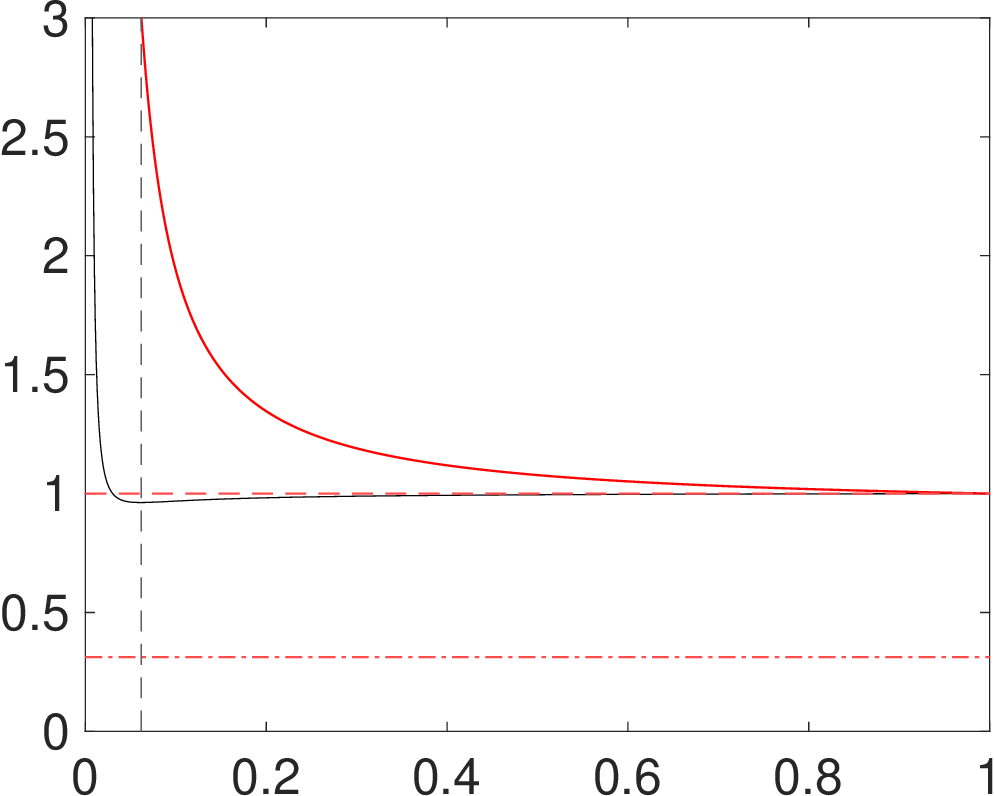}&
    \includegraphics[scale=0.3]{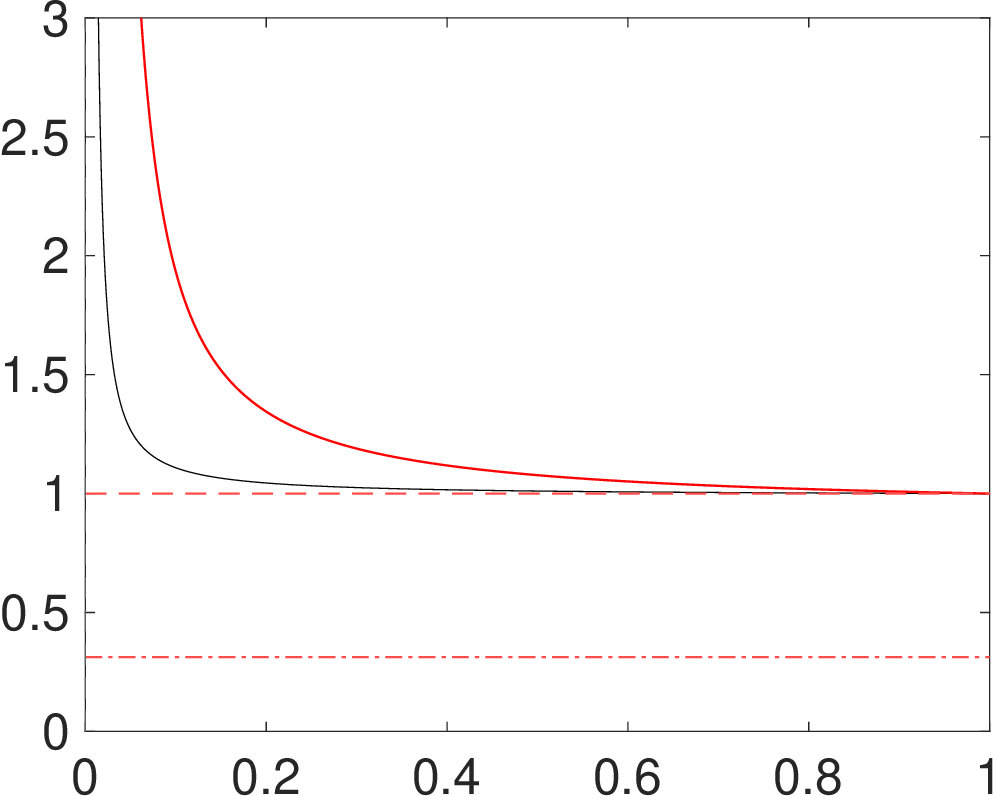}\\
\small    March 2020 & \small February 2021 & \small June 2021\\        \includegraphics[scale=0.3]{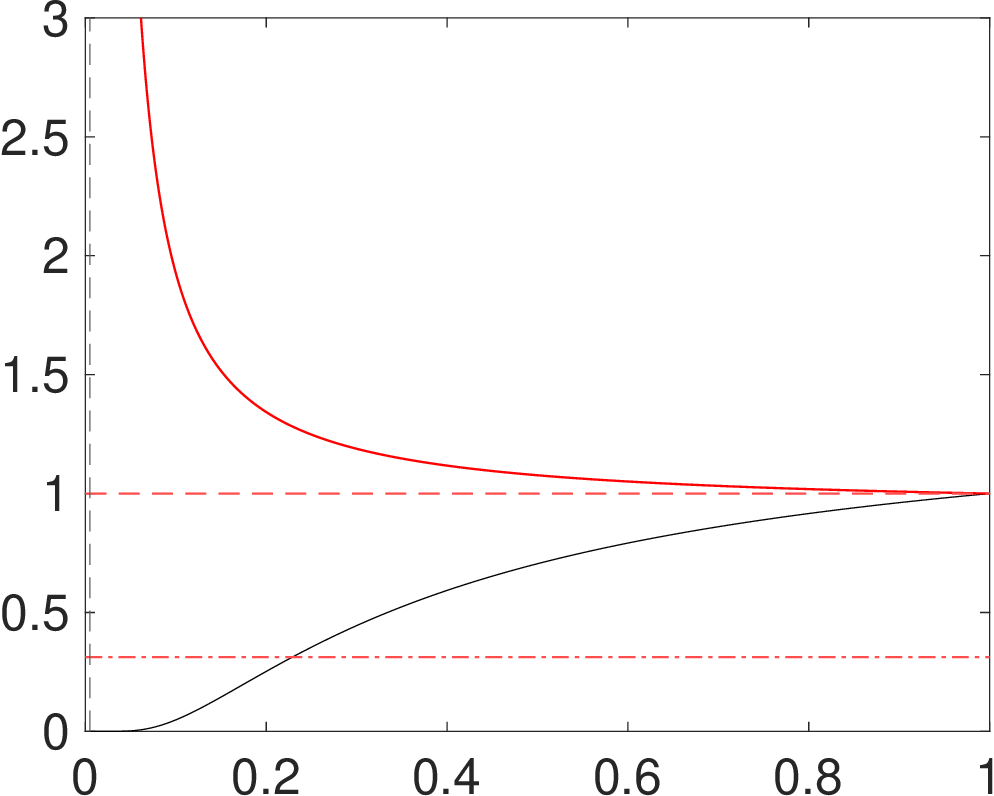}&
        \includegraphics[scale=0.3]{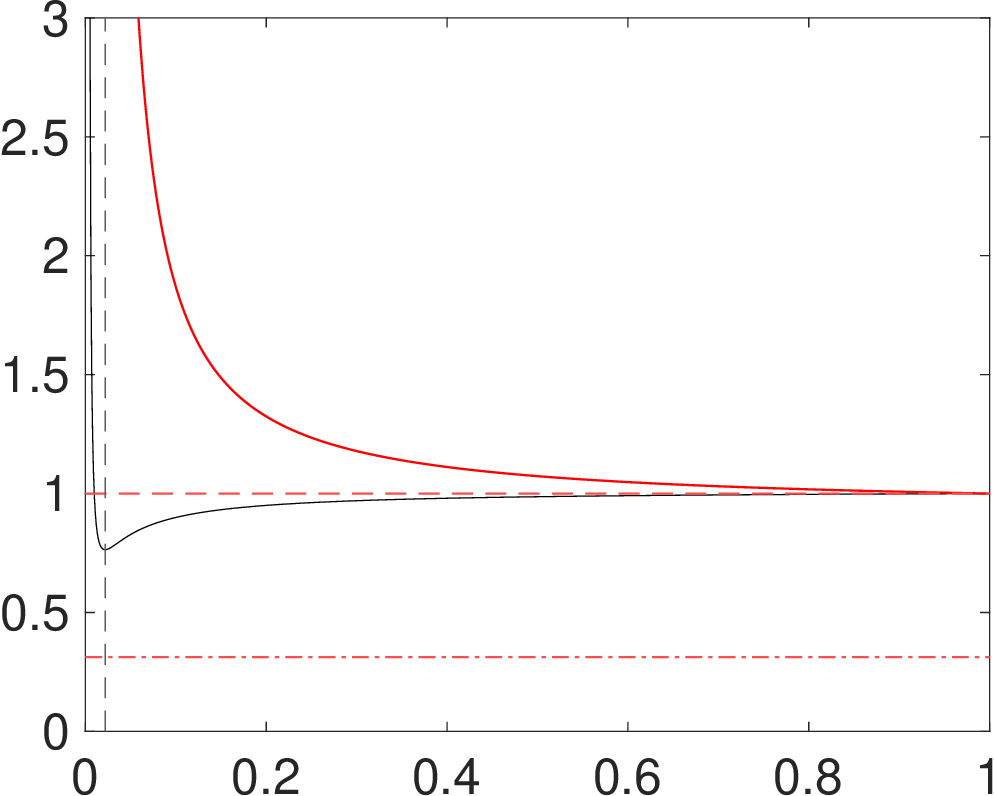}&
    \includegraphics[scale=0.3]{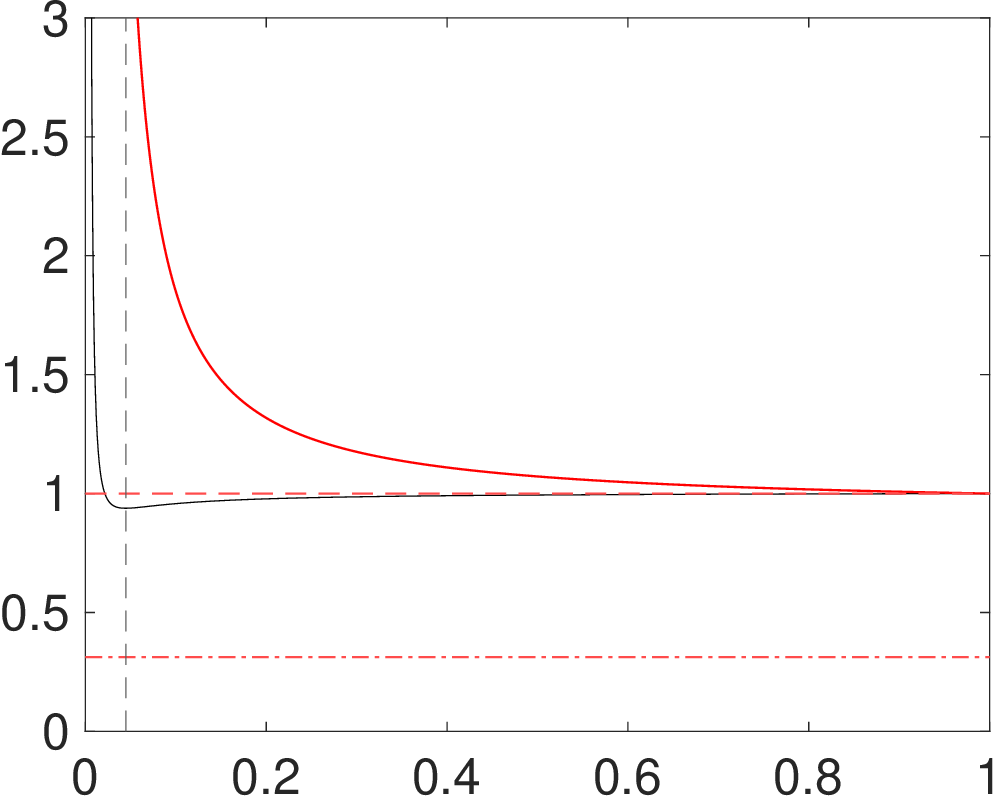}\\
    \end{tabular}
    \caption{Macroeconomic dataset. Value of the BF $H_t(\alpha_t)$ (solid black), the upper bound $\kappa_t(\alpha_t)$ (solid red) at different dates. In all plots, the reference lines at 1 ({\color{red}\protect\tikz[baseline]{\protect\draw[line width=0.2mm, dashed] (0,.6ex)--++(0.5,0) ;}}) and $10^{-1/2}$ ({\color{red}\protect\tikz[baseline]{\protect\draw[line width=0.2mm, dashdotted] (0,.6ex)--++(0.5,0) ;}}).}
    \label{emp1:figBFdatesBIS}
\end{figure}

\begin{figure}[h!]
    \centering
    \setlength{\tabcolsep}{5pt}
    \begin{tabular}{ccc}
\small    2008 & \small  2009 & \small  2010\\
    \includegraphics[scale=0.3]{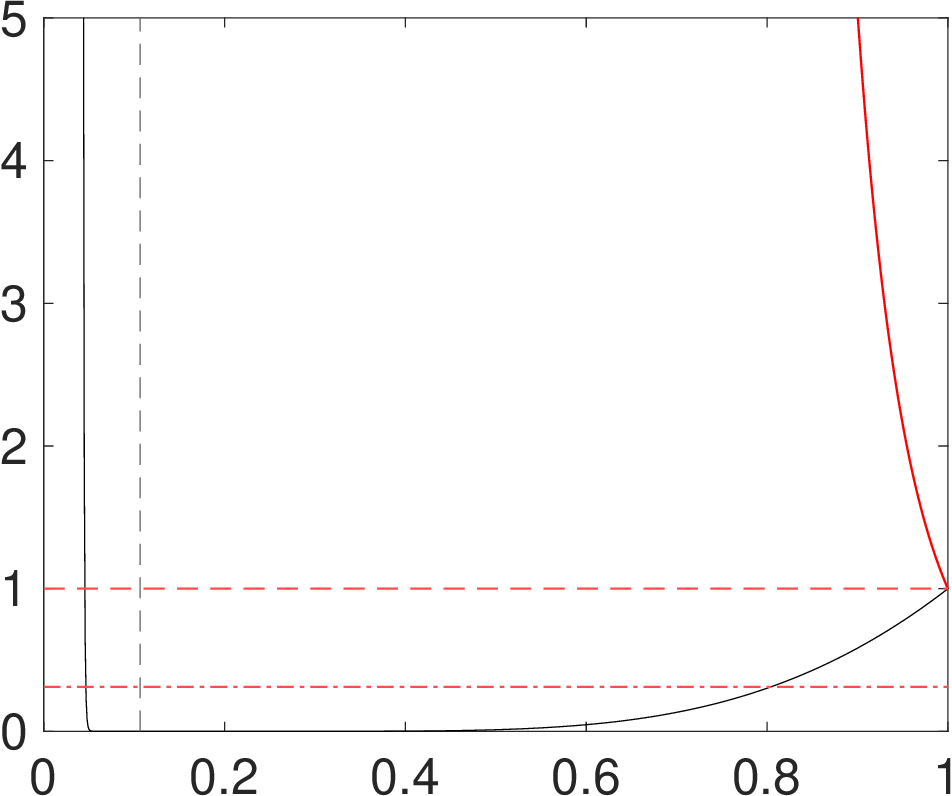}&
        \includegraphics[scale=0.3]{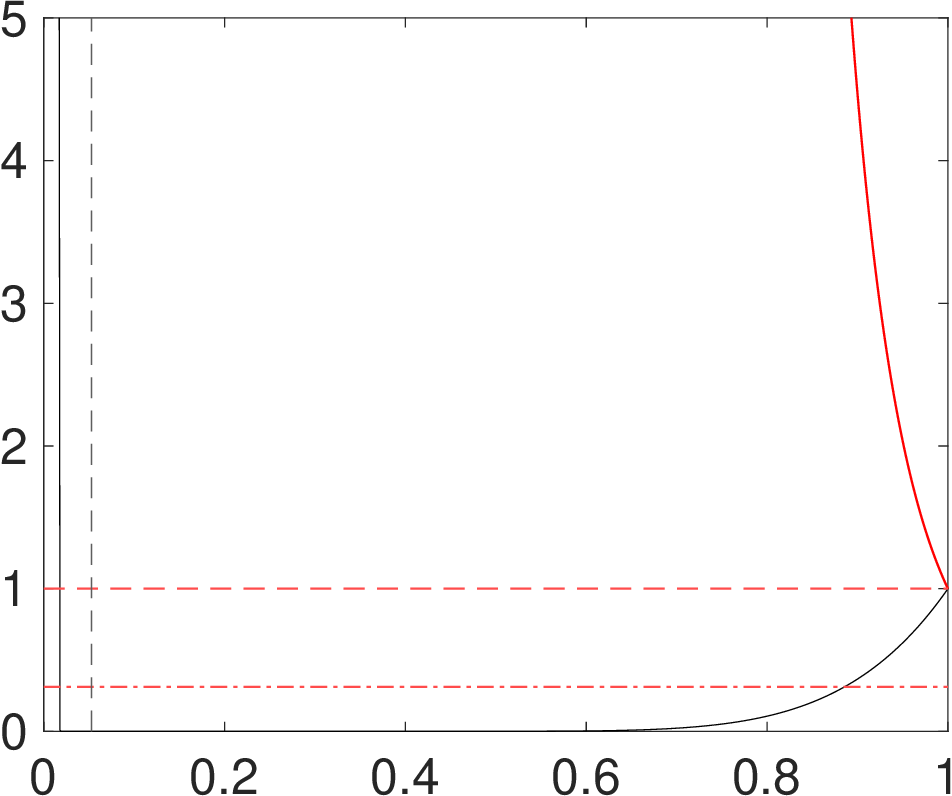}&
    \includegraphics[scale=0.3]{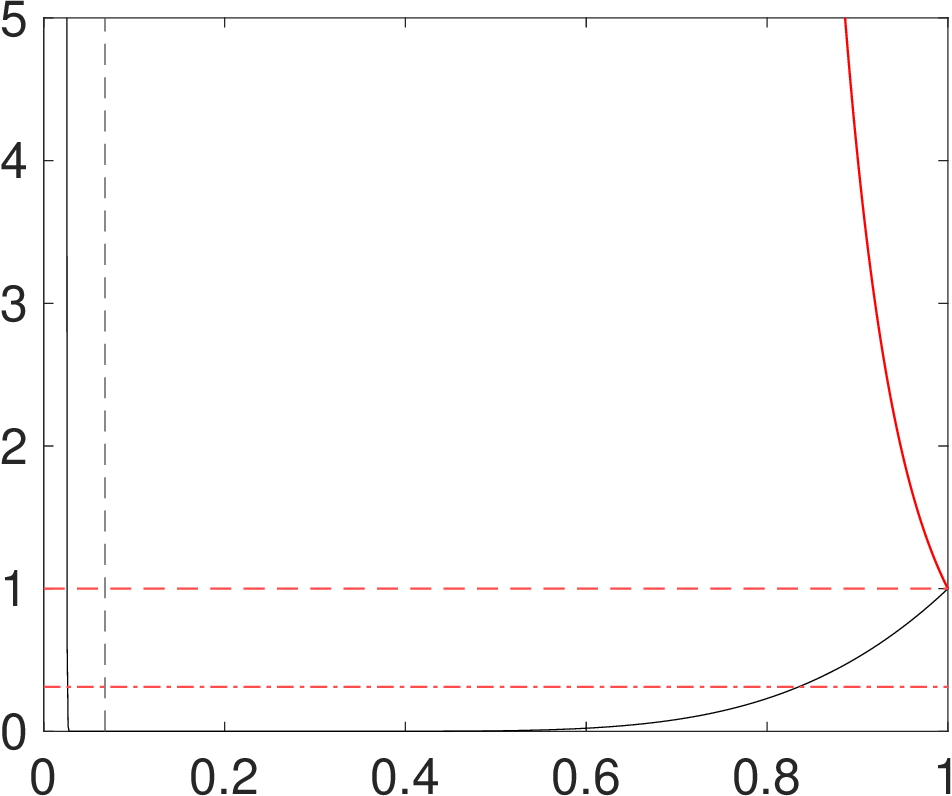}\\
\small     2011 & \small  2012 & \small  2013\\        \includegraphics[scale=0.3]{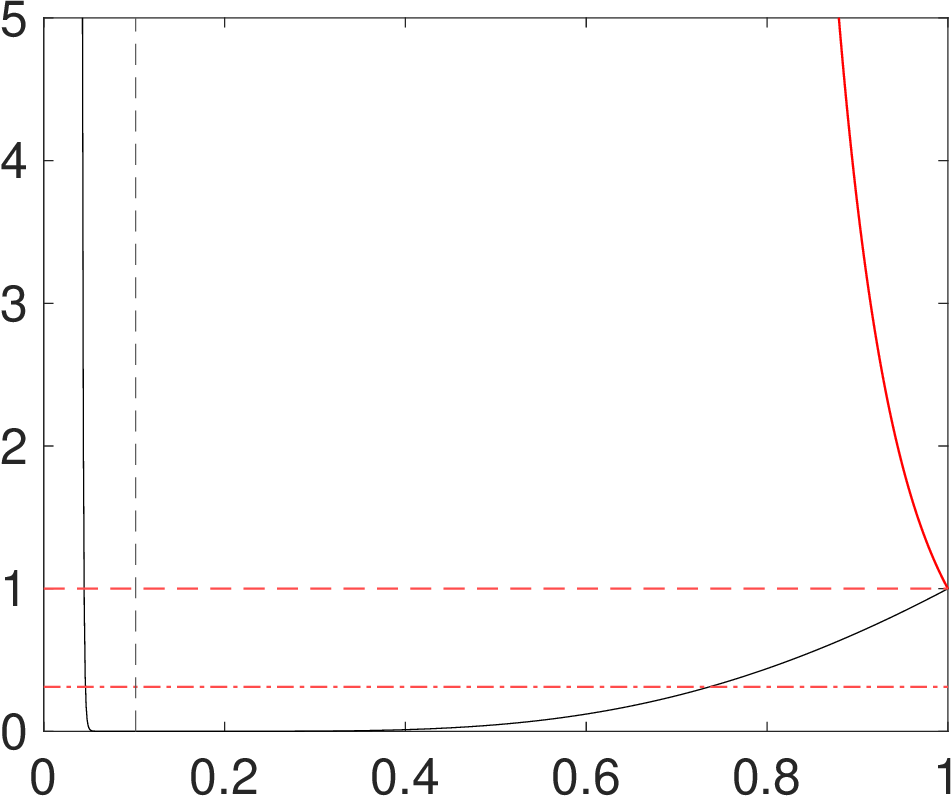}&
        \includegraphics[scale=0.3]{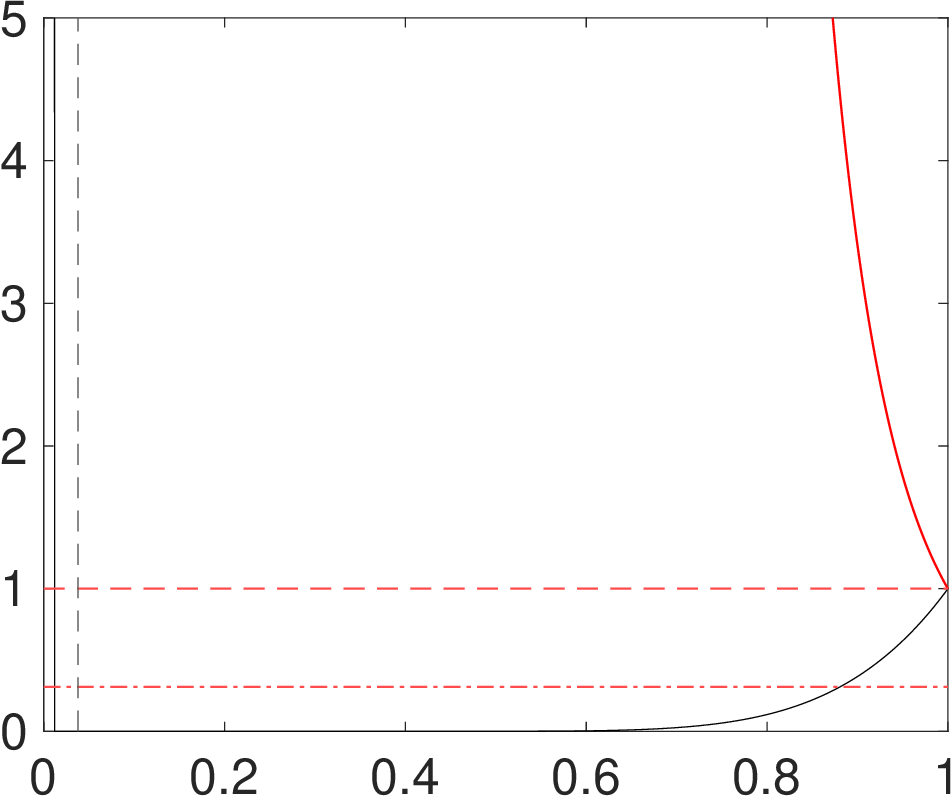}&
    \includegraphics[scale=0.3]{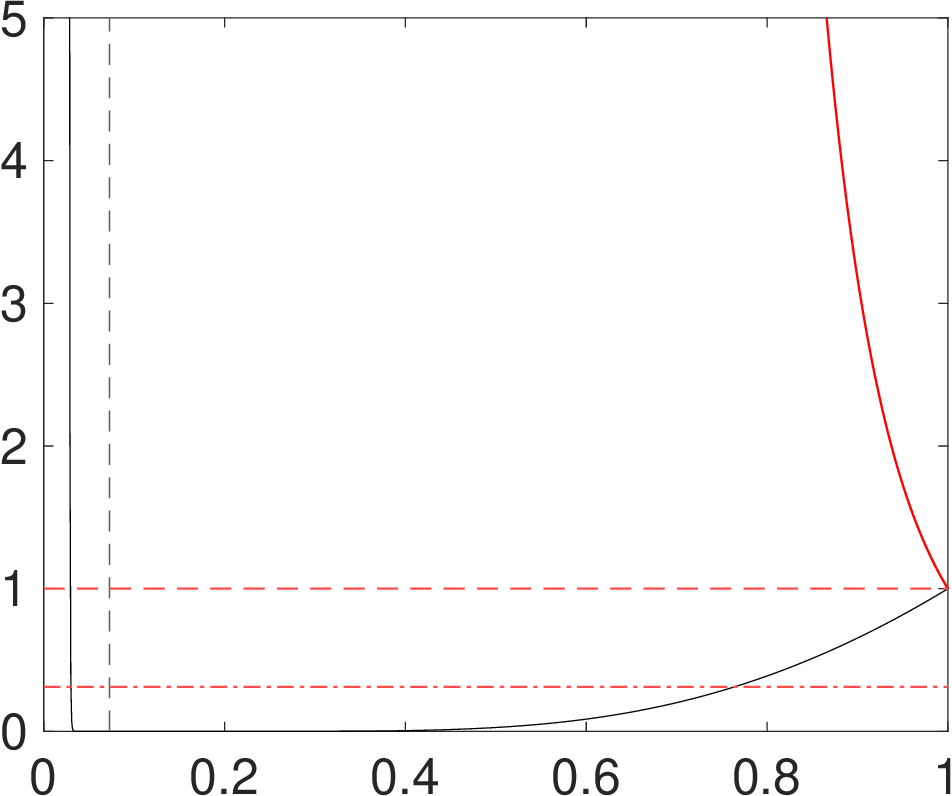}\\
\small     2014 & \small  2015 & \small  2016\\        \includegraphics[scale=0.3]{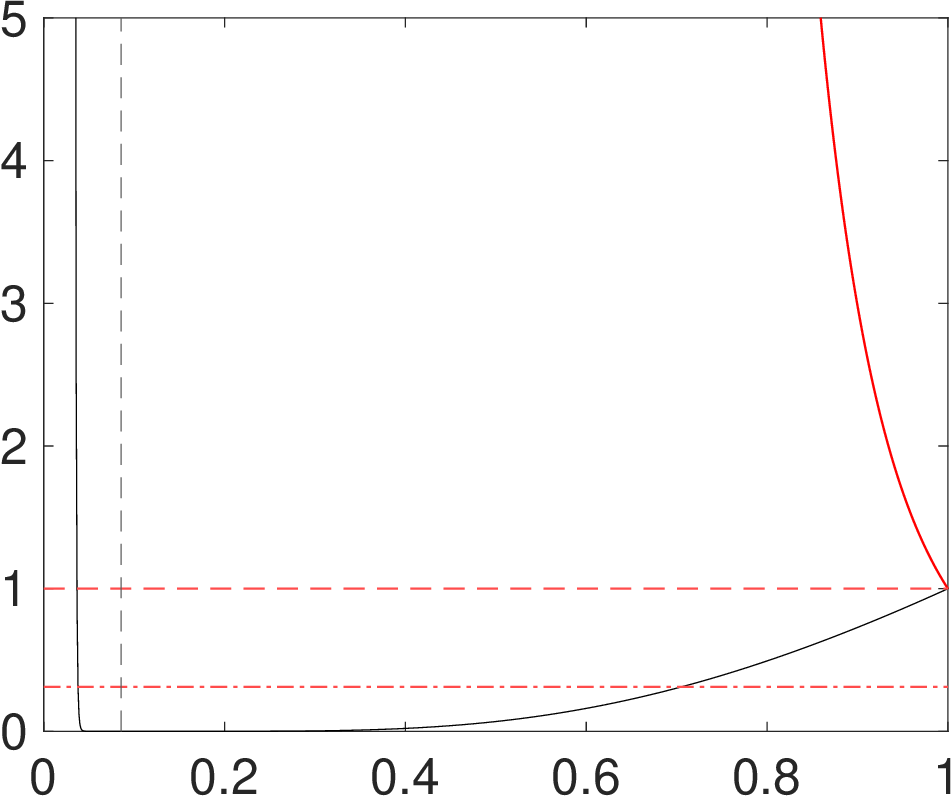}&
        \includegraphics[scale=0.3]{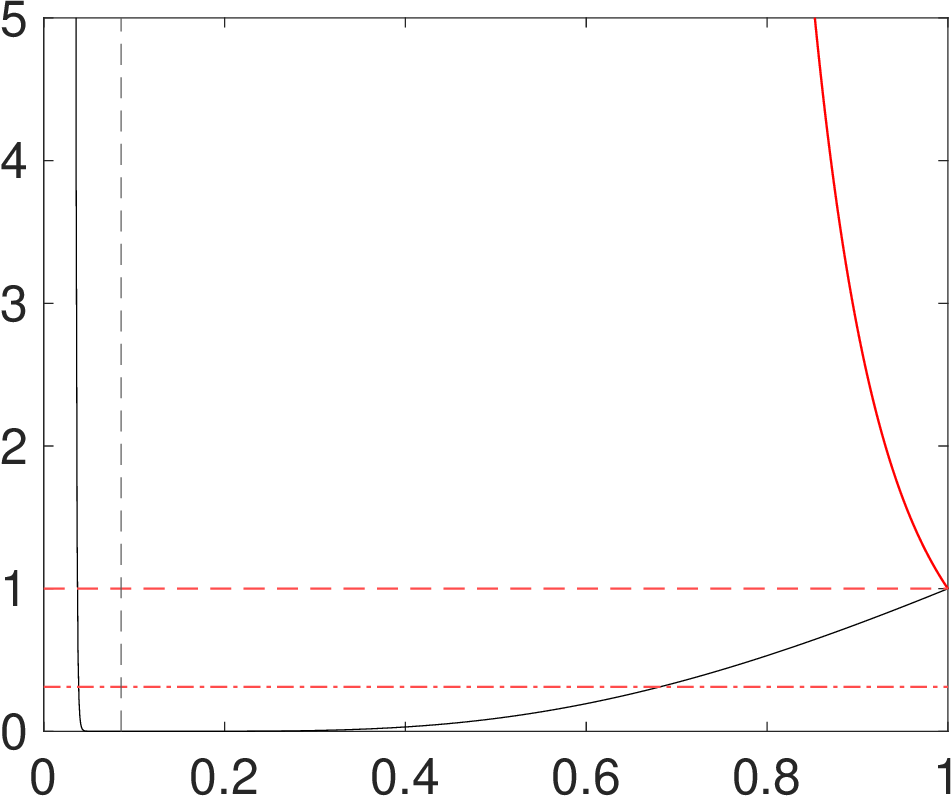}&
    \includegraphics[scale=0.3]{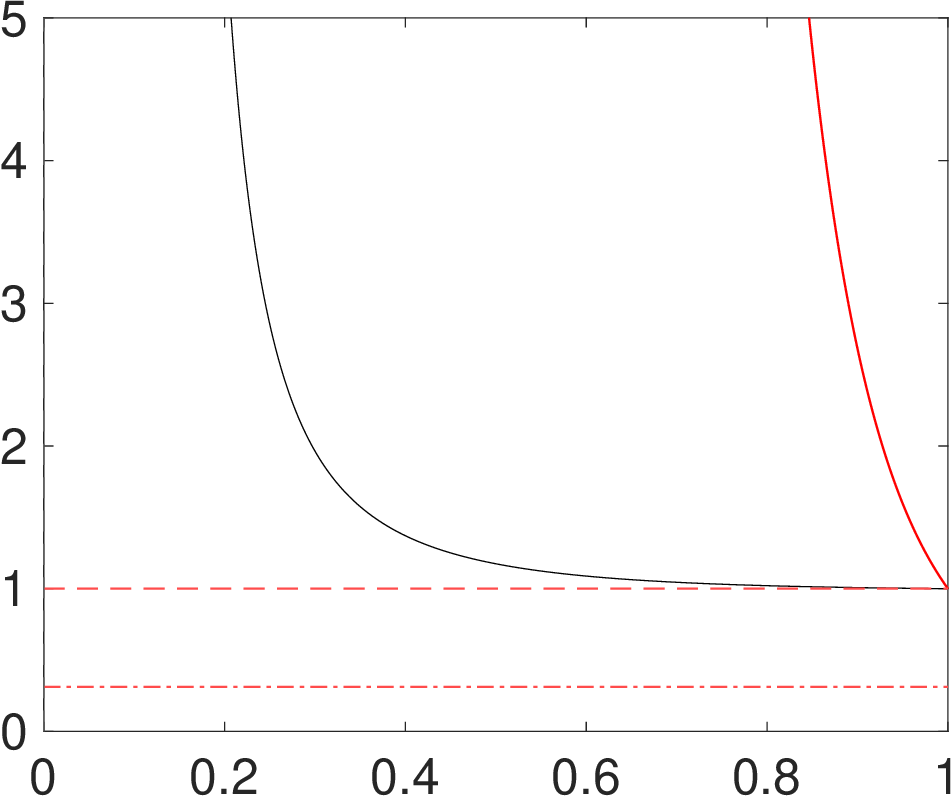}\\
    \end{tabular}
    \caption{International trade dataset. Value of the BF $H_t(\alpha_t)$ (solid black), the upper bound $\kappa_t(\alpha_t)$ (solid red) at different dates. In all plots, the reference lines at 1 ({\color{red}\protect\tikz[baseline]{\protect\draw[line width=0.2mm, dashed] (0,.6ex)--++(0.5,0) ;}}) and $10^{-1/2}$ ({\color{red}\protect\tikz[baseline]{\protect\draw[line width=0.2mm, dashdotted] (0,.6ex)--++(0.5,0) ;}}).}
    \label{emp1:figBFdatesBIStrade}
\end{figure}

\begin{figure}[h!]
    \centering
    \setlength{\tabcolsep}{5pt}
    \begin{tabular}{ccc}
\small    20 Sep 2019 & \small  11 Oct 2019 & \small  15 Nov 2019\\
       \includegraphics[scale=0.3]{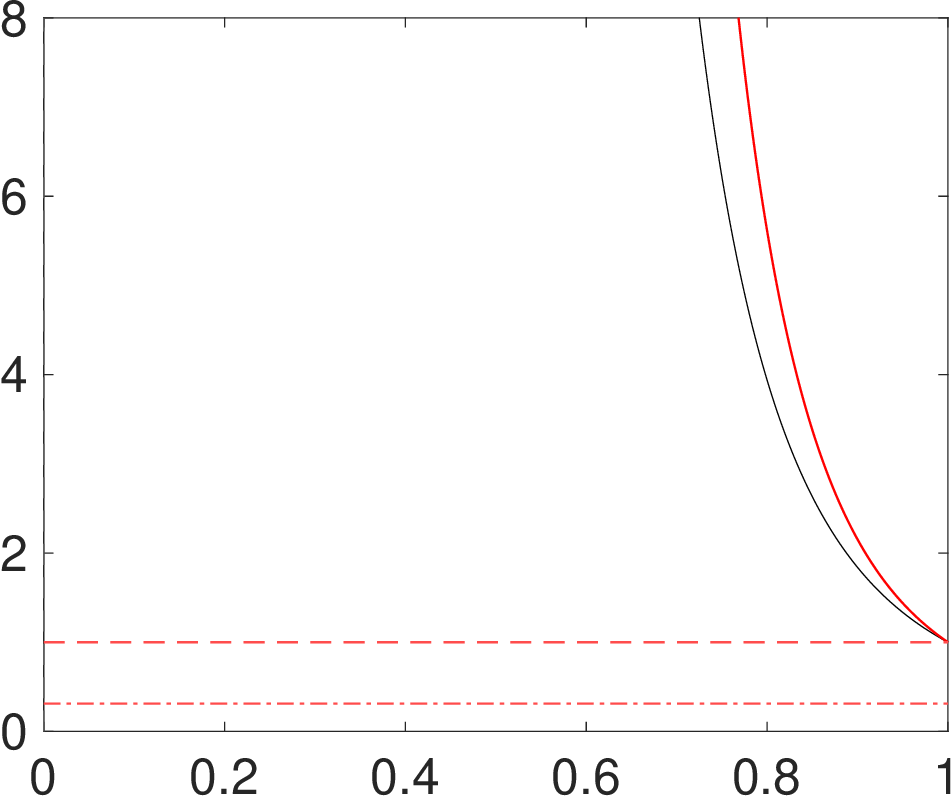}&
        \includegraphics[scale=0.3]{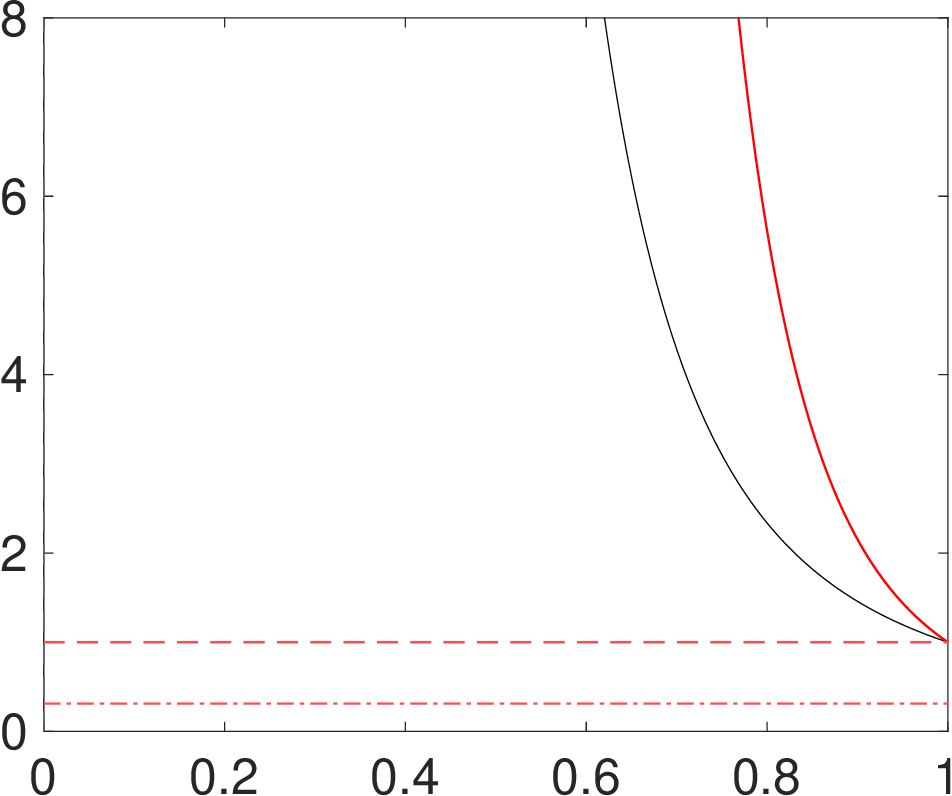}&
    \includegraphics[scale=0.3]{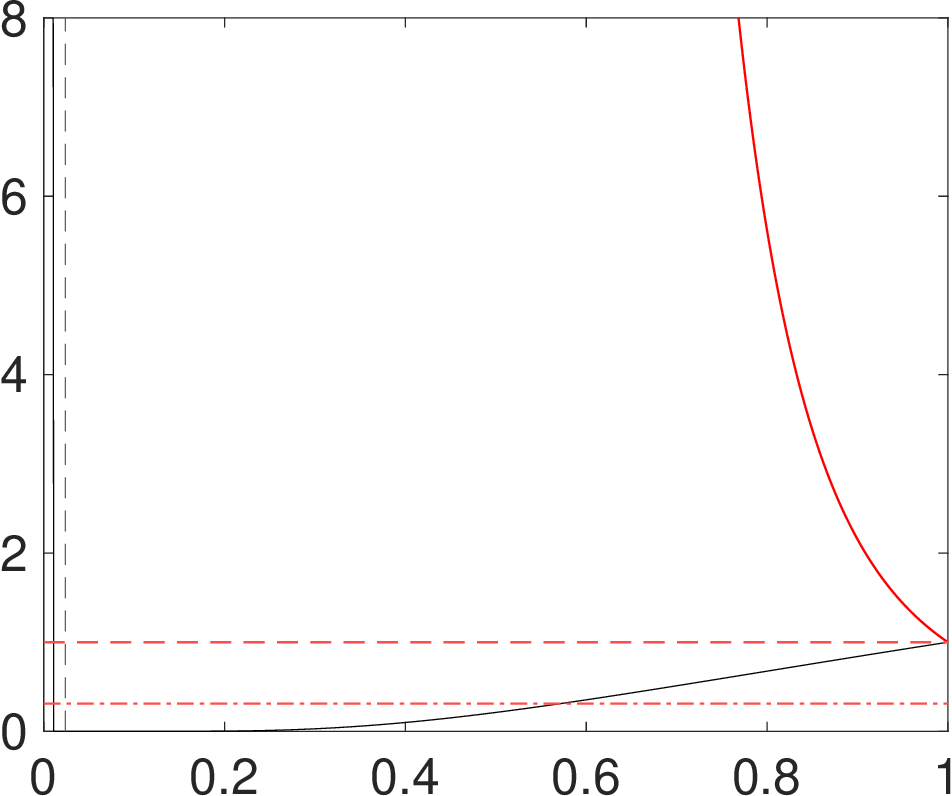}\\
\small     01 May 2020 & \small  14 Feb 2020 & \small  28 Feb 2020\\        \includegraphics[scale=0.3]{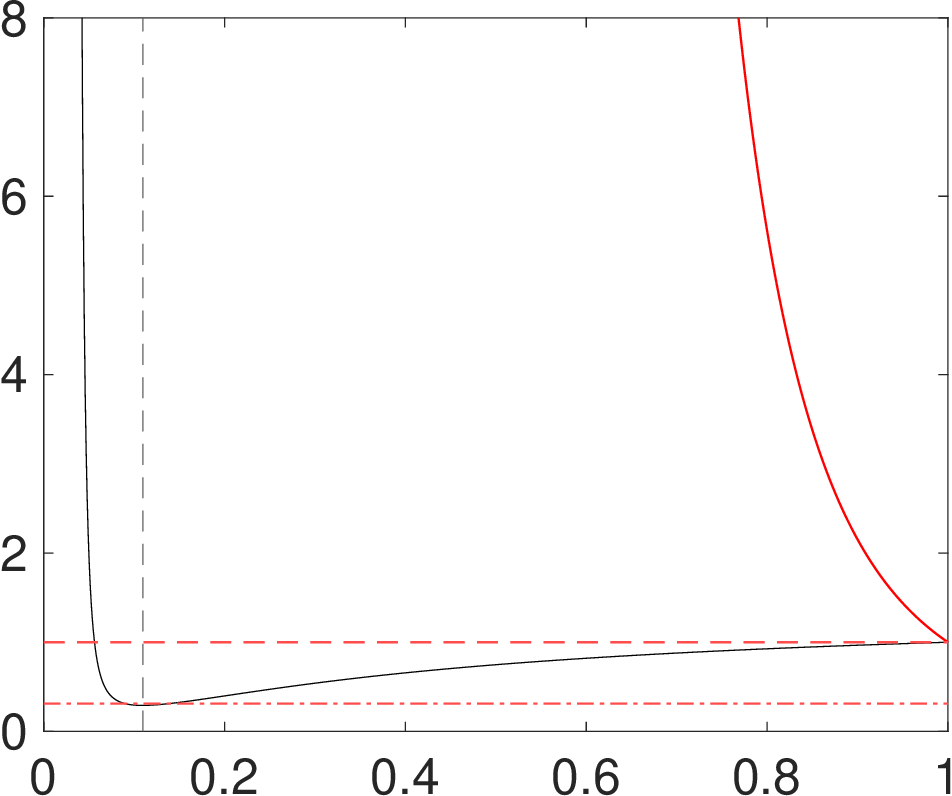}&
        \includegraphics[scale=0.3]{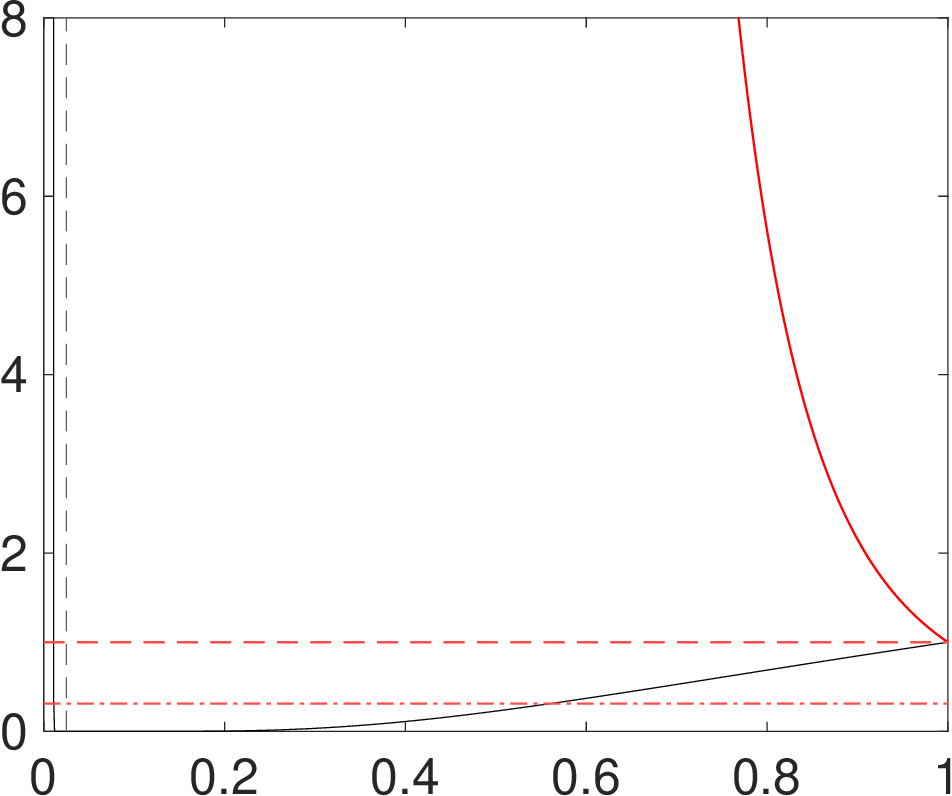}&
    \includegraphics[scale=0.3]{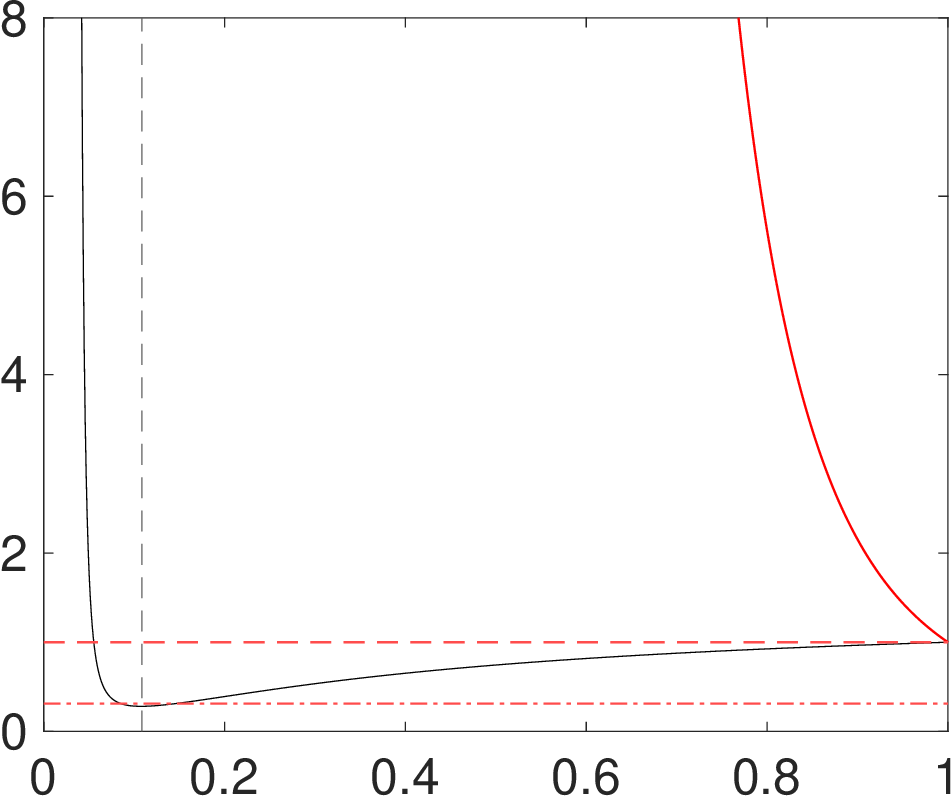}\\
\small  11 Sep 2020   & \small  18 Sep 2020  & \small  25 Sep 2020\\        \includegraphics[scale=0.3]{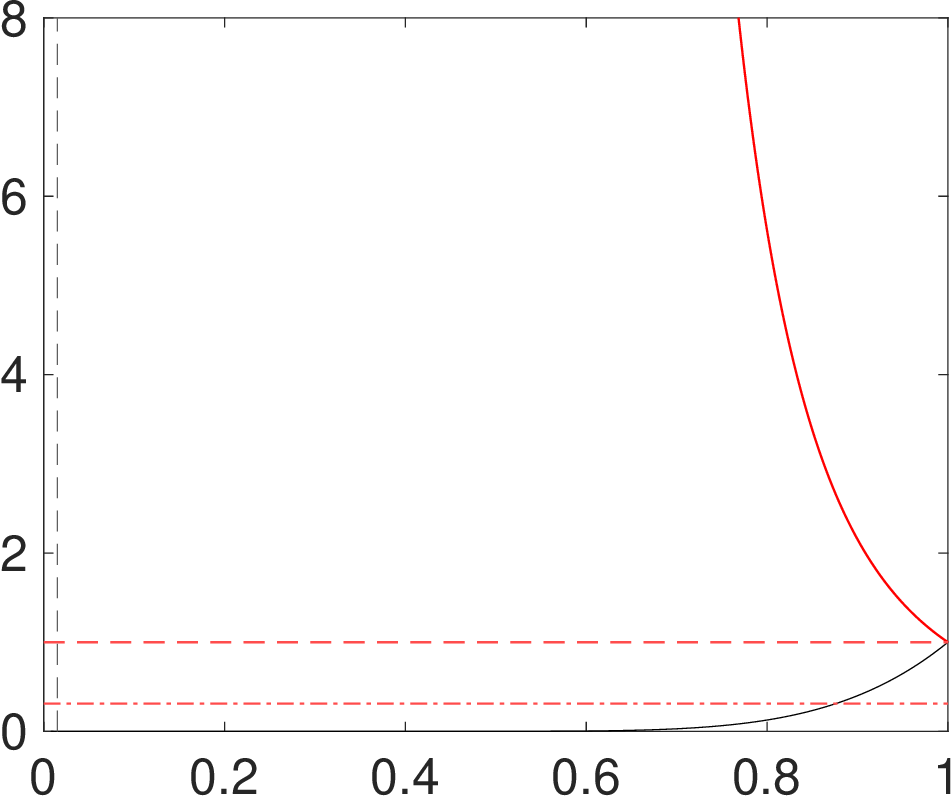}&
        \includegraphics[scale=0.3]{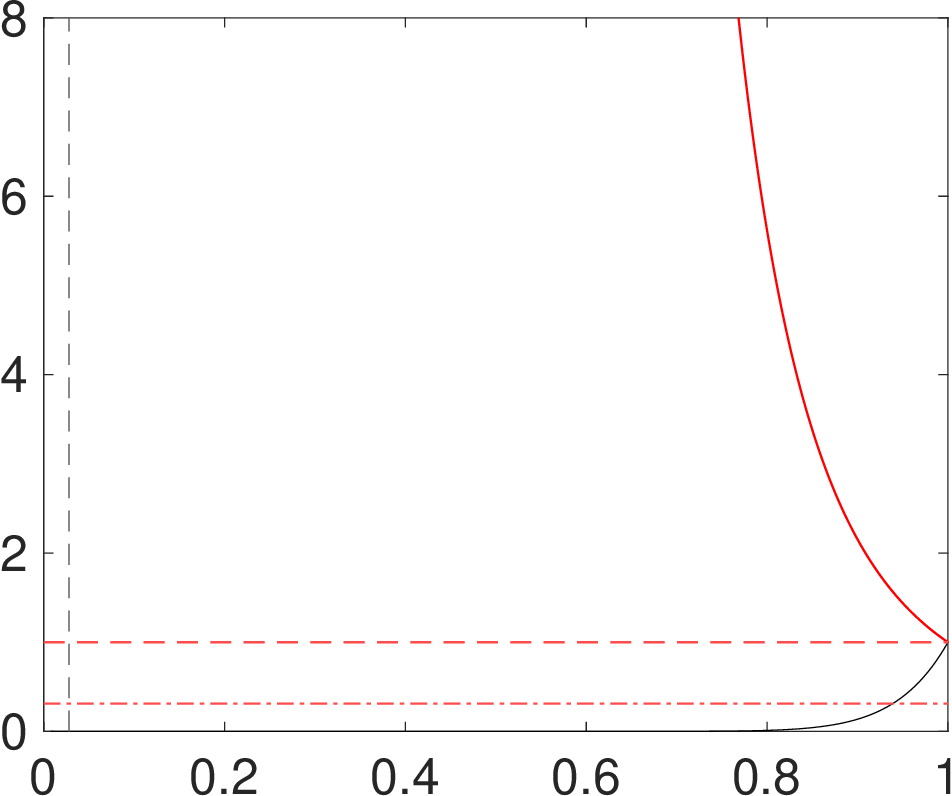}&
    \includegraphics[scale=0.3]{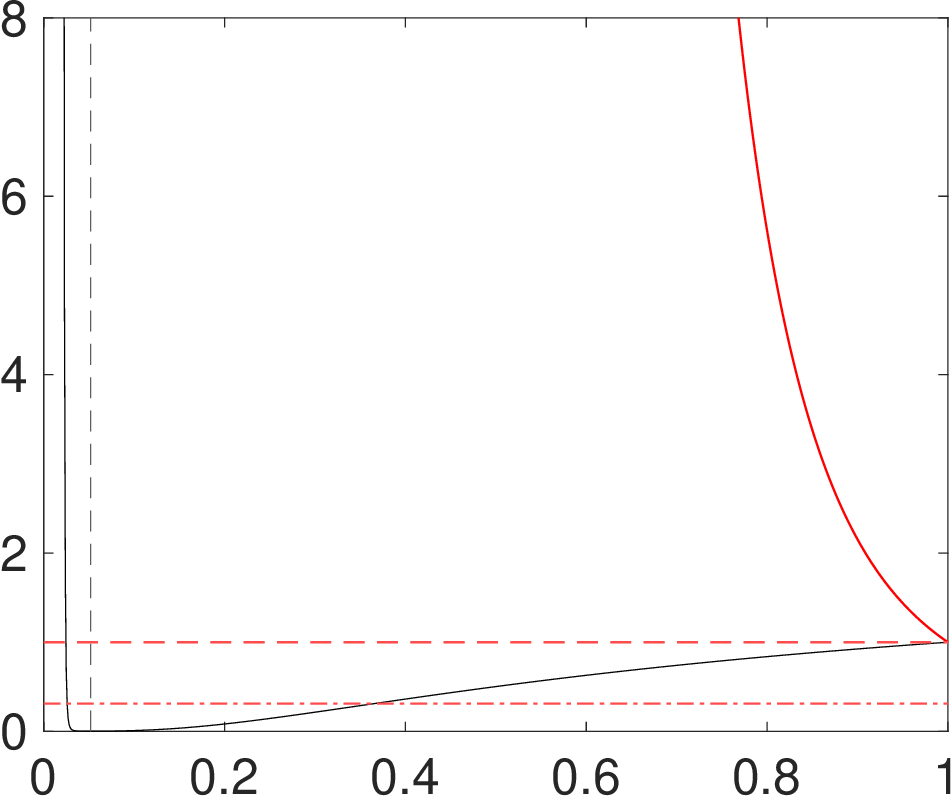}\\
    \end{tabular}
    \caption{Realized volatility network dataset. Value of the BF $H_t(\alpha_t)$ (solid black), the upper bound $\kappa_t(\alpha_t)$ (solid red) at different dates. In all plots, the reference lines at 1 ({\color{red}\protect\tikz[baseline]{\protect\draw[line width=0.2mm, dashed] (0,.6ex)--++(0.5,0) ;}}) and $10^{-1/2}$ ({\color{red}\protect\tikz[baseline]{\protect\draw[line width=0.2mm, dashdotted] (0,.6ex)--++(0.5,0) ;}}).}
    \label{emp1:figBFdatesBISvolatility}
\end{figure}

\pagebreak
\renewcommand{\thesection}{E}
\renewcommand{\theequation}{E.\arabic{equation}}
\renewcommand{\thefigure}{E.\arabic{figure}}
\renewcommand{\thetable}{D.\arabic{table}}
\renewcommand{\theproposition}{E.\arabic{proposition}}
\setcounter{table}{0}
\setcounter{figure}{0}
\setcounter{equation}{0}
\setcounter{proposition}{0}

\section{R Package BAYSFWATCH}
\label{sec:BAYSFWATCH}
A GitHub repository hosts the R package specifically developed for this work, providing tools to perform outlier detection. The package is available at:
\\
\url{https://github.com/BayesianEcon/BAYSFWATCH}.

\bibliographystyle{apalike}
\bibliography{biblio.bib}

\end{document}